%% file: main.tex
\pdfoutput=1
\documentclass[10pt]{article}
\usepackage[margin=1in]{geometry}
\usepackage{amsthm,amsfonts,amssymb,amsmath,comment,slashed,mathtools}
\usepackage[english]{babel}
\usepackage[T1]{fontenc}  
\usepackage[utf8]{inputenc}
\usepackage{bm}
\usepackage[shortlabels]{enumitem}
\usepackage[affil-it]{authblk}  
\usepackage{csquotes}
\usepackage{array}
\usepackage[dvipsnames]{xcolor}
\usepackage{graphicx}  
\usepackage{transparent}
\usepackage{subcaption}

\newcommand{\tv}{\tilde{v}}

\usepackage{epstopdf} 
\usepackage{xargs}          
\usepackage{float}

\usepackage{hyperref}
\hypersetup{colorlinks,
	linkcolor={green!30!black},
	citecolor={blue!40!black},
	urlcolor={black}}
\usepackage{bbold}
\usepackage[section]{placeins}

\usepackage[backend=biber,giveninits=true,doi=false,isbn=false,url=false,sorting=nyt,maxbibnames=99,date=year]{biblatex} 
\renewbibmacro{in:}{}
\bibliography{references.bib}
\newcommand{\ls}{\lesssim}
\newcommand{\omer}{\omega_{\textup{res}}}

\newcommand{\sbr}{\sigma_{br}}
\newcommand{\Rs}{\mathcal{R}}
\newcommand{\Sl}{\mathcal{SL}}

\newcommand{\HH}{\mathcal{H}^+}
\newcommand{\NN}{\mathcal{N}}
\newcommand{\EB}{\mathcal{EB}}
\newcommand{\LB}{\mathcal{LB}}

\newcommand{\phiH}{\phi_{\mathcal{H}^+}}
\newcommand{\phiHL}{(\phi'_{\mathcal{L}})_{|\mathcal{H}^{+}}}

\newcommand{\OO}{\mathcal{O}}
\newcommand{\NO}{\mathcal{NO}}
\newcommand{\OOp}{\mathcal{O}'}
\newcommand{\OOpp}{\mathcal{O}''}

\newcommand{\CH}{\mathcal{CH}_{i^+}}

\newcommand{\OmegaRN}{\Omega_{RN}}

\usepackage[capitalize]{cleveref}

\theoremstyle{plain}

\newtheorem{defn}{Definition} 
\newtheorem{prop}[defn]{Proposition}
\newtheorem{lem}[defn]{Lemma}
\newtheorem{thm}{Theorem}

\newtheorem*{thm*}{Theorem} 
\newtheorem{conj}{Conjecture}
\newtheorem{cor}[defn]{Corollary}
\newtheorem{rmk}[defn]{Remark}
 
\newtheorem*{exm*}{Example}
\newtheorem*{nonexm}{Non-example} 
\renewcommand{\d}{\mathrm{d}}

\newtheorem{theoa}{Theorem}
\newtheorem{theodeux}{Theorem}
\newtheorem{theob}{Theorem}
\newtheorem{theoc}{Theorem}
\newtheorem{theod}{Theorem}
\newtheorem{theol}{Theorem}

\newcommand{\Hp}{\mathcal{H}}

\newcommand{\Ch}{\mathcal{CH}}

\newcommand{\uhr}{u_{\Hp_R}}
\newcommand{\uhl}{u_{\Hp_L}}
\newcommand{\uchr}{u_{\Ch_R}}
\newcommand{\uchl}{u_{\Ch_L}}

\numberwithin{equation}{section}
\numberwithin{defn}{section}

\newcommand{\RR}{\mathbb{R}}

\newcommand{\phil}{\phi_{\mathcal{L}}'}
\newcommand{\phiNl}{\phi_{\mathcal{L}}}

\newcommand{\psil}{\psi_{\mathcal{L}}}
\newcommand{\dphi}{\delta \phi}

\title{Strong Cosmic Censorship in the presence of matter: the~decisive effect~of~horizon~oscillations~on~the~black~hole~interior~geometry}

\author[1]{Christoph~Kehle\thanks{christoph.kehle@eth-its.ethz.ch}}
\author[2]{Maxime~Van~de~Moortel\thanks{mmoortel@princeton.edu}}
\affil[1]{\small  Institute for Theoretical Studies, ETH Zürich, Clausiusstrasse~47,~8092~Zürich,~Switzerland \vskip.1pc \  }
\affil[2]{\small  Department of Mathematics, Princeton University, 
	Washington~Road,~Princeton~NJ~08544,~United~States~of~America \vskip.1pc \  }

\date{December 21, 2022}
\begin{document}
	\maketitle
	\thispagestyle{empty}
	
	\begin{abstract}
		\noindent Motivated by the Strong Cosmic Censorship Conjecture in the presence of matter, we study the Einstein equations coupled with a charged/massive scalar field with spherically symmetric characteristic data relaxing to a Reissner--Nordström event horizon.
		Contrary to the vacuum case, the relaxation rate is conjectured to be \emph{slow} (non-integrable), opening the possibility that the matter fields  and the metric coefficients \emph{blow up in amplitude} at the Cauchy horizon, not just in energy. 
		We show that whether this blow-up in amplitude occurs or not depends on a novel \emph{oscillation condition} on the event horizon which determines whether or not a resonance is excited dynamically: 
		\begin{itemize}
			\item If the oscillation condition is satisfied, then the resonance is not excited and we show  boundedness and continuous extendibility of the  matter fields and the metric across the  Cauchy horizon.
			\item   If the oscillation condition is violated, then by the \emph{combined effect of slow decay and the resonance being excited}, we show that the massive uncharged scalar field blows up in amplitude. 
			
			In our companion paper  \cite{MoiChristoph2}, we show that in that case a novel \textit{null contraction singularity} forms at the Cauchy horizon, across which the metric is not continuously extendible in the usual sense. 
		\end{itemize}
		Heuristic arguments in the physics literature  indicate that the oscillation condition should be satisfied generically on the event horizon. If these heuristics are true, then \emph{our result falsifies the $C^0$-formulation of Strong Cosmic Censorship by means of oscillation}. 
	\end{abstract}
	
	\newpage
	{\hypersetup{hidelinks}
		\tableofcontents
		\thispagestyle{empty}
	}
	\newpage
	
	\section{Introduction}
	
	Is General Relativity a deterministic theory? This fundamental question can only be addressed in the context of the initial value problem for the Einstein equations (see already \eqref{E1}) which govern the dynamics of spacetime in General Relativity. 
	Well-posedness for the initial value problem was   established in 1969 by Choquet-Bruhat and Geroch \cite{MGHD,GHD} proving that any suitably regular Cauchy data admit a unique maximal future development, the so-called \emph{Maximal Globally Hyperbolic Development} (MGHD). With this dynamical formulation at hand,  General Relativity can be considered deterministic  if the MGHD of \emph{generic} 
	Cauchy data for the Einstein equations is inextendible. The genericity stipulation is clearly necessary because the MGHD of Kerr \cite{MR0156674} Cauchy data (rotating black holes) and of Reissner--Nordstr\"{o}m \cite{reissner1916eigengravitation,nordstrom1918energy} Cauchy data (their charged analogs) admit a future boundary, the Cauchy horizon, across which the metric is smoothly  extendible.
	Heuristics of Penrose \cite{Penroseblue} however suggest the instability of the Kerr/Reissner--Nordstr\"{o}m Cauchy horizons and these led him to his famous \emph{Strong Cosmic Censorship Conjecture} \cite{penrose1974gravitational} supporting the idea of determinism in General Relativity. 
	The most definitive and perhaps most desirable formulation of Penrose's Strong Cosmic Censorship is the conjecture that the metric coefficients cannot be extended  as continuous functions, namely:
	\begin{conj}[$C^0$-formulation of  Strong Cosmic Censorship] \label{C0SCC}
		The MGHD of generic  asymptotically flat Cauchy data is inextendible as a continuous Lorentzian metric (we say the metric is $C^0$-inextendible).
	\end{conj}
	\noindent 
	\cref{C0SCC} is related to the expectation that physical observers approaching the boundary of the MGHD of generic Cauchy data are destroyed. If \cref{C0SCC} is false, then one may still be able to prove a weaker version of inextendibility, but this would correspond to a weaker version of determinism.

	\paragraph{\texorpdfstring{\cref{C0SCC}}{Conjecture 1} is false in the absence of matter.}
	In the celebrated work \cite{KerrStab}, Dafermos--Luk proved 
	that, in vacuum, small perturbations of Kerr still admit a Cauchy horizon across which the spacetime is $C^0$-extendible---thus falsifying \cref{C0SCC} in the absence of matter. The key ingredient to their proof is an \emph{integrable} inverse polynomial rate assumption for the decay of perturbations along the event horizon. Note however, that a weaker $H^1$-formulation is still expected to hold  \cite{ChristoSCC,KerrStab,Moi4}. If true, this would restore determinism at least in a weaker sense.
	
	\paragraph{Can \texorpdfstring{\cref{C0SCC}}{Conjecture 1} be salvaged  in the presence of matter?} 	In the present paper, we consider a	non-vacuum model: the Einstein--Maxwell--Klein--Gordon \eqref{E1}--\eqref{E5} system in spherical symmetry governing the dynamics of gravitation coupled to a charged/massive scalar field.
	Arguments in the physics literature \cite{HodPiran1,KoyamaTomimatsu,KonoplyaZhidenko,BurkoKhanna,OrenPiran} suggest that perturbations of the exterior of Reissner--Nordström in this model settle down merely at a slow, \emph{non-integrable} rate  (at least for massive and/or strongly charged perturbations), which is in stark contrast to the perturbations of Kerr in the vacuum case.  
	As such, the methods of  \cite{KerrStab} manifestly do not apply and the slow decay of perturbations may even raise hopes that for generic Cauchy data the metric is $C^0$-inextendible and thus, \cref{C0SCC} would be true after all for this  matter model.
	
	\paragraph{The question of $C^0$-extendibility across a  future null boundary $\CH$.} At first, it may appear that the slow decay in the above matter model in fact opens the possibility of a more drastic scenario where the singularity is everywhere spacelike inside the black hole. Notwithstanding, it was proven in \cite{Moi} that for this model, black holes  are bound to the future by a null boundary $\CH\neq \emptyset$ as depicted in  \cref{fig:cauchyhorizonexists}.  We will continue using the term ``Cauchy horizon'' for  $\CH$  by analogy with the Cauchy horizon of Reissner--Nordström, although the spacetime may or may not be $C^0$-extendible across the null boundary  $\CH$. Therefore, although the future boundary is null and in particular not spacelike, the question of $C^0$-extendibility of the spacetime across $\CH$, i.e.\ \cref{C0SCC}, remains  open. This is the question that we shall now address. 
	
	\begin{figure}
		\begin{center}
			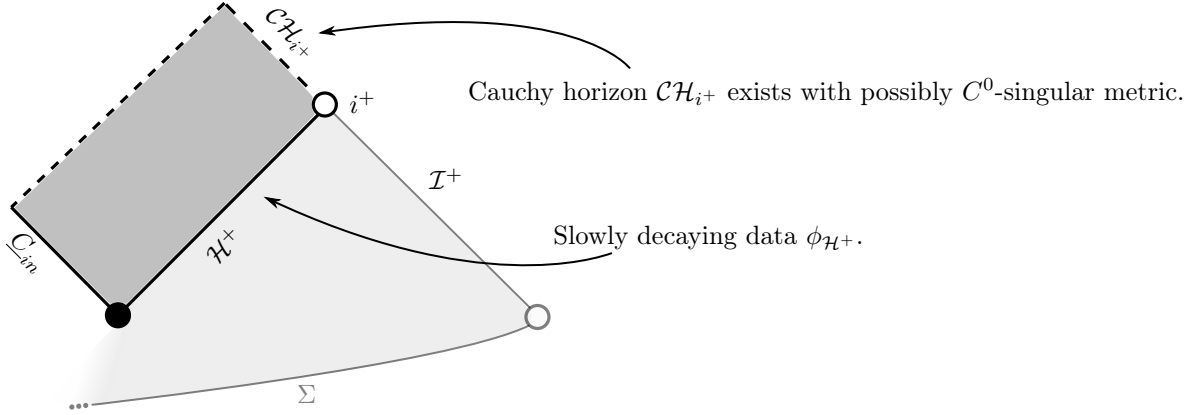	 
		\end{center}
		\caption{A Cauchy horizon $\CH$ exists for slowly decaying perturbations $\phiH$ as proven in \cite{Moi}, see already  \cref{Moi.intro}.}
		\label{fig:cauchyhorizonexists}
	\end{figure}

	\paragraph{Summary of our results.} 
	As we will show, the question of \cref{C0SCC} becomes unexpectedly subtle: In addition to the decay rates of perturbations 
	on the exterior, it turns out that the validity of \cref{C0SCC}  depends crucially on Fourier support properties of late time perturbations due to a scattering resonance associated to the Cauchy horizon $\CH$.  
	In our main   \cref{main.theorem.intro}   we identify an \emph{oscillation condition} on perturbations along the event horizon $\HH$:
	\emph{If the oscillation condition is satisfied by the perturbation, we show boundedness and continuous extendibility of the matter fields and the metric across the Cauchy horizon $\CH$} despite the obstruction created by slow decay.  
	On the other hand, in  \cref{main.theorem.intro2} we show that \emph{if the oscillation condition is violated on the event horizon $\HH$, the resonance is excited and  the uncharged scalar field blows up in amplitude} namely $|\phi|\to +\infty$  at the Cauchy horizon $\CH$. 
	
	Heuristic and numerical arguments in the physics literature \cite{HodPiran1,KoyamaTomimatsu,KonoplyaZhidenko,BurkoKhanna,OrenPiran} suggest that the oscillation condition is indeed satisfied on $\HH$ for generic perturbations of the black hole exterior. Assuming this, our result \cref{corollary.conj} \emph{falsifies the $C^0$-formulation of Strong Cosmic Censorship by means of oscillation}.

	In \cref{W11.thm.intro}, we show that for both oscillating and non-oscillating perturbations\footnote{up to a genericity condition in the charged scalar field case, which we can get rid of in the uncharged case, see Theorem~\ref{W11.thm.intro}.}, \emph{the  scalar field  blows up  in the $W^{1,1}_{loc}$-norm at the Cauchy horizon $\CH$}, i.e.\ $\int |D_v \phi | \d v= + \infty$ schematically.  This  $W^{1,1}$ blow-up  is in contrast to the vacuum case where the analogous statement is false \cite{KerrStab}. 
	This shows that for both oscillating and non-oscillating perturbations, the Cauchy horizon $\CH$ is more singular in the presence of matter than  in  vacuum. Moreover, the blow-up of the scalar field in $W^{1,1}$ indicates that our result cannot be captured using only physical space techniques which have been used previously.

	Finally,  in our companion paper \cite{MoiChristoph2} we will prove \cref{null.contraction.theorem} which shows that blow-up in amplitude of the scalar field indeed gives rise to a  $C^0$-inextendibility statement on the metric within a spherically symmetric class. \cref{null.contraction.theorem}, in conjunction with  \cref{main.theorem.intro2}, provides the first example of a dynamically formed singularity leading to a $C^0$-inextendibility statement of the metric across a null spacetime boundary (albeit within a restricted spherically symmetric class).	Whether this statement can be upgraded to the full $C^0$-inextendibility  of the spacetime remains open.\footnote{Unrestricted $C^0$-inextendibility results (even for spacelike singularities) are known to be notoriously difficult to show, see e.g.\ \cite{JanC0} for the proof of $C^0$-inextendibility of the Schwarzschild solution across the spacelike singularity $\{r=0\}$.}

	\paragraph{Similarities with the $\Lambda <0$ case.}
	In the asymptotically AdS case ($\Lambda < 0$),   solutions to the linear wave equation on AdS black holes also decay  at  a slow, non-integrable rate \cite{adslowerbound}. It turns out that in this context, oscillations also play a crucial role \cite{Kehle2019,Kehle2020}  to address of the question of the validity of  the linear analog of \cref{C0SCC}. The slow inverse logarithmic decay in the $\Lambda<0$ case however arises from the superposition of infinitely many high $\ell$ angular modes. This is different from the present problem for $\Lambda=0$, where the slow decay is inverse-polynomial (see already \cref{sec:intro1}) and already occurs in spherical symmetry.

	\paragraph{Outline of the Introduction.}
	In \cref{sec:intro1}  we introduce the Einstein--Maxwell--Klein--Gordon system and give a more detailed overview of our new results addressing the issue of Strong Cosmic Censorship within this matter model in spherical symmetry. Further, we present a  first version  of our main theorems. 
	In \cref{CH.intro} we outline the important differences  between the EMKG model and other models regarding the existence of a Cauchy horizon and the continuous extendibility of the metric. In \cref{Weaknullsection} we mention previous results on the dynamical formation of weak null singularities at the Cauchy horizon, which we compare to the new singularities that dynamically form in our setting. In \cref{lin.scat.sec} we present previous results on scattering inside Reissner--Nordstr\"{o}m black holes which are important for our proof. In \cref{sec:SCC.non.zero} we elaborate on the interior of black holes with $\Lambda<0$, in which oscillations turn out to  play an important role as well.
	In  \cref{Strategy.section}  we briefly discuss the strategy of the proof. 
	\subsection{Main results: First versions}\label{sec:intro1}
	\subsubsection{The EMKG system and existence of a Cauchy horizon for slowly decaying scalar fields} \label{relax}
	\paragraph{The EMKG model in spherical symmetry.} We  study the Einstein equations coupled to a charged massive scalar field: the Einstein--Maxwell--Klein--Gordon (EMKG) model in spherical symmetry
	\begin{align}
	\label{E1} & Ric_{\mu \nu}(g)- \frac{1}{2}R(g)g_{\mu \nu}= \mathbb{T}^{EM}_{\mu \nu}+  \mathbb{T}^{KG}_{\mu \nu}, \\
	\label{E2} &\mathbb{T}^{EM}_{\mu \nu}=2\left(g^{\alpha \beta}F _{\alpha \nu}F_{\beta \mu }-\frac{1}{4}F^{\alpha \beta}F_{\alpha \beta}g_{\mu \nu}\right),\\
	\label{E3} & \mathbb{T}^{KG}_{\mu \nu}= 2\left( \Re(D _{\mu}\phi \overline{D _{\nu}\phi}) -\frac{1}{2}(g^{\alpha \beta} D _{\alpha}\phi \overline{D _{\beta}\phi} + m ^{2}|\phi|^2  )g_{\mu \nu} \right),\\ \label{E4} & \nabla^{\mu} F_{\mu \nu}= \frac{ q_{0} }{2}i (\phi \overline{D_{\nu}\phi} -\overline{\phi} D_{\nu}\phi) , \; F=\d A ,\\
	\label{E5} & g^{\mu \nu} D_{\mu} D_{\nu}\phi = m ^{2} \phi ,  \; D_{\mu}= \nabla_{\mu}+i q_0 A_{\mu}	\end{align} for a quintuplet $(\mathcal M,g,F,A,\phi)$,  where $(\mathcal M,g)$ is a  3+1-dimensional Lorentzian manifold, $\phi$ is a   complex-valued scalar field, $A$ is a   real-valued 1-form, and $F$ is a   real-valued 2-form.	Here $q_0 \in \RR$ and   $m\geq 0$  are fixed constants representing respectively the charge and the mass of the scalar field.
	The EMKG model describes self-gravitating matter and provides a setting for studying spherical gravitational collapse of charged and massive matter if $q_0 \neq 0$ and $m^2\neq 0$ (see the discussion in \cref{weak.grav.col}). This model has attracted much attention in the literature \cite{An.charged.SF,Boson,scc1lambda>0,Kommemi,gajicluk,Moi,Moi4, Moi2,MoiThesis}, see also    \cite{YangYu,LindbladKG,KlainermanMachedon,Krieger2,TataruOh,RodTao} for work on the flat Minkowski background.
	
	\paragraph{Setting of the problem.} Consider the Maximal Globally Hyperbolic Development of suitably regular spherically symmetric Cauchy data prescribed on an asymptotically flat initial hypersurface $\Sigma$ as depicted in \cref{fig:cauchyhorizonexists}. General results for the EMKG model in spherical symmetry \cite{Kommemi} allow to define null infinity $\mathcal{I}^+$---a conformal boundary where idealized far away observers live, and the black hole interior region as the complement of the causal past of $\mathcal I^+$. If the black hole interior is non-empty, we also define the event horizon $\HH$ as the past boundary of the black hole interior which separates the black hole interior from the black hole exterior. 
	
	In the current paper we will only be interested in the dynamics of the black hole interior. In particular, instead of studying the Cauchy problem with  data on $\Sigma$, we will prescribe the scalar field $\phi$ and the metric on an ingoing cone $\underline{C}_{in}$ and on an outgoing cone $\HH$ emulating the event horizon of an already-formed black hole. This setting corresponds to a characteristic initial value problem with data imposed on $\HH\cup \underline{C}_{in}$, see \cref{fig:cauchyhorizonexists}.  Our study of this characteristic initial value problem will be entirely self-contained. We will however continue to depict $\Sigma$ on \cref{fig:cauchyhorizonexists} and subsequent figures for completeness. 
	Our assumptions on the characteristic initial data on $\HH\cup \underline{C}_{in}$ will be made in accordance with the conjectured late-time tails on the event horizon $\HH$ arising from generic Cauchy data   on   asymptotically flat $\Sigma$, see the discussion below.

	\paragraph{Conjectured late-time asymptotics on the event horizon $\HH$ and contrast with the vacuum case.}  Heuristic arguments regarding the black hole exterior in the physics literature (see \cite{HodPiran1,KoyamaTomimatsu,KonoplyaZhidenko,BurkoKhanna,OrenPiran}) indicate  that (spherically symmetric)   dynamical black holes arising from Cauchy data on $\Sigma$ for the EMKG model relax to Reissner--Nordstr\"{o}m along the event horizon $\HH$ at a slow\footnote{ Precisely, these slow rates hold conjecturally for a massive ($m^2\neq 0$) scalar field and/or  strongly charged ($|q_0 e|\geq \frac{1}{2}$) one.}, \emph{non-integrable} rate  $v^{-s}$, $s \in (\frac{1}{2},1]$ for large $v$, in a standard Eddington--Finkelstein coordinate $v$.  This is in contrast to the faster and integrable rate $s>1$  proved in the uncharged massless case $m^2=q_0=0$ \cite{PriceLaw}, or assumed in vacuum by Dafermos--Luk  \cite{KerrStab}  (see already \eqref{Kerr.decay}). This fast, integrable rate $v^{-s},\ s >1$ in vacuum is indeed sufficient to prove the existence of a Cauchy horizon $\CH$, across which the spacetime is continuously extendible: this led to a \emph{falsification of  \cref{C0SCC} in vacuum without symmetry assumptions}   \cite{KerrStab} (or for spherically symmetric models as in \cite{MihalisPHD,Moi}), see already \cref{downfall.section}.
	
	\paragraph{Existence of a Cauchy horizon $\CH$ for slowly decaying scalar fields.}
	Returning to the EMKG model, the first step in addressing \cref{C0SCC}  is  to understand whether for slowly decaying characteristic data on the event horizon $\HH$, the future boundary inside the black hole is null (a Cauchy horizon) or spacelike. In view of the slow decay on the event horizon $\HH$, the spacelike singularity scenario is plausible and indeed desirable (if it was true, then \cref{C0SCC} would likely be valid). Despite the obstruction created by the slow decay of event horizon perturbations, it turns out however that the black hole future boundary has a non-empty null component $\CH\neq \emptyset$ emanating from $i^+$, see \cref{fig:cauchyhorizonexists}, and is not everywhere spacelike as one might have hoped:
	\begin{thm}[M.VdM \cite{Moi}]\label{Moi.intro} \textup{[Rough version; precise version recalled in {\cref{existence.CH.main}}]}
		Consider spherically symmetric characteristic initial data  for \eqref{E1}--\eqref{E5} 
		on the event horizon $\HH$  (and on an ingoing cone).  Assume the following \textbf{slow decay} upper bound on  the scalar field   $\phiH$ 
		on the event horizon $\HH=[v_0,+\infty)$ as
		\begin{align} 
		|\phiH(v)|\leq C_0 v^{-s},\ |D_v \phiH|\leq C_0 v^{-s} 
		\end{align} for all $v\geq v_0$ in a standard Eddington--Finkelstein type $v$-coordinate on $\HH=[v_0,+\infty)$, for some $C_0>0$ and some  decay rate $s>\frac{1}{2}$.
		
		Then the spacetime is bound to the future by an ingoing null boundary $\CH \neq \emptyset$  (the Cauchy horizon) foliated by spheres of  positive radius and emanating from $i^+$, and the Penrose diagram is given by the dark gray region in \cref{fig:cauchyhorizonexists}.
	\end{thm}
	Since  by \cref{Moi.intro} the black hole future boundary is not everywhere spacelike and has a null component  $\CH \neq \emptyset$, one may at first expect continuous extendibility across $\CH$. It turns out however that the spacetime of \cref{Moi.intro} may or may not be continuously extendible across $\CH$. This is perhaps unexpected, since all previous instances of black hole spacetimes with a null future boundary component are  at least continuously extendible across that component \cite{KerrStab,MihalisPHD,JonathanStab}.
	Thus, \cref{Moi.intro} is not sufficient to fully address \cref{C0SCC} and the question of continuous extendibility across the null boundary $\CH$ has remained open.
	
	The slow rate $s> \frac{1}{2}$ assumed in \cref{Moi.intro} is indeed  \emph{too slow} to prove   the $C^0$-extendibility of spacetime across the Cauchy horizon $\CH$ using the same method as Dafermos--Luk \cite{KerrStab} in vacuum. The method of \cite{KerrStab}  requires the faster integrable decay assumption $s>1$ and does not extend to the  non-integrable case $s\leq 1$, a failure that may even raise the attractive possibility that \cref{C0SCC} is true after all for the EMKG matter model. \textbf{\emph{This could mean that determinism is in better shape in the presence of matter!}}

	\subsubsection{Theorem I: event horizon oscillations are decisive for the  $C^0$ extendibility of the metric}
	Our main result  however shows that the situation is more subtle than one may first think: assuming that the scalar field $\phi$ oscillates sufficiently on the event horizon $\HH$, we show in \cref{main.theorem.intro} that \emph{$\phi$ is uniformly bounded in the black hole interior and the metric is continuously extendible}. The event horizon oscillation assumption  is sharp in the following sense: conversely assuming that the scalar field $\phi$  does not oscillate sufficiently on the event horizon $\HH$, we show in \cref{main.theorem.intro2} that  \emph{$\phi$ blows up in amplitude at the Cauchy horizon $\CH$}. It turns out that the oscillation condition on the event horizon $\HH$, i.e.\ the main assumption  of \cref{main.theorem.intro}, is conjecturally satisfied for generic Cauchy data on an asymptotically flat $\Sigma$, and thus, the \emph{hope that determinism is in better shape in the presence of matter in the end does not come true!} (See already \cref{TheoremB.intro.section}.)
	
	\cref{main.theorem.intro} and \cref{main.theorem.intro2} show that uniform boundedness or blow-up of the matter fields unexpectedly relies on fine properties of the scalar field $\phi$ on the event horizon $\HH$ in both physical and Fourier space. At the heart of our novel oscillation condition lies the following resonant frequency
	\begin{align}  
	\label{omer}
	\omer(M,e,q_0) := \omega_- (M,e,q_0) - \omega_+(M,e,q_0) , 
	\end{align}
	where $\omega_- = \omega_-  (M,e,q_0) := \frac{ q_0 e }{r_-(M,e)}$, 
	$\omega_+=\omega_+  (M,e,q_0) := \frac{ q_0 e }{r_+(M,e)}$ for asymptotic black hole parameters  $0< |e|< M$. 
	
	In what follows we will give rough versions of \cref{main.theorem.intro} and \cref{main.theorem.intro2}. For the precise versions we refer the reader to \cref{sec:preciseA1} and \cref{sec:preciseA2}.

	\begin{theoa}[\textbf{Boundedness}]\textup{[Rough version; precise version in {\cref{sec:preciseA1}}]} \label{main.theorem.intro}
		Consider spherically symmetric characteristic initial data  for \eqref{E1}--\eqref{E5} 
		on the event horizon $\HH$ (and on an ingoing cone).   Assume the following \textbf{slow decay} upper bound on the scalar field   $\phiH$ 
		on the event horizon on the event horizon $\HH=[v_0,+\infty)$ as
		\begin{align} \label{decay.s}
		|\phiH(v)|\leq C v^{-s},\ |D_v \phiH|\leq C v^{-s} 
		\end{align}
		for all $v\geq v_0$ in a standard Eddington--Finkelstein type $v$-coordinate on $\HH=[v_0,+\infty)$, for $v_0>1$ sufficiently large and for some $C>0$ and some (non-integrable) decay rate
		\begin{align} \label{s.def}
		\frac 34 < s \leq 1.
		\end{align}
		
		By \cref{Moi.intro}, the spacetime, i.e. the dark gray region in \cref{fig:cauchyhorizonexists},  is bound to the future by a null boundary $\CH \neq \emptyset$ (the Cauchy horizon). Then, in the gauge $A_v=0$, the following holds true.
		\begin{itemize}
			\item \label{phibounded}If $\phiH$ satisfies the \textbf{qualitative oscillation condition} 	on $\HH=[v_0,+\infty)$, i.e.\ if for all $O(v^{1-2s})$ functions 
			\begin{align} \label{osc.cond.qual}
			\limsup_{\tilde{v}\to + \infty} \left| \int_{v_0}^{\tilde{v}} \phiH(v) e^{i \omer v ( 1 + O(v^{1-2s})) } \d v \right| < +\infty;
			\end{align} Then, the scalar field $\phi$  is  uniformly bounded in amplitude up to  and including  the Cauchy horizon $\CH$. 
			\item If $\phiH$ satisfies the \textbf{strong qualitative oscillation condition} 
			on $\HH=[v_0,+\infty)$, i.e. if for all $O(v^{1-2s})$ functions 
			\begin{align} \label{osc.cond.Squal}
			\lim_{\tilde{v}\to + \infty} \left| \int_{v_0}^{\tilde{v}} \phiH(v) e^{i \omer v (1+O(v^{1-2s})) } \d v \right| \text{ exists and is finite;}\end{align}Then,  additionally the metric $g$ and the scalar field $\phi$ are continuously extendible across the Cauchy horizon $\CH$.
			\item If $\phiH$ satisfies the \textbf{quantitative oscillation condition} on $\HH=[v_0,+\infty)$, i.e.\ if there exist $E>0$, $\epsilon > 1-s$ such that for all $O(v^{1-2s})$ functions 
			\begin{align} \label{osc.cond.quant}
			\lim_{\tilde{v}\to + \infty} \left| \int_{v_1}^{\tilde{v}} \phiH(v) e^{i \omer v(1  +O(v^{1-2s})) } \d v \right| \leq E  v_1^{-\epsilon} \text{ for all $v_1\geq v_0$;}
			\end{align}
			Then, additionally the Maxwell field contraction $F_{\mu \nu} F^{\mu \nu}$ is uniformly bounded in amplitude and continuously extendible across the Cauchy horizon $\CH$.
		\end{itemize}
	\end{theoa}
	We refer to  \cref{fig:theorema}   for an illustration of \cref{main.theorem.intro}.
	\begin{figure}[H]
		\centering
		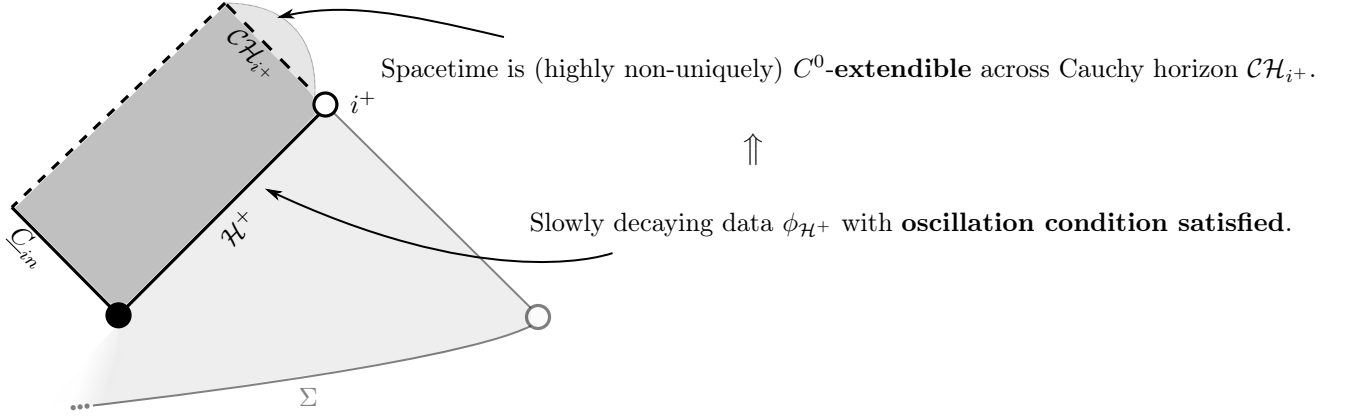
		\caption{\cref{main.theorem.intro}: If the strong qualitative oscillation condition is satisfied, then the spacetime is $C^0$-extendible across the Cauchy horizon $\CH$.\vspace{.6cm}}
		\label{fig:theorema}
	\end{figure}

	In the uncharged case $q_0 = 0$, where $\omer=0$, we show that the qualitative oscillation condition \eqref{osc.cond.qual} is sharp to obtain boundedness.

	\begin{theodeux}[\textbf{Blow-up}]\textup{[Rough version; precise version in {\cref{sec:preciseA2}}]} \label{main.theorem.intro2}
		Consider spherically symmetric characteristic initial data  for \eqref{E1}--\eqref{E5} 
		on the event horizon $\HH$ (and on an ingoing cone).   Assume the following \textbf{slow decay} upper bound on  the scalar field   on the event horizon $\HH$ (i.e\ $\phiH  $ satisfies \eqref{decay.s} where $s$ satisfies \eqref{s.def}).	Assume  additionally $q_0=0$ and let $m^2 > 0 $ be generic.  
		
		Then, $\phi$ blows up in amplitude at every point on the Cauchy horizon  $\CH$
		\begin{align}\limsup_{(u,v) \to \CH} |\phi(u,v)|= +\infty \label{eq:blowupphi} \end{align} if and only if
		\begin{align}
		\limsup_{\tilde{v}\to + \infty} \left| \int_{v_0}^{\tilde{v}} \phiH(v)  \d v \right| = +\infty,
		\end{align}
		i.e.\ if and only if  $\phiH$ violates the  qualitative oscillation condition \eqref{osc.cond.qual}.

		Further, in the case where the scalar field $\phi$ blows up at the Cauchy horizon $\CH$ as in \eqref{eq:blowupphi}, a null contraction singularity forms at the Cauchy horizon $\CH$ as stated in \cref{null.contraction.theorem}   and proved in \cite{MoiChristoph2}. 
	\end{theodeux} We refer to \cref{fig:theorema'} for an illustration of \cref{main.theorem.intro2}.
	
	\begin{figure}[ht]
		\centering
		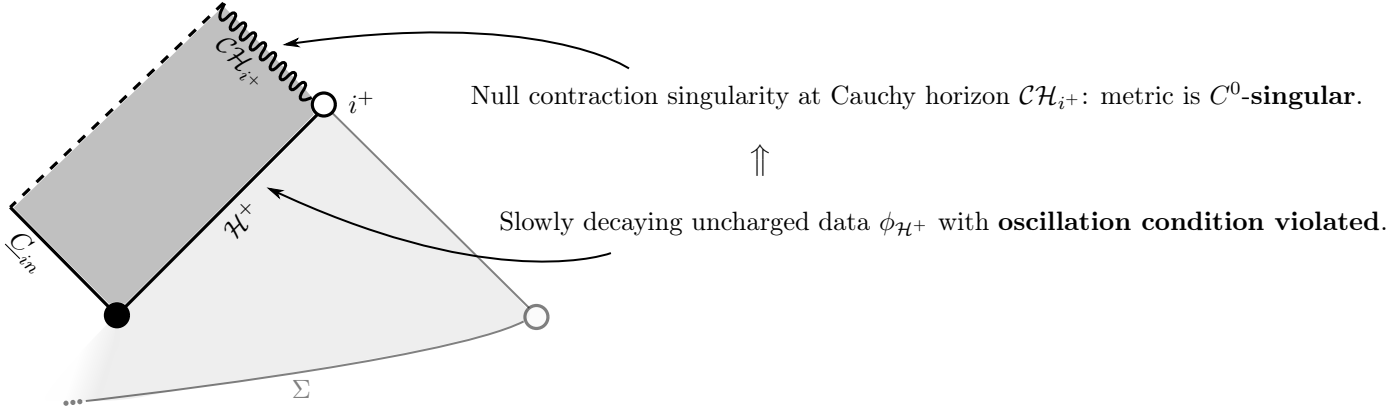	 
		\caption{\cref{main.theorem.intro2}: If the oscillation condition is violated in the uncharged case, then a novel null contraction singularity forms at the Cauchy horizon $\CH$ and the metric is $C^0$-singular at $\CH$.}   \label{fig:theorema'}
	\end{figure} 
	
	\cref{main.theorem.intro2} also shows that it is impossible to prove boundedness of the scalar field $\phi$ only under the assumptions of \cref{Moi.intro}. This motivates a posteriori the introduction of the oscillation conditions \eqref{osc.cond.qual}, \eqref{osc.cond.Squal}, \eqref{osc.cond.quant} which are thus necessary to obtain boundedness and $C^0$ extendibility as claimed in \cref{main.theorem.intro}. Anticipating \cref{null.contraction.section}, we note that it is also impossible to prove the continuous extendibility of the metric in the usual sense only under the assumptions of \cref{Moi.intro}, by \cref{null.contraction.theorem}.

	For concreteness, we will now give explicit examples of profiles $\phiH$ which satisfy (respectively violate) the oscillation condition \eqref{osc.cond.qual}, \eqref{osc.cond.Squal}, \eqref{osc.cond.quant} from above.

	\begin{exm*} 
		For any fixed $\omega  \neq \omer$   the profile $\phiH := e^{-i\omega v}   v^{-s}$  satisfies the quantitative oscillation condition \eqref{osc.cond.quant}.
	\end{exm*}
	
	\begin{nonexm} 
		The profile $\phiH:= e^{-i\omer v} v^{-s}$ violates the oscillation condition \eqref{osc.cond.qual}.
	\end{nonexm}

	\subsubsection{{\cref{corollary.conj}}: the $C^0$-formulation of Strong Cosmic Censorship is false}
	\label{TheoremB.intro.section}

	\paragraph{Slow decay on $\HH$ for generic Cauchy data on $\Sigma$.} We now return to  \cref{C0SCC}, which is formulated in terms of generic Cauchy data on an asymptotically flat $\Sigma$. First, the scalar field $\phi$ on the event horizon $\HH$ is indeed expected to decay slowly for generic Cauchy data on  $\Sigma$, i.e.\ $\phiH$ satisfies \eqref{decay.s} only for $s\leq 1$, at least for almost every parameters $(m^2,q_0)$, see already \cref{decay.conj}. This slow decay makes \cref{main.theorem.intro} and \cref{main.theorem.intro2} decisive to the study of Cauchy data on $\Sigma$ as above, since the validity of \cref{C0SCC} now  crucially depends on whether generic Cauchy data on $\Sigma$ give rise to solutions for  which the (slowly decaying) scalar field $\phi$ on the event horizon $\HH$  satisfies or violates the oscillation condition \eqref{osc.cond.qual} (or \eqref{osc.cond.Squal}, its stronger analogue).

	\paragraph{Oscillations on $\HH$ for generic Cauchy data on $\Sigma$.} As it turns out, $\phiH$ is expected to satisfy the  (even stronger) quantitative oscillation condition \eqref{osc.cond.quant} for generic regular Cauchy data on $\Sigma$. This expectation is based on works in the physics literature relying on heuristic analysis \cite{HodPiran1,KonoplyaZhidenko,KoyamaTomimatsu,KoyamaTomimatsu2} or numerics \cite{BurkoKhanna,OrenPiran} giving precise asymptotic tails on  the event horizon $\HH$. We formulate this as the following conjecture, where $\phiH$ is the scalar field $\phi$ restricted to the event horizon $\HH = [v_0,+\infty)$, $v$ is an Eddington--Finkelstein type coordinate (see the gauge choice later defined in \eqref{gauge1}), and electromagnetic gauge $A_v =0$ (see \eqref{GaugeAv}):
	\begin{conj} \label{decay.conj}
		Let $(\mathcal M,g,F,A,\phi)$ be a black hole solution of the system \eqref{E1}--\eqref{E5} arising from generic, spherically symmetric smooth Cauchy data on an asymptotically flat $\Sigma$. Then, the black hole exterior settles down to a Reissner--Nordström exterior with asymptotic mass $M$ and asymptotic charge $e$ satisfying $0<|e|<M$. Moreover, the scalar field   has the following late-time asymptotics on the event horizon $\HH=[v_0,+\infty)$: 
		\begin{enumerate}
			\item \label{conj.decay.1} In the massive uncharged case, i.e.\ $m^2 > 0$, $q_0=0$, \begin{equation} \label{conj.decay.1.eq}
			\phiH(v) = C(m\cdot M, D)  \sin( m v + \omega_{err}(v)) \cdot v^{-\frac{5}{6}} + \phi_{err},
			\end{equation} for fast decaying $\phi_{err}$ (i.e.\ $\phi_{err}$ satisfies \eqref{decay.s} for $s>1$), a constant $C(m\cdot  M,D) \neq 0$ depending on $m\cdot M$ and the initial data $D$, and a sublinear growing phase $\omega_{err}(v)= -\frac{3m}{2} (2\pi M)^{\frac{2}{3}}  v^{\frac{1}{3}} + \omega(m \cdot M) $.
			\item \label{conj.decay.2} In the massless charged case, i.e.\ $m^2 = 0$, $q_0 \neq 0$,
			\begin{equation}  \label{conj.decay.2.eq}
			\phiH(v) = C_H(q_0e,D) \cdot e^{ \frac{iq_0 e}{r_+} v} \cdot v^{-1-\delta}+ \phi_{err},
			\end{equation} where $C_H(q_0e,D)\neq 0$ is a constant depending on $q_0 e$ and the initial data $D$, $\delta(q_0 e):=\sqrt{1-4(q_0 e)^2}\in \mathbb{C}$, and  $\phi_{err}$ is fast decaying (i.e.\ $\phi_{err}$ satisfies \eqref{decay.s} for $s>1$).
			\item  \label{conj.decay.3} In the massive charged case, i.e.\ $m^2 > 0$, $q_0 \neq 0$,
			\begin{align} \label{conj.decay.3.eq}
			\phiH(v) =  C(M \cdot m,D) \cdot  e^{ \frac{iq_0 e}{r_+} v} \cdot \sin( m v + \omega_{err}(v)) \cdot v^{-\frac{5}{6}} + \phi_{err},
			\end{align} 
			where all the quantities are as above and generically, $|q_0 e| \neq r_- |m|$.
		\end{enumerate}
	\end{conj}
	\paragraph{Falsification of \cref{C0SCC} assuming \cref{decay.conj}.} We will show that the conjectured profiles in  \eqref{conj.decay.1.eq}, \eqref{conj.decay.2.eq} and \eqref{conj.decay.3.eq} indeed satisfy the quantitative oscillation \eqref{osc.cond.quant}. Thus, as a corollary  of our main result \cref{main.theorem.intro} we obtain a conditional, but otherwise definitive resolution of \cref{C0SCC}:
	
	\begin{theob}\textup{[Rough version; precise version in {\cref{sec:c0formulation}}]}\label{corollary.conj} 
		If $\phiH$ is as in \cref{decay.conj}, 
		then the metric $g$ and the scalar field  $\phi$ are continuously extendible across the Cauchy horizon $\CH$.
		
		In particular, if  \cref{decay.conj} is true, then \cref{C0SCC} is false for the Einstein--Maxwell--Klein--Gordon system in spherical symmetry. 
	\end{theob} 
	We refer to \cref{sec:c0formulation} for the precise statement of \cref{corollary.conj}.

	The conjectured decay rates for $\phiH$ in \cref{decay.conj} are non-integrable, i.e.\ $\phiH$ satisfies \eqref{decay.s} with $s$ in the range \eqref{s.def}, except for the massless charged case with $|q_0 e|<\frac 12$. We also recall that non-integrable decay of $\phiH$ is insufficient to prove continuous extendibility for $g$ and $\phi$ by means of decay and indeed even leads to the blow-up of $|\phi|$ as shown in \cref{main.theorem.intro2} in the case where the oscillation condition \eqref{osc.cond.qual} is violated. In that sense, under the assumption of \cref{decay.conj}, \cref{corollary.conj} shows that  ${C^0}$-Strong Cosmic Censorship for the EMKG model is false \textbf{only} by virtue of the oscillations of the scalar field $\phi$ on the event horizon $\HH$.
	
	\paragraph{Lack of oscillations for non-generic Cauchy data on $\Sigma$.} Having addressed the generic case in \cref{decay.conj}, there remains still the possibility that there exist  (non-generic) Cauchy data for which the scalar field $\phiH$ on the event horizon $\HH$ does not satisfy the (qualitative) oscillation condition \eqref{osc.cond.qual}.
	Indeed, on  the basis of certain scattering arguments \cite{scatteringAAG,scatteringDRSR,Hamed} we conjecture\footnote{We also note that \cref{scat.conj} is not specific to the EMKG system in spherical symmetry: similar conjectures  can be made for a rather general class of models, see for instance \cite{scatteringAAG,scatteringDRSR}.}
	
	\begin{conj} \label{scat.conj}
		For any suitable finite-energy profile $\phiH$ there exist sufficiently regular Cauchy data on $\Sigma$ for the EMKG system in spherical symmetry giving rise to a dynamical black hole for which the scalar field along the event horizon is given by $\phiH$.
	\end{conj}
	In particular, if \cref{scat.conj} is true, this means that there exist Cauchy data on $\Sigma$ for which the scalar field $\phiH$ on the event horizon  $\HH$ obeys  \eqref{decay.s} for $s >\frac{3}{4}$, but violates the oscillation condition \eqref{osc.cond.qual}, thus by \cref{main.theorem.intro2}, the scalar field $\phi$ blows up in amplitude at the Cauchy horizon $\CH$ (if $q_0=0$). 
	Such (non-generic) Cauchy data will be important in \cref{null.contraction.section} as they will constitute examples of null contraction singularities at $\CH$, see \cref{null.contraction.theorem}.
	Finding the precise regularity (c.f.\ \cite{dafermosSR,rough}) of such Cauchy data on $\Sigma$ is also part of the resolution of \cref{scat.conj}.

	\subsubsection{{\cref{W11.thm.intro}}: $W^{1,1}$-blow-up along outgoing cones---a complete contrast with the vacuum case}\label{W11.intro.section.}
	
	We remarked before that the falsification of the $C^0$-formulation of Strong Cosmic Censorship in vacuum \cite{KerrStab} by Dafermos--Luk---the vacuum analog of \cref{corollary.conj} outside spherical symmetry---crucially relies on integrable decay along   the event horizon $\HH$ for perturbations and their derivatives (see \eqref{Kerr.decay}). Indeed, in their work, Dafermos--Luk propagate this integrable decay towards $i^+$ with   suitable weighted energy  estimates into the black hole interior. This integrable decay for outgoing derivatives is then used to  show that  the metric is actually $W^{1,1}$-extendible along  outgoing null cones, i.e.\ with locally integrable Christoffel symbols. Note that this $W^{1,1}$-extendibility result of the metric is \emph{strictly stronger} than the $C^0$-extendibility which subsequently follows by integrating.  Mutatis mutandis, this robust physical space method of showing the stronger $W^{1,1}$-extendiblity result as an intermediate step has been applied in various previous contexts to show $C^0$-extendibility, e.g.\  \cite{MihalisPHD,Mihalis1,JonathanStab,KerrStab}, exploiting the null structure of the Einstein equations: In fact, this was the only known method to prove $C^0$-extendibility so far. For the  EMKG model, however, only in the case $m^2 =0,\ |q_0 e|<\frac 12$, do perturbations along the event horizon $\HH$ decay at an \emph{integrable} rate. For such integrable rates, the analog of \cref{corollary.conj} was shown already \cite{Moi} using the aforementioned physical space method and proving  $W^{1,1}$-extendibility as an intermediate step (schematically $\int |\partial_v g| \d v<\infty$):
	
	\begin{thm*}[M.VdM.\ \cite{Moi}] \label{Moi.theorem.intro} 
		Consider spherically symmetric characteristic initial data  for \eqref{E1}--\eqref{E5} 
		on the event horizon $\HH$ (and on an ingoing cone). Let the scalar field $ \phiH$  decay fast on the event horizon $\HH$ (i.e.\ $\phiH$ satisfies \eqref{decay.s} for $s>1$). 	Then $\phi$ 
		is uniformly bounded in amplitude and in $W^{1,1}$ i.e.\
		\begin{align}\label{eq:boundedessofphi}
		& \sup_{(u,v)} |\phi|(u,v) < +\infty,\ \hskip 22 mm \sup_{u} \int_{v_0}^{+\infty} |D_v \phi|(u,v) \d v < +\infty.
		\end{align} 
		Moreover the metric $g$ admits a $W^{1,1}$ extension $\tilde{g}$ across the Cauchy horizon $\CH$ and $\tilde{g}$ is $C^0$-admissible (\cref{C0admissible}). In particular, $g$ is $C^0$-extendible. 
	\end{thm*}
	Note that the $W^{1,1}$-extendibility method provides a so-called $C^0$-admissible extension, which is a continuous extension also admitting null coordinates (a slightly stronger result than general $C^0$-extendibility).
	
	Apart from the massless case $m^2 =0$ with  $|q_0 e|< \frac 12$,   the scalar field $\phi$ on the event horizon $\HH$ is expected to be \emph{non-integrable} along the event horizon $\HH$ (\cref{decay.conj}) and as such, the robust physical space methods of \cite{KerrStab,JonathanStab,MihalisPHD, Moi} showing the intermediate and stronger $W^{1,1}_{loc}$-extendibility fail. 
	
	We show in \cref{W11.thm.intro} below that indeed for \emph{generic non-integrable} scalar field $\phiH$ on the event horizon $\HH$, the scalar field $\phi$ \emph{blows up} in $W^{1,1}$ (i.e.\ $\int |D_v \phi| \d v=\infty$) at the Cauchy horizon $\CH$.

	This is yet another manifestation of the fact that the $C^0$-extendibility result for the  non-integrable perturbations is unexpectedly subtle and crucially relies on the precise oscillations of the perturbation on the event horizon $\HH$. In this sense, our result cannot be captured solely in physical space---making our mixed physical space-Fourier space approach  seemingly necessary.

	We now give a rough version of \cref{W11.thm.intro} and refer to  \cref{sec:w11blowup} for the precise formulation.
	\begin{theoc} [$\bm{W^{1,1}}$\textbf{-blow-up} along outgoing cones] \textup{[Rough version; precise version in \cref{sec:w11blowup}]} \label{W11.thm.intro} 	Consider spherically symmetric characteristic initial data  for \eqref{E1}--\eqref{E5} 
		on the event horizon $\HH=[v_0,+\infty)$ (and on an ingoing cone). 	Then the following hold true.
		\begin{itemize}
			\item Consider arbitrary $q_0 \in \RR, m^2\geq 0$. 
			
			Then, for generic $\phiH$ satisfying \eqref{decay.s} and \eqref{s.def}, the scalar field $\phi$ blows up  in $W^{1,1}_{loc}$ at the Cauchy horizon $\CH$ i.e.\ for all $u$
			\begin{align} \label{eq:W11blowup}
			\int_{v_0}^{+\infty} |D_v \phi|(u,v) \d v = +\infty.
			\end{align}
			\item Consider either the small charge case (i.e.\ $0<|q_0 e|< \epsilon(M,e,m^2)$ for $\epsilon(M,e,m^2)>0$ sufficiently small, $m^2 \geq 0$) or the uncharged case $q_0 =0$ for almost every mass $m^2 \in \mathbb R_{>0}$. 
			
			Then, for all non-integrable  $\phiH \notin L^1$ satisfying \eqref{decay.s} and \eqref{s.def}, the scalar field $\phi$ blows up  in $W^{1,1}_{loc}$ along outgoing cones at the Cauchy horizon $\CH$: i.e.\ for all $u$	\begin{align}
			\int_{v_0}^{+\infty} |D_v \phi|(u,v) \d v = +\infty.
			\end{align}
			
		\end{itemize}
	\end{theoc}

	\cref{W11.thm.intro} shows that the Cauchy horizon $\CH$ is already more singular in the slowly decaying case (i.e.\ $\phiH $ obeys \eqref{decay.s} for $s\leq 1$) than in the fast decaying case (i.e.\ $\phiH $ obeys \eqref{decay.s} for $s> 1$) as the comparison  with \eqref{eq:boundedessofphi} illustrates.
	
	Assuming that \cref{decay.conj} is true, as part of  our novel \cref{W11.thm.intro}, we also show that the $W^{1,1}$ blow-up of $\phi$ given by \eqref{eq:W11blowup} also occurs for generic and regular Cauchy data (for almost all parameters $(q_0,m^2)$).
	
	Further \cref{W11.thm.intro} strongly suggests that generically the metric itself is also   $W^{1,1}$-inextendible, i.e.\ does not admit locally integrable Christoffel symbols in any coordinate system. If true, this statement would be in dramatic contrast with  the vacuum perturbations of Kerr considered in  \cite{KerrStab} and the weak null singularities from \cite{JonathanWeakNull} (both enjoying the analog of fast decay on the event horizon $\HH$, see \cref{downfall.section}) in which the metric is shown to be $W^{1,1}$-extendible across the Cauchy horizon $\CH$. Extending \cref{W11.main.thm} to a full $W^{1,1}$-inextendibility result on the metric is however a difficult (albeit very interesting) open problem due to the geometric nature of such a statement, see \cite{KerrStab,JonathanWeakNull,JanC0,JanC1,MoiChristoph2} for related discussions.

	\subsubsection{{\cref{null.contraction.theorem}}: the null contraction singularity at the Cauchy horizon $\CH$ for perturbations violating the oscillation condition}\label{null.contraction.section}

	By \cref{main.theorem.intro2}, if $q_0=0$, then any scalar field $\phiH$ that violates on oscillation condition \eqref{osc.cond.qual} on the event horizon $\HH$ gives rise to $\phi$ that  blows up in amplitude at the Cauchy horizon $\CH$. A natural question then emerges: How does this blow up of the matter field translate geometrically, i.e.\ does the metric admit a singularity? 
	
	This question is answered in the affirmative in our companion paper \cite{MoiChristoph2}: We show that the metric admits a novel type of $C^0$-singularity at the Cauchy horizon $\CH$ that we call a \emph{null contraction singularity}. The main result of \cite{MoiChristoph2} is conditional: we show that the metric admits a null contraction singularity if $|\phi|$ blows up at the Cauchy horizon. Combining this result with \cref{main.theorem.intro2} (if $q_0=0$) shows that a null contraction singularity is formed dynamically for a scalar field $\phiH$ violating the oscillation condition \eqref{osc.cond.qual} on $\HH$.
	
	We emphasize that the null contraction singularity is a $C^0$-singularity and different (in particular stronger) from the usual blue-shift instability \cite{dafermosSR} for derivatives, which additionally occurs at the Cauchy horizon of dynamical EMKG black holes and triggers the blow up of curvature and of the Hawking mass (mass inflation), see \cite{Moi,Moi4} and the discussion in \cref{Weaknullsection}. Specifically, the null contraction singularity has the following novel characteristics.

	\begin{theod}[C.K.--M.VdM.\ \cite{MoiChristoph2}]  \label{null.contraction.theorem}
		Consider spherically symmetric characteristic initial data  for \eqref{E1}--\eqref{E5} 
		on the event horizon $\HH$  (and on an ingoing cone). Let the scalar field  $\phiH$  decay slowly on the event horizon $\HH$  (i.e.\ $\phiH  $ satisfies \eqref{decay.s}, \eqref{s.def}). 	Assume additionally that $\phi$ blows up in amplitude at the Cauchy horizon $\CH$ i.e.\ assume that $\underset{(u,v) \to \CH}{\limsup} |\phi|(u,v) = +\infty$. 
		
		\noindent Then the metric $g$ admits a null contraction singularity in the following sense:    \setlist{nolistsep}
		\begin{enumerate}[a)]
			\item \label{singA} The metric does not admit any $C^0$-admissible extension (as defined in   \cref{C0admissible}) across the Cauchy horizon $\CH$. 
			
			\item \label{singB}

			The affine parameter time on ingoing null geodesics  (with uniform but otherwise arbitrary normalization) between two radial causal curves with distinct endpoint at the Cauchy horizon $\CH$ tends to zero as the Cauchy horizon $\CH$ is approached.
			
			\item \label{singC}
			The angular tidal deformations of radial ingoing null geodesics (with uniform but otherwise arbitrary normalization) become arbitrarily large near the Cauchy horizon $\CH$.

		\end{enumerate}
		
	\end{theod} For the precise definitions of the terms employed in the statement of   \cref{null.contraction.theorem} we refer the reader to \cite{MoiChristoph2}. Note that the null contraction singularity is named in reference to Statement \ref{singB}, the most emblematic: physically, it means that the (suitably renormalized) affine parameter time in the ingoing null direction between two observers tends to zero as both observers approach the Cauchy horizon $\CH$. 
	
	\cref{null.contraction.theorem} is the first instance of a \textit{null contraction singularity}: Statements \ref{singA}--\ref{singC} have only been shown to occur  in the context of matter fields blowing up at the Cauchy horizon $\CH$, as we prove in \cite{MoiChristoph2}. In particular, Statements \ref{singA}--\ref{singC} are all false on the exact Reissner--Nordstr\"{o}m interior or on the spacetimes of \cref{main.theorem} for which $\phi$ is bounded.  
	
	In view of \cref{main.theorem.intro2}, we note that there exists a large class of characteristic data on $\HH\cup \underline{C}_{in}$ giving rise to a null contraction singularity at $\CH$, see \cref{fig:theorema'}. Moreover, assuming  \cref{scat.conj}, we also note that there exist Cauchy data on asymptotically flat $\Sigma$  which give rise to a null contraction singularity at $\CH$.
	
	Finally, we note that statement \ref{singA} of \cref{null.contraction.theorem} is, to the best of the authors knowledge, the first $C^0$-inextendibility result across a null boundary (in our case the Cauchy horizon $\CH$). The geometric statement \ref{singA} strongly suggests that the oscillation condition \eqref{osc.cond.qual} is indeed crucial to falsify \cref{C0SCC}. Note however that \cref{null.contraction.theorem} only proves the impossibility to extend the metric in a spherically symmetric $C^0$-class  (also used in \cite{gmadmissible}), where $C^0$ double null coordinates exist. It would be interesting to investigate whether statement \ref{singA} can be promoted to a full $C^0$-inextendibility statement. However such statements are notoriously difficult to obtain: even in the more singular case where the black hole boundary is spacelike\footnote{A spacelike singularity is indeed  widely associated to $C^0$-inextendibility, and viewed as a stronger singularity than a Cauchy horizon, notably because of the blow-up of tidal deformations experienced on timelike geodesics \cite{KerrStab,JanC0,JanC1}.}, the $C^0$-extendibility of the metric has only been proved for the Schwarzschild black hole \cite{JanC0}.

	\subsection{Cauchy horizons in other models: a comparison with our results}
	\label{CH.intro}
	
	Having introduced our main results on the EMKG model \eqref{E1}--\eqref{E5} in \cref{sec:intro1}, we will now mention selected results on the existence/regularity of Cauchy horizons and \cref{C0SCC} for different models, which will appear to be in dramatic contrast with the previous \cref{Moi.intro} and  our new results given in \cref{main.theorem.intro}, \cref{main.theorem.intro2}, \cref{W11.thm.intro} and \cref{null.contraction.theorem} on the EMKG model in spherical symmetry.

	\subsubsection{Spherically symmetric models with no Maxwell field: absence of a Cauchy horizon} \label{brief}
	
	Before turning to models admitting Cauchy horizons emanating from $i^+$, it is useful to recall that there exist models for which such Cauchy horizons do not form. An example of such a model is given by the Einstein-scalar-field system (i.e.\ \eqref{E1}--\eqref{E5} with $F \equiv 0$, $m^2=0$) in spherical symmetry. This model was studied in the seminal series of Christodoulou \cite{Christo1,Christo2,Christo3} who showed that the MGHD of generic spherically symmetric data is bound to the future by a spacelike boundary $\mathcal{S}=\{r=0\}$ (in particular, there exists no null component of the boundary) and observers approaching $\mathcal{S}=\{r=0\}$ experience infinite tidal deformations.
	
	From Christodoulou's work \cite{Christo1}, it follows that \cref{C0SCC} is \emph{true} for the Einstein-scalar-field system in spherical symmetry in the sense that there exists no spherically symmetric $C^0$-extension of the metric.

	\subsubsection{Stability of the Cauchy horizon and the downfall of  \cref{C0SCC} for massless fields and in vacuum} \label{downfall.section}
	
	\paragraph{The Einstein--Maxwell-uncharged-scalar-field in spherical symmetry.} Christodoulou's spherically symmetric spacetimes however fail to capture the repulsive effect that angular momentum exerts on the geometry in non-spherical collapse. One way to model this repulsive effect while remaining in the realm of spherical symmetry is to add a Maxwell field to the Einstein-scalar-field equations: The electromagnetic force then plays the role of angular momentum in non-spherical collapse \cite{Dafermos.Wheeler}. The resulting Einstein--Maxwell-uncharged-scalar-field system, i.e.\ \eqref{E1}--\eqref{E5} with $m^2=q_0=0$, admits a (spherically symmetric) stationary charged black hole, the Reissner--Nordstr\"{o}m metric (for which $\phi \equiv 0$) whose MGHD is bound to the future by a smooth Cauchy horizon $\mathcal{CH}_{i^+} $, see \cref{Fig2}.

	\paragraph{Falsification of \cref{C0SCC} for the Einstein--Maxwell-uncharged-scalar-field model in spherical symmetry.} The interior dynamics\footnote{For a discussion of the dynamics far away from $i^+$ in the context of gravitational collapse, see \cref{weak.grav.col}.} near $i^+$ for the Einstein--Maxwell-uncharged-scalar-field model  were studied  in the pioneering work of Dafermos \cite{Mihalis1,MihalisPHD} who proved that the interior  of the black hole admits a Cauchy horizon $\CH$ across which the metric is continuously extendible, under the crucial assumption of integrable decay of the scalar field  on the event horizon $\HH$. 
	Integrable decay for the scalar field on the event horizon $\HH$ (i.e.\ $\phiH$ satisfies \eqref{decay.s} for $s>1$)  was later proved for sufficiently regular Cauchy data by Dafermos--Rodnianski \cite{PriceLaw},   therefore  \cref{C0SCC} is \emph{false} for the Einstein--Maxwell-uncharged-scalar-field model in spherical symmetry  \cite{Mihalis1, MihalisPHD, PriceLaw} \emph{by means of fast decay} $s>1$. 
	
	Moreover, for this spherically symmetric model, Dafermos characterized entirely  the black hole future boundary \cite{nospacelike} for any small, two-ended perturbation of Reissner--Nordstr\"{o}m. He indeed showed that the resulting dynamical black hole has no spacelike singularity: its Maximal Globally Hyperbolic Development is bound to the future by a null bifurcate Cauchy horizon $\CH$, and has the Penrose diagram of \cref{Fig2}.

	\paragraph{Falsification of \cref{C0SCC} for the vacuum Einstein equations without symmetry.}\label{vacuum.section} As we already mentioned in \cref{sec:intro1},  \cref{C0SCC}  was also falsified in vacuum with no symmetry assumption  in the celebrated work of Dafermos--Luk \cite{KerrStab}. In this case as well, the crucial assumption in \cite{KerrStab} is the fast decay of metric perturbations along the event horizon, i.e.\  schematically in a standard choice of $v$ coordinate \begin{equation}  \label{Kerr.decay}
	\|v^{s-\frac{1}{2}} (g-g_{K}) \|_{L^2(\HH)} \leq \epsilon, \text{ for some } s>1
	\end{equation} 	where $g_K$ is the Kerr metric and $\epsilon>0$ is small. Remark that \eqref{Kerr.decay} shows that $|g-g_{K}|(v)\ls v^{-s}$ (at least along a sequence) and in that sense \eqref{Kerr.decay} is indeed the analog for $g-g_K$ of fast decay of the scalar field i.e.\ \eqref{decay.s} for $s>1$.

	The linear analog  of \eqref{Kerr.decay} for the black hole exterior stability problem  around Kerr has been established in \cite{KerrDaf,SRTdC2020boundedness}, see also the recent nonlinear \cite{SchwarzschildStab}. If \eqref{Kerr.decay} (and  related estimates)  are shown for the full Einstein equations in a neighborhood of Kerr, then the result of \cite{KerrStab} \emph{unconditionally falsifies \cref{C0SCC} in vacuum}, by means of fast decay $s>1$.
	
	\begin{figure}		\begin{center}			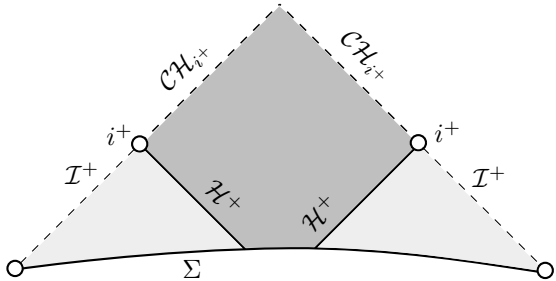		\end{center}		\caption{Penrose diagram of the subextremal Reissner--Nordstr\"{o}m spacetime.} \label{Fig2} 	\end{figure}

	\subsection{Weak null singularities at the Cauchy horizon and a weaker formulation of Strong Cosmic Censorship}
	\label{Weaknullsection}
	In this section, we mention briefly other types of singularities at the Cauchy horizon $\CH$, and how they compare with the new singularities at the Cauchy horizon $\CH$ from \cref{W11.thm.intro} and \cref{null.contraction.theorem}.
	
	\subsubsection{Weak null singularities and blue-shift instability} As discussed earlier, our new results exhibit the first examples  of  Cauchy horizons $\CH$ singular at the $C^0$ level (for non-oscillating scalar fields at $\HH$) and the $W^{1,1}$ level (for all slowly decaying scalar fields at $\HH$). This new singularity at the Cauchy horizon $\CH$ is very different from the well-known weak null singularity at $\CH$ \cite{JonathanWeakNull,r0,Ori2,Brady,PhysRevD.57.R7084,Ori3} which corresponds to blow-up in the energy class (i.e.\ $H^1$ norm) at $\CH$ due to the celebrated blue-shift instability \cite{Penroseblue, McNamara}. Blow-up in energy (i.e.\ $H^1$ norm in non-degenerate coordinate) at the Cauchy horizon  of Kerr and Reissner--Nordström has indeed been proven to occur for the linear wave equation in \cite{dafermosSR,KerrInstab,JonathanInstab}. Based on the blue-shift instability, Christodoulou  suggested an alternative formulation of Strong Cosmic Censorship that is weaker than \cref{C0SCC}. Specifically, he conjectured in \cite{ChristoSCC} that for generic asymptotically flat Cauchy data, the metric is $H^1$-inextendible i.e.\ admits no extension with square-integrable Christoffel symbols, see also \cite{Chrusciel,KerrStab}. 
	
	More generally, we say that the Cauchy horizon $\CH$ is a weak null singularity if already the metric is $C^2$-inextendible across $\CH$, a property which is generally obtained from the blow-up of some curvature component in an appropriate frame \cite{Moi4,JonathanStab,Kommemi}.

	\subsubsection{Dynamical formation of weak null singularities and known inextendibility results} While examples of weak null singularities have been constructed in vacuum \cite{JonathanWeakNull}, their dynamical formation from an ``open set'' of data with no symmetry assumption is still an open problem. Nevertheless, for the EMKG model in spherical symmetry, it was proven \cite{Moi,Moi4} that \emph{the Cauchy horizon $\CH$ of \cref{Moi.intro} is weakly singular}, i.e.\ the metric is $C^2$-inextendible across the Cauchy horizon $\CH$, under the assumptions of \cref{Moi.intro} and additional lower bounds on the scalar field consistent with \cref{decay.conj}.
	In the uncharged massless model $q_0=m^2=0$ of \cref{downfall.section}, the same result was previously proven unconditionally by Luk--Oh \cite{JonathanStab,JonathanStabExt} for generic asymptotically flat two-ended Cauchy data. Both for the EMKG and the $q_0=m^2=0$ model, the above $C^2$-inextendibility result was improved to a  $C^{0,1}$-inextendibility statement  in \cite{JanC1}.

	\subsubsection{Weak null singularities in gravitational collapse} \label{weak.grav.col}	We conclude this section by a brief discussion of the influence of a weak null singularity on the black hole geometry away from $i^+$. To study this question in the framework of gravitational collapse (i.e.\ one-ended spacetimes with a center $\Gamma$ as in \cref{Fig:spacelike}), we cannot study the Einstein--Maxwell-uncharged-scalar-field model of  \cref{downfall.section} because of a well-known \cite{Kommemi,r0} topological obstruction caused by the scalar field being uncharged, i.e.\ $q_0=0$, forcing the initial data $\Sigma$ to be two-ended \cite{nospacelike}. However in the EMKG model, where $q_0\neq0$, there is no such obstruction and one can study the one-ended global geometry of the black hole interior with a weak null singularity, even in spherical symmetry \cite{Kommemi}. The main known result in this context is that the \textit{weak null singularity $\CH$ breaks down} \cite{r0} before reaching the center: Consequently a so-called first singularity $b_{\Gamma}$ is formed at the center $\Gamma$, as depicted in  \cref{Fig:spacelike}. This is in complete contrast with the two-ended case where the future boundary is entirely null \cite{nospacelike} for a large class of spacetimes as we discussed in \cref{downfall.section}. In the conjecturally generic case where $b_{\Gamma}$ is not a so-called \emph{locally naked singularity} \cite{r0,Kommemi,MihalisSStrapped,Christo3}, then the breakdown  of the weak null singularity $\CH$  proven \cite{r0} implies that a stronger singularity $\mathcal{S}=\{r=0\}$ takes over and connects the weak null singularity $\CH$ to the center $\Gamma$ as depicted in \cref{Fig:spacelike}.
	
	\begin{figure}
		\begin{center}		
			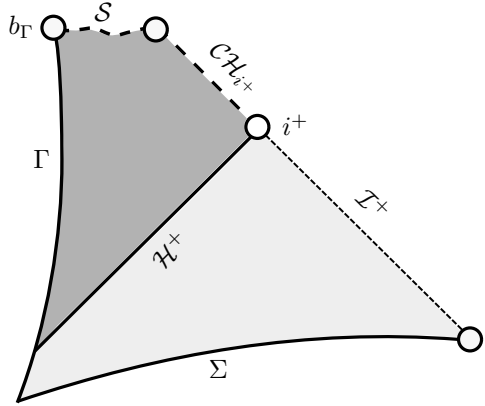
		\end{center}
		\caption{Conjectured Penrose diagram of a generic EMKG black hole with weakly singular $\CH$, \cite{r0}.} \label{Fig:spacelike} 
	\end{figure}

	\subsection{Scattering resonances associated to the   Reissner--Nordström   Cauchy horizon}\label{lin.scat.sec}
	
	We now turn to another result which is not directly concerned with the stability/instability of the Cauchy horizon but  turns out to be important for the proofs of our main theorems: the finite energy scattering theory for the linear wave equation on the interior of Reissner--Nordström developed in \cite{Kehle2018}. A key insight to the result in \cite{Kehle2018} was the absence  of scattering resonances associated to the Killing generator of the Cauchy horizon which is an exceptional feature of the massless and uncharged wave equation on exact Reissner--Nordström.  Indeed, for the massive wave equation with generic masses $m^2 \in \mathbb R_{>0} - D(M,e)$ or for the charged equation the scattering resonances are present and there does not exist an analogous scattering theory \cite{Kehle2018}. As we will show, these   scattering resonances are also the key sources of   blow-up in amplitude of $\phi$ at the Cauchy horizon if the scalar field along the event horizon is non-oscillating and slowly decaying and thus, sufficiently resonant. In view of this, for the blow-up statement of \cref{main.theorem.intro2}  these exceptional masses for which the scattering resonances are absent have   to be excluded. 
	Refer also to \cite{mokdad2021conformal,hafner2020scattering} for a scattering theory of the Dirac equation on the interior of Reissner--Nordström and to \cite{bachelot,dimock,scatteringDRSR,Hamed,alfrodf} for scattering theories on the exterior.

	\subsection{Connection to the linear analog of \texorpdfstring{\cref{C0SCC}}{Conjecture 1} for negative cosmological constant \texorpdfstring{$\Lambda < 0$}{lambda<0}} \label{sec:SCC.non.zero}
	In the discussion above we have studied the Einstein equations with cosmological constant $\Lambda =0$. Analogously, for $\Lambda \neq 0$, the Reissner--Nordström--(Anti-)de~Sitter and Kerr--(Anti-)de~Sitter spacetimes   admit a smooth Cauchy horizon and the issue of Strong Cosmic Censorship analogously arises   in this setting. 
	In particular, the case $\Lambda <0$ has some similarities with our case in the sense that linear perturbations also only decay at a non-integrable (inverse logarithmic for $\Lambda <0$) rate due to a stable trapping phenomenon \cite{decayads,adslowerbound,ads2}. A difference to our result is however that only perturbations consisting of a superposition of infinitely many high $\ell$ angular modes decay  slowly and thus,  the problem for $\Lambda <0$ cannot be reasonably studied in spherical symmetry. Nevertheless, as in our case,  this non-integrable rate of decay might raise hopes that in the case of negative cosmological constant $\Lambda <0$,  \cref{C0SCC} holds true. 
	
	On the one hand, for Reissner--Nordström--AdS, since stable trapping is a high-frequency phenomenon and uniform boundedness (on the linear level) is associated to zero-frequency scattering resonances of the Cauchy horizon,  it was shown in \cite{Kehle2019} that these difficulties decouple on Reissner--Nordström--AdS. (This decoupling can be seen as the analog of the fact that the oscillation condition of \eqref{osc.cond.qual} is satisfied.) As a consequence of this frequency decoupling, it is shown in \cite{Kehle2019} that, despite slow non-integrable decay on the exterior, linear perturbations remain uniformly bounded and extend continuously across the Reissner--Nordstr\"{o}m--AdS Cauchy horizon. This falsifies the linear analog of \cref{C0SCC} for Reissner--Nordström--AdS.
	
	On the other hand, for Kerr--AdS, in view of the rotation of the black hole, frequency mixing occurs and trapped high frequency perturbations on the exterior can at the same time be low frequency when frequency is measured with respect to the Killing generator of the Cauchy horizon. In \cite{Kehle2020,kehle2020diophantine} it is shown that this frequency mixing gives rise to a  resonance phenomenon and an associated small divisors problem. In particular, for a set of Baire-generic Kerr--AdS black hole parameters, which are associated to a Diophantine condition, it is shown that linear perturbations $\phi$ blow up in amplitude at the Cauchy horizon. This shows that the linear analog of \cref{C0SCC} holds true for Baire-generic Kerr--AdS black holes. 
	
	There is yet another possible scenario in which the exteriors of AdS black holes are nonlinearly unstable (cf.\ \cite{gm1,gm2,gmadmissible} and \cite{bizon}) and the question of Strong Cosmic Censorship would be thrown even more open.

	Let us finally also briefly mention the case of positive cosmological constant $\Lambda >0$, where perturbations on the exterior of Reissner--Nordström/Kerr--de~Sitter decay at an exponential rate as proved in \cite{dyatlovenergydecay,gergiosmkerr} for the linear wave equation and in \cite{Andras} for the vacuum Einstein equations. In view of this rapid decay, the theorem \cite{KerrStab} of Dafermos--Luk manifestly also applies and thus, \cref{C0SCC} is false for $\Lambda >0$. However, in view of this exponential decay, even weaker formulations such as the $H^1$-formulation of Strong Cosmic Censorship mentioned in \cref{Weaknullsection} may fail. We  refer to the works \cite{nospacelike,hintzvasyscc,scc1lambda>0,rough,sccinds,dafermosSR,costa1,Costa-Franzen} and to \cite{MR3969720,hollands,cardoso} for details.  
	
	\subsection{Summary of the strategy of the proof}\label{Strategy.section}
	
	We now turn to an outline of our proof and begin with the obstructions and difficulties encountered when attempting to prove boundedness of the scalar field at the Cauchy horizon $\CH$ and continuous extendibility of the metric.
	\begin{itemize}
		\item The physical space estimates used to show $\CH\neq \emptyset$ in the proof of \cref{Moi.intro}, under the assumption of a  slowly decaying $\phiH$ on $\HH$, i.e.\ obeying \eqref{decay.s} and \eqref{s.def}, are \emph{consistent with the blow-up of the scalar field  $\phi$} at the Cauchy horizon $\CH$ and the failure of $\partial_v \phi$ to be integrable in $v$. As our new result shows, these estimates from \cite{Moi} are sharp by \cref{W11.thm.intro} and blow-up in amplitude indeed occurs for some perturbations by \cref{main.theorem.intro2}.
		
		\item  The estimates of the proof of \cref{Moi.intro} however suggest that, \emph{if  $\partial_v\phi$  oscillates infinitely towards the Cauchy horizon} $\CH$ then \emph{$\phi$ is bounded} (see \cref{physical.space.intro.section}): the hope would be that,  although $\partial_v \phi$ is not Lebesgue-integrable  (i.e.\ $  \int^{+\infty}_{v_0} |\partial_v \phi| \d v =+\infty$), it has a semi-convergent Riemann integral  (i.e.\ $ \underset{\tilde v \to +\infty}{\lim} |\int_{v_0}^{\tilde v} \partial_v \phi \d v| <+\infty$ exists). A natural approach is then to attempt to propagate  the event horizon oscillations \eqref{osc.cond.qual} satisfied $\phiH$  towards the Cauchy horizon $\CH$ in a suitable sense and deduce the boundedness of $\phi$. However, this is not easy to show in physical space and prompts a Fourier space approach for the linearized equation.

		\item   A complete understanding of the linearized problem is however insufficient in itself to prove the boundedness of $\phi$  since   the \emph{nonlinear terms cannot be treated purely perturbatively} in view of the slow decay. Consequently the precise structure of these nonlinear terms has to be understood and plays an important role in the argument (in contrast to the fast decay case $s>1$) (see \cref{nonlinear.intro}).
		
		\item Even once $\phi$ is proven to be bounded in amplitude, there is no clear mechanism yielding the continuous extendibility of the metric, contrary to the fast decay case $s>1$ in which the mechanism is given by the integrability of  the Christoffel symbols \cite{KerrStab,JonathanStab} in a suitable sense (see \cref{C0.ext.method.intro} for a discussion).
	\end{itemize}

	\paragraph{Strategy.} To address and overcome these difficulties in order to prove our main theorems as stated in \cref{sec:intro1}, we proceed as follows: \begin{enumerate}
		\item \label{step1} We take advantage on the one hand of the previous result of \cref{Moi.intro}: the future black hole boundary is null, i.e.\ $\CH\neq \emptyset$ and the Penrose diagram is given by \cref{fig:cauchyhorizonexists}; and on the other hand of the nonlinear estimates  (see \cref{physical.space.intro.section}) that were already proven in \cite{Moi} for slowly decaying  $\phiH$.
		\item\label{Fourier} We consider the massive/charged  linear  wave equation $ g_{RN}^{\mu \nu} D^{RN}_{\mu} D^{RN}_{\nu} \phiNl= m^2 \phiNl$ on a fixed  Reissner--Nordstr\"{o}m background $g_{RN}$ which we view as the linearization of the EMKG system \eqref{E1}--\eqref{E5}.  Using Fourier methods and a scattering approach, we prove uniform boundedness (respectively blow-up in amplitude) of $\phiNl$ at the Cauchy horizon $\CH$ for an oscillating  scalar field $\phiH$ obeying \eqref{osc.cond.qual} at $\HH$ (respectively non-oscillating $\phiH$,  i.e.\ $\phiH$ violates \eqref{osc.cond.qual} at $\HH$), see \cref{linear.intro}. 
		
		\item \label{diff.step} Independently of Step \ref{Fourier}, we prove nonlinear difference estimates on $g-g_{RN}$. Although these estimates are, in a sense, weaker\footnote{In the sense that these estimates alone are insufficient to show that  $\CH \neq \emptyset$ as proven in \cite{Moi} (see  \cref{Moi.intro}).} than the nonlinear estimates of Step \ref{step1}, they are crucial in our proof that for all slowly decaying  $\phiH$, the linear solution $\phiNl$ is bounded if and only if  the nonlinear   $\phi$ is bounded (at least in the  $q_0=0$ case). 
		In the charged $q_0\neq 0$ case, we follow a similar logic but additional difficulties arise from the nonlinear backreaction of the Maxwell field. This step will be discussed in \cref{nonlinear.intro}. 
		
		\item \label{C0step} With the boundedness of $\phi$ at hand from the previous step, we prove the continuous extendibility of the metric for oscillating perturbations $\phiH$ satisfying \eqref{osc.cond.Squal}. For the proof, we introduce a crucial new quantity $\Upsilon$ (see already \eqref{Delta.def})  exploiting the exact algebraic\footnote{In contrast, when the decay is integrable as in vacuum,  the null structure of the Einstein equations is sufficient \cite{KerrStab}, \cite{JonathanStab}.} structure of the nonlinear terms in the Einstein equations, see \cref{C0.ext.method.intro}.
	\end{enumerate} 
	The proofs of \cref{main.theorem.intro} and \cref{main.theorem.intro2} are finally obtained	combining Step \ref{step1}, Step \ref{Fourier}, Step \ref{diff.step}, Step \ref{C0step}. \cref{corollary.conj} follows immediately. The proof of \cref{W11.main.thm} is  also derived from the strategy given by the same Step \ref{step1}, Step \ref{Fourier}, Step \ref{diff.step}, Step \ref{C0step}, see already the last paragraphs in \cref{C0.ext.method.intro}. We refer to \cref{Strategy.section.bulk} for a more detailed outline of the strategy of the proof.

	\subsection{Outline of the paper}
	
	In   \cref{preliminaries}, we set out notations, definitions and the geometric setting for the solutions of \eqref{E1}--\eqref{E5} under spherical symmetry. In \cref{coordinatechoice.section}, for any arbitrary slowly decaying scalar field $\phiH$, we construct and set up spherically symmetric characteristic data on the event horizon $\HH$ and an ingoing cone such that the scalar field is given by $\phiH$ on $\HH$. In \cref{precise.section}, we give the  precise formulations of our main results  \cref{main.theorem.intro}, \cref{main.theorem.intro2}, \cref{corollary.conj}, \cref{W11.thm.intro} and their assumptions. We end this section with a detailed outline of our proof in \cref{Strategy.section.bulk}. In   \cref{linearsection}, we develop the linear theory and show our main linear results in \cref{sec:main.lin}. 
	In \cref{nonlinearsection}, we develop the nonlinear theory and show the boundedness of the scalar field for the coupled \eqref{E1}-\eqref{E5} and the continuous extendibility of the metric. We first outline in \cref{recall.section} the estimates proved in \cite{Moi} which will be useful for the nonlinear EMKG system. Then in \cref{prelim.nonlin}, we establish the main estimates necessary for the continuous extendibility of the metric.  
	In \cref{difference.estimates.section}, we prove    difference estimates which we combine in   \cref{combining.section}  with the  linear estimates from \cref{linearsection}  to prove our main results \cref{main.theorem.intro}, \cref{main.theorem.intro2}, \cref{corollary.conj}, \cref{W11.thm.intro}. 
	
	\subsection{Acknowledgements}
	The authors would like to express their gratitude to Mihalis Dafermos and Jonathan Luk for their support and many valuable comments on the manuscript. They also thank Jan Sbierski for insightful conversations. M.VdM.\ is grateful to Andr\'as Vasy for interesting discussions in the early stage of this project. M.VdM. also wants to express gratitude to Amos Ori for very enlightening conversations related to this project. The authors are also grateful to an anonymous referee for several helpful comments and remarks to improve the manuscript.  C.K.\ is supported by Dr.\ Max Rössler, the Walter Haefner Foundation and the ETH Zürich Foundation.
	
	\section{Preliminaries} \label{preliminaries}
	
	\subsection{The Reissner--Nordstr\"{o}m interior} \label{RNsolution}
	Reissner--Nordstr\"{o}m black holes constitute a 2-parameter family of spherically symmetric   spacetimes indexed by charge and mass $(e,M)$, which satisfy the Einstein--Maxwell system (\eqref{E1}--\eqref{E5} with $\phi \equiv 0$) in spherical symmetry.
	We are interested in the interiors of {subextremal} Reissner--Nordstr\"{o}m black holes satisfying $0<|e|<M$.
	To define these  spacetimes, we first set 	\begin{equation} \label{OmegaRN}
	\OmegaRN^{2}(r_{RN}):=-\left(1-\frac{2M}{r_{RN}}+\frac{e^2}{r^2_{RN}}\right)
	\end{equation} 
	which is non-negative between the zeros given by $$ r_+(M,e) = M+ \sqrt{M^2-e^2}>0$$
	and
	$$ r_-(M,e) = M- \sqrt{M^2-e^2}>0.$$
	Now, we define the smooth manifold $\mathring{\mathcal{M}}_{{RN}}$ as a 4-dimensional smooth manifold diffeomorphic to $\mathbb R^2 \times \mathbb S^2$. Up to the well-known degeneracy of the spherical coordinates on $\mathbb S^2$, let $(r_{RN}, t, \theta, \varphi) \in (r_-, r_+) \times \mathbb R \times \mathbb S^2$ be a global chart. In that chart we define the smooth Lorentzian metric $g_{RN}$ and Maxwell 2-form $F_{RN}$  
	\begin{align}
	& g_{RN} := -  \OmegaRN^{-2}\d r_{RN}^2 +  \OmegaRN^{2} \d t^2 + r_{RN}^2 \left( \d \theta^2 + \sin^2\theta \d \varphi^2\right),\\ \label{F.RN} & F_{RN} = \d A^{RN} =\frac{e}{r^2} \d t \wedge \d r.
	\end{align}
	We time orient the Lorentzian manifold such that vector field $-\nabla r_{RN} $ is future-directed. 
	Further, we define the tortoise coordinate $r^\ast$ by $ \d r^\ast = -\Omega^{-2}_{RN} \d r_{RN}$ or more explicitly by 
	\begin{align}r^{*} = r^\ast (r_{RN}) =r_{RN}+ \frac{1}{4K_{+}} \log(r_{+}-r_{RN})+\frac{1}{4K_{-}}  \log(r_{RN}-r_{-}), \end{align}
	where $ K_+(M,e)$, $ K_-(M,e)$ are the surface gravities associated to the event/Cauchy horizon defined as
	\begin{align}K_+(M,e)= \frac{1}{ 2r_+^2}\left(M - \frac{e^2}{ r_+} \right)= \frac{ r_+- r_-}{ 4r_+^2}>0, \;\; K_-(M,e)= \frac{1}{2 r_-^2}\left(M - \frac{e^2}{ r_-}\right)= \frac{ r_- - r_+}{ 4 r_-^2}<0.\end{align}
	We further introduce the null coordinates $(u,v,\theta,\varphi) \in \mathbb{R}  \times  \mathbb{R}\times \mathbb S^2$ on  $\mathring{\mathcal{M}}_{{RN}}$ as
	\begin{align}\label{eq:coordinatesuvonRN}
	v= r^{*}(r)+t, \, u= r^{*}(r)-t, \, \theta=\theta, \,\varphi=\varphi. \end{align}
	In this coordinate system the metric $g_{RN}$ has the form 
	\begin{align}\label{RN} g_{RN} =- \frac{\OmegaRN^{2}}{2} (\d u  \otimes \d v+\d v  \otimes \d u)+r^{2}_{RN}[ \d\theta^{2}+\sin(\theta)^{2}\d \varphi^{2}].\end{align} Now, we  attach the (right) event horizon $ \mathcal{H}^{+}$, the past/future bifurcation sphere $\mathcal B_-$, $\mathcal B_+$, the left event horizon  $ \mathcal{H}^{+,L}$, the (right) Cauchy horizon $\mathcal{CH}_{i^+}$, and the left Cauchy horizon $   \mathcal{CH}_{i_L^+} $ to our manifold, formally defined as
	$\mathcal{H}^{+}=\{ u=  -\infty, v\in \mathbb{R} \}$, $\mathcal{H}^{+,L}=\{ v=  -\infty, u\in \mathbb{R} \}$, 
	$\mathcal{CH}_{i^+}=\{ v=  +\infty, u\in \mathbb{R} \}$, $\mathcal{CH}_{i_L^+}=\{ u=  +\infty, v\in \mathbb{R} \}$,
	$\mathcal{B}_- = \{ u =-\infty , v=-\infty\}$
	and 
	$\mathcal{B}_+ = \{ u = +\infty , v=+\infty \}$. 
	
	\textbf{A word of caution}: In the linear  theory of \cref{linearsection} we will indeed denote  with $\HH$  the Reissner--Nordström event horizon $\{ u=-\infty, v\in \mathbb R\}$. 	However, in the other parts of the paper we  denote with $\HH$ the dynamical event horizon $\{ u=-\infty , v \geq v_0\}$ in the nonlinear part of \cref{nonlinearsection} (see also the set-up of the characteristic data in \cref{coordinatechoice.section} and the main theorems stated in \cref{precise.section}). Similarly for the Cauchy horizon $\CH$. We also note that the left event and the left Cauchy horizon only play a minor role in the linear part of \cref{linearsection}   and we often omit ``right'' when referring to $\HH$ and $\CH$. 
	
	The metric $g_{RN}$ extends smoothly to the boundary and the resulting spacetime is a time-oriented Lorentzian manifold ($\mathcal{M}_{RN},g_{RN}$) with corners---the Reissner--Nordström interior. 
	We like to remark that \begin{align} \Omega^{2}_{RN} \underset{ r \rightarrow r_{+} }{\sim} C_{e,M} e^{4 K_{+}r^{*}}=C_{e,M} e^{2K_{+} (u+v)},\;\; \Omega^{2}_{RN} \underset{ r \rightarrow r_{-} }{\sim} C'_{e,M} e^{4 K_{-}r^{*}}= C'_{e,M} e^{2K_{-} (u+v)},\end{align}
	for some $C_{e,M}>0$, $C_{e,M}'>0$.
	Further, we introduce regular coordinates $(U,v)$ on $\mathring {\mathcal M}_{{RN}} \cup \HH$ as \begin{align}\d U = \frac 12 \Omega^2_{RN}(u,v_0) \d u,\ U(-\infty) = 0,\ v=v \label{eq:Ucoordinate} \end{align}  
	and note that $\HH = \{ U = 0\}$. Here $v_0=v_0(M,e,D_1,s)$ will be determined in \cref{prop:constraints} later. In these coordinates we have obtained a different lapse function $ (\Omega_{RN}^{2})_H = (\Omega_{RN}^{2})_H(U,v)= -2g_{RN}(\partial_U, \partial_v)$ and the metric reads \begin{align}\label{eq:metricincoordinates}
	g_{RN} =- \frac{(\Omega_{RN}^{2})_H}{2} (\d U  \otimes \d v+\d v  \otimes \d U)+r^{2}_{RN}[ \d\theta^{2}+\sin(\theta)^{2}\d \varphi^{2}].
	\end{align}
	Of course we can invert the   coordinate change \eqref{eq:Ucoordinate}   and obtain \begin{align} \label{defn:uintermsofUv} u = u(U),\  v = v.\end{align} 
	We also remark that $T:= \partial_t$ in $(r_{RN},t,\theta,\phi)$-coordinates is a Killing vector field which extends smoothly to ($\mathcal{M}_{RN},g_{RN}$).

	\subsection{Class of spacetimes, null coordinates, mass, charge} \label{Coordinates}
	
	\paragraph{Spherically symmetric solution to the EMGK system.}
	A smooth spherically symmetric solution of the EMKG system  is described by a quintuplet $(\mathcal M,g,F,A,\phi)$,  where $(\mathcal M,g)$ is a smooth 3+1-dimensional Lorentzian manifold, $\phi$ is a smooth complex-valued scalar field, $A$ is a smooth real-valued 1-form, and $F$ is a smooth real-valued 2-form satisfying \eqref{E1}--\eqref{E5} and admitting a free $SO(3)$ action on $(\mathcal M,g)$ which acts by isometry with spacelike 2-dimensional orbits (homeomorphic to $\mathbb S^2)$ and which additionally leaves  $F$, $A$ and $\phi$ invariant.\footnote{Note that we assume that the $SO(3)$ action is free, i.e.\ free of fixed points ``$r=0$'' as we are interested in the region near $i^+$, i.e.\ away from $r=0$. }
	In this case, the quotient $\mathcal{Q}=\mathcal M/SO(3)$ is a 2-dimensional manifold with projection	$\Pi : \mathcal M \rightarrow  \mathcal{Q} $    taking a point of $\mathcal M$ into its spherical orbit. As $SO(3)$ acts by isometry, $\mathcal{Q}$ inherits a natural metric, which we call $ g_{\mathcal{Q}}$. 
	The metric on $\mathcal M$ is then given by the warped product $g= g_{\mathcal{Q}}+ r^{2} \d \sigma_{\mathbb{S}^{2}}$,  where $r= \sqrt{\frac{Area(\Pi^{-1}(p))}{4 \pi}}$ for $ p \in \mathcal Q$ is the area radius of the orbit and  $\d\sigma_{\mathbb{S}^{2}}$ is the standard metric on the sphere.
	The Lorentzian metric 	$g_{\mathcal{Q}}$ over  the smooth  2-dimensional manifold $\mathcal Q$, can be written in null coordinates $(u,v)$ as a conformally flat metric
	\begin{align} g_{\mathcal{Q}}:= - \frac{\Omega^{2}}{2} (\d u  \otimes \d v+\d v  \otimes \d u),
	\label{defn:gQ}
	\end{align} 
	such that (in mild abuse of notation)  we have upstairs
	\begin{equation} \label{metric}
	g=- \frac{\Omega^{2}}{2} (\d u  \otimes \d v+\d v  \otimes \d u)+ r^2 \d \sigma_{\mathbb{S}^{2}}.
	\end{equation}
	On $(\mathcal Q, g_{\mathcal Q})$, we now define the Hawking mass as
	\begin{align} \rho := \frac{r}{2}(1- g_{\mathcal{Q}} (\nabla r, \nabla r ))\end{align}
	as well as  $\kappa$ and $\iota$ as 
	\begin{align} \label{kappa}
	\kappa := \frac{-\Omega^2}{2\partial_u r}\in \mathbb{R} \cup \{ \pm \infty\} ,
	\end{align}
	\begin{align} \label{iota}
	\iota := \frac{-\Omega^2}{2\partial_v r} \in  \mathbb{R} \cup \{ \pm \infty\}. 
	\end{align}
	
	\paragraph{Electromagnetic fields on $\mathcal Q$.}
	In what follows, we will abuse notation and denote by $F$ the 2-form over $\mathcal{Q}$ that is the push-forward by $\Pi$ of the electromagnetic 2-form originally on $\mathcal M$, and similarly for $A$ and $\phi$.
	If view of the $SO(3)$ symmetry of the potential $A$   we have (see \cite{Kommemi}) that $F$ has the  form
	\begin{align}\label{eq:Fnonlinear} F= \frac{Q}{2 r^{2}}  \Omega^{2} \d u  \wedge \d v ,\end{align}
	where $Q$ is a scalar function called the electric charge. From  $F=\d A$ we also obtain
	$$ [D_u,D_v] = iq_0 F_{u v}   =\frac{iq_0Q \Omega^2}{2r^2}  .$$
	Now we introduce the modified Hawking mass $\varpi$ that involves the charge $Q$:	\begin{align} \label{electromass}	\varpi := \rho + \frac{Q^2}{2r}.	\end{align} 
	An elementary computation relates  geometric quantities (on the left) to coordinate-dependent ones (on the right) gives:	\begin{align} \label{mu}	1 - \frac{2\rho}{r} = \frac{-4 \partial_u r \partial_v r}{\Omega^2} = \frac{-\Omega^2 }{ \iota \kappa}= 1- \frac{2 \varpi}{r}+ \frac{Q^2}{r^2}.	\end{align} 
	We also define the quantity \begin{align} \label{K.def}
	2K:=\frac{1}{r^2} \left(\varpi-\frac{Q^2}{r}\right),
	\end{align} and notice that, if $\varpi=M$ and $Q=e$, then $2K(r_{\pm})= 2K_{\pm}$.  Further, we introduce the following notation, first used by Christodoulou:
	$$ \lambda = \partial_v r ,$$
	$$ \nu = \partial_u r. $$

	Finally, note that \eqref{E4}--\eqref{E5} are invariant under electromagnetic gauge transformations (see \cref{{gaugechoice.section}}) and two solutions $(\phi, A)$ which differ by a gauge transformation represent the same physical behavior. 
	An equivalent formulation to express this gauge freedom is to consider electromagnetism as a $U(1)$ gauge theory with principal $U(1)$-bundle $\pi\colon P\to M$: the charged scalar field is a global section of the associated complex line bundle $P\times_\rho \mathbb C$ through the representation $\rho$ such that $\phi$ corresponds to an equivariant $\mathbb C$-valued map on $P$, i.e.\ $\phi(p g ) = \rho(g)^{-1} \phi$.  The representation $\rho: U(1) \to \textup{GL}(1, \mathbb C  )$ models the coupling of the scalar field and electromagnetic field. We refer to \cite[Section~1.1]{Kommemi} and stick to our equivalent and more concrete formulation of the EMKG system.

	\paragraph{$C^0$-admissible spacetimes and extensions.}
	Lastly, we define the notion of a $C^0$-admissible extension of the metric  (inspired from \cite[Definition~A.3]{gmadmissible}). For the sake of brevity and concreteness we will give neither the most geometric nor the most general formulation and we refer to \cite{gmadmissible} and \cite{MoiChristoph2} for further details.
	\begin{defn} \label{C0admissible}
		We call  $( \mathcal {M}, {g})$ an admissible $C^0$ spherically symmetric spacetime if the following holds.
		\begin{enumerate}
			\item $ \mathcal M$ is a $C^1$-manifold diffeomorphic to $\mathcal Q \times \mathbb S^2$ for an open domain $\mathcal Q \subset \mathbb R^2$,
			\item $g$ is an admissible $C^0$ spherically symmetric Lorentzian metric in the sense that   for a diffeomorphism $\Phi \colon \mathcal M \to \mathcal Q \times \mathbb S^2$ there exist $C^1$-coordinates $(u,v)$ on $\mathcal Q$ in which the metric $\Phi^\ast (g)$ on $\mathcal Q\times \mathbb S^2$ can be written as
			\begin{align}\label{eq:defnofg}
			\Phi^\ast (g) = - \frac{\Omega^2 }{2} (\d u\otimes \d v+ \d v\otimes \d u)  + r^2 g_{\mathbb S^2},
			\end{align}
			where  $g_{\mathbb S^2}$ is the standard round metric on $\mathbb S^2$ and $\Omega^2, r^2 \colon\mathcal Q \to (0,+\infty)$ are continuous.
			\item If $(\tilde u, \tilde v)$ is another $C^1$-coordinate system such that \eqref{eq:defnofg} holds with $\tilde \Omega^2$ in place of $\Omega^2$, then $\tilde u  = U(u) $ and $\tilde v = V(v)$ for some unique and strictly monotonic   $C^1$-functions $U, V$. 
		\end{enumerate}
	\end{defn}
	\begin{rmk} The pair $(u,v)$ as above is called a null coordinate system. In the case where the metric $g$ is locally Lipschitz such null coordinates always exist.
		Since we merely consider $C^0$ metrics, in our definition of admissible $C^0$ metric we additionally impose the existence and uniqueness (up to re-scaling) of such null coordinates.  
	\end{rmk}

	\begin{defn} 
		Let $( \mathcal M,g)$ and $(\tilde{\mathcal M},\tilde g)$ be time-oriented admissible $C^0$ spherically symmetric spacetimes. We say that $(\tilde{\mathcal M} ,\tilde g)$ is an admissible $C^0$ spherically symmetric future extension if
		\begin{enumerate}
			\item there exists a $C^1$ embedding $i: \mathcal M \rightarrow \tilde{\mathcal M}$ which is also a time-orientation-preserving isometry,
			\item there exists $p \in \tilde{\mathcal M}-i(\mathcal M)$ which is to the future of $i(\mathcal M)$.
		\end{enumerate}
	\end{defn}

	\subsection{Electromagnetic gauge choices}\label{gaugechoice.section}
	
	As remarked above, for a fixed metric $g$, the Maxwell--Klein--Gordon system of equations \eqref{E4}--\eqref{E5} is invariant under the following \textit{gauge} transform: 	
	\begin{equation} \label{gaugetransform1}
	\phi \rightarrow  \tilde{\phi}=e^{-i q_0 f } \phi,
	\end{equation}
	\begin{equation} \label{gaugetransform2}
	A \rightarrow  \tilde{A}=A+ \d f ,\end{equation}
	where $f$ is a smooth real-valued function.  Notice that for  $\tilde{D}:=\nabla+\tilde{A}$, we have 
	$$ \tilde{D} \tilde{\phi} = e^{-iq_0f} D\phi.$$
	Therefore the quantities $|\phi|$ and $|D\phi|$ are gauge invariant.
	In \cref{nonlinearsection}, we will use that these gauge-invariant quantities satisfy the following estimates which are an immediate consequence of the fundamental theorem of calculus, see e.g.\ \cite[Lemma~2.1]{gajicluk}.
	In any $(u,v)$-coordinate system and  for $u\geq u_1$ and $v \geq v_1$:
	\begin{align}
	&|f(u,v)| \leq  |f(u_1,v)| + \int_{u_1}^{u} |D_u f|(u',v) \d u,
	\\
	&|f(u,v)| \leq  |f(u,v_1)| + \int_{v_1}^{v} |D_v f|(u,v') \d v',\end{align} for any sufficiently regular function $f(u,v)$.

	Although we will mainly estimate gauge-invariant quantities, to set up the characteristic data it is useful to fix an electromagnetic gauge. For the analysis of the nonlinear system in  \cref{nonlinearsection} in double null coordinates $(u,v)$ we will impose \begin{equation} 
	\label{GaugeAv}
	A_v \equiv 0.
	\end{equation} 
	In this gauge, the condition $F=\d A$ from \eqref{E4} can be written (in any $(u,v)$ coordinate system) as
	\begin{align} 
	\partial_v A_u =- \frac{Q \Omega^2}{2r^2}.
	\end{align}  
	To estimate the dynamics of $A=A_u \d u$ in the coupled system it is useful to define a background electromagnetic field $A^{RN}$ which is governed by the fixed Maxwell form $F=F_{RN}$ as in \eqref{F.RN} on a fixed Reissner--Nordström background with mass and charge $(M,e)$. Using  coordinates $(u,v)$ as defined in \eqref{eq:coordinatesuvonRN} we impose the gauge 
	\begin{align}
	A_v^{RN} \equiv 0
	\end{align}
	such that $F_{RN} = \d A_{RN}$ becomes 
	\begin{align}\label{eq:backgroundA}
	\partial_v A_u^{RN} =- \frac{e \Omega^2_{RN}(u,v)}{2 r^2_{RN}(u,v)}.
	\end{align} 
	Moreover, we choose the normalization for $A^{RN}$ to obtain
	\begin{equation} \label{ARN}
	A^{RN}= \left(-\frac{e}{r_{RN}}+ \frac{e}{r_{+}} \right) \d u
	\end{equation} 
such that the $1$-form $A^{RN}$ extends smoothly to the right event horizon $\mathcal H^+$ on Reissner--Nordström.

	For the linear theory in \cref{linearsection} we will work with the $t$-Fourier transform. It that context it is useful to use a gauge which is \textbf{different} from \eqref{ARN} and which is given (see already \eqref{eq:gauge_lienartheory}) by 
	\begin{equation} \label{gauge_linear}
	A'_{RN}= \left(\frac{e}{r_{RN}}- \frac{e}{r_+}\right) \d t=  \left(\frac{e}{r_+}- \frac{e}{r_{RN}}\right) \frac{  \d u- \d v}{2}.
	\end{equation}

	\subsection{The Einstein--Maxwell--Klein--Gordon system in null coordinates} \label{eqcoord}
	
 We now express the EMKG system \eqref{E1}--\eqref{E5} in a double coordinate system $(u,v)$ on $\mathcal Q$  using the electromagnetic gauge \eqref{GaugeAv}. The unknown functions  ($r,\Omega^2,A_u,Q,\phi$) on $\mathcal Q$ are subject to the following system. 
	\begin{align}\label{Radius}
	&\partial_{u}\partial_{v}r =\frac{- \Omega^{2}}{4r}-\frac{\partial_{u}r\partial_{v}r}{r}
	+\frac{ \Omega^{2}}{4r^{3}} Q^2 +  \frac{m^{2}r }{4} \Omega^2 |\phi|^{2}  =-\frac{\Omega^2}{2}\cdot 2K +    \frac{m^{2}r }{4} \Omega^2 |\phi|^{2}, \\
	\label{Omega}
	&	\partial_{u}\partial_{v} \log(\Omega^2)=-2\Re(D_{u} \phi \partial_{v}\bar{\phi})+\frac{ \Omega^{2}}{2r^{2}}+\frac{2\partial_{u}r\partial_{v}r}{r^{2}}- \frac{ \Omega^{2}}{r^{4}} Q^2,
	\end{align}
	the Raychaudhuri equations: 
	\begin{align}\label{RaychU}&\partial_{u}\left(\frac {\partial_{u}r}{\Omega^{2}}\right)=\frac {-r}{\Omega^{2}}|  D_{u} \phi|^{2}, \\ \label{RaychV}
	&\partial_{v}\left(\frac {\partial_{v}r}{\Omega^{2}}\right)=\frac {-r}{\Omega^{2}}|\partial_{v}\phi|^{2},\end{align}
	the charged and massive Klein-Gordon equation: 
	\begin{align}\label{Field}
	\partial_{u}\partial_{v} \phi =-\frac{\partial_{u}\phi\partial_{v}r}{r}-\frac{\partial_{u}r \partial_{v}\phi}{r} +\frac{ q_{0}i \Omega^{2}}{4r^{2}}Q \phi
	-\frac{ m^{2}\Omega^{2}}{4}\phi- i q_{0} A_{u}\frac{\phi \partial_{v}r}{r}-i q_0 A_{u}\partial_{v}\phi,\end{align} 
	and the Maxwell equations: 
	\begin{align} \label{chargeUEinstein}
	&	\partial_u Q = -q_0 r^2 \Im( \phi \overline{ D_u \phi}),\\ \label{ChargeVEinstein}
	&	\partial_v Q = q_0 r^2 \Im( \phi \overline{\partial_v \phi}).
	\end{align}
	Finally, $F=\d A$ reads
	\begin{align} \label{eq:maxwellAu}
	\partial_v A_u = \frac{-Q \Omega^2}{2r^2}.
	\end{align}
	Note that \eqref{chargeUEinstein} and \eqref{ChargeVEinstein} can be equivalently formulated introducing the quantity $\psi:=r\phi$ as  	\begin{align} \label{chargeUEinstein2}
	&	\partial_u Q = -q_0 \Im( \psi \overline{ D_u \psi}) ,
	\\ \label{ChargeVEinstein2}
	&	\partial_v Q = q_0 \Im( \psi \overline{\partial_v \psi}). 	\end{align}
	Further, \eqref{Radius} is equivalent to 	\begin{align}\label{Radius2}\partial_{u}(r\partial_{v}r) =\frac{- \Omega^{2}}{4}
	+\frac{ \Omega^{2}}{4r^{2}} Q^2 +  \frac{m^{2}r^2 }{4} \Omega^2 |\phi|^{2}. \end{align}
	We can also rewrite \eqref{Field} to control $|\partial_{v} \phi|$ more easily: 
	\begin{align}\label{Field2}
	D_u \partial_v \phi=	e^{-iq_0 \int_{u_{0}}^{u}A_u} \partial_{u}( e^{iq_0 \int_{u_{0}}^{u}A_u}\partial_{v} \phi) =-\frac{\partial_{v}r D_u\phi}{r}-\frac{\partial_{u}r \partial_{v}\phi}{r} +\frac{ q_{0}i \Omega^{2}}{4r^2}Q \phi
	-\frac{ m^{2}  \Omega^{2}}{4}\phi. \end{align}
	We also have (recalling the notation $\psi =r\phi$):
	\begin{align}\nonumber
	e^{-iq_0 \int_{u_{0}}^{u}A_u}\partial_u (e^{i q_0 \int_{u_{0}}^{u}A_u}\partial_v \psi) &=D_{u}(\partial_{v} \psi) \\
	& =	\frac{- \Omega^{2} \phi}{4r}-\frac{\partial_{u}r\partial_{v}r \cdot \phi}{r}
	-\frac{ \Omega^{2} \phi }{4r^{3}} Q^2 +  \frac{m^{2}r }{4} \Omega^2  \phi |\phi|^{2}  -\frac{ m^{2}  \Omega^{2} r }{4}\phi-\frac{ q_{0}i \Omega^{2}}{4r}Q \phi, \label{Field5} \end{align}
	and
	\begin{align}\label{Field4}
	\partial_{v}(D_{u} \psi) =	\frac{- \Omega^{2} \phi}{4r}-\frac{\partial_{u}r\partial_{v}r \cdot \phi}{r}
	+\frac{ \Omega^{2} \phi }{4r^{3}} Q^2 +  \frac{m^{2}r }{4} \Omega^2  \phi |\phi|^{2}  -\frac{ m^{2}  \Omega^{2} r }{4}\phi-\frac{ q_{0}i \Omega^{2}}{4r}Q \phi. \end{align}
	\section{Setup of the characteristic data and the oscillation condition}  
	\label{coordinatechoice.section}
	We  first fix the following arbitrary quantities \begin{align} 
	\label{eq:defnem}
	&  \text{\textbullet} \text{ subextremal charge and mass parameters } 0<|e|<M,\\ \label{eq:defns}
	& \text{\textbullet}  \text{ a decay rate }\frac 34 < s \leq 1,\\ \label{eq:defnd}
	& \text{\textbullet} \text{ constants } D_1, D_2 >0.
	\end{align}
	These quantities will be kept fixed from now onward.
	\subsection{Characteristic cones \texorpdfstring{ $\underline{C}_{in}$, $\HH$}{Cin,H+} and underlying manifold \texorpdfstring{$\mathcal Q^+$}{Q+}}
	Our yet-to-be-constructed spacetime of study will be the future domain of dependence $\mathcal{Q}^+$ of the characteristic set $\underline C_{in}\cup_p \HH \subset \mathbb R^{1+1}$, where $\HH:=\{ U=0,\ v_0  \leq v < +\infty\}$ and $\underline C_{in}:=\{ 0 \leq U \leq U_s,\ v =v_0 \}$ which meet transversely at the common boundary point $p := \{ U = 0,v=v_0\}$.  Here, we use the convention that  $f\in C^1(\HH)$ means   that $f\in C^1((v_0,\infty))\cap C^0([v_0,\infty))$ with the property that that $\partial_v f$ extends continuously to $v_0 = \partial\HH$.  Analogously, we define  $C^1(\underline C_{in})$. Moreover, we  say that  $ f\in C^1(\underline C_{in}\cup_p \HH)$ if $f$ is continuous on $ \underline C_{in}\cup_p \HH$ and   $f|_{\HH} \in C^1(\HH), f|_{\underline C_{in}} \in C^1(\underline C_{in})$. In particular, note that if $f_1\in C^1(\HH)$ and $f_2\in C^1(\underline C_{in})$ satisfy $f_1(p) = f_2(p)$, then they define a function in $C^1(\underline C_{in}\cup_p \HH)$. Analogously, we define $C^k$ for $k\geq 2$.  We denote $\mathcal Q^+ := \{ 0\leq U \leq U_s,\ v_0 \leq v<+\infty \}$. Here $v_0  = v_0(M,e,s,D_1) \geq 1$ only depends on $M,e,s,D_1$ and  $U_s=U_s(M,e,s,D_2,D_1)$ only depends on     $M,e,s,D_2,D_1$---both of which  will be determined in \cref{prop:constraints} below.
	\paragraph{A new coordinate $u$.}
	We will make use of other coordinates $(u,v)$ on $\mathcal Q^+-\HH$ given by $u:=u(U), v=v$, where $u(U)$ is the function given through the condition \eqref{eq:Ucoordinate} and $(M,e)$ are as in \eqref{eq:defnem}. We also denote $u_s:= u(U_s)$. 
	\paragraph{An additional electromagnetic gauge freedom.}
	At this point we recall our global electromagnetic gauge choice $A_v\equiv 0$ in \cref{gaugechoice.section}. An additional electromagnetic gauge freedom we have is the specification of $A_U$ (or equivalently  $A_u$) on $\underline C_{in}$. We impose that $A_U$ on $\underline{C}_{in}=\{ 0 \leq U \leq U_s,\ v=v_0\}$ satisfies 
	\begin{align}\label{eq:gaugeforAU}
	A_U (U,v_0) =   \left(  -\frac{e}{r_{RN}(U,v_0)} + \frac{e}{r_+(e,M)}\right) \frac{\d u }{\d U} (U)= 2  \left(  -\frac{e}{r_{RN}(U,v_0)} + \frac{e}{r_+(e,M)}\right) \Omega^{-2}_{RN}(U,v_0)
	\end{align}
	where we used \eqref{eq:Ucoordinate} for the second identity and thus
	\begin{align}\label{gaugeAponctuelle}A_u(u,v_0) =    -\frac{e}{r_{RN}(u,v_0)} + \frac{e}{r_+(e,M)}.\end{align} Here, $r_{RN}$ is the $r$-value on Reissner--Nordström with parameters $(M,e)$ as given in \eqref{eq:defnem} and $r_+(M,e)=M^2+\sqrt{M^2-e^2}$.
	\subsection{Coordinate gauge conditions on \texorpdfstring{$\HH$ and $\underline{C}_{in}$}{H+ and Cin}} 
	
	On $\HH=\{ U=0,\ v_0 \leq v\}$ we will impose the gauge condition
	\begin{align} \label{gauge1}
	\frac{\partial_U r (0, v)}{ \Omega^2_H( 0, v)} = -\frac 12
	\end{align} 
	and on $\underline{C}_{in}=\{ 0 \leq U \leq U_s,\ v=v_0\}$ we will impose 
	\begin{align} \label{eq:gaugeforU}
	\partial_U r = -1. 
	\end{align} 
	
	\subsection{Free data \texorpdfstring{$\phi\in C^1(\underline{C}_{in} \cup_p \HH)$}{phi in C1(Cin cup H+)} with slow decay on \texorpdfstring{$\HH$}{H+} and construction of \texorpdfstring{$r,Q,\Omega_H^2$}{r,Q,OmegaH2}}
	Having set up the gauges we will now---additionally to the free prescription of   $0<|e|<M$ in \eqref{eq:defnem}---freely prescribe  data for $\phi$ on $  \underline{C}_{in}\cup_p \HH $. We recall \eqref{eq:defns} and \eqref{eq:defnd} and define the class of slowly decaying data $\Sl$ on the event horizon $\HH$ in the following. In order to  highlight that the definition does not depend on the gauge choice for the electromagnetic potential $A$ we formulate it  in  a gauge-invariant form (although we have already fixed the gauge $A_v\equiv0$ in \eqref{GaugeAv} and \eqref{gaugeAponctuelle}). 
	\begin{defn}[Set of slowly decaying data {$\Sl$}] \label{slowdecay.def} 
		We say that $\phiH\in C^1(\HH,\mathbb C)$ is \textbf{slowly decaying}, denoted $\phiH\in \Sl$, if  \begin{align} \label{polynomialdecay}
		|\phiH|(v)+ |D_v \phiH|(v) \leq D_1  v^{-s}
		\end{align} 
		for all	  $v \in  \HH $, where we recall $\frac{3}{4} < s \leq 1$ was introduced and fixed in \eqref{eq:defns}, and $D_1>0$ was introduced and fixed in  \eqref{eq:defnd}.
	\end{defn}

	\noindent   Similarly, on $\underline{C}_{in}$ we will also impose arbitrary (up to the corner condition) data  $\phi_{in}\in C^1(\underline{C}_{in})$ satisfying 
	\begin{align}\label{eq:estiamteonphiin}
	|D_U \phi_{in}|\leq  D_2. 
	\end{align}
	
	We will now finally conclude the setup of the initial data, where we recall that we freely prescribed subextremal $e,M$ and the scalar field $\phi$ on $\underline{C}_{in}\cup_p \HH$. In particular,  using standard results about o.d.e.'s  (recall that $s > \frac 34$; actually $s>\frac 12$ is sufficient to prove \cref{prop:constraints}) we obtain
	\begin{prop}\label{prop:constraints}
		There exist   $v_0(M,e,s,D_1)\geq 1$ sufficiently large and $U_s(M,e,s,D_2,D_1)>0$ sufficiently small such that the following holds true. Let   $\phiH\in \Sl$ and  $\phi_{in}\in C^1(\underline{C}_{in})$ satisfying \eqref{eq:estiamteonphiin} with $\phiH(p)=\phi_{in}(p)$ be arbitrary. Then, there exist unique solutions $r\in C^2 (\underline{C}_{in}\cup_p \HH )$, $\Omega_H \in C^1(\underline{C}_{in}\cup_p \HH)$ and $Q \in C^1(\underline{C}_{in} \cup_p \HH)$  of the o.d.e.\ system  consisting of the Raychaudhuri equation \eqref{RaychV}, equation \eqref{ChargeVEinstein}, the equation \eqref{Radius} using \eqref{gauge1} on $\HH$ and the o.d.e.\ system consisting of   \eqref{eq:gaugeforU}, \eqref{RaychU} and \eqref{chargeUEinstein} on $\underline{C}_{in}$ such that  
		\begin{align}
		& \lim_{v\to +\infty} r(0,v) = r_+(M,e)=M+\sqrt{M^2+e^2},\\
		& \lim_{v\to +\infty} Q(0,v) = e.
		\end{align}
		Moreover, $\HH$ is affine complete, i.e.\ $\int_{v_0}^{+\infty} \Omega^2_H(0,v) \d v = +\infty$. 
	\end{prop}
	This shows that our free data $(e,M,\phi)$ and the gauge conditions give rise to a full set of data $(r,Q,\Omega_H^2, \phi)$ on $\underline{C}_{in}\cup_p \HH$ satisfying the constraint equations.

	Further, note that \eqref{gauge1}  implies that  $$ \kappa_{|\mathcal{H}^+} \equiv 1$$ in view of \eqref{kappa}.
	
	\begin{rmk}We also  associate to $(u,v)$ a lapse function  $\Omega^2$ through 
		\begin{align}
		\Omega^2 := \Omega^2_{H} \frac{\d U}{\d u} \label{eq:relationomegas}
		\end{align}
		such that $\Omega^2 = -2 g(\partial_u,\partial_v)$ and $\Omega_H^2 = - 2g(\partial_U, \partial_v)$ once the spacetime is constructed.
	\end{rmk}
	\begin{rmk}
		In \cref{W11.main.thm} we will introduce generic properties of functions in $\Sl$. We  remark that $\Sl$ is the ball of size $D_1$ in the Banach space 
		\begin{align}\label{eq:defnsl0}
		\Sl_0 := \{ f\in C^1(\HH;\mathbb C): \underset{v\geq v_0}{\sup}(| v^s f| + |v^s D_vf | ) <+\infty \} .\end{align} In \cref{W11.main.thm} (more precisely in \cref{linear.cor.W11}) we identify a (exceptional) subspace $H_0 \subset \Sl_0$ of infinite co-dimension. We then call functions $\phiH\in\Sl $ generic if $\phiH \in \Sl-H$, where $H:= H_0 \cap \Sl$. \end{rmk}
	\subsection{Definitions of the oscillation spaces \texorpdfstring{$\OO$, $\OOp$, $\OOpp$}{O,O',O''}}
	\label{sec:oscillationspaces}
	We now define the subsets $\OO,\OOp,\OOpp\subset \Sl$ of slowly decaying data on the event horizon    describing the oscillation conditions. In order to  highlight that the definitions do not depend on the gauge choice for the electromagnetic potential $A$ we formulate them in  a gauge-invariant form (although we have already fixed the gauge $A_v\equiv0$ in \eqref{GaugeAv} and \eqref{gaugeAponctuelle}).

	\begin{defn}[Qualitative oscillation condition {$\mathcal O$}] \label{defn:qualoscillationcond} A function  $\phiH \in \Sl$   is said to satisfy the \textbf{qualitative oscillation condition}, denoted $\phiH \in \OO$, if the following qualitative 
		condition holds
		\begin{equation} \label{oscillation.condition}\limsup_{v \to +\infty} \left| \int_{v_0}^{v}  \phiH(v')  e^{i (\omer  v'+ q_0 \sbr(v'))} e^{iq_0\int_{v_0}^{v'} (A_v)_{|\mathcal{H}^+}(v'')\d v''}\d v' \right| < +\infty,
		\end{equation} 
		for all  $D_{br}>0$ and all functions $\sbr \in C^2([v_0,+\infty),\RR)$ satisfying  
		\begin{align} \label{sigma_err1}
		&|\sbr(v)| \leq D_{br} \cdot( v^{2-2s} 1_{s<1} + \log(1+v) 1_{s=1}),
		\\
		\label{sigma_err2}
		&|\sbr'(v) | +  |\sbr'' (v) | \leq D_{br}  v^{1-2s},
		\end{align} for  all $v\geq v_0$, where we recall that $v_0(M,e,s,D_1)>1$.
		
	\end{defn}
	
	We will also denote $\NO:= \Sl-\OO$ the space of $\phiH \in \Sl$ violating \eqref{oscillation.condition}.
	\begin{defn}[Strong qualitative oscillation condition {$\OOp$}] \label{defn:quantitativeosc}
		A function $\phiH \in \OO$   is said to satisfy the  \textbf{strong qualitative oscillation condition}, denoted    $\phiH\in \OOp$, if the limit \begin{equation} \label{eq:oscillation.condition}\lim_{v \to +\infty} \left| \int_{v_0}^{v}  \phiH(v')  e^{i (\omer  v'+ q_0 \sbr(v'))} e^{i q_0\int_{v_0}^{v'} (A_v)_{|\mathcal{H}^+}(v'')\d v''}\d v' \right|  			\end{equation}
		exists (and is finite) for all $D_{br}>0$ and all functions $\sbr\in C^2([v_0,+\infty),\RR) $ satisfying \eqref{sigma_err1} and \eqref{sigma_err2}. \end{defn}
	
	\begin{defn}[Quantitative oscillation condition   {$\OOpp$}]
		\label{defn:strongoscilcond} 
		A function  $\phiH\in \OOp$ is said to satisfy the \textbf{quantitative oscillation condition}, denoted $\phiH \in \OOpp$ if for all $D_{br}>0$, there exist $E_{\OOpp}(D_{br})>0$, $\eta_0(D_{br})>0$ such that\begin{equation} \label{eq:strong.oscillation.condition}
		\left|\int_{v}^{+\infty}e^{i (\omer \cdot v'+ q_0\sbr(v'))} e^{i q_0\int_{v_0}^{v'} (A_v)_{|\mathcal{H}^+}(v'')\d v''}\phiH(v') \d v' \right| \leq  E_{\OOpp} \cdot v^{s-1-\eta_0},
		\end{equation}  for all $v\geq v_0$  and  all functions $\sbr \in C^2([v_0,+\infty),\RR) $ satisfying \eqref{sigma_err1} and \eqref{sigma_err2}.  
	\end{defn}

	\begin{rmk} Note that we have by definition the    inclusions: $  \OOpp \subset \OOp \subset \OO \subset \Sl$. Moreover, note that $ \OOpp \not\subset L^1([v_0,+\infty))$; more generally, a generic  function of $\OOpp$ is not in $L^1([v_0,+\infty))$.
	\end{rmk}
	\begin{rmk} The condition \eqref{oscillation.condition} and its   stronger versions \eqref{eq:oscillation.condition}, \eqref{eq:strong.oscillation.condition}  guarantee sufficiently robust  non-resonant oscillations. These conditions are sufficient (our proof also suggests that they are necessary to some extent)   to avoid that the backreaction of the Maxwell field  (which, as we will show, creates unbounded but sublinear oscillations $\sbr$ obeying \eqref{sigma_err1}, \eqref{sigma_err2})   turns \textbf{linearly non-resonant} profiles into \textbf{nonlinearly resonant} profiles, see last paragraph of \cref{nonlinear.intro} for a discussion.
	\end{rmk}
	\begin{rmk} In the uncharged case $q_0=0$, the backreaction of the electric field is absent. In this case note that \eqref{oscillation.condition} simplifies to a ``finite average'' condition.
	\end{rmk}

	\section{Precise statements of the main theorems and outline of their proofs} \label{precise.section}

	\subsection{Existence of a Cauchy horizon \texorpdfstring{$\CH\neq \emptyset$}{CH=/=emptyset} and quantitative estimates in the black hole interior from \texorpdfstring{\cite{Moi}}{VdM '18}} \label{existence.CH.main}
	In \cite{Moi}, the second author proved (among other results) that spherically symmetric EMKG black holes converging to a subextremal Reissner--Nordstr\"{o}m admit a null boundary $\CH\neq \emptyset$ that we still call a Cauchy horizon. The proof of this main result in \cite{Moi} required many quantitative estimates that will be useful in the analysis of the current  paper.
	
	\setcounter{thm}{0}
	\begin{thm}[\cite{Moi}] \label{CH.stab.thm} Consider the characteristic data on $\underline C_{in} \cup_p \HH$ as described in \cref{coordinatechoice.section} and fix the electromagnetic gauge  \eqref{GaugeAv} as in \cref{gaugechoice.section}. 
		Let   $\phiH\in \Sl$ be arbitrary, and let  $\phi_{in}\in C^1(\underline{C}_{in})$ satisfying \eqref{eq:estiamteonphiin} with $\phi_{in}(p)=\phiH(p)$ be arbitrary.

		Then, by   choosing  $U_s(M,e,s,D_2,D_1)>0$ potentially   smaller, the characteristic data give rise to the unique  $C^1$ Maximal Globally Hyperbolic Development $(r,\Omega^2_H,A,Q,\phi)$   on  $\mathcal Q^+$   solving the EMKG system of \cref{eqcoord}.    In addition, an (ingoing) null boundary $\CH \neq \emptyset$ (the Cauchy horizon) can be attached to $\mathcal Q^+$ on which $r$ extends as a continuous function $r_{CH}$ which remains bounded away from zero, depicted in Penrose diagram \cref{fig:cauchyhorizonexists}. Note that $(r,\Omega^2_H,A,Q,\phi)$ on $\mathcal Q^+$  defines $(\mathcal M,g,A,F,\phi)$ which solves \eqref{E1}--\eqref{E5}.
		
		Moreover, all the quantitative estimates stated in   \cref{prop.HH.estimates.Moi},   \cref{prop.RS.estimates.Moi},   \cref{prop.NN.estimates.Moi},   \cref{prop.EB.estimates.Moi} and   \cref{prop.LB.estimates.Moi} are satisfied.
		
		If we additionally assume fast decay (i.e.\ $\phiH$ satisfies \eqref{polynomialdecay} for $s>1$), then $\phi$ is in $W^{1,1}_{loc} \cap L^{\infty}$ at the Cauchy horizon $\CH$ and extends as a continuous function across the Cauchy horizon $\CH$. Moreover, in this case, the metric admits a $C^0$-admissible extension $\tilde{g}$ across the Cauchy horizon $\CH$ in the sense of Definition \ref{C0admissible} and $\tilde{g}$ has locally integrable Christoffel symbols.
	\end{thm}
	\begin{rmk}
		We note that the above \cref{CH.stab.thm} showing $\CH \neq \emptyset$, together with all the quantitative estimates stated \cref{prop.HH.estimates.Moi},   \cref{prop.RS.estimates.Moi},   \cref{prop.NN.estimates.Moi},   \cref{prop.EB.estimates.Moi} and   \cref{prop.LB.estimates.Moi}, actually holds under the weaker assumption of decay rate $s>\frac 12$ as opposed to $s>\frac 34$, see \cite{Moi}. For the purpose of extendibility across the Cauchy horizon $\CH$ for oscillating data as stated in our main result below, the decay assumption $s>\frac 34  $ is needed and appears to be crucial, see the discussion in \cref{Strategy.section.bulk}.
	\end{rmk}
	
	\subsection{\texorpdfstring{\cref{main.theorem.intro}}{Theorem I(i)}: Scalar field boundedness and continuous extendibility for oscillating data}\label{sec:preciseA1}
	In this section we give the precise version of   \cref{main.theorem.intro} which is proved as \cref{boundedness.cor} in \cref{boundedness.combining.section}. 
	\begin{theoa}[\textbf{Boundedness}] \label{main.theorem} Let the assumptions of   \cref{CH.stab.thm} hold.  \begin{enumerate}
			\item If $\phiH$ satisfies the qualitative oscillation condition $\phiH \in \OO$ (see  \cref{defn:qualoscillationcond}), then
			\begin{align}\label{phi.bounded.main.thm}			   \sup_{(u,v) \in \mathcal Q^+} |\phi(u,v)| <+\infty.
			\end{align}

			\item If $\phiH$ satisfies the strong qualitative oscillation condition $\phiH \in \OOp$ (see \cref{defn:quantitativeosc}), then \eqref{phi.bounded.main.thm} is true and moreover $\phi$ admits a continuous extension to $\CH$ and $g$ admits a $C^0$-admissible extension to $\CH$ in the sense of Definition \ref{C0admissible}. In particular, $g$ is continuously extendible.
			
			\item If $\phiH$ satisfies the quantitative oscillation condition $\phiH \in \OOpp$ (see  \cref{defn:strongoscilcond}), then \eqref{phi.bounded.main.thm} is true, $\phi$ admits a continuous extension to $\CH$ and $g$ admits a $C^0$-admissible extension to $\CH$. Moreover, $Q$ is uniformly bounded on $\mathcal{Q}^+$ and admits a continuous extension to $\CH$. Further, there exists a constant $ \tilde C=\tilde C(D_1,D_2,E_{\OOpp},\eta_0,e,M,m^2,q_0,s) >0$ 
			such that for all $(u,v) \in \mathcal{LB} \subset  \mathcal{Q}^+$:\begin{align} \label{phi.decay.main.thm}
			|\phi|(u,v) & \leq \tilde C  \cdot |u|^{s-1-\eta_0},\\
			\label{Q.decay.main.thm}
			|Q-e|(u,v) 	& \leq \tilde  C \cdot |u|^{-\eta_0},
			\end{align} where $E_{\OOpp} = E_{\OOpp}(D_{br})>0$, $\eta_0 = \eta_0(D_{br})>0$ are as in \eqref{eq:strong.oscillation.condition} and  ${D_{br}:=D_{br}(D_1,D_2,e,M,m^2,q_0,s)>0}$ is defined in the proof of \cref{boundedness.prop}.   Here $\mathcal{LB}$ denotes the late blue-shift region (see \cref{Fig.regions}),  a neighborhood of the Cauchy horizon which is defined in \cref{recall.section}.
		\end{enumerate}
	\end{theoa}

	\subsection{\texorpdfstring{\cref{main.theorem.intro2}}{Theorem I(ii)}: Blow-up in amplitude of the uncharged scalar field for non-oscillating data}
	\label{sec:preciseA2}
	In this section we give the precise version of   \cref{main.theorem.intro2} which is proved as \cref{cor:proofofthmaii} in \cref{blow.up.section}.
	\begin{theodeux}[\textbf{Blow-up}]  \label{main.theorem2}
		Let the  assumptions of \cref{CH.stab.thm} hold and let $q_0=0$ and $m^2 \in \mathbb R_{>0} -{D}(M,e)$, where $D(M,e)$ is the discrete set of exceptional non-resonant masses as defined in \cite[Theorem~7]{Kehle2018}. In addition, assume that      $\phiH$ violates the qualitative oscillation condition as in  \cref{slowdecay.def}, i.e.\ assume that  $\phiH \in \NO :=\Sl-\OO$.
		
		Then, for all $u \leq u_s$, the scalar field blows up in amplitude at the Cauchy horizon $\CH$:
		\begin{align}\limsup_{v\rightarrow+\infty} |\phi|(u,v)=+\infty.\end{align}
	\end{theodeux}
	
	\subsection{\texorpdfstring{\cref{corollary.conj}}{Theorem II}: Falsification of \texorpdfstring{$C^0$-formulation}{C0-formulation} of Strong Cosmic Censorship if \texorpdfstring{\cref{decay.conj}}{Conjecture 2} is true}\label{sec:c0formulation}
	We now give the precise version of \cref{corollary.conj} which is proved as \cref{cor:proofofthmb} in \cref{computation.proof.section}.
	\begin{theob}
		Let the assumptions of \cref{CH.stab.thm} hold. Additionally assume that    \cref{decay.conj} is true, i.e.\ $\phiH$  is given by \eqref{conj.decay.1.eq} (if $q_0=0$, $m^2 > 0$), \eqref{conj.decay.2.eq} (if $q_0\neq0$, $m^2 =0$), or \eqref{conj.decay.3.eq} (if $q_0\neq 0$, $m^2 > 0$) in the $v$-coordinate defined by \eqref{gauge1} and that the generic condition  $ |q_0 e| \neq r_-(M,e) |m|$ holds.
		
		Then  $|\phi|$, $Q$ and the metric $g$ admit a continuous extension to $\CH$ and the extension of $g$ can be chosen to be $C^0$-admissible.

		In the above sense, assuming that \cref{decay.conj} is true, then \cref{C0SCC} is false  for the Einstein--Maxwell--Klein--Gordon system in spherical symmetry.
	\end{theob}
	\subsection{\texorpdfstring{\cref{W11.thm.intro}}{Theorem III}: \texorpdfstring{$W^{1,1}$ blow-up}{W11 blow-up} of the scalar field for non-integrable data}
	\label{sec:w11blowup}
	In this section we give the precise version of  \cref{W11.thm.intro} which is proved  in \cref{W11.blowup.section}. 
	To state the theorem we first define the set \begin{align}  \label{eq:defnofz}\mathcal{Z}_{ \mathfrak{t}}(M,e,q_0,m^2):=\{\omega \in \RR \colon \mathfrak{t}(\omega,M,e,q_0,m^2)=0\} \subset \mathbb R \end{align} which is the zero set of the renormalized transmission coefficient $\mathfrak{t}(\omega)$  defined in \eqref{def.t}. At this point we already note that $\mathcal Z_{\mathfrak t}(M,e,q_0,m^2)$ is discrete and, depending on the parameters $(M,e,q_0,m^2)$, possibly empty. For small $\delta>0 $ we also define the smeared out set $\mathcal Z^\delta_{\mathfrak t}(M,e,q_0,m^2) \subset \mathbb R $ as the set of all $\omega \in \mathbb R $ with $\operatorname{dist}(\omega,\mathcal Z_{\mathfrak t}(M,e,q_0,m^2)) < \delta $. We remark that $\mathcal Z^\delta_{\mathfrak t}(M,e,q_0,m^2) = \emptyset$ if $\mathcal Z_{\mathfrak t}(M,e,q_0,m^2)=\emptyset$.
	
	Associated to $\mathcal Z_{\mathfrak t}^\delta(M,e,q_0,m^2)$ we now define a family (parametrized by $\delta>0$) of Fourier projection operators $P_\delta\colon   f \in L^2([v_0,+\infty)) \mapsto \mathcal{F}^{-1}[  \chi_\delta  \mathcal F[\tilde{f}]] \in L^2(\RR)$, where $\tilde{f} \in L^2(\RR)$ is the extension of $f$ by the zero function for $v<v_0$. Here, $ \chi_\delta(\omega)$ is a family (parametrized by $\delta>0$) of smooth functions which are positive on  $  \mathcal Z^\delta_{\mathfrak t}(M,e,q_0,m^2) $ and vanish otherwise. In the case where $\mathcal Z^\delta_{\mathfrak t}(M,e,q_0,m^2)=\emptyset$, also  $ \chi_\delta \equiv 0$. Further, for the Fourier transform, we use the convention $\mathcal F[\tilde f](\omega) = \frac{1}{\sqrt{2\pi}}\int_{\mathbb R} \tilde f(v) e^{i \omega v} \d v$. 
	Finally, we are in the position to state \cref{W11.main.thm} which is proved in \cref{W11.blowup.section}. The first part is shown as \cref{W11.OO.cor}, the second part is shown as \cref{lem:w11blowupofconjrate}.
	\begin{theoc} \label{W11.main.thm} Let the assumptions of \cref{CH.stab.thm} hold.
		
		\textbf{Part 1.} Let $\phiH \in \Sl - L^1([v_0,+\infty))$ and let at least one of the follow assumptions hold: 
		\begin{enumerate}[a)]
			\item  $P_{\delta} \phiH \in L^1(\RR)$ for some $\delta >0$,
			\item or $0 <  |q_0 e|\leq \epsilon(M,e,m^2)$ for some $ \epsilon(M,e,m^2)>0$ sufficiently small or $q_0=0,\ m^2\notin D(M,e)$.
		\end{enumerate} 
		Then, the scalar field $\phi$ blows up in $W^{1,1}$ along outgoing cones at the Cauchy horizon $\CH$ in the sense that for all $u \leq u_s$  \begin{equation} \label{blowup}
		\int_{v_0}^{+\infty} |D_v \phi|(u,v) \d v = +\infty.
		\end{equation} 
		In particular, for any $q_0 \in \RR$ and $m^2 \geq 0$, the set $H$ of data $\phiH \in \Sl$ for which \eqref{blowup} is not  satisfied for all $u \leq u_s$ is exceptional in the sense that $H= H_0 \cap \Sl$, where $ H_0\subset \Sl_0 $ is a subspace of infinite co-dimension within $\Sl_0$ (recall the definition of $\Sl_0$  from \eqref{eq:defnsl0}). In the above sense, $\Sl-H$ is a generic set and thus $W^{1,1}$-blow-up of the scalar field at the Cauchy horizon $\CH$ is a generic property of the data $\phiH \in \Sl$.
		
		\textbf{Part 2.} Assume that $\phiH$  is given by \eqref{conj.decay.1.eq} (if $q_0=0$, $m^2 > 0$), \eqref{conj.decay.2.eq} (if $q_0\neq0$, $m^2 =0$), or \eqref{conj.decay.3.eq} (if $q_0\neq 0$, $m^2 > 0$) in the $v$-coordinate defined by \eqref{gauge1}.  Assume the conditions  \begin{align}\label{set.condition}
		    \mathcal Z_{\mathfrak t}\cap \varTheta  = \emptyset
		    \end{align} 
		    and
		    \begin{align}(m^2,|q_0 e|)\notin \{0\} \times [0,\frac 12), \end{align}
		where  $\mathcal Z_{\mathfrak t}(M,e,q_0,m^2)$ is defined in \eqref{eq:defnofz} and where
		\begin{align*}
		\varTheta(M,e,q_0,m^2): =\begin{cases}\{ -m, +m \} & \text{if }  q_0=0, m^2\neq 0,\\
		\{ - \frac{q_0 e }{r_+}\}& \text{if } |q_0 e| \geq \frac 12,  m^2=0 ,\\
		\{ -m - \frac{q_0 e}{r_+},  m - \frac{q_0 e}{r_+} \} & \text{if } q_0\neq 0, m^2\neq 0.
		\end{cases} 
		\end{align*}
		Then, the scalar field $\phi$ blows up in $W^{1,1}$ along outgoing cones at the Cauchy horizon $\CH$ i.e.\ \eqref{blowup} holds for all $u \leq u_s$.
		
		Moreover, \eqref{set.condition} is satisfied generically in the sense that for given parameters $m^2\geq 0 ,\ q_0\in \mathbb R$, with   $m^2\neq q_0^2$ , the condition \eqref{set.condition} is satisfied for     $$(M,e)\in \{ (M,e) \in \RR^2,\ 0<|e|<M\}-E_{m^2,q_0}, $$  where $E_{m^2,q_0}\subset \RR^2$ is the zero set of an analytic function.

		In particular, for fixed  $m^2\geq 0 ,q_0\in \mathbb R$  with $ m^2 \neq q_0^2$ and $(m^2,|q_0 e|)\notin \{0\} \times [0,\frac 12)$ and for almost all parameters  $$(M,e)\in \{ (M,e) \in \RR^2,\ 0<|e|<M\},$$  assuming  $\phiH $ is as above, then \eqref{blowup}  holds for all $u \leq u_s$. 
		
	\end{theoc}
	\begin{rmk}
		Note that \eqref{blowup} also implies the blow-up of the spacetime $W^{1,1}$ norm in $(u,v)$ coordinates, i.e.\ for  all $u_1< u_2 \leq u_s$ 
		\begin{align*} \int_{u_1}^{u_2} \int_{v_0}^{+\infty} |D_v \phi|(u,v) \d v \d u = +\infty.
		\end{align*}
	\end{rmk}
	~

	The precise formulation and the proof of \cref{null.contraction.theorem} will be given in our companion paper \cite{MoiChristoph2}.
	\subsection{Outline of the proofs}\label{Strategy.section.bulk}
	
	In this section, we elaborate on Step \ref{step1}, Step \ref{Fourier}, Step \ref{diff.step}, Step \ref{C0step} originally presented in  \cref{Strategy.section}. The reader may wish to come back to the current section while consulting the proofs given in \cref{linearsection} and \cref{nonlinearsection}. For convenience, we will conclude this section with a guide for the reader, see \cref{guide.section}.
	\subsubsection{A first approach in physical space and the difficulties associated to slow decay (Step \ref{step1})} \label{physical.space.intro.section}
	\paragraph{Physical space estimates for the nonlinear problem.} \cref{CH.stab.thm} proving  $\CH \neq \emptyset$ also comes with many quantitative stability estimates (see   \cref{recall.section}) for the nonlinear problem \eqref{E1}-\eqref{E5} under the assumption of slowly decaying $\phiH$ satisfying \eqref{decay.s} on $\HH$ (not only for  $s>\frac{3}{4}$ but also $s>\frac{1}{2}$). These estimates  already proven in \cite{Moi} will be our starting point in \cref{nonlinearsection}. Although these estimates are sharp, they are however not sufficient to prove the boundedness of $\phi$ in amplitude, in view of the slow decay obstruction if $s \leq 1$ as we shall explain below.
	To illustrate our point, we start with one of the main estimates\footnote{The main difficulty in obtaining \eqref{Dvbounded.intro} is  nonlinear in nature: its proof in \cite{Moi} exploits the structure of the Einstein equations to address the delicate issue of controlling the metric for a slow rate $s\leq1$. In contrast, the null condition suffices if $s>1$.} obtained by physical space methods in \cite{Moi}: \begin{equation} \label{Dvbounded.intro}
	|D_v \phi|(u,v) \ls v^{-s}.
	\end{equation}
	\paragraph{Boundedness/continuous extendibility in the integrable case.} In the integrable case $s>1$, integrating \eqref{Dvbounded.intro} gives immediately boundedness  \begin{equation}\label{Dvbounded.intro2}
	\| \phi \|_{L^{\infty}} \lesssim data+  \sup_{u}\| D_v\phi(u,\cdot ) \|_{L^{1}_v} \lesssim   data+   \| \langle v \rangle^{-s} \|_{L^{1}_v} < +\infty,
	\end{equation}
	and also gives the $W^{1,1}$-extendibility of the metric (i.e.\ locally integrable Christoffel symbols). From the estimates giving the $W^{1,1}$-extendibility of the metric, one can immediately deduce the continuous extendibility of the metric (see the discussion in \cref{C0.ext.method.intro}). All the  known previous proofs of continuous extendibility of the metric indeed proceed via this method \cite{JonathanStab,MihalisPHD,KerrStab}.

	\paragraph{Slow decay obstruction in the non-integrable case.} 
	In present paper we however have to deal with the non-integrable case $ s\leq 1$, where we note that the above method fails as $\|  \langle v \rangle ^{-s}  \|_{L^{1}_v}$ (the RHS of \eqref{Dvbounded.intro2}) is infinite, even suggesting that the LHS $\| D_v\phi(u,\cdot ) \|_{L^{1}_v}$ could be infinite as well. Indeed, we  prove   \emph{blow-up} of $\| D_v\phi(u,\cdot ) \|_{L^{1}_v}$ (the so-called $W^{1,1}$ norm on outgoing cones)  for generic data $\phiH \in \Sl$  (\cref{W11.thm.intro}, see \cref{C0.ext.method.intro} for a description of its proof), which illustrates the obstruction to proving boundedness by the standard method previously used in the $s>1$ case.
	
	\paragraph{Summary of the rate numerology.} To summarize, square-integrable decay (i.e.\ \eqref{polynomialdecay} with $s>\frac{1}{2}$) is sufficient to show that the black hole boundary admits a null component $\CH$ (the Cauchy horizon) by \cref{CH.stab.thm}, but is in general insufficient for $W^{1,1}$ extendibility and boundedness of the matter fields and metric coefficients (for which integrable decay, i.e.\  \eqref{polynomialdecay} with $s>1$, is sufficient).  In the rest of the section, we explain how to deal with the broader range $\frac{3}{4} < s \leq 1$ ($s>\frac{3}{4}$ is important for the new nonlinear estimates, see already \cref{C0.ext.method.intro} and \cref{s>3/4}).

	\paragraph{An ingoing derivative estimate.} Yet another particularity of the non-integrable case $s \leq 1$ is that $| D_u \phi|$ \textit{may potentially} blow up in amplitude at the Cauchy horizon \cite{Moi} (there are indeed known examples for which $|D_u\phi|$ blows up, see \cite{Moi4}).
	Nevertheless, assuming $s>\frac{3}{4}$, we show that $D_u(r\phi)$ is \textit{uniformly bounded} (\cref{Dupsi.prop}), although not integrable  i.e.\ we prove that for all $\phiH$ satisfying \eqref{polynomialdecay}: \begin{equation} \label{Dupsi.est.intro}
	|D_u(r\phi)|(u,v) \ls |u|^{-s}.
	\end{equation} Note that, consistently with our result that $|\phi|$ blows up for some data, \eqref{Dupsi.est.intro} cannot be integrated in $u$.
	
	\paragraph{Compensate the failure of integrability with oscillations.} Slow decay of the data, as we explained, leads to a lack of integrability of the metric and fields {derivatives} which are roughly of the form, for $ \frac 34 < s\leq 1$ \begin{equation} \label{non.int.est.intro}
	|D_v \phi| \approx v^{-s},
	\end{equation}  which is not integrable as $v\rightarrow+\infty$ (i.e.\ towards the Cauchy horizon $\CH$). Nevertheless, boundedness of $\phi$ could be obtained \textit{by means of the oscillations},   i.e.\ if we could propagate an estimate of the form   \begin{equation} \label{osc.est.intro}
	D_v \phi \approx e^{i \omega v}  \cdot v^{-s},
	\end{equation}  for some $\omega \in \RR-\{0\}$. However, the propagation of such oscillations, if present on the event horizon characteristic data $\phiH$ requires further estimates in Fourier space that we introduce in the following section.

	\subsubsection{The linear problem (Step \ref{Fourier})} \label{linear.intro}
	
	In this section, we discuss how to prove   boundedness or blow-up of $\phiNl$ solving the linearized equation. This step corresponds to the proof of our main linear result \cref{prop:representation} in \cref{linearsection}.
	\paragraph{Representation formula using the Fourier transform.} For the linear (charged massive) wave equation  $ g_{RN}^{\mu \nu} D^{RN}_{\mu} D^{RN}_{\nu} \phiNl= m^2 \phiNl$  on a fixed subextremal Reissner--Nordstr\"{o}m interior metric \eqref{RN}, the physical space estimates of \cref{physical.space.intro.section} also apply, but a Fourier approach is also possible, taking advantage of the Killing vector field $\partial_t$. 
	Taking the Fourier transform in $t$, the wave equation then reduces to the so-called radial o.d.e.\ (see already \eqref{eq:radialode}). Using this,  we will view aspects of the interior propagation from the event horizon to the Cauchy horizon as a scattering problem mapping data on the event horizon to their evolution restricted to the Cauchy horizon, c.f.\ \cite{Kehle2018,Kehle2020}. 
	Formally, we have, 	in a suitable regular electromagnetic gauge at the Cauchy horizon: 
	\begin{align} \nonumber
	\phiNl \restriction_{\CH} (u)  =  & \frac{r_+}{\sqrt{2\pi} r_-} \textup{p.v.} \int_{\mathbb R} \frac{\mathfrak r(\omega) }{\omega - \omer} \mathcal{F}[\phiH](\omega) e^{i (\omega-\omer) u} \d \omega \\ &+  \lim_{v\to\infty}  \frac{r_+}{\sqrt{2\pi} r_-} \textup{p.v.} \int_{\mathbb R} \frac{\mathfrak t(\omega) }{\omega - \omer} \mathcal{F}[\phiH](\omega) e^{-i (\omega-\omer) v} \d \omega  +  \text{Error},\label{eq:scatteringoperator}
	\end{align}
	where $\text{Error}$ is uniformly bounded by the energy of $\phiH$ along the event horizon $\HH$ and $\omer(M,e,q_0)$ is   as in \eqref{omer}.
	Here, $\mathfrak r(\omega)$ and $\mathfrak t(\omega)$ are  the (renormalized) scattering coefficients (see already \cref{defn:trans}).
	
	Using that  $\mathcal{F}\left[ \textup{p.v.}(\frac{1}{x}) \right] = i \pi \textup{sgn} $  and $\mathfrak t(  \omer) = -\mathfrak r ( \omer) $ (see already \eqref{eq:tomegromegr}) we formally obtain
	\begin{align}\label{eq:representationformulainintro}
	\phiNl \restriction_{\CH} (u) = \frac{\sqrt{2\pi} i r_+}{r_-} \mathfrak r ( \omer) \lim_{v\to\infty} \int_{-u}^v \phiH(\tilde v ) e^{i\omer \tilde v} \d \tilde v + \textup{Error}.	\end{align}
	Note that $\mathfrak r(\omega)$ is  real-analytic and in the charged case when $\omer \neq 0$, then always $\mathfrak r(\omega=\omer) \neq 0$. In this charged case, the formal scattering operator \eqref{eq:scatteringoperator} has a resonance at $\omega = \omer$.  However, in the uncharged case $q_0 = \omer = 0$, there exists a discrete set of non-resonant masses $m^2 \in D(M,e)$ (particularly $0 \in D(M,e) $), such that $\mathfrak  r(\omega  = \omer= 0) =0$   for $m^2 \in D(M,e)$ as shown in \cite{Kehle2018}.  
	In that case the scattering pole is absent and this can be seen as a key observation towards the $T$-energy scattering theory on the interior of Reissner--Nordström for the uncharged massless wave equation developed in \cite{Kehle2018}. However, it is shown in \cite{Kehle2018} that for generic masses $m^2 \in \mathbb R_{>0} - D(M,e) $, the resonance is present and scattering fails.
	\paragraph{A sharp condition for boundedness or blow-up at the Cauchy horizon.}	
	Restricting to parameters  $q_0  \neq 0$ or $q_0 =0,\ m^2 \in \mathbb R_{>0}- D(M,e)$, the resonance is present and from the formal computation and  \eqref{eq:representationformulainintro} we   read off that $|\phiNl| \leq C$ if the data $\phiH$ satisfy $\phiH \in L^1_v$ (in addition to having finite energy to control the error terms). Thus, in particular for fast decaying data (i.e.\ $\phiH$ satisfies \eqref{polynomialdecay} for $s>1$), we formally obtain uniform boundedness of $\phiNl$ at the Cauchy horizon.

	For general $\phiH \in \Sl-L^1$, the above reasoning does not hold, and blow-up in amplitude is possible. For concreteness, first consider the uncharged and massive case $q_0=0$, $m^2 \notin D(M,e)$. Then, $\omer =0$ and, as we will show,  $\phiNl$ is uniformly bounded at the Cauchy horizon if and only if $\phiH$ satisfies 
	\begin{equation}  \label{mean}
	\sup_{v\in [v_0,+\infty) }\left| \int_{v_0}^v \phiH(v') \d v' \right|<+\infty.
	\end{equation}   
	For instance, \eqref{mean} gives boundedness of $\phiNl$ for data  $\phiH$ of the form  \begin{equation}  \label{example.intro}
	\phiH(v) \approx e^{-i  \omega v}  \cdot v^{-s},
	\end{equation} where we recall $\frac{3}{4}<s\leq 1$,  provided $\omega \in \RR-\{0\} $: in this case, $\phiH$ obeys the quantitative oscillation condition $\phiH \in \OOpp$ as defined in \cref{defn:strongoscilcond}. If, however, $\omega=0$  then $\phiH$ violates the oscillation condition i.e.\  $\phiH \in \NO = \Sl - \OO$, and thus, $|\phiNl|$ blows up at the Cauchy horizon $\CH$ in view of  \eqref{mean} (still assuming $q_0=0$).  
	
	In the charged case $q_0 \neq 0$, the resonance is always present and  uniform boundedness of $\phiNl$ at the Cauchy horizon is true for profiles satisfying the oscillation condition $\phiH\in \OO$, e.g. profiles of the form
	\begin{equation}  \label{example.intro2}
	\phiH \approx e^{-i (\omega+\omer)\cdot v}  \cdot v^{-s}
	\end{equation}
	where $\frac{3}{4}<s\leq 1$,  provided $\omega \in \RR-\{0\} $. If however $\omega=0$ then $|\phiNl|$ blows up at the Cauchy horizon $\CH$. We refer to \cref{linearcor} for a precise statement of the results of this paragraph.

	\paragraph{Improved decay for $ \phiH \in \OOpp$ to obtain the boundedness of the Maxwell field.} Note that for the nonlinear EMKG system \eqref{E1}--\eqref{E5}, the charge $Q(u,v)$ from 
	\eqref{eq:Fnonlinear} is a dynamical quantity (assuming $q_0\neq 0$) that is nonlinearly coupled to $\phi$ and $g$, hence the boundedness of $Q$ is not guaranteed. Proving the boundedness of $Q$ in amplitude indeed requires to establish further decay estimates proved in  \cref{linearcor}, part~\ref{part3}, whose proof we now outline.
	In the case where $\phiH$ satisfies the   quantitative oscillation condition, i.e. $\phiH \in \OOpp$, the main term in \eqref{eq:representationformulainintro} enjoys  decay in $|u|$ as $u\to -\infty$ (corresponding to $i^+$ in \cref{fig:cauchyhorizonexists}). In particular, for $\phiH \in \OOpp$ we  will show (see already \cref{prop:representation} Part \textbf{B}) the quantitative control 
	\begin{equation} \label{u.decay.intro2}
	|\phiNl|(u,v)
	\ls |u|^{-1+s-\eta_0}\end{equation}
	for some $\eta_0 >0$. This (linear) quantitative estimate will  be later useful to the boundedness proof of $Q$ in the coupled case (see already \cref{C0.ext.method.intro}).
	
	\paragraph{Towards the  $W^{1,1}$-inextendibility.}
	To illustrate the obstruction caused by slow decay explained in  \cref{physical.space.intro.section}, we show in \cref{W11.main.thm} that $\phi$ does not have locally outgoing integrable derivatives near the Cauchy horizon, i.e.\ $\int |D_v \phi|(u,v) \d v =+ \infty$ for all $u$, consistently with the expectation given by \eqref{non.int.est.intro}. This blow-up in $W^{1,1}$ norm on outgoing cones justifies that, in the case where $\phi$ remains bounded, the reason is oscillation and not decay.
	
	To show the $W^{1,1}$ blow-up in linear theory (see \cref{linear.cor.W11}), we prove a representation formula for $\partial_v \phiNl (u_0,v)$ (see already \eqref{eq:formulaforpartialvphi}) and show that $\partial_v\phiNl(u_0,v) \notin L^1_v$ for fixed $u_0$. Expressed in a regular gauge on the Cauchy horizon and neglecting error terms, we formally have
	\begin{align}  
	\partial_v    	\phiNl(u_0,v)  \approx  -i \frac{r_+  e^{i \omer u_0 }  }{ \sqrt{2\pi}r_- } \int_{\mathbb R}  \mathcal F[\phiH    ] (\omega)   \; \mathfrak t (\omega)    e^{-i (\omega-\omer) v}  \d \omega \label{eq:fouriermultiplier} 
	\end{align}
	close to the Cauchy horizon. We interpret \eqref{eq:fouriermultiplier} as a formal Fourier multiplication operator with multiplier $\mathfrak t(\omega)$, i.e. $T_{\mathfrak t} \colon \phiH(v)\mapsto  \partial_v\phiNl(u_0,v) $.    Since our data $\phiH(v)$ are not  integrable ($\phiH \notin L^1$) along the event horizon $\HH$ and we aim to show that $\partial_v \phiNl(u_0,v)$ is not in $L^1_v$, it is natural to consider to inverse operator $T_{\mathfrak t}^{-1} = T_{\frac{1}{\mathfrak t}}$ with Fourier multiplier $\frac{1}{\mathfrak t(\omega)}$. Formally, by Young's convolution inequality we have 
	\begin{align} \label{W11.intro.est}\| \phiH\|_{L^1} = \| T_{\mathfrak t}^{-1}[ \partial_v\phiNl] \|_{L^1}  = \| \mathcal{F}[\mathfrak t^{-1} ]  \ast \partial_v\phiNl\|_{L^1} \leq  \| \mathcal{F}[\mathfrak t^{-1} ]  \|_{L^1}  \|    \partial_v\phiNl\|_{L^1}.\end{align}
	Since our data $\phiH$ are assumed to be non-integrable (i.e.\ $\phiH \notin L^1$), the above formal argument shows $W^{1,1}$ blow-up for $\phiNl(u_0,\cdot)$ if $\mathcal{F}[\mathfrak t^{-1} ]  \in L^1$. The above formal computation is made rigorous in the proof of  \cref{prop:representation}~Part~\textbf{E}. 
	Further, we will prove that the only obstruction to $\mathcal{F}[\mathfrak t^{-1} ]  \in L^1$ are potential zeros of $\mathfrak t(\omega)$. In the uncharged case $q_0=0$, however, the o.d.e.\ analog of the $T$-energy identity yields that 
	\begin{align}\label{eq:odeenergyidentity} 
	| \mathfrak t(\omega)|^2 = |\mathfrak r(\omega)|^2 + |\omega|^2.
	\end{align}
	Moreover, since we exclude non-resonant masses (i.e.\ $m^2 \in \mathbb R_{>0} - D(M,e)$), we have $\mathfrak t(0) \neq 0$ and as such, $\mathfrak t(\omega)$ is nowhere zero. As a result,  we show  $\mathcal{F}[\mathfrak t^{-1} ]  \in L^1$. For the uncharged case with resonant masses, this shows that all characteristic data $\phiH$ on the event horizon $\HH$ that are not integrable give rise to solutions which blow up in $W^{1,1}$ along outgoing cones at the Cauchy horizon $\CH$. 
	
	In the charged case, however, the analog of \eqref{eq:odeenergyidentity} becomes
	\begin{align}
	| \mathfrak t(\omega)|^2 = |\mathfrak r(\omega)|^2 +  \omega (\omega - \omer)
	\end{align}
	such that $\mathfrak t(\omega)$ may have zeros for $\omega \in (0, \omer)$ or $\omega\in (\omer,0)$. For small charges, a perturbation argument shows that $\mathfrak t(\omega) $ does not have zeros but for general charges, the set of zeros $\mathcal{Z}_t(M,e,q_0,m^2)=\{ \omega \in \RR \colon  \mathfrak t(\omega,M,e,q_0,m^2)=0  \}  \subset \{ 0 < |\omega| < |\omer|\}  $ could be (and in general will be)   non-empty. 
	In view of this, for non-integrable data (i.e.\ $\phiH\notin L^1$) which satisfy $P_{\delta} \phiH \in L^1$ (recall the definition of $P_{\delta}$ from \cref{sec:w11blowup}), we   show that the arising solution blows up in $W^{1,1}$ along outgoing cones. 
	It follows $\phiNl$ blows up in $W^{1,1}$ along outgoing cones  for all $\phiH \in \Sl-H$, where $H \subset \Sl$  is an exceptional subset first introduced in the statement of \cref{W11.thm.intro}.

	\subsubsection{The nonlinear problem I: Physical space estimates of the difference (Step \ref{diff.step})} \label{nonlinear.intro}
	As we explained, the physical space method  does not capture the oscillations of the field which are crucial to our proof. On the other hand, the (global) frequency analysis used for the linear equation Klein--Gordon equation on Reissner--Nordstr\"{o}m  (see already \eqref{eq:chargedKGlinear} and as explained above) relies on two key properties: the existence of the Killing vector field $\partial_t$ and the linearity of the equation---none of which extends to the coupled system \eqref{E1}--\eqref{E5}.

	In the present paper we overcome these limitations by controlling the difference between the nonlinear evolution and its linear counterpart in physical space (i.e.\ $g-g_{RN}$ and $\phi-\phiNl$, see below). In the uncharged case $q_0 =0$, this is exactly the strategy we adopt, see the first paragraph below. In the case $q_0 \neq 0$, unbounded backreaction oscillations of the Maxwell field however require a more sophisticated nonlinear scheme, see \cref{C0.ext.method.intro} and the second paragraph below. These unbounded backreaction oscillations motivate the precise definition of the oscillations spaces $\OO$, $\OOp$ and $\OOpp$ from \cref{sec:oscillationspaces}, see the third paragraph below.

	The proof of the nonlinear differences estimates will be carried out in \cref{difference.estimates.section} and follows the splitting of spacetime into four different regions depicted in \cref{Fig.regions}   used already in  \cite{Moi}, see \cref{Fig.regions_intro} (a similar splitting was first introduced by Dafermos \cite{MihalisPHD}  and subsequently used in \cite{Franzen,KerrStab,JonathanStab}).   More specifically we refer the reader to the four \cref{RS.diff.prop}, \cref{N.diff.proposition}, \cref{EB.diff.proposition}, \cref{LB.prop}.
	
		\begin{figure}
		\centering
		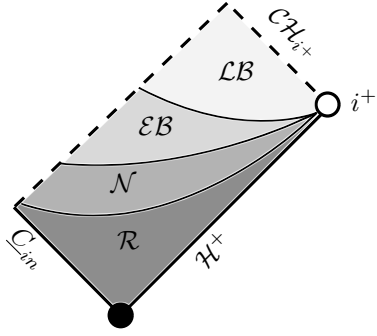			
		\caption{Division of a rectangular neighborhood of $i^+$ into  four spacetime regions.}
		\label{Fig.regions_intro}
	\end{figure}
	
	It is important to note that the difference estimates described in this section (and proved in \cref{difference.estimates.section}) are completely independent of the estimates of \cref{linearsection} (whose description was outlined in \cref{linear.intro}), with the notable exception of the final formula \eqref{Rep3} that uses the linear formula \eqref{eq:representationformulainintro} ``as a black box''.

	\paragraph{Difference estimates near $i^+$ for $q_0=0$.}
	Near the Cauchy horizon $\CH$ and close to $i^+$ as in   \cref{fig:cauchyhorizonexists} (i.e.\ for $u$ close to $-\infty$) we obtain difference estimates of the schematic form: \begin{align} \label{diff.intro1}
	&|\phi-\phiNl|(u,v)+ |u|^{-s} \cdot (|g-g_{RN}|+ |\partial_u(g-g_{RN})|)(u,v) \ls |u|^{1-3s},\\
	\label{diff.intro2}
	&\bigl|	\partial_v (\phi-\phiNl) \bigr|(u,v)+  v^{-s} \cdot |\partial_v (g-g_{RN})|(u,v) \ls v^{1-3s},
	\end{align} where $(g,F,A,\phi)$ solve \eqref{E1}--\eqref{E5} with data $\phiH \in \Sl$  and $\phiNl$ solves \eqref{E5}  with same data $\phiH \in \Sl$ on a fixed Reissner--Nordstr\"{o}m background \eqref{RN} (corresponding to the one $g$ is converging towards $i^+$). 
	The key point is that $\phi-\phiNl$, unlike $\phi$,  will turn out to be $\dot{W}^{1,1}$ along outgoing cones at $\CH$ namely \eqref{diff.intro2} gives  $$ \sup_{u,v} |\phi - \phiNl|(u,v) \ls \sup_{u} \int_{v_0}^{+\infty} |\partial_v (\phi - \phiNl)|(u,v) \d v \ls \int_{v_0}^{+\infty} v^{1-3s} \ls v_0^{2-3s} < \infty  $$ as $s>\frac{3}{4}>\frac{2}{3}$. Therefore \emph{$\phi-\phiNl$ is bounded}. In particular,  in the uncharged case $q_0=0$, uniform boundedness  of $\phi$ in the region of   \cref{fig:cauchyhorizonexists} is equivalent to that of $\phiNl$. As we will see below, this is no longer true if $q_0 \neq 0$.
	
	\paragraph{Difference estimates near $i^+$ for $q_0\neq0$.} \label{difference.intro.section}If $q_0\neq 0$, the metric differences are similar, but the scalar field difference is now impacted by the Maxwell backreaction. In particular, the first term of \eqref{diff.intro2} is replaced by an estimate of the schematic form  (in the gauge \eqref{GaugeAv} where $A_v =A_v^{RN}= 0$) 	\begin{equation} \label{diff.intro3}
	\bigl| e^{-i\sbr(u,v)}\partial_v \phi - \partial_v \phiNl\bigr|(u,v) \ls v^{1-3s},
	\end{equation} 
	\begin{equation} \label{diff.intro4}
	\sbr(u,v):= \int_{u_{\gamma}(v)}^u \left( (A_u)^{CH}(u')-  (A_u^{RN})^{CH}(u') \right)\d u',
	\end{equation}  where $u_{\gamma}(v) \sim -v$ and $(A_u)^{CH}(u')$,  $(A_u^{RN})^{CH}(u')$ are defined  as the extensions of $A_u(u,v)$,  $A_u^{RN}(u,v)$ to $\CH$, see \cref{LB.prop} for a precise statement.
	The difficulty is that $\sbr$ is unbounded in general; nevertheless, we prove sublinear growth estimates (in \cref{LB.prop} again): \begin{align} \label{diff.intro5}
	&	|\sbr(u,v)| \ls v^{2-2s} 1_{s<1}+ (1+\log(v))1_{s=1},\\ \label{diff.intro6}
	&	|\partial_v	\sbr(u,v)|+ |\partial_v^2	\sbr(u,v)| \ls v^{1-2s}. \end{align}
	Note that this is not a gauge issue: in fact, $\sbr$ is a gauge-independent quantity obtained by the expression \begin{equation}\label{sbr.exp}
	\sbr(u,v):= \iint_{[u_{\gamma}(v),u] \times [v_0,+\infty)} \left(\frac{\Omega^2 Q}{r^2}- \frac{\Omega^2_{RN} e}{r^2_{RN}} \right)\d u \d v,
	\end{equation} assuming \eqref{gaugeAponctuelle}. As a consequence, \emph{it is no longer true that $\phi-\phiNl$ is uniformly bounded}. Instead, the consequence of \eqref{diff.intro3} is that the following quantity is in $\dot{W}^{1,1}$ along outgoing cones and hence bounded: \begin{equation} \label{phi.diff.strategy}
	|\phi(u,v)- \int_{v_{\gamma}(u)}^v e^{i\sbr(u,v')} \partial_v \phiNl(u,v') \d v'| \lesssim |u|^{2-3s},
	\end{equation} where $v_{\gamma}(u)\sim -u$. Therefore,  boundedness of $\phi$ is now down to the boundedness of   $\int_{v_{\gamma}(u)}^v e^{i\sbr(u,v')} \partial_v \phiNl(u,v')  \d v'$. By our representation formula \eqref{eq:representationformulainintro}, this expression becomes, up to error, \emph{an explicit integral of the data}\begin{align} 
	\phi(u,v)= 	\int_{v_{\gamma}(u)}^v e^{i\sbr(u,v')+i\omer v'}   \phiH(v') \d v' +O(|u|^{2-3s}).
	\label{Rep3}
	\end{align} 
	Thus, the nonlinear representation formula \eqref{Rep3}  gives boundedness of $\phi$ up to  and including the Cauchy horizon $\CH$ for characteristic  event horizon data $\phiH \in \OO$, one of the main goals of \cref{main.theorem}  (see \cref{boundedness.combining.section}).

	Further, \eqref{Rep3} will  also show  blow-up of $\phi$ in amplitude at the Cauchy horizon $\CH$ for event horizon characteristic data $\phiH \notin \OO$. We postpone the related discussion to the last paragraph of \cref{C0.ext.method.intro}.

	\paragraph{The motivation to introduce $\sbr$ in the definition of the spaces $\OO$, $\OOp$, $\OOpp$.}  
	
	As explained above, the Maxwell field exerts a nontrivial backreaction with in general \emph{unbounded oscillation} $\sbr$ (recall \eqref{diff.intro5}). Recalling that $\phiNl$ is bounded if and only if the RHS of \eqref{eq:representationformulainintro} is finite (where $\sbr$ is as in \eqref{sbr.exp}), and that $\phi$  is bounded if and only if the RHS of \eqref{Rep3} is finite, it becomes clear that \textbf{the Maxwell backreaction may turn some linearly non-resonant profiles into nonlinearly resonant ones and vice versa} (a phenomenon which is absent in the uncharged case $q_0=0$ where the nonlinear estimates show that $\phi$ is bounded if and only if $\phiNl $ is bounded).

	Therefore, to ensure that our class of oscillating data $\phiH \in \OO$ (and analogously $\OOp$, $\OOpp$) gives rise to a bounded $\phi$ (and not only bounded $\phiNl$), we must define $\phiH\in \OO$ (and analogously $\OOp$, $\OOpp$) as a stronger condition than the RHS of \eqref{eq:representationformulainintro} being finite. This stronger condition is to impose \emph{sufficiently robust oscillations} that yield finiteness of the RHS of \eqref{Rep3} for {all} functions $\sbr$ satisfying \eqref{sigma_err1}, \eqref{sigma_err2}. 
	In particular, for $\sbr$ given by the formula \eqref{sbr.exp} (which obeys \eqref{sigma_err1}, \eqref{sigma_err2}, as we show, see \eqref{diff.intro5}, \eqref{diff.intro6}), the condition $\phiH \in \OO$ (and analogously $\OOp$, $\OOpp$) shows that the oscillations in the initial data are sufficiently robust to not be over-powered by the nonlinear backreaction of the Maxwell field in evolution.

	\subsubsection{The nonlinear problem II: boundedness/blow-up of matter fields and metric extendibility (Step \ref{C0step}).}  	\label{C0.ext.method.intro}
	Earlier we explained how the nonlinear difference estimates, culminating with \eqref{Rep3}, show that qualitatively oscillating $\phiH \in \OO$ on the event horizon $\HH$ give rise to uniformly bounded scalar field $\phi$ up to and including $\CH$. 
	In this section, we outline the proof of the following results that conclude the proof of our main theorems: \begin{itemize}
		\item $C^0$-extendibility of the metric (within a certain spherically symmetric class) is \emph{equivalent} to boundedness of $|\phi|$ in amplitude (first paragraph below, see also statements \ref{claimpart1} and \ref{claimpart2}). From the above equivalence given by \ref{claimpart1} and \ref{claimpart2}, we deduce the main statement of \cref{main.theorem}: the $C^0$-extendibility of the metric across $\CH$ holds under the strong qualitative oscillation condition $\phiH \in \OOp$ on the event horizon $\HH$ (see the proof in \cref{prelim.nonlin}). In our companion paper \cite{MoiChristoph2}, the implication \ref{claimpart2}  that  ''blow-up of $\phi$ implies $C^0$-inextendibility'' will be used to prove \cref{null.contraction.theorem}.
		
		\item The charge $Q(u,v)$ of the Maxwell Field is bounded for quantitatively oscillating $\phiH \in \OOpp$ on the event horizon $\HH$ (second paragraph below, proved in \cref{boundedness.combining.section}): one of the statements of \cref{main.theorem}.
		
		\item The scalar field $\phi$ blows up in $W^{1,1}$ i.e.\ $\int |D_v \phi|(u,v) \d v = \infty$ for generic slowly decaying  $\phiH \in \Sl$ on the event horizon $\HH$ (third paragraph below, proved in \cref{W11.blowup.section}): this is \cref{W11.main.thm}.
		
		\item The scalar field $\phi$ blows up in $L^{\infty}$ i.e.\ $\underset{(u,v)}{\sup}\ |\phi| (u,v)=\infty$ for non-oscillating $\phiH \in \NO=\Sl-\OO$ on the event horizon, assuming $q_0=0$ (fourth paragraph below, proved in \cref{blow.up.section}): this is \cref{main.theorem2}.

	\end{itemize} 
	
	\paragraph{Continuous extendibility of the metric as a consequence of scalar field boundedness.} We explained above how to prove boundedness/blow-up  of the scalar field depending on the data $\phiH$. Now we explain how to prove that $C^0$-extendibility on the metric is in a sense equivalent to the boundedness of $\phi$ up to  and including  $\CH$, as it turns out! Combining this novel conditional result with the previously discussed boundedness theorem for $\phi$ will give the main result of \cref{main.theorem} i.e.\ the $C^0$-extendibility of the metric for any characteristic data $\phiH \in \OOp$. The proof relies on a nonlinear scheme adapted to the slow decay of the solutions and taking advantage of the algebraic structure of the Einstein equations as explained below.

	We begin by recalling from \cite{Moi} that the following estimates for $\phiH \in \Sl$   hold true near the Cauchy horizon  $\CH$ and for some $\alpha>0$ (see \cref{recall.section} for details) \begin{align} 
	&\Omega^2(u,v) \ls e^{ -\alpha v},\\ \label{partialu.intro}
	&|\partial_u \log(\Omega^2)| \ls |u|^{1-2s},\\
	\label{partialv.intro}
	&	|\partial_v \log(\Omega^2)| \ls v^{1-2s},
	\\
	&|\partial_u r| \ls |u|^{-2s}, \label{r1}
	\\
	&	|\partial_v r| \ls v^{-2s}. \label{r2}
	\end{align}  
	The $r$ estimates \eqref{r1} \eqref{r2} being integrable, it can be shown that $r(u_n,v_n)$ is a Cauchy sequence for any  $u_n \rightarrow u$, $v_n \rightarrow +\infty$: Therefore, $r$ extends to a continuous function. In contrast, the  (conjecturally sharp) decay for $\log(\Omega^2)$  is too weak  to adopt the same reasoning since $s\leq 1$ (\eqref{partialu.intro}, \eqref{partialv.intro} are non-integrable).
	
	Nevertheless,  $\partial_u \partial_v \log(\Omega^2)+ 2 \Re( \overline{D_u \phi} D_v \phi)$ enjoys a better decay (see already \eqref{logomega1}) i.e.\ the weak decay from \eqref{partialu.intro}, \eqref{partialv.intro} comes from a $\Re( \overline{D_u \phi} D_v \phi)$ term in the Einstein equations. It was first noticed by the second author in \cite{MoiThesis} that it is useful to write the weakly decaying term  as a \emph{total derivative}, up to error: $$ 2 \Re( \overline{D_u \phi} D_v \phi) = \partial_u \partial_v ( |\phi|^2)+ ...$$
	
	Exploiting the ideas of \cite{MoiThesis}, we introduce the following new quantity $\Upsilon$, which is nonlinear and non-local: 	\begin{align} \label{Delta.def} 
	\Upsilon(u,V):=\log(\Omega^2)(u,V)+ |\phi|^2(u,V)+\int_{u}^{u_s}\frac{|\partial_u r|(u',V)}{r(u',V)}|\phi|^2(u',V) \d u',
	\end{align} 	where $\Omega^2:=-2g(\partial_u,\partial_V)$ for a suitably renormalized $(u,V)$ coordinate system. We then prove that $\Upsilon$ is bounded and admits a continuous extension (see \cref{coordinatesphibounded} for the proof).

	\begin{rmk} \label{s>3/4}
		To show that the RHS of \eqref{Delta.def} is bounded, we need the assumption $s>\frac{3}{4}$, which among other things, explains the numerology in the definition of $\Sl$ (\cref{slowdecay.def}), compare with \cref{Moi.intro}.
	\end{rmk}
	
	It turns out that the boundedness of $\Upsilon$ ultimately makes $C^0$-extendibility \emph{equivalent} to the boundedness of $\phi$ in the following sense (see \cite{MoiThesis}).
	\begin{enumerate}[A)]
		\item \label{claimpart1} If $|\phi|$ is bounded, then there exists a coordinate system $(u,V)$ such that $\log(\Omega^2)$  is bounded. 
		\item \label{claimpart2} Conversely, if $|\phi|$ blows up,  there \textbf{exists no coordinate system} $(u,V)$ such that $\log(\Omega^2)$ is bounded.
	\end{enumerate}
	
	Part \ref{claimpart1}  follows from the definition \eqref{Delta.def} and the (unconditional) boundedness of $\Upsilon$ (since $\frac{\partial_u r }{r}$ is also bounded). Moreover, because $\Upsilon$ is continuously extendible, if $|\phi|$ is continuously extendible, then  $\log(\Omega^2)$ is also continuously extendible (hence so is $\Omega^2$). In particular for  data $\phiH \in \OOp$, since we previously showed that $|\phi|$ is continuously extendible across $\CH$, we then obtain the continuous extendibility of $g$ (see \cref{sec:metri.ext.section} for the proof), and a slightly improved statement: the existence of a $C^0$-admissible extension (\cref{C0admissible}) i.e.\ a continuous extension admitting regular double null coordinates $(u,V)$ given by the above pair $(r,\Omega^2)$.

	Part  \ref{claimpart2}  is more delicate and is proven in \cite[Theorem~2.3.5]{MoiThesis}  (and used in \cite{MoiChristoph2} to prove \cref{null.contraction.theorem}): It implies that  if $|\phi|$ blows up, then $g$ does not admit any $C^0$-admissible extension.

	\paragraph{Boundedness of the Maxwell field $Q$.} We now outline the proof of the boundedness of the charge $Q(u,v)$ for $\phiH \in \OOpp$ given in \cref{boundedness.combining.section}. To prove boundedness of $Q$, we will actually need \emph{decay as $u\rightarrow -\infty$ for $\phi$} (in addition to its uniform boundedness already obtained assuming $\phiH \in \OO$): this motivates the introduction of the space $\OOpp   \subset  \OO$ from \cref{sec:oscillationspaces}. We start taking advantage of the structure of the Maxwell equation: $$ \partial_u Q = r^2 \Im ( \overline{\phi} D_u \phi)=  \Im ( \overline{r\phi} D_u (r\phi) ).$$ Moreover, we use \eqref{Dupsi.est.intro} to obtain the estimate (using also the boundedness of $r$): $$ |\partial_u Q| \lesssim |\phi| \cdot  |D_u(r\phi)| \lesssim  |\phi| \cdot |u|^{-s}.$$
	To obtain boundedness, we integrate in $u$. For this, we   take advantage of the quantitative $|u|$ decay of $|\phi|$ which is   true if $\phiH$ satisfies the quantitative oscillation condition $\phiH \in \OOpp$.  
	Combining both the linear estimate \eqref{u.decay.intro2} on $\phiNl$ and the nonlinear estimate  \eqref{diff.intro1} on $\phi -\phiNl$,  we obtain $|\phi| \ls |u|^{s-1-\eta_0}$ and thus $$ |\partial_u Q| \lesssim  |u|^{-1-\eta_0},$$ which is integrable and thus sufficient to conclude the boundedness   and continuous extendibility of $Q$.
	
	\paragraph{$W^{1,1}$ blow-up of the scalar field.} We now turn to the proof of $W^{1,1}$ blow-up on outgoing cones of $\phi$ for generic $\phiH \in \Sl - L^1$(proof in  \cref{W11.blowup.section}). One of our nonlinear difference estimates gives near the Cauchy horizon $\CH$ and uniformly in $u$: $$ \bigl | |D_v \phi| (u,v)- |D_v^{RN} \phiNl| (u,v) \bigr| \ls v^{1-3s},$$ which is integrable, since $s>\frac{3}{4}>\frac{2}{3}$. Therefore, $\| D_v \phi(u, \cdot) \|_{L^1} =+\infty$ if and only if $\| D_v^{RN}\phiNl(u, \cdot) \|_{L^1} =+\infty$.
	For $|q_0 e|$ small enough, \eqref{W11.intro.est} gives blow up of $\| D_v \phi(u, \cdot) \|_{L^1}$ for any $\phiH \in \Sl-L^1$ (and for any  $\phiH \in \Sl-H$ in the case $q_0 \neq 0$, what we call the generic case, recalling the discussion at the end of   \cref{linear.intro}).
	
	\paragraph{Blow-up in amplitude of the scalar field $\phi$ if $\phiH \notin \OO$.} We now explain how the nonlinear representation formula \eqref{Rep3} can be used to prove the blow-up in amplitude of $\phiH$ for $\phiH \in \NO=\Sl-\OO$ (see \cref{blow.up.section} for the proof). Recall indeed that \eqref{Rep3} formally states that the uniform boundedness of $\phi$ up to  and including  the Cauchy horizon $\CH$ is equivalent to the finiteness of the characteristic data integral on the event horizon $\HH$, i.e.\ for all $|u| \geq v_0$  \begin{equation} \label{int.eq}
	\sup_{v} |\phi|(u,v) = \infty \iff   \sup_v \left|\int_{-u}^{v} e^{i\sbr(u,v')+i\omer v'}  \phiH(v') \d v'\right|=\infty
	\end{equation} for  $\sbr$ defined by \eqref{diff.intro4} and in the gauge \eqref{GaugeAv}. If for given characteristic data $\phiH \in \Sl -\OO$ on the event horizon $\HH$, the upper bounds \eqref{diff.intro5}, \eqref{diff.intro6} also hold as \emph{lower bounds} up to the Cauchy horizon $\CH$, \eqref{int.eq} shows that $\phi$  \emph{blows up at the Cauchy horizon $\CH$}: for instance, one can check that for $\frac{2}{3}<s<1$ $$ \text{ for the choice } \phiH(v)= e^{-iq_0 \omega_{res}v} v^{-s}, \hskip 5 mm \limsup_{v \rightarrow+\infty}\ \bigl| \int_{v_0}^v  e^{iq_0 (v')^{2-2s} }(v')^{-s}dv'\bigr|=+\infty.$$ Unfortunately, while we conjecture that such lower bounds are true\footnote{The identity \eqref{sbr.exp} indeed suggests that $\sigma_{br}$ is comparable schematically to $|g-g_{RN}|$ which is formally of order $\alpha\  v^{1-2s}+o(v^{1-2s})$ for some $\alpha\in \mathbb{R}$. The case $\alpha=0$ is presumably non-generic but leads to faster decay for $\sigma_{br}'$ and $\sigma_{br}''$ notably.} for \textit{most solutions}, it seems that fine-tuned ones could violate them.  When these lower bounds are violated and $\sigma_{br}'$ or $\sigma_{br}''$ decay faster, we have a  linearly resonant profile ($\phi_{\mathcal{H}^+} \notin \mathcal{O}$) become nonlinearly non-resonant (meaning $\phi$ is bounded at the Cauchy horizon)  (for instance: if $\sbr''$ decays faster,  say $\sbr''(v)=O(v^{-5s+3})$,  then the RHS of \eqref{int.eq} is finite for  the choice $\phiH(v)= e^{-iq_0 \omega_{res}v} v^{-s}$). To sum up: the difficulty to control precisely these backreaction oscillations explains the absence of blowing-up examples for $q_0\neq 0$ in the present paper, but not their plausibility!

	In the case $q_0 = 0$, and for   $m^2 \notin D(M,e)$, we obtain blow-up for all data $\phiH \in \Sl-\OO$.
	As mentioned before, the restriction of the mass parameter $m^2$ is due to ``exceptional'' so-called \emph{non-resonant masses} (see \cite{Kehle2018}) for which boundedness of the linearized $\phiNl$ (hence of the EMKG-coupled scalar field $\phi$, by our result) is true, even though $\phiH \notin \OO$. Nevertheless, the set of non-resonant masses $D(M,e)$ is the zero set of a nontrivial analytic function as proved in \cite{Kehle2018}, and as such, it is   discrete and of zero Lebesgue measure.

	\subsubsection{Guide to the reader} \label{guide.section}
	We conclude this section with a short guide to help the reader read through the proofs of \cref{linearsection} and \cref{nonlinearsection}. While the above outline of the proof was organized thematically to highlight the resolution of various difficulties, for technical reasons the rest of the paper is organized slightly differently as follows: \begin{enumerate}
		\item In \cref{linearsection} we study the solution $\phiNl$ of the linear charged and massive Klein--Gordon equation $ g_{RN}^{\mu \nu} D^{RN}_{\mu} D^{RN}_{\nu} \phiNl= m^2 \phiNl$  on a fixed Reissner--Nordström metric     with slowly decaying  characteristic data $\phiH \in \Sl$ on the event horizon $\HH$. The approach is mostly focused on Fourier analysis, capturing the oscillations of $\phiNl$ towards the Cauchy horizon $\CH$.
		\begin{enumerate}
			\item In \cref{sec:separationofvariables}, we set up the radial ODE satisfied by   the Fourier transform of $\phiNl$ associated to the timelike Killing vector field $\partial_t$ on \eqref{RN}.
			\item In \cref{sec:radialode2}, we first show the existence of a scattering resonance (i.e.\ a pole at the  resonant frequency $\omega=\omer$). Moreover, we show suitable resolvent estimates associated  to the radial o.d.e. This allows us to prove   properties of the (renormalized) scattering  coefficients $\mathfrak{r}(\omega)$, $\mathfrak{t}(\omega)$. 
			
			\item In \cref{Representation.section}, we show a first representation formula involving $\mathfrak{r}(\omega)$ and $\mathfrak{t}(\omega)$ for $\phiNl$ in terms of the event horizon data $\phiH$.
			\item In \cref{sec:main.lin}, we take the limit of the  representation formula   to the Cauchy horizon of Reissner--Nordström which eventually yields our main linear result \cref{prop:representation}. 
		\end{enumerate}
		\item In \cref{nonlinearsection} we estimate the solution $(g,F,A,\phi)$ of the nonlinear Einstein--Maxwell--Klein--Gordon system \eqref{E1}--\eqref{E5} with slowly decaying  characteristic data $\phiH \in \Sl$ on the event horizon $\HH$. The approach is mostly focused on physical space estimates, capturing the effect of $\phi$ on the metric $g$.
		
		\begin{enumerate}
			\item 	In \cref{recall.section} we recall the nonlinear estimates from \cite{Moi}. They are essential to the analysis, both to show the continuous extendibility of $g$ and for the nonlinear difference estimates, see below.
			\item In \cref{prelim.nonlin}, we show that: assuming $\phi$ is uniformly bounded, then the metric $g$ is continuously extendible. The proof exploits the special structure of the nonlinearity in the Einstein equations.
			\item In  \cref{difference.estimates.section}, we estimate together the differences $g- g_{RN}$  and $\phi- \phiNl$. If $ q_0=0$ this shows that boundedness of $\phi$ is equivalent to boundedness of $\phiNl$. If $q_0\neq0$, we have \eqref{phi.diff.strategy} as a substitute.
			\item In \cref{combining.section}, we combine the results of \cref{linearsection} and \cref{difference.estimates.section} to obtain the nonlinear representation formula \eqref{Rep3}. From \eqref{Rep3} we can read off boundedness/blow-up of $\phi$ from the event horizon data $\phiH$. Combining with \cref{prelim.nonlin} gives the $C^0$-extendibility of $g$ for oscillating event horizon data $\phiH \in \OOp$ (\cref{main.theorem}). The other results follow from similar considerations.

		\end{enumerate}
	\end{enumerate}

	\section{Linear theory: the charged/massive Klein--Gordon equation on the Reissner--Nordström interior} \label{linearsection}

	We begin by  studying the charged and massive scalar fields on the \emph{fixed} subextremal Reissner--Nordström interior \eqref{RN} with the subextremal parameters $0<|e|<M$ from \eqref{eq:defnem}.    In this section, the connection $\nabla$ and the metric $g_{RN}$ are the Reissner--Nordström connection and metric, respectively. As mentioned in \cref{gaugechoice.section}, we also use the electromagnetic gauge condition 	\begin{align}\label{eq:gauge_lienartheory}
	A'_{RN} = \left( \frac{e}{r} - \frac{e}{r_+} \right) \d t =  \frac 12\left( \frac{e}{r} - \frac{e}{r_+} \right) \d v  - \frac 12  \left( \frac{e}{r} - \frac{e}{r_+} \right)  \d u
	\end{align} which satisfies  $F_{RN}=\d A'_{RN}$ for
	\begin{align}\label{eq:linearversionofF}
	F_{RN} = \frac{e}{2 r^2} \OmegaRN^2 \d u \wedge \d v. 
	\end{align}
	Note that $F_{RN}$ satisfies the homogeneous Maxwell equations  $\d\ast F_{RN} =0, \;  \d F_{RN}=0$ and remark that \eqref{eq:linearversionofF} is the corresponding linear version of \eqref{eq:Fnonlinear}.
	
	We now consider solutions $\phil$ of the charged Klein--Gordon equation \eqref{E5} which reads
	\begin{align}
	\label{eq:chargedKGlinear}
	(\nabla_\mu +i q_0 (A'_{RN})_\mu)(\nabla^\mu + iq_0 (A'_{RN})^\mu) \phil  - m^2 \phil  =0,
	\end{align}
	where $q_0 \in \mathbb R  $, $m^2\geq 0 $ are the charge and mass parameters of the field.
	We also recall  \begin{align} \omega_r  = \frac{q_0 e}{r},\omega_+  = \frac{q_0 e }{r_+},  \omega_-  = \frac{q_0 e }{r_-}, \omer  = \omega_- - \omega_+.\end{align}
	Note that in the gauge \eqref{eq:gauge_lienartheory}, we have  \begin{align}
	D_v^{RN} = \partial_v + i q_0 (A'_{RN})_v = \partial_v + \frac i 2  (\omega_r - \omega_ +)\\
	D_u^{RN}  = \partial_u + i q_0 (A'_{RN})_u = \partial_u - \frac i 2  (\omega_r - \omega_ +)
	\end{align}
	such that for any $C^1$ function we have
	\begin{align}\label{eq:gaugederivative}
	e^{- i \omer r^\ast} \partial_v (e^{i \omer r^\ast} f ) & =  \partial_v f +    i (\omega_- - \omega_+) (\partial_v r^\ast) f  =   D_v^{RN} f + \frac{i}{2}(\omega_- - \omega_r) f  
	\end{align}
	and similarly for $D_u^{RN}$. 
	For $q_0 = m^2 =0$, the field is uncharged and massless, and equation \eqref{eq:chargedKGlinear} reduces to the well-known wave equation
	\begin{align}
	\Box_{g_{RN}} \phil =0.
	\end{align}
	For $q_0 \neq 0$, $m^2=0$, the field is charged and massless and is governed by 
	\begin{align}
	(\nabla_\mu + i q_0 (A'_{RN})_\mu)(\nabla^\mu + i q_0 (A'_{RN})^\mu) \phil  =0. 
	\end{align}
	Finally, for $q_0 =0$, $m^2\neq 0$, the field is uncharged and massive and governed by the Klein--Gordon equation
	\begin{align}
	\Box_{g_{RN}} \phil - m^2 \phil =0.
	\end{align}
	
	\noindent \textbf{Notation.} Throughout \cref{linearsection} we will use the following notation.  If $X$ and $Y$ are two (typically non-negative) quantities, we use $X \lesssim Y$ or $Y\lesssim X$ to denote that
	$ X \leq C (M,e,m^2,q_0,s) Y$ for some constant $ C (M,e,m^2,q_0,s) $ depending   on the  
	parameters $ (M,e,m^2,q_0,s)$. If $C$ depends on an additional parameter $p$, we also use the notation $\lesssim_p$, $\gtrsim_p$.  We also use  $ X = O(Y )$ for $|X| \lesssim Y$.  We
	use $X\sim Y$  for $X \lesssim Y \lesssim X$. 	We also recall that throughout \cref{linearsection} we use the convention that  $\HH= \HH_R = \{ u=-\infty, v\in \mathbb R\}$ as stated in \cref{RNsolution}.

	\subsection{Separation of variables and radial o.d.e.} \label{sec:separationofvariables}
	Since $T=\partial_t$  is a Killing field of the Reissner--Nordström spacetime and in view of the specific choice of electromagnetic gauge $A'_{RN}$,   equation~\eqref{eq:chargedKGlinear} admits a separation of variables. Formally, let $\phil = \phil(t,r)$ be a solution to  \eqref{eq:chargedKGlinear}. Then, we define the $t$-Fourier transform
	\begin{align}
	\mathcal{F}[ \phil](r,\omega) = \hat{\phil}  =  \frac{1}{\sqrt{2\pi}} \int_{\mathbb R} \phil (r,t) e^{i \omega t } \d t.
	\end{align}
	Formally, since $\phil$ solves \eqref{eq:chargedKGlinear}, 
	we have that 
	\begin{align}
	u(r^\ast)= u (\omega,r^\ast) := r(r^\ast)\mathcal{F}[ \phil ](r(r^\ast),\omega)
	\end{align}
	solves
	\begin{align}
	-u'' - (\omega - ( \omega_r - \omega_+) )^2u + V u =0,\label{eq:radialode}
	\end{align}
	where  
	\begin{align}\label{eq:potential}
	V = - \OmegaRN^2 (r_\ast)\left( \frac{2M}{r^3} - \frac{2e^2}{r^4} + m^2 \right).
	\end{align}
	The radial o.d.e.\ \eqref{eq:radialode} admits the following fundamental pairs of solution associated to the event horizon ($r^\ast \to -\infty$) and the Cauchy horizon ($r^\ast \to +\infty$). 
	\begin{defn}
		\label{defn:ua}
		Let $\uhr$, $\uhl$, $\uchr$ and $\uchl$ be the unique smooth solutions to \eqref{eq:radialode} satisfying 
		\begin{align}
		&\uhr (r^\ast) = e^{-i\omega r^\ast}  + O(\OmegaRN^2) \text{ as } r^\ast \to -\infty \\
		&\uhl (r^\ast) = e^{i \omega r^\ast} + O(\OmegaRN^2)  \text{ as } r^\ast \to -\infty\\
		&\uchr (r^\ast) = e^{i(\omega -\omer ) r^\ast}  + O(\OmegaRN^2) \text{ as } r^\ast \to +\infty \\
		&\uchl (r^\ast) = e^{- i (\omega - \omer  )  r^\ast} + O(\OmegaRN^2)  \text{ as } r^\ast \to +\infty
		\end{align}
		for $ \omega \in \mathbb R$. The pairs $(\uhr, \uhl)$ and ($\uchr, \uchl$)  span the solution space of \eqref{eq:radialode} for $\omega \in \mathbb R - \{ 0 \}$ and $\omega \in \mathbb R - \{ \omer \}$, respectively.
	\end{defn}
	Using the fact that the Wronskian
	\begin{align}
	\mathfrak{W} (f,g) := f g' - f' g
	\end{align}
	of two solution of \eqref{eq:radialode} is independent of $r^\ast$, we define transmission and reflection coefficients $\mathfrak{T} (\omega)$ and $\mathfrak{R} (\omega)$ as follows.
	\begin{defn}\label{defn:trans}
		For $\omega \in \mathbb R - \{ \omer \} $, we define the transmission and reflection coefficients $\mathfrak T$ and $\mathfrak R$ as 
		\begin{align}
		&\mathfrak{T}(\omega) := \frac{\mathfrak{W}(\uhr, \uchr)}{\mathfrak{W}(\uchl, \uchr)} =  \frac{\mathfrak{W}(\uhr, \uchr)}{2i (\omega - \omer )}\\
		&\mathfrak{R}(\omega) := \frac{\mathfrak{W}(\uhr, \uchl)}{\mathfrak W (\uchr, \uchl)} =  \frac{\mathfrak W(\uhr, \uchl )}{-2i (\omega - \omer )} ,
		\end{align}
		where $\uhr$, $\uhl$, $\uchr$ and $\uchl$ are defined in \cref{defn:ua}. Indeed, this allows us to write 
		\begin{align}
		\uhr = \mathfrak T \uchl + \mathfrak R  \uchr
		\end{align} 
		for $\omega \in \mathbb R - \{ \omer \} $.
		Moreover, we define the  normalized transmission and reflection coefficients as
		\begin{align}
		\label{def.t}	& \mathfrak t(\omega)  = ( \omega - \omer)  \mathfrak T(\omega) = \frac{\mathfrak W(\uhr, \uchr )}{2i},\\
		&   \mathfrak r(\omega) = ( \omega - \omer) \mathfrak R (\omega)=  \frac{\mathfrak{W}(\uhr, \uchl)}{-2i}
		\end{align}
		which manifestly satisfy
		\begin{align}\label{eq:tomegromegr}
		\mathfrak t(\omer) = - \mathfrak r(\omer).
		\end{align}
	\end{defn}
	\begin{rmk}
		Note that the radial o.d.e.\ \eqref{eq:radialode} depends analytically on $\omega$. Thus, $\uhr$, $\uhl$, $\uchr$ and $\uchl$ are real-analytic functions for $\omega$ for fixed $r^\ast$. In particular, this means that the Wronskians $\mathfrak{W}(\uhr, \uchr)$, $\mathfrak W (\uchr, \uchl) $ etc.\ are real-analytic functions for $\omega \in \mathbb R$ which can be extended holomorphically to a neighborhood of the real line. 
	\end{rmk}
	We will also define the re-normalized functions
	\begin{defn}\label{defn:tildeuhretc}
		We define
		\begin{align}
		&\tilde \uhr (r^\ast,\omega)  :=e^{i\omega r^\ast} \uhr (r^\ast,\omega)  ,  \\
		&\tilde \uhl (r^\ast,\omega)  :=e^{-i\omega r^\ast} \uhl (r^\ast,\omega)  , \\
		&\tilde \uchr  (r^\ast,\omega) :=e^{-i (\omega - \omer ) r^\ast} \uchr (r^\ast,\omega) ,   \\
		&\tilde \uchl (r^\ast,\omega)  :=e^{i (\omega - \omer )  r^\ast} \uchl (r^\ast,\omega)    .
		\end{align}
	\end{defn}

	\subsection{Analysis for the radial o.d.e.}\label{sec:radialode2}
	\begin{prop}\label{prop:tr1}
		Let either of the two assumptions hold true. 
		\begin{itemize}
			\item Either $ q_0  \neq 0$. 
			\item or $q_0 =0$ but $m^2 \notin  D(M,e)$, where $ D(M,e)$ is the discrete set of \cite[Theorem~7]{Kehle2018}. 
		\end{itemize} Then, the transition and reflection coefficients $\mathfrak T(\omega)$ and $\mathfrak R(\omega)$ as defined in \cref{defn:trans} have (non-removable) poles of first order at $\omega = \omer$.
	\end{prop}
	\begin{proof}
		First, note that $(\operatorname{Im}(u' \bar u) )' = 0$ holds true for any $C^1$ solution of \eqref{eq:radialode}. Applying this to $\uhr$ and expanding $\uhr $  as $\uhr = \mathfrak T \uchl + \mathfrak R \uchr$, we conclude the \textit{o.d.e.\ energy identity}
		\begin{align}
		|\mathfrak{T}|^2 - |\mathfrak{R}|^2 = \frac{\omega}{\omega - \omer}.
		\end{align}
		If $q_0\neq 0$ and thus, $\omer \neq 0$,   we have $ |\mathfrak{T}|^2 \geq   \frac{\omega}{\omega - \omer}  $ for $|\omega| > \omer$. Sending $\omega \to\omer $, we conclude that $\mathfrak T$ blows up and since $\mathfrak T$ is meromorphic in a complex neighborhood of $\omer$, the claim follows. In particular, we have that $\mathfrak W (\uhr, \uchr) (\omega = \omer) \neq 0$ and $\mathfrak W (\uhr, \uchl) (\omega = \omer) \neq 0$. For $q_0 = 0$ and $m^2\notin D(M,e)$, the claim follows from \cite[Theorem~7]{Kehle2018}.
	\end{proof}
	\begin{prop}\label{prop:tr2}
		The solutions $\uhr$, $\uchl$, $\uchr$ and the renormalized functions  $\tilde\uhr$, $\tilde\uchl$, $\tilde \uchr$ as defined in \cref{defn:ua} and \cref{defn:tildeuhretc}, respectively, satisfy for $\omega \in \mathbb R$
		\begin{align}\label{eq:431}
		& \sup_{r^\ast \in (-\infty,r^\ast_0]} |\uhr (\omega, r^\ast)|\lesssim_{r^\ast_0} 1, \\
		& \sup_{r^\ast \in (-\infty,r^\ast_0]} |\uhr^\prime (\omega, r^\ast)|\lesssim_{r^\ast_0} |\omega|
		\end{align}
		for  any fixed  $r^\ast_0 \in \mathbb R$ and 
		\begin{align}
		&  |\tilde \uhr (\omega, r^\ast) -1 |  \lesssim_{r_0^\ast} |\OmegaRN^2(r^\ast) | , \\
		&  | {\tilde\uhr}^\prime (\omega, r^\ast)|\lesssim_{r^\ast_0} |\OmegaRN^2(r^\ast)| \label{eq:434}
		\end{align}
		uniformly for $r^\ast \leq r_0^\ast$. Moreover,  for $\omega \in \mathbb R$ and any fixed $r^\ast_0 \in \mathbb R$
		\begin{align}
		& \sup_{r^\ast \in [r^\ast_0,+\infty)} | \uchl (\omega, r^\ast)|\lesssim_{r^\ast_0} 1, \\
		& \sup_{r^\ast \in [r^\ast_0 ,+\infty)} |{ \uchl}^\prime (\omega, r^\ast)|\lesssim_{r^\ast_0}  |\omega|,
		\end{align}
		\begin{align}
		& \sup_{r^\ast \in [r^\ast_0,+\infty)} | \uchr (\omega, r^\ast)|\lesssim_{r^\ast_0} 1, \\
		& \sup_{r^\ast \in [r^\ast_0 ,+\infty)} |{ \uchr}^\prime (\omega, r^\ast)|\lesssim_{r^\ast_0}  |\omega|,
		\end{align}
		and uniformly for $r^\ast \geq r_0^\ast$,
		\begin{align}
		&   | \tilde \uchl (\omega, r^\ast) -1 |\lesssim_{r_0^\ast} |\OmegaRN^2(r^\ast)| , \\
		&   |{\tilde \uchl}^\prime (\omega, r^\ast)| \lesssim_{r_0^\ast} |\OmegaRN^2|,
		\end{align} 
		\begin{align}
		&  | \tilde \uchr (\omega, r^\ast) -1 |\lesssim_{r_0^\ast} | \OmegaRN^2(r^\ast) | , \\
		&   |{\tilde \uchr}^\prime (\omega, r^\ast)|\lesssim_{r_0^\ast} |\OmegaRN^2 (r^\ast) |.
		\end{align}
		The transition and reflection coefficients as defined in \cref{defn:trans} satisfy
		\begin{align}\sup_{|\omega - \omer | \geq 1} \left( |\mathfrak T(\omega)| + |\mathfrak R(\omega)| \right) \lesssim 1. \label{eq:boundednessofTandR}
		\end{align}
		\begin{proof} It suffices to show the results for $\uhr$ and $\tilde \uhr$ as the other cases follow completely analogously. We will consider the cases $|\omega|\leq \omega_0 := |\omer| + 1$ and $|\omega | >\omega_0$ independently. First, for $|\omega|\leq \omega_0$, we note that $\uhr$ is the unique solution to the Volterra equation 
			\begin{align}
			\uhr(r^\ast,\omega) = e^{-i\omega r^\ast} + \int_{-\infty}^{r^\ast}  \frac{\sin(\omega (r^\ast - y)) }{\omega} \left(  2 \omega(\omega_r - \omega_+) -(\omega_r -\omega_+)^2  + V(y) \right) \uhr(y,\omega) \d y.
			\end{align}
			For $\omega=0$, we mean $\frac{\sin(\omega (r^\ast -y))}{\omega } = r^\ast -y$. Now, since 
			\begin{align}
			& \int_{-\infty}^{r_0^\ast}  \sup_{y \leq r^\ast<r_0^\ast} |K(r^\ast,y)| \d y  \lesssim \OmegaRN^2(r^\ast_0) , 
			\end{align}
			where \begin{align}
			K(r^\ast,y) = \frac{\sin(\omega (r^\ast - y)) }{\omega} \left( 2 \omega(\omega_r - \omega_+) -(\omega_r -\omega_+)^2  + V(y) \right) ,
			\end{align}
			we have by standard estimates on Volterra equations (e.g. \cite[Proposition 2.3]{Kehle2018} or \cite[\S10]{olver})  that for $|\omega|\leq \omega_0$, 
			\begin{align}\| \uhr \|_{L^\infty(-\infty, r_0^\ast) } \lesssim_{r_0^\ast} 1\end{align}
			as well as 
			\begin{align}
			| \uhr - e^{-i\omega r^\ast} |\lesssim  |\OmegaRN^2(r^\ast)|
			\end{align}
			uniformly for $r^\ast \leq 0$.
			Similarly, we obtain 
			\begin{align}
			\| \uhr' \|_{L^\infty(-\infty, r_0^\ast) } \lesssim_{r_0^\ast} 1+|\omega|\lesssim_{r_0^\ast} 1.
			\end{align}
			Note that this also shows that 	for $|\omega|\leq \omega_0$, we have 
			\begin{align}    
			&\| \tilde \uhr' \|_{L^\infty(-\infty, r_0^\ast) } \lesssim_{r_0^\ast} 1
			\end{align}
			and 
			\begin{align}
			| \tilde  \uhr - 1|  \lesssim |\OmegaRN^2(r^\ast)|
			\end{align}
			uniformly for $r^\ast \leq 0$. 
			
			Now, we consider the case $|\omega|\geq \omega_0$. Note that in this frequency regime, the frequency dependent potential \begin{align}W:= - (\omega - ( \omega_r - \omega_+ ) )^2  \end{align} satisfies 
			\begin{align} \label{eq:estimatesonW1}
			&  - W \gtrsim  \omega^2 \\
			&  \left|\frac{W'}{W}\right|\lesssim \frac{\OmegaRN^2}{|\omega|} \label{eq:estimatesonW2}\\
			&  \left|\frac{W''}{W}\right|\lesssim \frac{\OmegaRN^2}{|\omega|} \label{eq:estimatesonW3}
			\end{align} 
			and the radial potential $V$ satisfies
			\begin{align}
			|V|, |V'|, |V''|\lesssim \OmegaRN^2
			\end{align}
			uniformly on $r^\ast \in \mathbb R$.
			
			Now we will use a WKB approximation for $\uhr$. First, we will estimate the total variation $\mathcal V_{-\infty,+\infty}$ associated to the error-control function 
			\begin{align}
			F_{\uhr} (r^\ast, \omega) := \int_{-\infty}^{r^\ast} \frac{1}{|W|^{\frac 14}}  \frac{\d^2}{\d x^2 } |W|^{- \frac 14} - \frac{V}{|W|^{\frac 12} } \d y. 
			\end{align}
			In view of \eqref{eq:estimatesonW1}--\eqref{eq:estimatesonW3}, we   estimate 
			\begin{align}
			\mathcal V_{-\infty, + \infty}(F_{\uhr } ) = \int_{-\infty}^{+\infty} \left| \frac{1}{|W|^{\frac 14}}  \frac{\d^2}{\d x^2 } |W|^{- \frac 14} - \frac{V}{|W|^{\frac 12} } \right| \d y\lesssim  \frac{1}{|\omega|}.
			\end{align} 
			Thus, applying \cite[Theorem~2.2, p.196]{olver}   we obtain
			\begin{align}
			\uhr (r^\ast, \omega)  & =   \frac{|\omega|^{\frac 12} }{|W(r^\ast,\omega)|^{\frac 14}} e^{- i    \omega r^\ast  +i \int_{-\infty}^{r^\ast}  \omega_r  - \omega_+ \d y } \left( 1+ \eta_{\uhr}\right),
			\end{align}
			where the error function $\eta_{\uhr}$ satisfies  
			\begin{align}
			& |   \eta_{\uhr} (r^\ast,\omega) |  \lesssim \frac{1}{|\omega|},\\
			& |   \eta'_{\uhr} (r^\ast, \omega) |  \lesssim |W(r^\ast,\omega)|^{\frac 12}  \frac{1}{|\omega|} \lesssim 1
			\end{align}
			uniformly for $r^\ast \in \mathbb R$ and $|\omega|\geq \omega_0$ as well as 
			\begin{align}
			& |   \eta_{\uhr} (r^\ast,\omega) |  \lesssim \frac{\OmegaRN^2}{|\omega|},\\
			& |   \eta'_{\uhr} (r^\ast, \omega) |  \lesssim \OmegaRN^2
			\end{align}
			uniformly for $r^\ast < 0  $ and $|\omega|\geq \omega_0$.
			This shows that for $|\omega|\geq \omega_0$ we have 
			\begin{align}\label{equhr1}
			&\| \uhr \|_{L^\infty(\mathbb R) } \lesssim 1,\\
			& \| \uhr ' \|_{L^\infty(\mathbb R) } \lesssim |\omega|. \label{equhr2}
			\end{align}
			
			Note also that $\tilde \uhr = e^{i \omega r^\ast} \uhr$ similarly satisfies 
			\begin{align}
			&\|\tilde  \uhr \|_{L^\infty(\mathbb R) } \lesssim 1,\\
			&\| \tilde \uhr ' \|_{L^\infty(\mathbb R) } \lesssim 1
			\end{align}
			and 
			\begin{align}
			& |\tilde \uhr (r^\ast, \omega) -1 |\lesssim_{r_0^\ast} \OmegaRN^2 \\
			&   | \tilde \uhr' (r^\ast, \omega)| \lesssim_{r_0^\ast} \OmegaRN^2
			\end{align}
			uniformly for $r^\ast \leq r_0^\ast$ and $\omega \in \mathbb R$.  The other results for $\uchl$ and $\uchr$ are shown completely analogous.

			Now, we will show the bounds on the transmission and reflection coefficients $\mathfrak T$ and $\mathfrak R$. The bound \eqref{eq:boundednessofTandR} follows from the fact that for $|\omega|$ sufficiently large, $|\mathfrak W(\uhr, \uchr)|,  | \mathfrak W(\uhr, \uchl)|\lesssim |\omega|$ in view of \eqref{equhr1}, \eqref{equhr2} and computing the Wronskian as $r^\ast \to +\infty$.  For $|\omega|$ small, the bound follows from continuity of $|\mathfrak W(\uhr, \uchr)|$ and $|\mathfrak W(\uhr, \uchl)|$.
		\end{proof}
	\end{prop}
	\begin{lem}\label{lem:estimatesonderiativesonuchrl}
		The bounds
		\begin{align}
		& |\partial_\omega\tilde  \uchr (\omega, r^\ast) | \lesssim \OmegaRN^2, \\
		& |\partial_\omega\tilde  \uchl (\omega, r^\ast) | \lesssim \OmegaRN^2
		\end{align}
		and
		\begin{align}
		& |\partial_{r^\ast}\partial_\omega\tilde  \uchr (\omega, r^\ast) | \lesssim \OmegaRN^2 \langle \omega \rangle , \\
		& |\partial_{r^\ast} \partial_\omega\tilde  \uchl (\omega, r^\ast) | \lesssim \OmegaRN^2 \langle \omega \rangle
		\end{align}
		hold uniformly for $r^\ast \geq 0$ and $\omega \in \mathbb R$.  (We recall that $\langle \omega \rangle:= \sqrt{1+\omega^2}$).
		
		Moreover, 
		\begin{align}
		& |\partial_\omega\tilde  \uhr (\omega, r^\ast) | \lesssim \OmegaRN^2
		\end{align}
		and
		\begin{align}
		& |\partial_{r^\ast}\partial_\omega\tilde  \uhr  (\omega, r^\ast) | \lesssim\OmegaRN^2 \langle \omega \rangle  
		\end{align}
		hold uniformly for $r^\ast \leq  0$ and $\omega \in \mathbb R$. 
		\begin{proof} First, we consider the range   $|\omega - \omer |\leq 1$. 
			First, note that $\tilde \uchr$ solves the Volterra integral equation
			\begin{align} \nonumber 
			\tilde \uchr (r^\ast, \omega)=  1 +  \int_{r^\ast}^{+\infty} & \frac{\sin[ (\omega - \omer)(r^\ast - y) ]}{\omega -\omer} e^{- i (\omega - \omer) (r^\ast - y)} \\
			& \left[V(y) - (\omega_- - \omega_{r(y)})(2 \omega + 2\omega_+ - \omega_- - \omega_{r(y)} ) \right] \tilde \uchr (\omega, y) \d y.
			\end{align}
			Thus, 
			$\partial_\omega \tilde \uchr$ solves  
			\begin{align} \nonumber 
			\partial_\omega 	\tilde \uchr (r^\ast, \omega) = \int_{r^\ast}^{+\infty} & \frac{\sin[ (\omega - \omer  )(r^\ast - y) ]}{\omega -\omer } e^{- i (\omega - \omer) (r^\ast - y)} \\ \nonumber 
			& \left[V(y) - (\omega_- - \omega_{r(y)})(2 \omega + 2\omega_+ - \omega_- - \omega_{r(y)} ) \right] \partial_\omega \tilde \uchr (\omega, y) \d y \\ \nonumber  +  \int_{r^\ast}^{+\infty} & \frac{ \partial_\omega \left( \operatorname{sinc}[ (\omega - \omer )(r^\ast - y) ] e^{- i (\omega -\omer) (r^\ast - y)} \right)}{r^\ast - y} (r^\ast -y)^2  \\\nonumber
			& \left[V(y) - (\omega_- - \omega_{r(y)})(2 \omega + 2\omega_+ - \omega_- - \omega_{r(y)} ) \right] \tilde \uchr (\omega, y) \d y\\ \nonumber  +  \int_{r^\ast}^{+\infty} & \frac{\sin[ (\omega - \omer  )(r^\ast - y) ]}{\omega - \omer} e^{- i (\omega - \omer) (r^\ast - y)}  \\
			& 2\left[V(y) - (\omega_- - \omega_{r(y)}) ) \right] \tilde \uchr (\omega, y) \d y.
			\label{eq:formulaforpartialwuchr}
			\end{align}
			Now, we have the following bounds uniformly for $r^\ast \geq 0$
			\begin{align}
			&\left|  \frac{\sin[ (\omega - \omer )(r^\ast - y) ]}{\omega - \omer} e^{- i (\omega - \omer) (r^\ast - y)} \right|\lesssim (r^\ast -y),\\
			&\left| \frac{ \partial_\omega \left( \operatorname{sinc}[ (\omega - \omer)(r^\ast - y) ] e^{- i (\omega - \omer) (r^\ast - y)} \right)}{r^\ast - y} \right|\lesssim 1, \\
			& \left|   V(y) - (\omega_- - \omega_{r(y)})(2 \omega + 2\omega_+ - \omega_- - \omega_{r(y)} )   \right|\lesssim \OmegaRN^2,\\
			& \left| V(y) - (\omega_- - \omega_{r(y)}) ) \right|\lesssim \OmegaRN^2. 
			\end{align}
			With these bounds, standard results (e.g.\ \cite[\S10]{olver}) on estimates of solutions of Volterra integral equations, show that 
			\begin{align}
			\left| \partial_\omega \tilde \uchr (r^\ast,\omega) \right|  \lesssim\OmegaRN^2
			\end{align}
			uniformly for $r^\ast \geq 0$.
			Similarly, we have \begin{align}
			\left| \partial_\omega \tilde \uchl (r^\ast,\omega) \right|  \lesssim  \OmegaRN^2
			\end{align}
			uniformly for $r^\ast \geq 0$.
			
			Differentiation of \eqref{eq:formulaforpartialwuchr} with respect to $r^\ast$ also gives
			\begin{align}
			&|  \partial_{r^\ast} \partial_{\omega} \tilde \uchr |\lesssim\OmegaRN^2
			\end{align}
			and analogously we obtain
			\begin{align} |\partial_{r^\ast} \partial_{\omega} \tilde \uchl |\lesssim \OmegaRN^2.
			\end{align}
			
			Now, we consider the range $|\omega - \omer| \geq 1$. Then, for $r^\ast \geq 0$, we have the bounds 	\begin{align}
			&\left|  \frac{\sin[ (\omega - \omer )(r^\ast - y) ]}{\omega -\omer} e^{- i (\omega - \omer) (r^\ast - y)} \right|\lesssim \langle \omega \rangle^{-1}, \\
			&\left|  \partial_\omega \left( \operatorname{sinc}[ (\omega - \omer )(r^\ast - y) ] e^{- i (\omega - \omer) (r^\ast - y)} \right)  \right| \lesssim \langle \omega \rangle^{-1} \frac{1+|r^\ast -y|}{|r^\ast -y|}, \\
			& \left|   V(y) - (\omega_- - \omega_{r(y)})(2 \omega + 2\omega_+ - \omega_- - \omega_{r(y)} )   \right|\lesssim \OmegaRN^2\langle \omega \rangle, \\ &
			\left| V(y) - (\omega_- - \omega_{r(y)}) ) \right|\lesssim\OmegaRN^2.
			\end{align}
			Thus, analogously to the above, this gives 	uniformly for $r^\ast \geq 0$
			\begin{align}
			&	\left| \partial_\omega \tilde \uchr (r^\ast,\omega) \right|  \lesssim \OmegaRN^2, \\
			&	\left| \partial_\omega \tilde \uchl (r^\ast,\omega) \right|  \lesssim \OmegaRN^2,
			\end{align}
			as well as 
			\begin{align}
			&|  \partial_{r^\ast} \partial_{\omega} \tilde \uchr |\lesssim\OmegaRN^2 \langle \omega \rangle,
			\\ & |\partial_{r^\ast} \partial_{\omega} \tilde \uchl |\lesssim \OmegaRN^2\langle \omega \rangle.
			\end{align}
			The result on $\uhr$ follows completely analogous.
		\end{proof}
	\end{lem}
	
	\begin{cor}\label{cor:tarnsmissionandreflection}
		The normalized transmission and reflection coefficients   satisfy \begin{align}  |\mathfrak t(\omega) | + |\mathfrak r(\omega)| \lesssim 1+ |\omega|. \end{align}
		\begin{proof}
			This is a consequence of \cref{prop:tr1} and \cref{prop:tr2}.
		\end{proof}
	\end{cor}
	\begin{lem}\label{lem:estimatesonderiativesofrandt}
		We have 
		\begin{align}
		&|\partial_\omega \mathfrak r (\omega) | \lesssim \langle \omega \rangle \\  
		&| \partial_\omega \mathfrak t (\omega) | \lesssim \langle \omega \rangle.
		\end{align}
		\begin{proof}
			We estimate \begin{align} & |\partial_\omega \mathfrak r |\lesssim |\partial_\omega  \mathfrak W(\uhr, \uchr)|   \lesssim | \mathfrak W(\partial_\omega \uhr, \uchr) (r^\ast =0) |  +| \mathfrak W(\uhr,\partial_\omega  \uchr) (r^\ast =0) | \lesssim \langle \omega \rangle \end{align} in view of \cref{lem:estimatesonderiativesonuchrl} and \cref{prop:tr2}. Analogously the same holds for $\mathfrak t$. 
		\end{proof}
	\end{lem}
	Towards the $W^{1,1}$ inextendibility at the Cauchy horizon we need to analyze the zeros of the transmission coefficient $\mathfrak t$. To do so, we recall the definition of $\mathcal Z_{\mathfrak t}(M,e,q_0,m^2)$ from \eqref{eq:defnofz}.
	\begin{lem}\label{lem:zerosornot}
		\begin{enumerate}
			\item Let $q_0 e \neq 0$. Then,   $    \mathcal Z_{\mathfrak t}\subset (0,\omer)$ if $q_0e > 0$ or  $    \mathcal Z_{\mathfrak t}\subset (\omer, 0)$ if $q_0 e < 0$. 
			\item Let $0 < |q_0 e | < \epsilon(M,e,m^2) $ for some $\epsilon(M,e,m^2)$ sufficiently small, then $\mathfrak t$ does not have any zeros, i.e. $    \mathcal Z_{\mathfrak t} = \emptyset$. 
			\item Let $q_0 =0$ and let $m^2 \notin D(M,e)$, where $D(M,e)$ is the discrete set as in \cite[Theorem~7]{Kehle2018}. Then, $\mathfrak t(\omega)$ does not have any zeros, i.e.    $ \mathcal Z_{\mathfrak t}(M,e,0,m^2)=\emptyset$ if $m^2 \notin D(M,e)$. 
		\end{enumerate} 
	\end{lem}
	\begin{proof}
		The first statement follows from the fact that $|\mathfrak t|^2 = |\mathfrak r |^2 + \omega ( \omega - \omer) \geq \omega (\omega - \omer)$, \cref{prop:tr1} and the fact that $\mathfrak t(\omega =0) \neq 0$. Indeed, if $\mathfrak t(\omega=0) = 0$, then $\mathfrak r(\omega =0)=0$ and thus $\mathfrak T(\omega =0) =\mathfrak R(\omega =0) = 0$. But this cannot be true,   since otherwise $\uhr =  \mathfrak R \uchr + \mathfrak T \uchl$ would be trivial.
		The second statement just follows from continuity of $\mathfrak t$ as a function of the parameters $q_0 e$. 
		The third statement is shown in \cite[Theorem~7]{Kehle2018}. 
	\end{proof}
	\begin{rmk}
		Note that for $q_0 =0$ and $m^2=0$, we have that  $\mathfrak t( \omega =0) = 0$. This is a crucial observation for the existence of a $T$-energy scattering theory as established in \cite{Kehle2018}.
	\end{rmk}
	\subsection{Representation formula}  \label{Representation.section}
	We recall that throughout \cref{linearsection} we consider the event horizon $\mathcal H^+$ as the set $\{ u=-\infty\}\times \{ v \in \mathbb R \}$ as in \cref{RNsolution}. 
	\begin{defn}
		For $f \in L^2(\mathcal{H}^+)$ we define the Fourier transform along the event horizon as 
		\begin{align}
		\mathcal F_{\mathcal{H}^+} [f](\omega) :=  r_+ \mathcal{F} [f] (\omega) = \frac{r_+}{\sqrt{2\pi} }  \int_{\mathbb R} f(\tilde v) e^{i \omega \tilde v} \d \tilde v 
		\end{align}
		in    mild abuse of notation. 
	\end{defn}
	\begin{lem}\label{lem:representation}
		Let $  \phiHL \in C^\infty (\mathcal{H}^+) $ be spherically symmetric smooth data on the event horizon and assume that $  \phiHL$ is supported away from the past bifurcation sphere. Assume vanishing data on the left event horizon and let $  \phil$ be the arising smooth solution to \eqref{eq:chargedKGlinear} attaining that data. Then, for any fixed $ v_1$ and any $u \in \mathbb R$, $v\leq v_1$ we have 
		\begin{align}\label{eq:representationformula}
		\phil  (u,v) = \frac{1}{\sqrt{2\pi}r} \int \mathcal{F}_{\mathcal{H}^+} [ \phiHL \chi_{\leq v_1}] (\omega) \tilde \uhr (r^\ast(u,v),\omega) e^{-i\omega v} \d \omega 
		\end{align}
		and 
		\begin{align}
		\label{eq:representationformulapartialv}
		\partial_v	( r  \phil (u,v) )= \frac{1}{\sqrt{2\pi}} \int \mathcal{F}_{\mathcal{H}^+} [  \phiHL \chi_{\leq v_1}] (\omega)  \partial_v \left( \tilde  \uhr (r^\ast(u,v),\omega) e^{-i\omega v} \right) \d \omega \\\label{eq:representationformulapartialu}
		\partial_u (r	 \phil (u,v)) = \frac{1}{\sqrt{2\pi}} \int \mathcal{F}_{\mathcal{H}^+} [ \phiHL \chi_{\leq v_1}] (\omega)  \partial_u \left( \tilde  \uhr (r^\ast(u,v),\omega) e^{-i\omega v} \right) \d \omega,
		\end{align}
		where $\chi_{\leq v_1} (v)= \chi_0 ( v - v_1)$ and $\chi_0 \colon \mathbb R \to [0,1]  $  is a smooth cut-off which  satisfies $\chi_0 (x) = 1$ for $x \leq 0 $ and $\chi_0(x) = 0$ for $x \geq 1$. 
		
		By a standard density argument,  \eqref{eq:representationformula}, \eqref{eq:representationformulapartialv} and \eqref{eq:representationformulapartialu} hold also for spherically symmetric data  $ \phiHL \in C^1(\mathcal H^+)$ with $ \phiHL $ supported away from the past bifurcation sphere.
		
		\begin{proof}
			Fix any $ v_1$ and let $(u,v)$ with $v \leq  v_1$ be arbitrary. By the domain of dependence property, we have that $ \phil$ satisfies $ \phil  =  {\phil}_{\leq v_1}  $ on $(u,v)$ with $v\leq v_1$, where $ {\phil}_{\leq v_1}  $ is the unique solution arising from data  $\phiHL  \chi_{\leq v_1} \in C_c^\infty (\mathcal H^+) $ on the right event horizon $\mathcal H^+$ together with vanishing data on the left event horizon.  
			Now, since $\mathcal{F}_{\mathcal H^+} [\phiHL \chi_{\leq v_1}]$ is Schwartz, $ \uhr$ satisfies \eqref{eq:radialode}, and $ \uhr$ obeys the bounds as in \cref{prop:tr2}, we can differentiate under the integral sign on the right hand side of \eqref{eq:representationformula} and conclude that indeed the right hand side of \eqref{eq:representationformula} solves \eqref{eq:chargedKGlinear}. Finally, to show that $ \phil  =  {\phil}_{\leq v_1}  $ it suffices to show that the right hand side assumes the data from which $ {\phil}_{\leq v_1}$ arises. But again, since $\mathcal{F}_{\mathcal H^+} [\phiHL  \chi_{\leq v_1}]$ is Schwartz, we immediately obtain that the right-hand side of \eqref{eq:representationformula} converges to $ \phiHL \chi_{\leq v_1}$ towards the right event horizon, and---after an application of the Riemann--Lebesgue lemma---to 0 towards the left event horizon. 
			Now, \eqref{eq:representationformula} follows from uniqueness of the characteristic initial value problem.
			The formulae 	\eqref{eq:representationformulapartialv} and 	\eqref{eq:representationformulapartialu} now follow  from differentiating under the integral sign which can be applied as $\mathcal{F}_\mathcal H^+ [ \phiHL \chi_{\leq v_1}]$ is a Schwartz function.
		\end{proof}
	\end{lem}
	Note that the above proposition immediately implies 
	\begin{cor}\label{cor:respres}
		Let $ \phiHL$ be as in \cref{lem:representation} and  assume vanishing data on the left horizon. Let $ \phil$ be the arising smooth solution attaining that data. Then, 
		\begin{align}
		\phil(u,v) = \frac{1}{\sqrt{2\pi} r} \int_{\mathbb R} \mathcal F_{\mathcal H^+} [ \phiHL \chi_{\leq v} ](\omega) \tilde \uhr (r(u,v), \omega) e^{-i \omega v} \d \omega
		\end{align}
		and 
		\begin{align} \label{eq:representationforpartialv}
		&   \partial_v (  r  \phil(u,v)) = \frac{1}{\sqrt{2\pi}} \int_{\mathbb R} \mathcal F_{\mathcal H^+} [ \phiHL \chi_{\leq v} ](\omega) \partial_v \left(   \tilde \uhr (r(u,v), \omega) e^{-i \omega v} \right) \d \omega \\
		&   \partial_u (r \phil(u,v)) = \frac{1}{\sqrt{2\pi}} \int_{\mathbb R} \mathcal F_{\mathcal H^+} [ \phiHL \chi_{\leq v} ](\omega) \partial_u \left(  \tilde \uhr (r(u,v), \omega) e^{-i \omega v} \right) \d \omega
		\end{align}
		for $u,v \in \mathbb R$, where $\chi_{\leq v} $ is as in \cref{lem:representation}.
		\begin{proof}
			Choosing $v = v_1$ in \cref{lem:representation} yields the result.
		\end{proof}
		
	\end{cor}
	\subsection{Main results from the linear theory} \label{sec:main.lin}
	Before we state the main proposition about the linear theory, we define the following norms for sufficiently regular functions.
	\begin{align}
	& E_1[f]:= \left( \int_{\mathbb R} |f(v)|^2 + |\partial_v f (v) |^2 \d v \right)^{\frac 12},
	\\
	& E_1^\beta[f] := \left(  \int_{\mathbb R}  ( |f(v)|^2 + |\partial_v f(v)|^2 ) \langle v \rangle^{2\beta} \d v \right)^{\frac 12}  , \\  
	& F^\beta [f]:= \sup_{v\geq 0} \langle v\rangle^\beta \left|  \int_v^{+\infty} f(\tilde v ) e^{i \omer \tilde v} \d \tilde v  \right|  .
	\end{align}
	Further, for   Part~\textbf{E.} of the following proposition, we will use the Fourier projection operator $P_{\delta}$ defined in \cref{sec:w11blowup}.
	We will further state estimates in the so-called late blue-shift region $\mathcal{LB}$. This region is defined as   $\mathcal{LB} =  \{\Delta' + \frac{2s}{2|K_-|} \log(v) \leq  u+v \}$  for some $\Delta'\geq 0$ chosen in \cref{recall.section}. (Note that the estimate below involving $\mathcal{LB}$ actually holds true uniformly for all $\Delta'\geq 0$.) For given $u$, we also define $v_\gamma(u)$ to satisfy  $\Delta' + \frac{2s}{2|K_-|} \log(v_\gamma(u)) =  u+v_\gamma(u) $. Note that the estimate $\Omega^2_{RN}(u,v)\lesssim  v^{-2s}$ is satisfied in $\mathcal{LB}$.  We refer to \cref{Fig.regions_intro} for a visualization of the region $\mathcal{LB}$ near $i^+$. In fact, in the region $\mathcal{LB}$  all the following estimates apply and $\mathcal{LB}$ is also the region in which we will make use of the linear theory for the nonlinear theory.  
	\begin{theol}
		\label{prop:representation}
		Let $\phiHL \in C^1(\mathcal H^+)$  be spherically symmetric and assume that $\phiHL$ is   supported away from the past bifurcation sphere. Assume further that $\phiHL$ has finite energy  along the event horizon, i.e.\ that  \begin{align}
		& E_1[\phiHL] < +\infty.
		\end{align}
		Let $\phil$ be the arising solution on the black hole interior with no incoming radiation from the left event horizon.
		
		\textbf{\textup{A.}} Then, for  $v\geq 0$ and $u \in \mathbb R$ with $ r^\ast = \frac 12 (u + v)  \geq 0$, we have
		\begin{align}\label{eq:formultaatcauchy}
		e^{i  \omer r^\ast } 	\phil(u,v) =\frac{\sqrt{2\pi} i r_+}{ r}
		\mathfrak r_{\omer}(0) e^{i \omer u} \left( \int_{-u}^{v}  \phiHL(\tilde v) e^{i\omer  \tilde v}  \d \tilde v \right) + \phi_{\textup{r}}(u,v) +\phi_{\textup{err}}(u,v) ,
		\end{align}
		where $\phi_{\textup{r}}(u,v) $ and $\phi_{\textup{err}}(u,v)$  satisfy the  quantitative bounds
		\begin{align}
		&|\phi_{\textup{r}}(u,v)| \lesssim E_1[\phiHL], \\
		& |\phi_{\textup{err}}(u,v)| \lesssim_\alpha E_1[ \phiHL] \Omega_{RN}^{2-\alpha}(u,v)
		\end{align} uniformly for $v\geq 0, u \in \mathbb R$, $2 r^\ast = v+u \geq 2$ and any fixed $0<\alpha <2$. Further, $\phi_{\textup{r}}(u,v)$ and $\phi_{\textup{err}}(u,v)$
		extend continuously to the right Cauchy horizon. In particular,  $\lim_{n\to+\infty}\phi_{\textup{r}}(u_n,v_n)$ exists for any sequence $(u_n, v_n)\to (u,+\infty)$.

		\textbf{\textup{B.}} If additionally    $\phiHL$ satisfies
		\begin{align}\label{eq:assumptiononfinitenessofE}
		& E_1^\beta[\phiHL]  < +\infty \\ \label{eq:assumptiononfinitenessofFbeta}
		& F^\beta [\phiHL]  < + \infty
		\end{align}
		for some $0 < \beta  \leq 1$, then
		\begin{align}
		\langle u \rangle^\beta |\phil| (u,v) \lesssim  \langle u \rangle^\beta 
		\left|	\mathfrak r_{\omer}(0)    \int_{-u}^{v}  \phiHL(\tilde v) e^{i\omer  \tilde v}  \d \tilde v \right| + E_1^\beta[\phiHL] + F^\beta[\phiHL]
		\end{align}
		uniformly for all $v\geq 2$, $u \in \mathbb R$ such that $v \geq v_\gamma(u)$. 
		
		\textbf{\textup{C.}} Moreover,
		\begin{align} \label{eq:formulaforpartialvphi}
		\partial_v \left(  r e^{i \omer  r^\ast }	\phil(u,v) \right)  =   -i \frac{r_+  e^{i\omer u} }{ \sqrt{2\pi} } \int_{\mathbb R}  \mathcal F[\phiHL   e^{i \omer \cdot} ] (\omega)    \mathfrak t_{\omer}(\omega)     e^{-i\omega v}  \d \omega + \Phi_\textup{error},
		\end{align}
		where $\Phi_\textup{error}$ satisfy the quantitative bounds 
		\begin{align}
		|\Phi_\textup{error}|(u,v)\lesssim_{\alpha} E_1[\phiHL] \Omega_{RN}^{2-\alpha}(u,v)
		\end{align} 
		for any fixed $0<\alpha <2$ and every $(u,v)$ such that $r^\ast(u,v) \geq 1$. 
		
		\textbf{\textup{D.}} Additionally to the assumptions in \textbf{\textup{A.}} and  \textbf{\textup{B.}}, let $\sbr= \sbr(u,v) \in C^1_{u,v} $ with $|\partial_v \sbr|\lesssim \langle v \rangle^{1-2s}$ be arbitrary. Assume further that  \begin{align} & G^s [\phiHL] := \|  \langle v\rangle^s \phiHL  \|_{L^\infty}   + \| \langle v\rangle^s \partial_v \phiHL\|_{L^\infty} < +\infty. 
		\end{align} Then, for all $v\geq v_{\gamma}(u)$
		\begin{align} \nonumber 
		&	\left| \int_{v_{\gamma}(u)}^v e^{i \sbr(u,v')} \partial_{v'} (e^{i\omer r^\ast} r \phil (u,v') ) \d v'\right| \\ &  \;\;\; \;\;\;\; \lesssim \left| \int_{v_{\gamma}(u)}^v  e^{i\sbr(u,v')} e^{i \omer v'} \phiHL   (v')     \d v' \right| +  \langle u \rangle^{2-3s} ( G^s[\phiHL] + E_1[\phiHL]).
		\end{align}
		
		\textbf{\textup{E.}}    Let  $u\in \mathbb R $ be arbitrary  and assume that $\phiHL$ is such that $\left\| 	\partial_v \left(  r e^{i \omer  r^\ast }	\phil(u,v) \right) \right\|_{L^1_v} < +\infty$.
		
		\begin{itemize}
			\item Assume in addition that $P_{\delta} \phiHL \in L^1_v(\mathbb R)$ for some $\delta >0$. Then,
			$$    
			\|\phiHL\|_{L^1_v} \lesssim_\delta \left\| 	\partial_v \left(  r e^{i \omer  r^\ast }	\phil(u,v) \right) \right\|_{L^1_v}  + E_1[\phiHL]  +  \| P_{\delta}  \phiHL \|_{L^1_v}. 
			$$
			\item If $0< |q_0 e| <\epsilon(M,e,m^2)$  or $(q_0,m^2)\in   \{0\} \times \mathbb R - D(M,e)$ as in \cref{lem:zerosornot}, then 
			$$    
			\|\phiHL\|_{L^1_v} \lesssim \left\| 	\partial_v \left(  r e^{i \omer  r^\ast }	\phil(u,v) \right) \right\|_{L^1_v}  + E_1[\phiHL]. 
			$$
		\end{itemize}

		\begin{proof}[Proof of \cref{prop:representation}]
			
			\textbf{Part A.} We use the representation formula \eqref{eq:representationformula} in \cref{lem:representation}  and have
			\begin{align}\nonumber
			\phil(u,v) =  \frac{r_+}{ \sqrt{2\pi} r}  \textup{p.v.}\int_{\mathbb R} & \Big[  \mathcal F[\phiHL \chi_{\leq v}] (\omega) \\ & \frac{ \mathfrak r(\omega) \tilde \uchr(\omega, r^\ast) e^{i(\omega-\omer)r^\ast} + \mathfrak t(\omega)  \tilde \uchl(\omega,r^\ast) e^{-i(\omega - \omer)r^\ast} }{\omega - \omer}
			e^{- i \omega t} \Big] \d \omega. \end{align}
			After a change of variables $\omega \mapsto \omega+\omer$, we obtain
			\begin{align} \nonumber
			\phil(u,v) = \frac{r_+ e^{-i \omer r^\ast} e^{i\omer u} }{ \sqrt{2\pi} r}  \textup{p.v.}\int_{\mathbb R} & \Big[ \mathcal F[\phiHL \chi_{\leq v} e^{i \omer \cdot} ] (\omega)  \\ & \frac{ \mathfrak r_{\omer}(\omega) \tilde \uchr(\omega+\omer , r^\ast) e^{i \omega u} + \mathfrak t_{\omer}(\omega)  \tilde \uchl(\omega+\omer,r^\ast) e^{-i\omega v} }{\omega } \Big] 
			\d \omega, \label{eq:equationafterchangeofvar}
			\end{align}
			where $\mathfrak r_{\omer}(\omega) = \mathfrak r( \omega + \omer)$ and $\mathfrak t_{\omer} (\omega) = \mathfrak t (\omega + \omer)$.
			
			We now expand the numerator and obtain
			\begin{align}\label{eq:expansionofr}
			\mathfrak r_{\omer} (\omega) \tilde\uchr(\omega+\omer, r^\ast)  & = \mathfrak r_{\omer}(0)  + (\mathfrak r_{\omer} (\omega) - \mathfrak  r_{\omer}(0) ) + \mathfrak r_{\omer}(\omega) (\tilde \uchr(\omer + \omega, r^\ast) - 1)\\
			& = \mathfrak r_{\omer}(0) \label{eq:principal} \\ & \;\; + (\mathfrak r_{\omer} (\omega) - \mathfrak  r_{\omer}(0) )   \label{eq:error1} \\ & \;\; +\mathfrak r_{\omer}(\omega) (\tilde \uchr(\omer + \omega, r^\ast) - \tilde \uchr(\omer, r^\ast) )
			\label{eq:error2}	\\ & \;\;+ \mathfrak r_{\omer}(0) ( \tilde \uchr(\omer  , r^\ast) - 1)
			\label{eq:error3}	\\ & \;\;+ (\mathfrak r_{\omer}(\omega) - \mathfrak r_{\omer}(0) ) ( \tilde \uchr(\omer  , r^\ast) - 1)\label{eq:error4}
			\end{align}
			as well as  \begin{align}\label{eq:expansionoft}
			\mathfrak t_{\omer} (\omega) \tilde\uchl(\omega+\omer, r^\ast)  & = \mathfrak t_{\omer}(0)  + (\mathfrak t_{\omer} (\omega) - \mathfrak  t_{\omer}(0) ) + \mathfrak t_{\omer}(\omega) (\tilde \uchl(\omer + \omega, r^\ast) - 1)\\
			& = \mathfrak t_{\omer}(0) \label{eq:principalt} \\ & \;\; + (\mathfrak t_{\omer} (\omega) - \mathfrak  t_{\omer}(0) )   \label{eq:errort1} \\ & \;\; +\mathfrak t_{\omer}(\omega) (\tilde \uchl(\omer + \omega, r^\ast) - \tilde \uchl(\omer, r^\ast) )
			\label{eq:errort2}	\\ & \;\;+ \mathfrak t_{\omer}(0) ( \tilde \uchl(\omer  , r^\ast) - 1)
			\label{eq:errort3}	\\ & \;\;+ (\mathfrak t_{\omer}(\omega) - \mathfrak t_{\omer}(0) ) ( \tilde \uchl(\omer  , r^\ast) - 1).\label{eq:errort4}
			\end{align} 
			We  write 
			\begin{align}
			\frac{\mathfrak r_{\omer} (\omega)}{\omega} = \frac{\mathfrak r_{\omer}(0)}{\omega} + \mathfrak r^{\textup{re}}_{\omer} (\omega), \frac{\mathfrak t_{\omer} (\omega)}{\omega} = \frac{\mathfrak t_{\omer}(0)}{\omega} + \mathfrak t^{\textup{re}}_{\omer} (\omega)
			\end{align}
			where \begin{align}
			\mathfrak r^{\textup{re}}_{\omer}(\omega) := \frac{\mathfrak r_{\omer} (\omega) -\mathfrak r_{\omer} (0) }{\omega} \text{ and } 	\mathfrak t^{\textup{re}}_{\omer}(\omega) := \frac{\mathfrak r_{\omer} (\omega) -\mathfrak t_{\omer} (0) }{\omega} 
			\end{align}
			are real-analytic.
			
			In the following we will estimate each term from \eqref{eq:principal}--\eqref{eq:errort4} independently. We start with the main term coming from \eqref{eq:principal}.
			\begin{lem}\label{lem:principaltermr1}
				We have \begin{align}
				e^{i \omer r^\ast} \phi_{\textup{mainR}} (u,v) &:=   \frac{r_+ e^{i\omer u} }{ \sqrt{2\pi} r}  \textup{p.v.}\int_{\mathbb R} \mathcal{F}[\phiHL \chi_{\leq v} e^{i \omer \cdot} ] (\omega) \frac{  \mathfrak r_{\omer} (0)  }{\omega }
				e^{i \omega u} 	\d \omega
				\end{align} 
				satisfies
				\begin{align}
				e^{i \omer r^\ast} \phi_{\textup{mainR}}(u,v) =   i \pi  \frac{r_+ e^{i\omer u} \mathfrak r_{\omer} (0) }{ \sqrt{2\pi} r} \int_{\mathbb R} \phiHL(\tilde v) \chi_{\leq v}(\tilde v) e^{i \omer  u } \textup{sgn}(\tilde v + u) \d \tilde v 		\end{align}
				\begin{proof}
					This follows directly from the fact that $\mathcal F [ \textup{p.v.} (\frac{1}{x}) ]  = i \pi \textup{sgn} $.
				\end{proof}
			\end{lem} 
			\begin{lem}\label{lem:estphierrorr1}
				We have that \begin{align}
				e^{i \omer r^\ast} \phi_{\textup{errorR1}} (u,v) &:=   \frac{r_+ e^{i\omer u} }{ \sqrt{2\pi} r}  \textup{p.v.}\int_{\mathbb R} \mathcal{F}[\phiHL \chi_{\leq v} e^{i \omer \cdot} ] (\omega) \frac{\mathfrak r_{\omer} (\omega) - \mathfrak r_{\omer} (0)  }{\omega }
				e^{i \omega u} 	\d \omega
				\\	& = \frac{r_+   e^{i\omer u} }{ \sqrt{2\pi} r}   \int_{\mathbb R} \mathcal{F}[\phiHL \chi_{\leq v} e^{i \omer \cdot} ] (\omega) \mathfrak r^{\textup{re}}_{\omer}(\omega)
				e^{i \omega u} 		\d \omega 
				\end{align}
				extends continuously to the Cauchy horizon and satisfies \begin{align}\label{eq:decayofphierrore1R1}
				|\phi_{\textup{errorR1}}(u,v) |\lesssim E_{1}[\phiHL].
				\end{align}
				If additionally,  $E_1^\beta[\phiHL]<+\infty$ for some $0 < \beta\leq 1$, we further have 
				\begin{align}\label{eq:decayofphierrorR1}
				|  \langle u \rangle^\beta \phi_{\textup{errorR1}}(u,v) |\lesssim E_{1}^\beta[\phiHL]
				\end{align}
				for all $r^\ast\geq 0$. 
				\begin{proof}
					It suffices to show both claims  for $\int_{\mathbb R} \mathcal{F}[\phiHL \chi_{\leq v} e^{i \omer \cdot} ] (\omega) \mathfrak  r^{\textup{re}}_{\omer}(\omega) e^{i \omega u}  \d \omega $. We begin by showing  \eqref{eq:decayofphierrorR1} under the assumption $E_1^\beta[\phiHL] <\infty$.   We will  use the notation $\langle \partial_\omega\rangle^\beta$ to denote the Fourier multiplier with $(1+|u|^2)^{\frac \beta 2}$, where $u$ is the dual variable to $\omega$.   Using this, we estimate
					\begin{align}\nonumber 
					\left| \langle u \rangle^\beta \int_{\mathbb R} \mathcal{F}[\phiHL \chi_{\leq v} e^{i \omer \cdot} ] (\omega)  \mathfrak r^{\textup{re}}_{\omer}(\omega)e^{i \omega u}  \d  \omega  \right| & =  \left| \int_{\mathbb R} \langle \partial_\omega\rangle^\beta\left(  \mathcal{F}[\phiHL \chi_{\leq v} e^{i \omer \cdot} ] (\omega) \mathfrak r^{\textup{re}}_{\omer}(\omega) \right) e^{i \omega u}  \d  \omega  \right| \\ \nonumber 
					& \leq \left\|  \langle \partial_\omega \rangle^\beta \left(  \mathcal{F}[\phiHL \chi_{\leq v} e^{i \omer \cdot} ] \mathfrak r^{\textup{re}}_{\omer}  \right) \right\|_{L^1_\omega}  \\ \nonumber & \leq \left\|  \langle \partial_\omega \rangle^\beta \left(  \langle \omega \rangle  \mathcal{F}[\phiHL \chi_{\leq v} e^{i \omer \cdot} ]  \right) \right\|_{L^2_\omega}  \left\| \langle \omega \rangle^{-1} \mathfrak r^{\textup{re}}_{\omer} \right\|_{L^2_\omega} 
					\\ \nonumber & +     \left\|  \langle \omega \rangle    \mathcal{F}[\phiHL \chi_{\leq v} e^{i \omer \cdot} ]   \right\|_{L^2_\omega}  \left\| \langle \partial_\omega \rangle^\beta \left( \langle \omega \rangle^{-1} \mathfrak r^{\textup{re}}_{\omer} \right) \right\|_{L^2_\omega} 
					\\ \nonumber & \lesssim  \| \langle v \rangle^\beta   \phiHL\chi_{\leq v} e^{i \omer \cdot} \|_{L^2_v} + \| \langle v \rangle^\beta   \partial_v (\phiHL \chi_{\leq v} e^{i \omer \cdot} )  \|_{L^2_v}\\
					& \lesssim E_1^\beta[\phiHL] \label{eq:estiamtesusedkatoponce}
					\end{align}
					in view of a Kato--Ponce inequality (see e.g. \cite[Theorem~1]{MR3200091})  and 
					\begin{align}\label{eq:boundedofr}
					&   \left\|   \langle \omega \rangle^{-1} \mathfrak r^{\textup{re}}_{\omer}  \right\|_{L^2(\mathbb R_\omega)}\lesssim 1,
					\\ &
					\left\| \langle \partial_\omega \rangle^\beta \left( \langle \omega \rangle^{-1} \mathfrak r^{\textup{re}}_{\omer} \right) \right\|_{L^2(\mathbb R_\omega)} \lesssim 1,
					\end{align}
					which follow from the definition of $\mathfrak r^{\textup{re}}_{\omer},$ $\mathfrak t^{\textup{re}}_{\omer}$ as well as \cref{lem:estimatesonderiativesofrandt}.
					Now, note that the previous estimates for $\beta=0$ give \eqref{eq:decayofphierrore1R1}.  
					
					For the continuous extendibility across the Cauchy horizon we need to show that for $(u_n, v_n) \to (u_0, +\infty)$, the limit \begin{align}\lim_{n\to\infty} \int_{\mathbb R} \mathcal{F}[\phiHL \chi_{\leq v_n} e^{i \omer \cdot} ] (\omega) \mathfrak r^{\textup{re}}_{\omer}(\omega)
					e^{i \omega u_n} 		\d \omega \end{align}
					exists and that the limiting function is continuous. In view of the triangle inequality we have 
					\begin{align} \nonumber
					& \left|   \int_{\mathbb R} \mathcal{F}[\phiHL \chi_{\leq v_n} e^{i \omer \cdot} ] (\omega) \mathfrak r^{\textup{re}}_{\omer}(\omega)
					e^{i \omega u_n} 		\d \omega  -  \int_{\mathbb R} \mathcal{F}[\phiHL e^{i \omer \cdot} ] (\omega) \mathfrak r^{\textup{re}}_{\omer}(\omega)
					e^{i \omega u_0} 		\d \omega  \right| \\
					&\lesssim  \int_{\mathbb R} | \mathcal{F}[\phiHL  e^{i \omer \cdot} ] (\omega) \mathfrak r^{\textup{re}}_{\omer}(\omega)|
					|e^{i \omega u_n} - e^{i \omega u_0} |		\d \omega +  \int_{\mathbb R} |\mathcal{F}[\phiHL (1-\chi_{\leq v_n}) e^{i \omer \cdot} ] (\omega)| |\mathfrak r^{\textup{re}}_{\omer}(\omega)|\d \omega . \label{eq:twotermslefttoest}
					\end{align}
					In the first term of \eqref{eq:twotermslefttoest} we  apply dominated convergence to interchange the limit with the integral which is justified as 
					\begin{align} \label{eq:limititcontinuous}
					\int_{\mathbb R} | \mathcal{F}[\phiHL  e^{i \omer \cdot} ] (\omega) \mathfrak r^{\textup{re}}_{\omer}(\omega)|
					|e^{i \omega u_n} - e^{i \omega u_0} |		\d \omega \lesssim  \int_{\mathbb R} | \mathcal{F}[\phiHL  e^{i \omer \cdot} ] (\omega) \mathfrak r^{\textup{re}}_{\omer}(\omega)|
					\d \omega \lesssim E_1[\phiHL]		\end{align} in view of
					\eqref{eq:boundedofr}.  
					For the second term in \eqref{eq:twotermslefttoest} we have that  
					\begin{align}\label{eq:estimateforhatfwithr}
					\int_{\mathbb R} |\mathcal{F}[\phiHL (1-\chi_{\leq v_n}) e^{i \omer \cdot} ] (\omega)| |\mathfrak r^{\textup{re}}_{\omer}(\omega)|\d \omega \lesssim \left( \int_{\mathbb R} |\partial_{\tilde v}  \left( \phiHL(\tilde v ) (1-\chi_{\leq v_n } (\tilde v)) \right) |^2 \d \tilde v \right)^{\frac 12}   \to 0  
					\end{align}
					as $n \to \infty$ since $E_1[\phiHL] < +\infty$. That the limit is continuous also follows from \eqref{eq:limititcontinuous}. 
				\end{proof}
			\end{lem}
			\begin{lem}\label{lem:errorr2}
				We have that \begin{align}
				\nonumber	e^{i \omer r^\ast} \phi_{\textup{errorR2}} (u,v) :=   \frac{r_+ e^{i\omer u} }{ \sqrt{2\pi} r}  \textup{p.v.}\int_{\mathbb R} & \mathcal{F}[\phiHL \chi_{\leq v} e^{i \omer \cdot} ] (\omega)
				\\ & \cdot \frac{\mathfrak r_{\omer}(\omega) (\tilde \uchr(\omer + \omega, r^\ast) - \tilde \uchr(\omer, r^\ast) )}{\omega }
				e^{i \omega u} 	\d \omega
				\end{align}
				converges to zero towards the Cauchy horizon and satisfies the quantitative bound 			\begin{align}\label{eq:decayofphierrorr2}
				|  \phi_{\textup{errorR2}}(u,v) |\lesssim \Omega^2_{{RN}}(u,v) E_{1}[\phiHL]
				\end{align}
				for $r^\ast \geq 1$. 
			\end{lem}
			\begin{proof}
				We estimate
				\begin{align}
				&  \left|  \frac{\mathfrak r_{\omer}(\omega) (\tilde \uchr(\omer + \omega, r^\ast) - \tilde \uchr(\omer, r^\ast) )}{\omega } \right| \\  & \lesssim \sup_{|\omega|\leq 1} |\partial_\omega \tilde \uchr(\omer + \omega, r^\ast)| +  \sup_{|\omega|\geq 1 } |\tilde \uchr(\omer + \omega, r^\ast) - \tilde \uchr(\omer, r^\ast)|
				\\ & \lesssim \Omega_{RN}^2 \end{align}
				in view of \cref{lem:estimatesonderiativesonuchrl} and  \cref{prop:tr2}. Now, \eqref{eq:decayofphierrorr2} follows from a direct application of the Cauchy--Schwarz inequality. 
			\end{proof} 
			\begin{lem}\label{lem:lemmaphir3}
				We have that \begin{align}
				e^{i \omer r^\ast} \phi_{\textup{errorR3}} (u,v) &:=   \frac{r_+ e^{i\omer u} }{ \sqrt{2\pi} r}  \textup{p.v.}\int_{\mathbb R} \mathcal{F}[\phiHL \chi_{\leq v} e^{i \omer \cdot} ] (\omega) \frac{\mathfrak r_{\omer}(0) ( \tilde \uchr(\omer  , r^\ast) - 1)}{\omega }
				e^{i \omega u} 	\d \omega\\
				&	=   \mathfrak r_{\omer}(0) ( \tilde \uchr(\omer  , r^\ast) - 1) \frac{r_+ e^{i\omer u} }{ \sqrt{2\pi} r}  \textup{p.v.}\int_{\mathbb R} \mathcal{F}[\phiHL \chi_{\leq v} e^{i \omer \cdot}] (\omega) \frac{1}{\omega }
				e^{i \omega u} 	\d \omega
				\end{align}
				converges to zero towards the Cauchy horizon and satisfies the quantitative bound 			\begin{align}\label{eq:decayofphierrorr3}
				|  \phi_{\textup{errorR3}}(u,v) |\lesssim \Omega^2_{RN}(u,v) \| \phiHL \chi_{\leq v+u} \|_{L^1_v} \lesssim \Omega^2_{{RN}}(u,v) E_1[\phiHL] \langle r^\ast \rangle^{\frac 12} \lesssim_\alpha \Omega_{RN}^{2-\alpha}(u,v) E_1[\phiHL]
				\end{align}
				for $r^\ast \geq 1$ and any $\alpha >0$. 
			\end{lem}
			\begin{proof}
				It suffices to control the principal value integral. A direct computation using that $\mathcal F [ \textup{p.v.} (\frac{1}{x}) ]  = i \pi \textup{sgn} $  yields
				\begin{align}
				\left|   \textup{p.v.}\int_{\mathbb R} \mathcal{F}[\phiHL \chi_{\leq v} e^{i \omer \cdot} ] (\omega) \frac{1}{\omega }
				e^{i \omega u} 	\d \omega \right| \lesssim  \int_{\mathbb R} | \phiHL(\tilde v - u ) \chi_{\leq v} (\tilde v-u) | \d \tilde v \leq \| \phiHL \chi_{\leq v + u } \|_{L^1(\mathbb R)}.
				\end{align}
				The second inequality in \eqref{eq:decayofphierrorr3}  is now a consequence of the Cauchy--Schwarz inequality.
			\end{proof}
			Now, we are in the position to control the last term as follows.
			\begin{lem}\label{lem:lemmar4}
				We have that \begin{align}
				e^{i \omer r^\ast} \phi_{\textup{errorR4}} (u,v) &:=   \frac{r_+ e^{i\omer u} }{ \sqrt{2\pi} r}  \textup{p.v.}\int_{\mathbb R} \mathcal{F}[\phiHL \chi_{\leq v} e^{i \omer \cdot} ] (\omega) \frac{(\mathfrak r_{\omer}(\omega) - \mathfrak r_{\omer}(0) ) ( \tilde \uchr(\omer  , r^\ast) - 1)}{\omega }
				e^{i \omega u} 	\d \omega\\
				&	=   ( \tilde \uchr(\omer  , r^\ast) - 1)	e^{i \omer r^\ast} \phi_{\textup{errorR1}} 
				\end{align}
				converges to zero towards the Cauchy horizon and satisfies the quantitative bound 			\begin{align}\label{eq:decayofphierrorr4}
				|  \phi_{\textup{errorR4}}(u,v) |\lesssim \Omega^2_{{RN}}(u,v) E_1[\phiHL]
				\end{align}
				for $r^\ast \geq 0$.
			\end{lem}
			\begin{proof}
				This follows immediately from \cref{lem:estphierrorr1}. 
			\end{proof}
			
			Now, we turn to the terms arising from the transmission coefficient. Completely analogous to \cref{lem:principaltermr1} we obtain
			\begin{lem}
				We have that \begin{align}
				e^{i \omer r^\ast} \phi_{\textup{mainT}} (u,v) &:=   \frac{r_+ e^{i\omer u} }{ \sqrt{2\pi} r}  \textup{p.v.}\int_{\mathbb R} \mathcal{F}[\phiHL \chi_{\leq v} e^{i \omer \cdot} ] (\omega) \frac{  \mathfrak t_{\omer} (0)  }{\omega }
				e^{-i \omega v} 	\d \omega
				\end{align} 
				satisfies
				\begin{align}
				e^{i \omer r^\ast} \phi_{\textup{mainT}}(u,v) = i \pi  \frac{r_+ e^{i\omer u} \mathfrak t_{\omer} (0) }{ \sqrt{2\pi} r} \int_{\mathbb R} \phiHL(\tilde v) \chi_{\leq v}(\tilde v) e^{i \omer  \tilde v  } \textup{sgn}(\tilde v -v) \d \tilde v .		\end{align}
			\end{lem} 
			\begin{lem}\label{lem:estphierrort1} 
				We have that \begin{align}
				e^{i \omer r^\ast} \phi_{\textup{errorT1}} (u,v) &:=   \frac{r_+ e^{i\omer u} }{ \sqrt{2\pi} r}  \textup{p.v.}\int_{\mathbb R} \mathcal{F}[\phiHL \chi_{\leq v} e^{i \omer \cdot} ] (\omega) \frac{\mathfrak t_{\omer} (\omega) - \mathfrak t_{\omer} (0)  }{\omega }
				e^{-i \omega v} 	\d \omega
				\\	& = \frac{r_+   e^{i\omer u} }{ \sqrt{2\pi} r}   \int_{\mathbb R} \mathcal{F}[\phiHL \chi_{\leq v} e^{i \omer \cdot} ] (\omega) \mathfrak t^{\textup{re}}_{\omer}(\omega)
				e^{-i \omega v} 		\d \omega 
				\end{align}
				extends continuously to zero at the right Cauchy horizon, i.e. for $v \to +\infty$ and $u\to u_0$. If in addition $E_1^\beta[\phiHL]<\infty$, then we have the quantitative decay  
				\begin{align}\label{eq:decayofphierrorr1}
				|\phi_{\textup{errorT1}}(u,v) |\lesssim \langle v \rangle^{-\beta}  E_1^\beta[\phiHL].
				\end{align}
				\begin{proof}
					We first show the first claim without assuming that $E_1^\beta[\phiHL] < +\infty$. Doing the analogous estimate as in \eqref{eq:estimateforhatfwithr} it suffices to show that \begin{align}
					\left| 	 \int_{\mathbb R} \mathcal{F}[\phiHL e^{i \omer \tilde v} ] (\omega) \mathfrak t^{\textup{re}}_{\omer}(\omega)
					e^{-i \omega v} 		\d \omega \right|   
					\end{align}
					tends to zero as $v\to+\infty$. Thus, it suffices to show that  $ v \mapsto  \int_{\mathbb R} \mathcal{F}[\phiHL e^{i \omer \tilde v} ] (\omega) \mathfrak t^{\textup{re}}_{\omer}(\omega)
					e^{-i \omega v} \d \omega $ is an $H^1$ function. This again follows from \begin{align}
					\int_{\mathbb R} (1+\omega^2) |\mathcal{F}[\phiHL e^{i \omer \tilde v} ] (\omega)|^2 |\mathfrak t^{\textup{re}}_{\omer}(\omega)|^2 \d \omega \lesssim E_1[\phiHL] \sup_{\omega \in \mathbb R} |\mathfrak t^{\textup{re}}_{\omer}(\omega)| \lesssim E_1[\phiHL].
					\end{align} 
					
					We will now proceed to show the quantitative decay assuming that $E^\beta[\phiHL] < \infty.$ In this case we have \begingroup
					\allowdisplaybreaks
					\begin{align*}
					&	\left| 	 \int_{\mathbb R} \mathcal{F}[\phiHL \chi_{\leq v} e^{i \omer \tilde v} ] (\omega) \mathfrak t^{\textup{re}}_{\omer}(\omega)
					e^{-i \omega v} 		\d \omega \right|  = \left| \frac{1}{v}   \int_{\mathbb R} \mathcal{F}[\phiHL \chi_{\leq v} e^{i \omer \tilde v} ] (\omega) \mathfrak t^{\textup{re}}_{\omer}(\omega)
					\partial_{\omega} 	e^{-i \omega v} 		\d \omega \right| 
					\\ & \hspace{1cm}  \lesssim   \left| \frac{1}{v}   \int_{\mathbb R} 	\partial_{\omega}	 \mathcal{F}[ \phiHL\chi_{\leq v} e^{i \omer \tilde v} ] (\omega) \mathfrak t^{\textup{re}}_{\omer}(\omega)
					e^{-i \omega v} 		\d \omega \right| \\ & \hspace{1cm} +
					\left| \frac{1}{v}   \int_{\mathbb R} 		 \mathcal{F}[\phiHL \chi_{\leq v} e^{i \omer \tilde v} ] (\omega) (\partial_{\omega}\mathfrak t^{\textup{re}}_{\omer}(\omega) ) 
					e^{-i \omega v} 		\d \omega \right| \\
					& 	\hspace{1cm}  \lesssim  \left| \frac{1}{v}   \int_{|\omega|\leq 1} 		 \mathcal{F}[\tilde v \phiHL \chi_{\leq v} e^{i \omer \tilde v} ] (\omega) \mathfrak t^{\textup{re}}_{\omer}(\omega)
					e^{-i \omega v} 		\d \omega \right|  \\ &\hspace{1cm}  +
					\left| \frac{1}{v}   \int_{|\omega|\leq 1} 		 \mathcal{F}[\phiHL \chi_{\leq v} e^{i \omer \tilde v} ] (\omega) (\partial_{\omega}\mathfrak t^{\textup{re}}_{\omer}(\omega) ) 
					e^{-i \omega v} 		\d \omega \right| 
					\\	& \hspace{1cm}+ \left| \frac{1}{v}   \int_{|\omega|\geq 1} 		 \mathcal{F}[\tilde v   \phiHL \chi_{\leq v} e^{i \omer \tilde v} ] (\omega) \mathfrak t^{\textup{re}}_{\omer}(\omega)
					e^{-i \omega v} 		\d \omega \right|  \\
					& \hspace{1cm} +	\left| \frac{1}{v}   \int_{|\omega|\geq 1} 		 \mathcal{F}[\phiHL \chi_{\leq v} e^{i \omer \tilde v} ] (\omega) (\partial_{\omega}\mathfrak t^{\textup{re}}_{\omer}(\omega) ) 
					e^{-i \omega v} 		\d \omega \right| \\   & \hspace{1cm}
					\lesssim \frac{1}{v} \|\tilde v   \phiHL \chi_{\leq v} e^{i \omer \tilde v } \|_{L^2_{\tilde v} }  \|\partial_\omega \mathfrak t^{\textup{re}}_{\omer}\|_{L^2_\omega[-1,1]} 
					+ \frac{1}{v} \| \phiHL \chi_{\leq v} e^{i \omer \tilde v} \|_{L^2_{\tilde v} } \|\mathfrak t^{\textup{re}}_{\omer}\|_{L^2_\omega[-1,1]}  				\\& \hspace{1cm}  + \frac{1}{v} \|\partial_{\tilde v} (\tilde v   \phiHL \chi_{\leq v} e^{i \omer \tilde v } ) \|_{L^2_{\tilde v} }  \|\omega^{-1}\partial_\omega \mathfrak t^{\textup{re}}_{\omer} \|_{L^2_\omega(\mathbb R - [-1,1] ) } 
					\\
					&\hspace{1cm}  + \frac{1}{v} \| \partial_{\tilde v} (\phiHL \chi_{\leq v} e^{i \omer \tilde v }) \|_{L^2_{\tilde v} } \| \omega^{-1} \mathfrak t^{\textup{re}}_{\omer}\|_{L^2_\omega(\mathbb R-[-1,1])} 
					\\
					&\hspace{1cm}	\lesssim \frac{1}{v^\beta }  E_1^\beta[\phiHL]   
					\end{align*}
					\endgroup
					since $   \|\partial_\omega \mathfrak t^{\textup{re}}_{\omer}\|_{L^2_\omega[-1,1]}, \|\mathfrak t^{\textup{re}}_{\omer}\|_{L^2_\omega[-1,1]} , \|\omega^{-1}\partial_\omega \mathfrak t^{\textup{re}}_{\omer} \|_{L^2_\omega(\mathbb R - [-1,1] ) } , \| \omega^{-1} \mathfrak t^{\textup{re}}_{\omer}\|_{L^2_\omega(\mathbb R-[-1,1])} \lesssim 1$. 
				\end{proof}
			\end{lem}
			Analogously to \cref{lem:errorr2} we have
			\begin{lem}
				We have that \begin{align} \nonumber
				& e^{i \omer r^\ast} \phi_{\textup{errorT2}} (u,v) \\ \;\;\;\; &:=   \frac{r_+ e^{i\omer u} }{ \sqrt{2\pi} r}  \textup{p.v.}\int_{\mathbb R} \mathcal{F}[\phiHL \chi_{\leq v} e^{i \omer \cdot} ] (\omega) \frac{\mathfrak t_{\omer}(\omega) (\tilde \uchl(\omer + \omega, r^\ast) - \tilde \uchl(\omer, r^\ast) )}{\omega }
				e^{- i \omega v} 	\d \omega
				\end{align}
				converges to zero towards the Cauchy horizon and satisfies the quantitative bound 			\begin{align}\label{eq:decayofphierrort2}
				|  \phi_{\textup{errorT2}}(u,v) |\lesssim \Omega^2_{{RN}}(u,v) E_{1}[\phiHL]
				\end{align}
				for $r^\ast \geq 1$.
			\end{lem}
			Analogously to \cref{lem:lemmaphir3} we further obtain
			\begin{lem}
				We have that \begin{align}
				e^{i \omer r^\ast} \phi_{\textup{errorT3}} (u,v) &:=   \frac{r_+ e^{i\omer u} }{ \sqrt{2\pi} r}  \textup{p.v.}\int_{\mathbb R} \mathcal{F}[\phiHL \chi_{\leq v} e^{i \omer \cdot} ] (\omega) \frac{\mathfrak t_{\omer}(0) ( \tilde \uchl(\omer  , r^\ast) - 1)}{\omega }
				e^{- i \omega v} 	\d \omega\\
				&	=   \mathfrak t_{\omer}(0) ( \tilde \uchr(\omer  , r^\ast) - 1) \frac{r_+ e^{i\omer u} }{ \sqrt{2\pi} r}  \textup{p.v.}\int_{\mathbb R} \mathcal{F}[\phiHL \chi_{\leq v} e^{i \omer \cdot}] (\omega) \frac{1}{\omega }
				e^{- i \omega v} 	\d \omega
				\end{align}
				converges to zero towards the Cauchy horizon and satisfies the quantitative bound 			\begin{align}\label{eq:decayofphierrort3}
				|  \phi_{\textup{errorT3}}(u,v) |\lesssim_\alpha \OmegaRN^{2-\alpha}(u,v) E_1[\phiHL]
				\end{align}
				for $r^\ast \geq 1$. 
			\end{lem}
			Finally, completely analogous to \cref{lem:lemmar4} we have
			\begin{lem}
				We have that \begin{align} \nonumber
				e^{i \omer r^\ast} \phi_{\textup{errorT4}} (u,v) &:=   \frac{r_+ e^{i\omer u} }{ \sqrt{2\pi} r}  \textup{p.v.}\int_{\mathbb R} \mathcal{F}[\phiHL \chi_{\leq v} e^{i \omer \cdot} ] (\omega) \\ \nonumber &\hspace{2cm} \cdot \frac{(\mathfrak t_{\omer}(\omega) - \mathfrak t_{\omer}(0) ) ( \tilde \uchl(\omer  , r^\ast) - 1)}{\omega }
				e^{- i \omega v} 	\d \omega\\
				&	=   ( \tilde \uchr(\omer  , r^\ast) - 1)	e^{i \omer r^\ast} \phi_{\textup{errorT1}} \label{lem:estiamtesont4}
				\end{align}
				converges to zero towards the Cauchy horizon and satisfies the quantitative bound 			\begin{align}\label{eq:decayofphierrort1}
				|  \phi_{\textup{errorR4}}(u,v) |\lesssim \Omega^2_{{RN}}(u,v) E_1[\phiHL]
				\end{align}
				for $r^\ast \geq 1$.
			\end{lem}
			
			Having estimated each term independently in the integral appearing in \eqref{eq:equationafterchangeofvar} and noting that 
			\begin{align}
			e^{i\omer r^\ast} ( \phi_{\textup{mainR}} + \phi_{\textup{mainT}} )(u,v) = \frac{\sqrt{2\pi} i r_+}{ r}
			\mathfrak r_{\omer}(0) e^{i \omer u} \left( \int_{-u}^{v}  \phiHL(\tilde v) e^{i\omer  \tilde v}  \d \tilde v \right)
			\end{align} in view of $\mathfrak r_{\omer}(0) = - \mathfrak t_{\omer}(0)$, we finally obtain \eqref{eq:formultaatcauchy} with 
			\begin{align}
			\phi_\textup{r} = e^{i \omer r^\ast} \phi_{\textup{errorR1}}
			\end{align}
			and 
			\begin{align}
			\phi_{\textup{error}} = e^{i\omer r^\ast} ( \phi_{\textup{errorR2}} 
			+\phi_{\textup{errorR3}}
			+\phi_{\textup{errorR4}} 
			+\phi_{\textup{errorT1}}
			+\phi_{\textup{errorT2}}
			+\phi_{\textup{errorT3}}
			+\phi_{\textup{errorT4}} ). \label{eq:secondline}
			\end{align}
			The bounds and continuity statement for  $\phi_\textup{r}$ and $\phi_{\textup{error}}$ now follow from \cref{lem:estphierrorr1} and \eqref{lem:estiamtesont4}.
			
			\textbf{Part B.} 
			In view of Part A and the fact that $\Omega_{RN}$ decays exponentially in $r^\ast = \frac 12 (u + v)$ towards the Cauchy horizon, it suffices to show that 
			\begin{align}
			\langle u \rangle^\beta \left| \int_{-u}^{v+1} \chi_{\leq v}(\tilde v) \phiHL(\tilde v) e^{i\omer  \tilde v}  \d \tilde v \right|+ \langle u \rangle^\beta | \phi_{\textup{r}}(u,v) |\lesssim   F^\beta[\phiHL] + E_1^\beta[\phiHL] 
			\end{align}
			as we consider the region $v \geq |u| +  \frac{\log(v)}{2|K_-|}$ in which $\Omega^2_{RN}(u,v) \lesssim \langle v \rangle^{-1}$.
			Now, the claim is a direct consequence of the second parts of \cref{lem:estphierrorr1} and \cref{lem:estphierrort1} together with the assumptions \eqref{eq:assumptiononfinitenessofFbeta} and \eqref{eq:assumptiononfinitenessofE}.
			
			\textbf{Part C.}
			We will now consider $\partial_v(r e^{i \omer r^\ast} \phil)$. We use the second part of \cref{lem:representation} and end up with 
			\begin{align}\nonumber
			\partial_v(e^{i \omer r^\ast} r	\phil(u,v) )  = 	&\frac{r_+  e^{i\omer u} }{ \sqrt{2\pi} }  \textup{p.v.}\int_{\mathbb R}  \Big[ \mathcal F[\phiHL \chi_{\leq v_1} e^{i \omer \cdot} ] (\omega) \\ &   \cdot \frac{ \mathfrak r_{\omer}(\omega) \partial_{v}  \tilde \uchr(\omega + \omer , r^\ast) e^{i \omega u} + \mathfrak t_{\omer}(\omega)  \partial_{v} ( \tilde \uchl(\omega+\omer,r^\ast) e^{-i\omega v} )}{\omega } \Big] \d \omega\label{eq:termsofpartialvphi}
			\end{align}
			for   $v_1>v$. 
			Since $\partial_v \tilde \uchr$ and $\partial_v \tilde \uchl$ are bounded uniformly in absolute value by $\Omega^2_{{RN}}$ in view of \cref{prop:tr2}, the terms of \eqref{eq:termsofpartialvphi} which arise thereof are bounded by $\Omega_{{RN}}^{2-\alpha} E_1[\phiHL]$ for any $\alpha>0$ as  in Part~\textbf{A}. Similarly, $\tilde \uchl - 1$ is bounded by $\Omega^2_{RN}$ and thus, the main term arises from $\partial_{v} ( e^{-i \omega v})$ and we obtain
			\begin{align}
			\partial_v(e^{i \omer r^\ast} r	\phil(u,v) ) =&  -i \frac{r_+  e^{i\omer u} }{ \sqrt{2\pi} } \int_{\mathbb R}  \mathcal F[\phiHL \chi_{\leq v_1} e^{i \omer \cdot} ] (\omega)    \mathfrak t_{\omer}(\omega)     e^{-i\omega v}  \d \omega + \Phi^{v_1}_\textup{error},
			\end{align}
			where $|\Phi^{v_1}_\textup{error}|\lesssim_\alpha \Omega_{RN}^{2-\alpha} E_1[\phiHL]$. Note that $\Phi^{v_1}_\textup{error}$ depends on $v_1$ but the upper bound is uniform in $v_1$. 
			Since $\langle \omega \rangle \mathcal F[\phiHL e^{i \omer \cdot} ]  \in L^2_{\omega} $ and $\langle \omega \rangle^{-1} \mathfrak t_{\omer} \in L_\omega^\infty $ we can take the limit $v_1 \to \infty$ and obtain
			\begin{align}\label{eq:partialvphil}
			\partial_v(e^{i \omer r^\ast} r	\phil(u,v) ) =&  -i \frac{r_+  e^{i\omer u} }{ \sqrt{2\pi} } \int_{\mathbb R}  \mathcal F[\phiHL  e^{i \omer \cdot} ] (\omega)    \mathfrak t_{\omer}(\omega)     e^{-i\omega v}  \d \omega +  \Phi_\textup{error},
			\end{align}
			where $| \Phi_\textup{error}(u,v)|\lesssim_\alpha \OmegaRN^{2-\alpha}(u,v) E_1[\phiHL]$.

			\textbf{Part D.} Note that $\Phi_\textup{error}$ as in Part~\textbf{C} decays proportional to $\Omega_{RN}^{2-\alpha}$ for any $\alpha>0$ and thus,
			\begin{align}\int_{v_\gamma(u)}^{v} |\Phi_\textup{error}| \d v' \lesssim_{\alpha} (\Omega_{RN})^{2-\alpha}(u,v_\gamma(u)) E_1[\phiHL] \lesssim_{\alpha} \langle u \rangle^{-s (2-\alpha)}E_1[\phiHL] \lesssim \langle u \rangle^{2-3s}E_1[\phiHL]\end{align} choosing $\alpha>0$ sufficiently small (recall that $s \leq 1$ therefore $2s > 3s-2$). Thus, it suffices to show the result for the main part in  \eqref{eq:formulaforpartialvphi}. We further write \begin{align}
			\mathfrak t_{\omer} (\omega)  =	    \mathfrak t^0_{\omer}  +  \omega \mathfrak t^1_{\omer}  + \omega^2  \tilde{\mathfrak t}_{\omer} (\omega), \end{align}
			where we note that $ |\tilde{\mathfrak t}_{\omer}| = \left| \frac{	\mathfrak t_{\omer} (\omega) - \mathfrak t^0_{\omer}  -  \omega \mathfrak t^1_{\omer} }{\omega^2} \right|  \lesssim \langle \omega \rangle^{-1}$ and $|\partial_\omega \tilde{\mathfrak t}_{\omer}|\lesssim \langle \omega \rangle^{-1}   $ in view of \cref{cor:tarnsmissionandreflection} and \cref{lem:estimatesonderiativesofrandt}.  Hence,   \begin{align}  \langle v \rangle  \mathcal{F} \left( \tilde{\mathfrak t}_{\omer} \right) (v)   \in L^2(\mathbb R_v) \label{eq:lpestimateont} \end{align}
			and thus, $\ \mathcal{F} \left( \tilde{\mathfrak t}_{\omer} \right) \in L^1(\mathbb R) $ by the Cauchy--Schwarz inequality.
			
			Now, using \eqref{eq:formulaforpartialvphi} we  obtain
			\begin{align} \nonumber
			\left| \int_{v_\gamma(u)}^v e^{i \sbr(u,v')} \partial_{v'} (e^{i\omer r^\ast} r \phil) \d v' \right| \lesssim &    \left| \int_{v_\gamma(u)}^v e^{i \sbr(u,v')} \int_{\mathbb R}  \mathcal F[\phiHL   e^{i \omer \cdot} ] (\omega)    \mathfrak t_{\omer}^0     e^{-i\omega v'}  \d \omega  \d v'\right| \\\nonumber  &  +   \left| \int_{v_\gamma(u)}^v e^{i \sbr(u,v')} \int  \mathcal F[\phiHL   e^{i \omer \cdot} ] (\omega)  \omega   {\mathfrak t}_{\omer}^1 e^{-i\omega v'}  \d \omega \d v' \right|\\
			&  +   \left| \int_{v_\gamma(u)}^v e^{i \sbr(u,v')} \int  \mathcal F[\phiHL   e^{i \omer \cdot} ] (\omega)  \omega^2   \tilde{\mathfrak t}_{\omer}(\omega)     e^{-i\omega v'}  \d \omega \d v' \right|.
			\end{align}
			For the first term we directly take the inverse Fourier transform and estimate
			\begin{align}
			\left| \int_{v_\gamma(u)}^v e^{i \sbr(u,v')} \int_{\mathbb R}  \mathcal F[\phiHL   e^{i \omer \cdot} ] (\omega)    \mathfrak t_{\omer}^0    e^{-i\omega v'}  \d \omega  \d v'\right| & \lesssim   \left| \int_{v_\gamma(u)}^v  e^{i \sbr(u,v')} e^{i \omer v'} \phiHL  (v')     \d v' \right|.
			\end{align}
			Similarly, for the second term we integrate by parts and obtain 
			\begin{align}\nonumber
			&	\left| \int_{v_\gamma(u)}^v e^{i \sbr(u,v')} \int_{\mathbb R}  \mathcal F[\phiHL  e^{i \omer \cdot} ] (\omega)    \omega \mathfrak t_{\omer}^1 e^{-i\omega v'}  \d \omega  \d v'\right|  \lesssim   \left| \int_{v_\gamma(u)}^v  e^{i \sbr(u,v')}\partial_{v'} ( e^{i \omer v'} \phiHL  (v')    ) \d v' \right|\\\nonumber
			& \lesssim \langle u \rangle^{-s} \| \langle v\rangle^s \phiHL\|_{L^\infty}   + \left| \int_{v_\gamma(u)}^v  |\partial_{v'} \sbr(u,v')|  |  \phiHL  (v')   |  \d v' \right| \\ & \lesssim \langle u \rangle^{2-3s}  \| v^s \phiHL\|_{L^\infty}  .
			\end{align}
			Using the same method as above, the third term satisfies 
			\begin{align}\nonumber
			& \left| \int_{v_\gamma(u)}^v e^{i \sbr(u,v')} \int \mathcal F[\phiHL  e^{i \omer \cdot} ] (\omega)  \omega^2   \tilde{\mathfrak t}_{\omer}(\omega)     e^{-i\omega v'}  \d \omega \d v' \right|\\  &\lesssim \left| \int_{v_\gamma(u)}^v \partial_{v'} (e^{i \sbr(u,v')} )\int  \mathcal F[\partial_{\tilde v} (\phiHL   e^{i \omer \tilde v } ) ] (\omega)     \tilde{\mathfrak t}_{\omer}(\omega)     e^{-i\omega v'}  \d \omega \d v' \right|\label{eq:term1ofpartialphi}\\ &+ \left|   \int \mathcal F[\partial_{\tilde v} (\phiHL   e^{i \omer \tilde v } ) ]  (\omega)     \tilde{\mathfrak t}_{\omer}(\omega)     e^{-i\omega v}  \d \omega  \right|  \label{eq:term2ofpartialphi}\\ & + \left|  \int \mathcal F[\partial_{\tilde v} (\phiHL   e^{i \omer \tilde v } ) ]  (\omega)     \tilde{\mathfrak t}_{\omer}(\omega)     e^{ i\omega v_{\gamma}(u) }  \d \omega  \right|. \label{eq:term3ofpartialphi}
			\end{align}
			We will now estimate the three terms individually. 
			
			We start with integrand of \eqref{eq:term1ofpartialphi} and note that the other terms \eqref{eq:term2ofpartialphi} and \eqref{eq:term3ofpartialphi} are treated analogously. We write 	
			\begin{align} \nonumber
			&\left|   \int  \mathcal F[\partial_{\tilde v} (\phiHL   e^{i \omer \tilde v } ) ] (\omega)     \tilde{\mathfrak t}_{\omer}(\omega)     e^{-i\omega v'}  \d \omega \right|   \lesssim \left|  \left[ \partial_{\tilde v} (\phiHL   e^{i \omer \tilde v } )\right] \ast \mathcal{F}(\tilde{\mathfrak t}_{\omer})  \right|( v') \\ 
			& = \left| \int_{\mathbb R}  \partial_{\tilde v} (\phiHL   e^{i \omer \tilde v } ) (\tilde v) \mathcal{F}(\tilde{\mathfrak t}_{\omer})(v'-\tilde v)  \d \tilde v\right| 
			\end{align}
			To estimate the convolution, we note that for $v' \geq 2R$, either $|\tilde v| \geq R$ or $|\tilde v - v'| \geq R$. Thus,
			\begin{align} \nonumber
			& \left| 	\int_{\mathbb R}  \partial_{\tilde v} (\phiHL   e^{i \omer \tilde v } ) (\tilde v) \mathcal{F}(\tilde{\mathfrak t}_{\omer})(v'-\tilde v)  \d \tilde v \right| \\  \nonumber
			&\lesssim \left| 	\int_{|\tilde v |\geq R }  \partial_{\tilde v} (\phiHL   e^{i \omer \tilde v } ) (\tilde v) \mathcal{F}(\tilde{\mathfrak t}_{\omer})(v'-\tilde v)  \d \tilde v \right| +\left| 	\int_{|\tilde v - v'|\geq R}  \partial_{\tilde v} (\phiHL   e^{i \omer \tilde v } ) (\tilde v) \mathcal{F}(\tilde{\mathfrak t}_{\omer})(v'-\tilde v)  \d \tilde v \right| \\ \nonumber
			&\lesssim R^{-s} \left| \int_{|\tilde v| \geq R}  |\mathcal{F}(\tilde{\mathfrak t}_{\omer})(v'-\tilde v) |  \d \tilde v \right| ( \| v^s \phiHL\|_{L^\infty} + \| v^s \partial_v \phiHL\|_{L^\infty})\\   \nonumber
			&  + R^{-1} \left| 	\int_{|\tilde v - v'|\geq R}  | \partial_{\tilde v} (\phiHL   e^{i \omer \tilde v } ) (\tilde v) | |v' - \tilde v|  |\mathcal{F}(\tilde{\mathfrak t}_{\omer})(v'-\tilde v) |  \d \tilde v \right| 
			\\ \nonumber & \lesssim \langle  v' \rangle^{-s} \|  ( \| v^s \phiHL\|_{L^\infty} + \| v^s \partial_v \phiHL\|_{L^\infty}) \|  \mathcal{F}(\tilde{\mathfrak t}_{\omer})\|_{L^1}   + \langle \tilde v \rangle^{-1} E_1[\phiHL] \| \langle v \rangle \mathcal{F}(\tilde{\mathfrak t}_{\omer})\|_{L^2} \\ & \lesssim \langle v' \rangle^{-s} (  ( \| v^s \phiHL\|_{L^\infty} + \| v^s \partial_v \phiHL\|_{L^\infty}) + E_1[\phiHL]),
			\end{align}
			where we used \eqref{eq:lpestimateont}. 
			Now, plugging these estimates in \eqref{eq:term1ofpartialphi}  \eqref{eq:term2ofpartialphi} and \eqref{eq:term3ofpartialphi} and using that  $|\partial_v \sbr|\lesssim \langle v \rangle^{1-2s}$, we obtain, since $\frac{3}{4}< s \leq 1$ 
			\begin{align}  \nonumber
			\left| \int_{v_\gamma(u)}^v e^{i \sbr(u,v')} \partial_{v'} (e^{i\omer r^\ast} r \phil)\right| \d v' &   \lesssim  \langle u \rangle^{2-3s} ( \| v^s \phiHL\|_{L^\infty} + \| v^s \partial_v \phiHL\|_{L^\infty}  + E_1[\phiHL]). 
			\end{align}
			This shows  \textbf{Part D.}
			
			\textbf{Part E.} Assume that $\phiHL $ is such that the arising solution $\phil$ satisfies  $\partial_v (e^{i\omer r^\ast} r \phil )(u,\cdot) \in L^1_v $ on some constant $u$ surface.   Then, in view of \eqref{eq:partialvphil}, we have that \begin{align}\nonumber 
			\left\|  \int_{\mathbb R}  \mathcal F[\phiHL  e^{i \omer \cdot} ] (\omega)    \mathfrak t_{\omer}(\omega)     e^{-i\omega v}  \d \omega \right\|_{L^1_v} & \lesssim \left\| 	\partial_v \left(  r e^{i \omer  r^\ast }	\phil(u,v) \right) \right\|_{L^1_v}+ \left\|\Phi_\textup{error}\right\|_{L^1_v} \\ & \lesssim \left\| 	\partial_v \left(  r e^{i \omer  r^\ast }	\phil(u,v) \right) \right\|_{L^1_v} + E_1[\phiHL]. \label{eq:estimateonfouriertransformofphih}
			\end{align}
			We will first consider the cases for which 
			$\mathfrak t_{\omer}$ does not have any zeros (i.e.\ $\mathcal Z_{\mathfrak t} = \emptyset$), see \cref{lem:zerosornot}. Then  $\frac{1}{\mathfrak t_{\omer}}\lesssim \langle \omega \rangle^{-1} $ since $|\mathfrak t|^2 = |\mathfrak r|^2 + \omega ( \omega - \omer) $. For that, also recall $\mathfrak t_{\omer} (\omega) = \mathfrak t(\omega + \omer)$.  Moreover, in this case, $
			\mathcal F^{-1}  \left[\frac{1}{\mathfrak t_{\omer}} \right] \in L^1_v
			$
			since $ \frac{1}{\mathfrak t_{\omer}} \in L^2_\omega$, $\partial_\omega \frac{1}{\mathfrak t_{\omer}} \in L^2_\omega$. Thus, $\frac{1}{\mathfrak t_{\omer}}$ is a $L^1$ bounded Fourier multiplier. Hence, using that $1 =\mathfrak t_{\omer} \frac{1}{\mathfrak t_{\omer}}$ and  \eqref{eq:estimateonfouriertransformofphih}, we obtain
			\begin{align}
			\| \phiHL\|_{L^1_v} \lesssim \left\|  \int_{\mathbb R}  \mathcal F[\phiHL  e^{i \omer \cdot} ] (\omega)    \mathfrak t_{\omer}(\omega)     e^{-i\omega v}  \d \omega \right\|_{L^1_v} \lesssim \left\| 	\partial_v \left(  r e^{i \omer  r^\ast }	\phil(u,v) \right) \right\|_{L^1_v} + E_1[\phiHL]. \label{eq:fouriermulitplierestimate}
			\end{align}
			
			Now, we consider the case, where $\mathfrak t$ potentially has zeros, all of which have to lie in $\mathcal Z_{\mathfrak t}^\delta$. Then, by the inverse triangle inequality applied to \eqref{eq:estimateonfouriertransformofphih} we obtain
			\begin{align}\nonumber
			& \left\| 	\partial_v \left(  r e^{i \omer  r^\ast }	\phil(u,v) \right) \right\|_{L^1_v}  + E_1[\phiHL] \gtrsim  \left\| \int_{\mathbb R}  \mathcal F[\phiHL  e^{i \omer \cdot} ] (\omega)    \mathfrak t_{\omer}(\omega)     e^{-i\omega v}  \d \omega  \right\|_{L^1_v} 
			\\ \nonumber
			&  \geq      \left\| \int_{\mathbb R}  \mathcal F[\phiHL e^{i \omer \cdot} ] (\omega)    ( 1- \chi_{\delta}(\omega + \omer) ) \mathfrak t_{\omer}(\omega)     e^{-i\omega v}  \d \omega  \right\|_{L^1_v} \\ & - \left\| \int_{\mathbb R}  \mathcal F[ \phiHL  e^{i \omer \cdot} ] (\omega)      \chi_{\delta}(\omega + \omer)    \mathfrak t_{\omer}(\omega)     e^{-i\omega v}  \d \omega  \right\|_{L^1_v}   ,
			\end{align}
			where we recall that $\chi_{\delta}$ is supported in $\mathcal Z^\delta_{\mathfrak t} $. 
			For the first term we use $|\frac{1}{\mathfrak t} |\lesssim_\delta \langle \omega\rangle^{-1}$ on 
			$\mathbb R - \mathcal Z_{\mathfrak t}^\delta$ and obtain
			\begin{align} \nonumber
			\left\| \int_{\mathbb R}  \mathcal F[\phiHL e^{i \omer \cdot} ] (\omega)    ( 1- \chi_{\delta}(\omega + \omer) ) \mathfrak t_{\omer}(\omega)     e^{-i\omega v}  \d \omega  \right\|_{L^1_v} &  \gtrsim_\delta \left\| (1-P_{\delta} ) \phiHL\right\|_{L^1_v} \\ & \geq  \|\phiHL\|_{L^1_v} -  \| P_{\delta} \phiHL \|_{L^1_v}.
			\end{align}
			For the second term we use $ \mathfrak t \cdot   \chi_{\delta} \in C_c^\infty$ and obtain 
			\begin{align}
			\left\| \int_{\mathbb R}  \mathcal F[ \phiHL  e^{i \omer \cdot} ] (\omega)      \chi_{\delta}(\omega + \omer)    \mathfrak t_{\omer}(\omega)     e^{-i\omega v}  \d \omega  \right\|_{L^1_v} \lesssim \| P_{\delta} \phiHL \|_{L^1_v}.
			\end{align}
			Putting everything together  yields
			\begin{align}
			\|\phiHL\|_{L^1_v} \lesssim_\delta \left\| 	\partial_v \left(  r e^{i \omer  r^\ast }	\phil(u,v) \right) \right\|_{L^1_v}  + E_1[\phiHL]  +  \| P_{\delta} \phiHL \|_{L^1_v}. 
			\end{align}
			This  shows \textbf{Part~E.}\ and concludes the proof of  \cref{prop:representation}.
		\end{proof}
	\end{theol}
	
	To connect with the nonlinear theory and the various oscillation spaces  from  \cref{sec:oscillationspaces}  we state  the following corollaries from \cref{prop:representation}.
	We will also introduce a smooth positive cut-off supported only on $v\geq v_0 +  2 $ and such that $\chi_{ \geq v_0 +  3} =1 $ for $v\geq v_0 + 3$.  We assume that $  |\partial_v \chi_{\geq v_0 +  3} |\leq 2$.
	We also recall the notation $\psil' = r \phil$. 
	
	\begin{cor} \label{linearcor}
		Let $\phiH \in \Sl$ be arbitrary and define  $ \phiHL (v):=  \chi_{\geq v_0 + 3}(v) \phiH(v)$ which we trivially extend for $v\leq v_0$.
		Let $\phil$ be the unique solution of \eqref{eq:chargedKGlinear} with data $\phiHL$ on $\mathcal{H}^+ $ and no incoming data from the left event horizon. Note that by definition of $\Sl$ (recalling $s\in (\frac{3}{4},1]$) we have that  for  all $v\geq v_0 $\begin{align}
		&	v^s (|\phiHL|(v)+ |\partial_v \phiHL|(v) )\leq 4 D_1.
		\end{align}
		\begin{enumerate}
			\item If $\phiH\in \OO$,   then \begin{align} \sup_{v \geq v_0, u_0 \leq u_s} \left| \int_{v_0}^{v} e^{iq_0\sbr(v')} e^{i q_0\int_{v_0}^v (A'_{RN})_v(u_0,v') \d v'}	D_v^{RN} \psil'(u_0,v') \d v' \right|<+\infty\end{align} for all $\sbr$ satisfying \eqref{sigma_err1}, \eqref{sigma_err2}.
			
			\item If  $\phiH \in \OOp$, then  additionally for all $u_0\leq u_s$ \begin{align} \lim_{v \rightarrow  +\infty} \left| \int_{v_0}^{v} e^{iq_0\sbr(v')} e^{i q_0\int_{v_0}^v (A'_{RN})_v(u_0,v') \d v'}	D_v^{RN} \psil'(u_0,v') \d v'\right|\end{align} exists and is finite for all $\sbr$ satisfying \eqref{sigma_err1}, \eqref{sigma_err2}.
			
			\item \label{part3} If $\phiH \in \OOpp$, then additionally, for all $D_{br}>0$ there exists $D' = D'(e,M,D_1,s,q_0,m^2,D_{br})>0$ and  $\tilde \eta_0(e,M,D_1,s,q_0,m^2,D_{br}) >0 $ such that for all $\sbr$ satisfying \eqref{sigma_err1}, \eqref{sigma_err2} and for all $(u,v) \in \LB$  \begin{align}\left| \int_{v_{\gamma}(u)}^{v} e^{iq_0\sbr(v')} e^{i q_0\int_{v_0}^v (A'_{RN})_v(u,v') \d v'}	D_v^{RN} \psil'(u,v') \d v'\right| \ls D' \cdot |u|^{s-1-\tilde \eta_0}.
			\end{align}
			\item Assume that $q_0=0$, $m^2\notin D(M,e)$ and that $\phiH  \in \NO= \Sl-\OO$. Then for all $u \in \RR$ \begin{equation} \label{blow.up.lin}
			\limsup_{v\rightarrow+\infty}|\phil|(u,v)=+\infty.\end{equation}	
		\end{enumerate}
	\end{cor}
	\begin{rmk}\label{rmk:initialdataonlyhavetoagreeeventually}
		It should be noted that for the nonlinear problem we will impose non-zero data on $\underline{C}_{in}$. For the difference estimates it however suffices if the linear data and the nonlinear data agree eventually on $\HH$.   
	\end{rmk}
	\begin{proof}
		We begin by noting that $\phiH \in \OO, \OOp, \OOpp$, respectively, if and only if $\frac 14\phiHL (v) = \frac 14  \chi_{\geq v_0 + 3}(v) \phiH(v) \in \OO,\OOp,\OOpp$, respectively.\footnote{The factor $\frac 14$ is just to make sure that $ \frac 14   \chi_{\geq v_0 + 3}(v) \phiH(v)  \in \Sl$ if $\phiH\in \Sl$.} 
		
		Now, the first statement is a consequence of  Part~\textbf{D.} of \cref{prop:representation}, the expression for the gauge derivative in \eqref{eq:gaugederivative} and the fact that for some bounded function  $f(u)$: \begin{align} \nonumber q_0 & \int_{v_0}^v (A'_{RN})_{v'}(u,v') \d v' = - \frac 12  \int_{v_0}^{v}  (\omega_- - \omega_r) \d v' + \frac 12  \omer\cdot (v-v_0)  \\ \nonumber & = - \frac 12  \int_{v_0}^{+\infty}  (\omega_- - \omega_r) \d v' + \frac 12 \int_v^{+\infty}  (\omega_- - \omega_r) \d v'   + \frac 12  \omer \cdot (2 r^\ast - u -v_0)\\
		& =  \omer r^\ast + f(u) + O(\OmegaRN^2(r^\ast)). \end{align} 
		
		The second statement follows completely analogously. For the third statement, we use Part~\textbf{D.} of \cref{prop:representation}, and that, defining $0<\tilde{\eta}\eta_0=\min\{\eta_0,\frac{3s-4}{10}\}$  (where $\eta_0$ is as in the definition of $\mathcal O''$) we have  $\min(1-s +\tilde \eta_0, 2s-3) = 1-s+ \tilde \eta_0$ for some $\tilde \eta_0 >0$ as $s>\frac 34$. 
		
		Now, we proceed to the last statement. Indeed, under the assumption $q_0 =0$ and $m^2\notin D(M,e)$, we have that $\mathfrak r( \omega =0) \neq 0$. Thus, from \cref{prop:representation}, Part \textbf{A.}, and the assumption $\phiH \in \mathcal{NO}$, the claim follows. \end{proof}
	Moreover, we also deduce a  result of $\dot{W}^{1,1}$ blow-up along outgoing cones for the linearized solution in the following sense. To state the following corollary we recall the definition of $P_{\delta}$ as in \cref{sec:w11blowup}.
	\begin{cor} \label{linear.cor.W11}
		Let the assumptions of \cref{linearcor} hold.
		\begin{enumerate}
			\item   Assume   that  $P_{\delta} ( \phiH  ) \in L^1$  for some $\delta >0$. Then, for all $u\leq u_s$,  we have \begin{align}
			\int_{v_0}^{+\infty} |\phiH|(v') \d v' \lesssim_\delta \int_{v_0}^{+\infty} |D_v^{RN} \psil'|(u,v) \d v + \| P_{\delta} ( \phiH ) \|_{L^1_v}+    D_1 ,
			\end{align} recalling the definition $\psil'= r_{RN} \phil$. In particular,  if \begin{align}
			\label{eq:assumptiontoblowup}
			&	\phiH  \in \Sl -  L^1(\HH) \text{ with }    P_{\delta} ( \phiH   ) \in L^1(\RR) \text{ for some } \delta >0,
			\end{align}
			then for all $u\leq u_s$,  \begin{equation} \label{blowup.lin.eq}
			\int_{v_0}^{+\infty} |D_v^{RN}\psil'|(u,v') \d v'=+\infty.
			\end{equation} 
			Thus,  the set of data $\phiH \in \Sl$ leading to blow-up for each $u\leq u_s$ as in \eqref{blowup.lin.eq}  is generic in the sense that its complement $H$ is the set  $H=H_0\cap \Sl$ for some vector space $H_0\subset \Sl_0$ of infinite co-dimension in $\Sl_0$, where we recall \eqref{eq:defnsl0} for the definition of $\Sl_0$.
			
			\item Assume  $0<|q_0 e|<\epsilon(M,e,m^2)$ or $q_0  = 0$ and $m^2 \notin D(M,e)$. Then, for all $u \leq u_s$, we have  \begin{equation}
			\int_{v_0}^{+\infty} |\phiH|(v') \d v' \lesssim \int_{v_0}^{+\infty} |D_v^{RN} \psil'|(u,v') \d v' + D_0.
			\end{equation}
			In particular, if $\phiH \in \Sl- L^1(\HH)$, then  $$ \int_{v_0}^{+\infty} |D_v^{RN}\psil'|(u,v') \d v'=+\infty.$$
		\end{enumerate} 
		\begin{proof}
			The statements  follow from \cref{prop:representation}, Part \textbf{E.} The genericity of $\Sl-H$ in the  first statement is a direct consequence of  \eqref{eq:assumptiontoblowup}. We have also used that $P_{\delta} ( \phiHL  ) \in L^1$ if and only if $P_{\delta} ( \phiH  ) \in L^1$. 
		\end{proof}
	\end{cor}
	
	\section{Nonlinear estimates for the EMKG system and extendibility properties of the metric} \label{nonlinearsection}

	We give a brief outline of   \cref{nonlinearsection}: \begin{enumerate}
		\item 	In \cref{recall.section} we recall the time-decay estimates that where established in the nonlinear setting by the second author in \cite{Moi} (see   \cref{CH.stab.thm}). These estimates play a crucial role in the proof of the Cauchy horizon (in)-stability and will also be essential to the analysis of the present paper. 
		Recall that the various gauges were defined in   \cref{coordinatechoice.section} and \cref{gaugechoice.section}.
		\item In \cref{prelim.nonlin} and \cref{ext.section}, we provide some useful nonlinear estimates, and show how to deduce the continuous extendibility of the metric from the boundedness of the scalar field. To do so, we will in particular exploit the algebraic structure of the nonlinear terms in the Einstein equations.
		
		\item In  \cref{difference.estimates.section}, we estimate the difference of the dynamical metric $g$ with the Reissner--Nordström metric $g_{RN}$ and the difference of the scalar field $\phi$ and its linear counterpart $\phiNl$ ($\phiNl$ differs from $\phil$ of \cref{linearsection} by a gauge change, see \cref{difference.estimates.section}). If $q_0 = 0$, we show that these differences are bounded, thus showing the coupled $\phi$ is bounded if and only if  its linear counterpart $\phiNl$ is bounded. If $q_0 \neq 0$, the estimates are more involved and include a backreaction contribution from the  Maxwell field, see   \cref{diff.blue.section}.
		
		\item In \cref{combining.section}, we combine the results from the linear theory (\cref{linearsection}) with the results above to prove \cref{main.theorem} (\cref{boundedness.combining.section}), \cref{main.theorem2} (\cref{{blow.up.section}}), \cref{corollary.conj} (\cref{computation.proof.section}) and \cref{W11.main.thm} (\cref{W11.blowup.section}).
	\end{enumerate}
	\textbf{	 Throughout \cref{nonlinearsection} we will work under the assumptions of   \cref{CH.stab.thm}.}

	\subsection{The existence of a Cauchy horizon for the EMKG system and previously proven nonlinear estimates}
	\label{recall.section}

	We use five different regions which partition the domain $[-\infty,u_s]\times [v_0,+\infty]$, see   \cref{Fig.regions}. To this effect, we first introduce the function $h(v)$ as in \cite[Proposition 4.4]{Moi},  namely we define $h(v)$ by the relation \begin{equation} \label{h.def}
	\Omega^2_H(U=0,v)= e^{2K_+ \cdot(v+h(v)-v_0)}.
	\end{equation} 
	Note that $h(v_0)=0$ by gauges \eqref{eq:gaugeforU}, \eqref{gauge1}.	It is proven in \cite{Moi} that as $v \rightarrow +\infty$: 
	\begin{equation}\label{h.bound}
	h(v)= O(v^{2-2s}) 1_{s < 1} + O( \log(v)) 1_{s=1}, \hskip 5 mm h'(v)=O(v^{1-2s}), \hskip 5 mm h''(v) = O(v^{-2s}).
	\end{equation}		Now we can introduce the five regions partitioning our spacetime $\{0 \leq U \leq U_s,\ v\geq v_0\}$:
	\begin{enumerate}		\item The event horizon  $\mathcal{H}^+ =  \{ u=-\infty\}= \{ U=0\}$.

		\item The red-shift region $\mathcal{R} =  \{ u+v+h(v) \leq -\Delta\}$.

		\item The no-shift region $\mathcal{N}:= \{ -\Delta \leq u+v+h(v) \leq \Delta_N \}$.

		\item The early blue-shift region $\mathcal{EB}:=  \{\Delta_N \leq  u+v+h(v) \leq -\Delta'+
		\frac{2s}{2|K_-|} \log(v)\}$, assuming that $|u_s|$ is sufficiently large  so that $\Delta_N+\Delta' < 	\frac{2s}{2|K_-|} \log(v)$ in $\mathcal{EB}$.
		
		\item The late blue-shift\footnote{Note that the late blue-shift differs slightly from \cite{Moi} where it was defined to be  $\mathcal{LB}:=  \{-\Delta'+
			\frac{2s}{2|K_-|} \log(v) \leq  u+v+h(v) \}$.} region $\mathcal{LB}:=  \{-\Delta'+
		\frac{2s}{2|K_-|} \log(v+h(v)) \leq  u+v+h(v) \}$.
	\end{enumerate}

	In  the proof of   \cref{CH.stab.thm}, it was shown  that there exists a large constant $\Delta_0(M,e,q_0,m^2,s,D_1,D_2)>0$ such that, if $\Delta, \Delta_N, \Delta' > \Delta_0$, the following estimates (as enumerated below) are true. In the course of the proof of the new result, we will implicitly always assume that $\Delta, \Delta_N, \Delta'> \Delta_0$ and choose when necessary $\Delta, \Delta_N, \Delta'> \Delta_1$ for some   $\Delta_1(M,e,q_0,m^2,s,D_1,D_2)>\Delta_0$ that will be defined later.  
	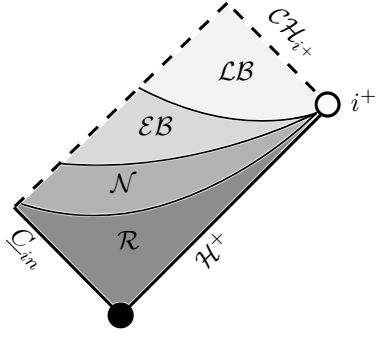
\begin{figure}
		\centering
		\input{regions.pdf_tex}			
		\caption{Division of a rectangular neighborhood of $i^+$ into  five spacetime regions.}
		\label{Fig.regions}
	\end{figure}
	\begin{prop}[Nonlinear estimates on the event horizon $\HH$, \cite{Moi}] \label{prop.HH.estimates.Moi}
		There exists a constant $D_H=D_H(M,e,q_0,m^2,s,D_1,D_2)>0$ such that the following estimates hold true on $\HH=\{U=0,\ v\geq v_0\}$: 
		
		\begin{equation} \label{Q.HH} |Q(0,v)-e| \leq D_H \cdot  v^{1-2s} ,
		\end{equation}
		\begin{equation} \label{M.HH} |\varpi(0,v)-M| \leq D_H \cdot   v^{1-2s} ,
		\end{equation}
		\begin{equation} \label{lambda.HH}
		0 \leq \lambda(0,v) \leq D_H \cdot  v^{-2s} ,    
		\end{equation} 		
		\begin{equation} \label{r.HH}
		0 \leq r_+ -r(0,v)  \leq D_H \cdot  v^{1-2s}  ,   
		\end{equation} 
		\begin{equation}\label{Omega.HH.1}
		|	\partial_v \log(\Omega^2_{H})(0,v)- 2K(0,v) | \leq D_H \cdot  v^{-2s}
		,	\end{equation}	
		\begin{equation}\label{Omega.HH.1.2}
		| 2K_+ h'(v)+ [2K_+- 2K(0,v)]| \leq D_H \cdot  v^{-2s},	\end{equation}
		\begin{equation}\label{Omega.HH.2}
		|	\partial_U \log(\Omega^2_{H})|(0,v)  \leq D_H \cdot   \Omega^2_{H}(0,v),
		\end{equation}	\begin{equation}\label{phi.HH}
		|	\partial_U \phi|(0,v)   \leq D_H \cdot  \Omega^2_{H}(0,v)\cdot v^{-s},
		\end{equation}		
		\begin{equation}\label{A.HH}
		|	A_U|(0,v)   \leq D_H \cdot  \Omega^2_{H}(0,v).
		\end{equation}

	\end{prop}
	\begin{prop}[Nonlinear estimates in the red-shift region $\Rs$, \cite{Moi}] \label{prop.RS.estimates.Moi}
		
		There exists a constant $D_R = D_R(M,e,q_0,m^2,s,D_1,D_2)>0$ such that  the following estimates hold true for all $(u,v) \in \Rs$:
		\begin{equation} \label{phivRS}
		|\phi|(u,v)+|D_v \phi|(u,v) \leq D_R \cdot v^{-s},
		\end{equation}			\begin{equation} \label{phiURS}
		|D_u \phi|(u,v) \leq D_R \cdot  e^{ 2K_+ \cdot (u+v+h(v))} \cdot v^{-s},
		\end{equation}
		\begin{equation} \label{OmegaPropRedshift}
		| \log(\Omega^2(u,v))- 2K_+ \cdot (u+v+h(v))| \leq D_R \cdot\Omega^2(u,v),
		\end{equation}
		\begin{equation} \label{kappaRedshiftprop}
		0 \leq 1-\kappa(u,v) \leq D_R \cdot  \Omega^2(u,v) \cdot v^{-2s},
		\end{equation}
		\begin{equation} \label{partialuRSOmegaprop}
		|\partial_u \log\Omega^2(u,v) |  \leq D_R \cdot    \Omega^2(u,v),
		\end{equation}
		\begin{equation} \label{partialvRSOmegaprop}
		|\partial_v \log(\Omega^2)(u,v)-2K(u,v) |  \leq D_R \cdot v^{-2s},
		\end{equation} \begin{equation} \label{rRedShiftprop}0 \leq r_+-r(u,v) \leq D_R \cdot  \Omega^2(u,v)+ v^{1-2s} ,
		\end{equation}
		\begin{equation} \label{QRedShiftprop} |Q(u,v)-e| \leq D_R \cdot  v^{1-2s} ,
		\end{equation}
		\begin{equation} \label{MRedShiftprop} |\varpi(u,v)-M| \leq D_R \cdot  v^{1-2s} ,
		\end{equation}
		\begin{equation} \label{KRedShiftprop} |2K(u,v)-2K_+| \leq D_R \cdot   \Omega^2(u,v)+   v^{1-2s}.
		\end{equation}
		
	\end{prop}
	\begin{prop}[Nonlinear estimates in the no-shift region $\NN$, \cite{Moi}]  \label{prop.NN.estimates.Moi}
		There exists a constant $D_N = D_N(M,e,q_0,m^2,s,D_1,D_2)>0$ such that the following estimates hold true for all $(u,v) \in \NN$:
		\begin{equation} \label{phivNSprop}
		|\phi(u,v)|+|D_v \phi(u,v)| \leq D_N \cdot  v^{-s},
		\end{equation}	\begin{equation} \label{phiUNSprop}
		|D_u \phi(u,v)|  \leq D_N \cdot v^{-s} ,
		\end{equation}
		\begin{equation} \label{OmegaPropNSprop}
		|\log\Omega^2(u,v)-\log\left(-(1-\frac{2M}{r(u,v)}+\frac{e^2}{r^2(u,v)})\right)| \leq D_N \cdot v^{1-2s} ,\end{equation}  \begin{equation} \label{kappaNStpropprop}
		0 \leq 1-\kappa(u,v)  \leq D_N \cdot v^{-2s},
		\end{equation}						\begin{equation} \label{iotaNStpropprop}
		|1-\iota(u,v)|  \leq D_N \cdot v^{1-2s},
		\end{equation}
		\begin{equation} \label{partialuNSOmegaprop2}
		|\partial_u \log(\Omega^2)(u,v)-2K(u,v) |  \leq D_N \cdot v^{1-2s},
		\end{equation}
		\begin{equation} \label{partialvNSOmegaprop2}
		|\partial_v \log(\Omega^2)(u,v)-2K(u,v) |   \leq D_N \cdot v^{-2s},
		\end{equation}
		\begin{equation} \label{QNS2} |Q(u,v)-e|   \leq D_N \cdot v^{1-2s}. 
		\end{equation}
		\begin{equation} \label{MNS2} |\varpi(u,v)-M| \leq D_N \cdot v^{1-2s} .
		\end{equation}	
		\begin{equation} \label{Omega.r.boring} |\log(\Omega^2)|(u,v)+ |\log(r)|(u,v) \leq D_N .
		\end{equation}	
		
		Moreover, denoting $\gamma_N:=\{ u+v +h(v)= \Delta_N \}$ the future boundary of $\NN$, we have on $\gamma_N$: \begin{equation} \label{Omega.gamma.N.Moi}
		\Omega^2(u_{\gamma}(u),v) \leq D_N \cdot e^{ 2K_- \cdot \Delta_N}.
		\end{equation}
	\end{prop}

	\begin{prop}[Nonlinear estimates in the early blue-shift region $\EB$, \cite{Moi}]  \label{prop.EB.estimates.Moi}
		There exists a constant $D_E = D_E(M,e,q_0,m^2,s,D_1,D_2)>0$ such that the following estimates hold true for all $(u,v) \in \EB$:
		\begin{equation} \label{phiTransition}
		|\phi(u,v)| \leq D_E \cdot  v^{-s} \log(v),
		\end{equation}	
		\begin{equation} \label{phiVTransition}
		|D_v \phi(u,v)| \leq D_E \cdot v^{-s},
		\end{equation}	\begin{equation} \label{phiUTransition}
		|D_u \phi(u,v)|\leq D_E \cdot v^{-s} ,
		\end{equation}\begin{equation} \label{Omegatransition}
		|\log\Omega^2(u,v)-2K_- \cdot (u+v+h(v))|\leq D_E \cdot \Delta \cdot e^{-2K_+ \Delta}  <1 ,
		\end{equation}  
		\begin{equation} \label{kappatransition}
		0 \leq 1-\kappa(u,v) \leq \frac{1}{3},
		\end{equation}			\begin{equation} \label{iotatransition}
		|1-\iota(u,v)|\leq \frac{1}{3},
		\end{equation}
		\begin{equation} \label{partialuOmegatransition}
		|\partial_u \log(\Omega^2)(u,v)-2K(u,v) |  \leq D_E \cdot   v^{1-2s}\log(v)^3 ,
		\end{equation}
		\begin{equation} \label{partialvOmegatransition}
		|\partial_v \log(\Omega^2)(u,v)-2K(u,v) |  \leq D_E \cdot   v^{-2s}\log(v)^3 ,
		\end{equation}
		\begin{equation} \label{Ktransition}
		|2K(u,v)- 2K_- |  \leq \frac{|K_-|}{1000} ,
		\end{equation}
		\begin{equation} \label{QTrans} |Q(u,v)-e| \leq D_E \cdot  v^{1-2s} ,
		\end{equation}
		\begin{equation} \label{MTrans} |\varpi(u,v)-M| \leq D_E \cdot v^{1-2s} .
		\end{equation}
		\begin{equation} \label{r.nodiff.EB}
		|r(u,v)-r_-(M,e)| \leq D_E \cdot( v^{1-2s}+ \Omega^2(u,v)).
		\end{equation}
		
		Moreover, denoting  $\gamma:=\{u+v+h(v)=-\Delta'+ \frac{s}{2|K_-|} \log(v)\}$ the future boundary of $\EB$, we have on $\gamma$:
		
		\begin{equation} \label{Omega.gamma.EB}
		\Omega^2(u_{\gamma}(v),v) \leq D_E \cdot v^{-2s}.
		\end{equation}
		
	\end{prop}
	\begin{prop}[Nonlinear estimates in the late blue-shift region $\LB$, \cite{Moi}]  \label{prop.LB.estimates.Moi}
		There exists a constant $D_L = D_L(M,e,q_0,m^2,s,D_1,D_2)>0$ such that  the following estimates hold true: for all $\eta>0$, there exists $C_{\eta}>0$ such that for all $(u,v) \in \LB$  
		\begin{equation} \label{phiLB}
		\Omega^{2\eta}(u,v)|\phi|(u,v)\leq   C_{\eta} \cdot v^{-s} ,
		\end{equation}		\begin{equation} \label{QLB}  \Omega^{2\eta}(u,v)|Q-e|(u,v) \leq  C_{\eta}  \cdot v ^{1-2s},	\end{equation}
		\begin{equation} \label{phiblowupVLB}
		|\phi|^2(u,v) + Q^2(u,v) \leq D_L \cdot   v^{2-2s}1_{ \{s<1\}}  + D_L \cdot  [\log(v)]^2 1_{ \{s=1\}},
		\end{equation}\begin{equation} \label{phiVLB}
		|D_v \phi|(u,v) \leq D_L \cdot v^{-s},
		\end{equation}		\begin{equation} \label{partialvOmegaLB}
		|\partial_v \log(\Omega^2_{CH}) |(u,v)  \leq D_L \cdot    v^{1-2s}1_{ \{s<1\}}  +D_L \cdot   \log(v) \cdot v^{-1} 1_{ \{s=1\}},
		\end{equation}			\begin{equation} \label{lambdaLB}
		0< \Omega^2(u,v) \leq	 -\lambda(u,v) \leq D_L \cdot  v^{-2s},
		\end{equation}
		\begin{equation} \label{nuLB}
		0< -\nu(u,v) \leq D_L \cdot |u|^{-2s}.
		\end{equation}
	\end{prop}

	\subsection{Nonlinear estimates exploiting the algebraic structure}   \label{prelim.nonlin}
	We emphasize that we \textbf{do not necessarily assume that} $\phiH \in \OO$ in this section. The specific assumptions of this type are made in   \cref{combining.section} only. In fact, we use many of these estimates in our companion paper \cite{MoiChristoph2} as well (where it is assumed that $\phiH \notin \OO$). Throughout \cref{prelim.nonlin} to  \cref{combining.section} we use the notation $|f(u,v)| \ls |g(u,v)|$ if there exists a constant $\Gamma(M,e,m^2,q_0,D_1,D_2,s)>0$ such that $|f(u,v)| \leq \Gamma \cdot |g(u,v)|$ for all $(u,v)$ in the spacetime region of interest. 
	
	\subsubsection{Boundedness and continuous extendibility of \texorpdfstring{$D_u{\psi}$}{Dupsi}} \label{Dupsi}
	
	To reach the goals of this section, we must first prove preliminary estimates on $D_u \psi$, where $\psi:=r\phi$ is (what is called in the black hole exterior) the radiation field. Since $r$ is upper and lower bounded in our region of interest, it may be very surprising to consider this quantity in the black hole \textit{interior}. However, as it turns out, $D_u \psi$ is always bounded, while $D_u \phi$ is bounded if and only if $ \phi$ is (providing $\liminf_{v \rightarrow +\infty}|\nu|(u,v)>0$, which is conjecturally a generic condition, see \cite{Moi4} for a discussion and proof of this result).

	\begin{prop}\label{Dupsi.prop}
		
		We have the following (gauge-independent) estimate for all $(u,v) \in\mathcal{LB} $:
		
		\begin{equation} \label{Dupsiestimate}
		| D_u \psi|(u,v) \lesssim  |u|^{-s} .
		\end{equation}
		Moreover,  in the gauge \eqref{GaugeAv}, both $D_u \psi$ and $A_u$ admit a bounded extension to the Cauchy horizon, denoted $(D_u \psi)_{CH}$ and $(A_u)^{CH}$,
		respectively.
	\end{prop}
	\begin{proof}

		Using  \eqref{Field4}   with the estimates of   \cref{prop.LB.estimates.Moi} we have	
		
		$$ |\partial_v ( D_u \psi)| \lesssim |\lambda|  \cdot	 |\nu|  \cdot  |\phi|+ \Omega^{1.99} \cdot v^{-s}.$$ Finally with \eqref{lambdaLB} and \eqref{phiblowupVLB} we get 
		
		$$ |\partial_v ( D_u \psi)| \lesssim v^{1-3s}  	\cdot  |u|^{-2s}+ \Omega^{1.99} \cdot v^{-s}.$$ Now the left hand side is integrable in $v$  since $s>\frac{2}{3}$  so $D_u \psi$ admits a bounded extension by integrability and integrating from $\gamma$ we obtain the estimate, in view of the estimate on $\gamma$ from   \cref{prop.EB.estimates.Moi}. To conclude, the extendibility of $A_u$ follows from \eqref{eq:maxwellAu} and the estimates of   \cref{prop.LB.estimates.Moi} that show that $|\partial_v A_u|$ is integrable in $v$.
	\end{proof}	
	
	\subsubsection{Key estimates for a candidate coordinate system $(u,V)$ for a continuous extension} \label{coordinatesphibounded}
	
	In this section, we construct an adequate coordinate system $(u,V)$, in which the boundedness of the metric coefficient $\log(\Omega^2_{CH})$ related to $(u,V) $ by $\Omega^2_{CH}=-2g(\partial_u,\partial_V)$ follows from the boundedness of the scalar field $\phi$.
	
	\begin{prop} \label{coordinatesphiboundedprop} There exists a coordinate system $(u,V)$ for which $V(v) <1$, and $\lim_{v \rightarrow +\infty} V(v)=1$ and for which, defining the metric coefficient $ \Omega^2_{CH} \d u \d V = \Omega^2 \d u \d v$, we have for all $(u,v) \in \LB$: \begin{equation} \label{logomega1}
		\left|\partial_v \left( \log(\Omega^2_{CH})(u,v)+ |\phi|^2(u,v)+\int_{u}^{u_s}\frac{|\nu|}{r}|\phi|^2(u',v) \d u' \right) \right| \lesssim v^{2-4s}+v^{-2s} |\log(v)|^3,
		\end{equation} 
		\begin{align} \nonumber
		&	\left|\partial_v \partial_u \left( \log(\Omega^2_{CH})(u,v)+ |\phi|^2(u,v)+\int_{u}^{u_s}\frac{|\nu|}{r}|\phi|^2(u',v) \d u' \right) \right| \\   & \;\;\;\;  \lesssim   |u|^{-2s} \cdot (v^{2-4s}+v^{-2s} |\log(v)|^3)+ |u|^{-s} 	 \cdot v^{1-3s}. \label{logomega2}
		\end{align}
		As a consequence, the quantity $\Upsilon$ defined as 
		\begin{align}\Upsilon(u,v):=\log(\Omega^2_{CH})+ |\phi|^2+\int_{u}^{u_s}\frac{|\nu|}{r}|\phi|^2 \d u'\end{align} admits a continuous extension $\Upsilon_{CH}(u)$  across $\CH$ and \begin{align}\partial_u \Upsilon= \partial_u \left(\log(\Omega^2_{CH})+ |\phi|^2+\int_{u}^{u_s}\frac{|\nu|}{r}|\phi|^2 \d u' \right)\end{align} admits a bounded extension across $\CH$.
	\end{prop}			  	
	
	\begin{proof}
		We first use \eqref{Field2} to establish the following two formulae: 
		$$\frac{\partial_u \partial_v (r|\phi|^2)}{r} = \partial_u \partial_v (|\phi|^2) + \frac{\nu}{r}\partial_v (|\phi|^2)+ \frac{1}{r}\partial_u (\lambda|\phi|^2),$$
		$$ -2\Re (D_u \phi \overline{\partial_v \phi}) = \frac{-\partial_u \partial_v (r|\phi|^2)}{r}+\left(\frac{\partial_u \partial_v r}{r}-\frac{m^2 \Omega^2}{2} \right)|\phi|^2.$$
		
		Now we define $2K_{\gamma}(v):= 2K(u_{\gamma(v)},v)$ and we rewrite \eqref{Omega}  using the two last formulae 
		$$ |\partial_u \left( \partial_v \log(\Omega^2)-2K_{\gamma}(v) + \partial_v (|\phi|^2)\right)+\frac{\nu}{r}\partial_v (|\phi|^2)+ \frac{1}{r}\partial_u (\lambda|\phi|^2) | \lesssim |\lambda \nu|(1+|\phi|^2) + \Omega^2 (1+Q^2+m^2 |\phi|^2).$$
		First note that the right hand side is $O(|u|^{-2s} \cdot v^{2-4s}+ |u|^{-2s} \cdot v^{-2s})$, using the estimates of   \cref{prop.LB.estimates.Moi}. Using \eqref{Radius}, \eqref{Dupsiestimate} and the other estimates of   \cref{prop.LB.estimates.Moi} we get $$ |\partial_u (\lambda |\phi|^2)|= |\partial_u (r^{-2}\lambda |\psi|^2)| \lesssim |u|^{-2s} v^{2-4s}+|u|^{-s} \cdot v^{1-3s}.$$ This gives
		\begin{equation} \label{logomega3}
		|\partial_u \left( \partial_v \log(\Omega^2)-2K_{\gamma}(v) + \partial_v (|\phi|^2)\right)+\frac{\nu}{r}\partial_v (|\phi|^2) | \lesssim |u|^{-2s} \cdot v^{-2s}+|u|^{-2s} v^{2-4s}+   |u|^{-s}  \cdot v^{1-3s}.
		\end{equation}
		Now we want to integrate both sides on $[u_{\gamma}(v),u]$. Recall that on $\gamma$, $|\partial_v \log(\Omega^2)(u_{\gamma}(v) ,v)-2K_{\gamma}(v)|\lesssim v^{-2s} |\log(v)|^3$ and $|\partial_v (\phi^2)|\lesssim v^{-2s} |\log(v)|$, as established in   \cref{prop.EB.estimates.Moi}. Thus, we obtain
		
		\begin{equation}
		| \partial_v \log(\Omega^2)-2K_{\gamma}(v) + \partial_v (|\phi|^2)+\int_{u_{\gamma(v)}}^{u}\frac{\nu}{r}\partial_v (|\phi|^2)\d u'| \lesssim v^{2-4s}+ v^{-2s} |\log(v)|^3.
		\end{equation}
		Now we write 
		$$\int_{u_{\gamma(v)}}^{u}\frac{\nu}{r}\partial_v \left(|\phi|^2\right)\d u' = \int_{u_{\gamma(v)}}^{u_s}\frac{\nu}{r}\partial_v (|\phi|^2)\d u'- \partial_v(\int_{u}^{u_s}\frac{\nu}{r}|\phi|^2\d u')+\int_{u}^{u_s}\partial_v(\frac{\nu}{r}) |\phi|^2\d u'.$$
		Using \eqref{Radius} and the estimates of   \cref{prop.LB.estimates.Moi} again, we see that   
		$$ \left|\int_{u}^{u_s}\partial_v(\frac{\nu}{r}) |\phi|^2\d u'\right| \lesssim \int_{u}^{u_s} (|\nu||\lambda|+\Omega^2 (1+Q^2+|\phi|^2)) |\phi|^2\d u'\lesssim v^{2-4s}.$$
		Therefore we actually showed that 
		\begin{equation}
		| \partial_v \log(\Omega^2)-2K_{\gamma}(v) +\int_{u_{\gamma}(v)}^{u_s}\frac{\nu}{r}\partial_v (|\phi|^2)\d u' +\partial_v (|\phi|^2)- \partial_v(\int_{u}^{u_s}\frac{\nu}{r}|\phi|^2\d u')| \lesssim v^{2-4s}+ v^{-2s} |\log(v)|^3.
		\end{equation}
		Note that the second and the third term of the left-hand-side only depend on $v$ and not on $u$.
		
		We define a new coordinate system $(u,V)$ with the following equations:
		\begin{equation} \label{Vdef1}
		\frac{\d V}{\d v}=e^{f(v)},
		\end{equation} \begin{equation} \label{Vdef2}
		f'(v)=2K_{\gamma}(v) +\int_{u_{\gamma(v)}}^{u_s}\frac{|\nu|}{r}\partial_v (|\phi|^2) (u',v)\d u'.
		\end{equation}
		By the estimates of   \cref{prop.LB.estimates.Moi}, note that $|f'(v)-2K_-| \lesssim v^{1-2s}$ and we recall that $K_-<0$; thus $V'(v)$ is integrable as $v \rightarrow+\infty$, and $V(v)$ increases towards a limit $V_{\infty}$ which we can choose to be $1$ without loss of generality. Therefore, we also have upon integration, as $v \rightarrow+\infty$:  $$ 1-V(v) \approx e^{f(v)}  .$$
		We also denote $\Omega^2_{CH}$ the metric coefficient in this system defined by $\Omega^2_{CH}=-2g(\partial_u,\partial_V)$, i.e.\
		$$ \Omega^2_{CH} \d u \d V = \Omega^2 \d u \d v, \text{ hence } \Omega^2_{CH}(u,v)= \Omega^2(u,v) e^{-f(v)}.$$
		We then have the claimed estimate \eqref{logomega1}
		\begin{equation*}
		\left| \partial_v \left( \log(\Omega^2_{CH})+|\phi|^2+ \int_{u}^{u_s}\frac{|\nu|}{r}|\phi|^2\d u' \right) \right| \lesssim v^{2-4s}+ v^{-2s} |\log(v)|^3.
		\end{equation*}
		
		Clearly, \eqref{logomega3} is a reformulation of \eqref{logomega2}.
		Since the right hand sides of \eqref{logomega1}  and \eqref{logomega2} are integrable in $v$ for $s>\frac{3}{4}$, a standard Cauchy sequence argument shows that $\Upsilon(u,v)$ admits a continuous extension, and $\partial_u \Upsilon(u,v)$ has a (locally) bounded extension.
	\end{proof}

	\subsubsection{Metric extendibility conditional on the boundedness of the scalar field}  \label{sec:metri.ext.section}
	
	Now that we have built the quantity $\Upsilon$ and proven its extendibility, we will prove that the continuous extendibility of $|\phi|$ implies the continuous extendibility of the metric (conversely, the blow-up of $|\phi|$ implies that there exists no coordinate system $(u,v)$ in which $\log(\Omega^2)$ is even bounded, see \cite{MoiChristoph2} and \cite{MoiThesis}).
	\label{ext.section}
	
	\begin{lem} \label{lemma.ext.easy} Assume that the function $(u,v) \in \LB \rightarrow |\phi|(u,v)$ extends continuously to $\CH \cap \{u\leq u_s\}$ as a continuous function $|\phi|_{CH}(u)$.	Then $\int_{u}^{u_s} \frac{\nu}{r} |\phi|^2(u',V)\d u'$  extends continuously to $\CH \cap \{u\leq u_s\}$ as a continuous function. Moreover, $ \nu(u,v)$ extends to $\CH \cap \{u\leq u_s\}$ as a bounded function $\nu_{CH}(u)$.
	\end{lem}
	
	\begin{rmk}
		In fact, we do not prove directly that $\nu$ extends as continuous function across the Cauchy horizon, as we do not control $\partial_u \nu$. However,   even though $\nu_{CH}$ might not be continuous in $u$, it is clearly in $L^1_{loc}$ (and even in $L^1(\CH \cap \{u\leq u_s\})$, as $|\nu_{CH}|\lesssim |u|^{-2s}$) which is sufficient for our purpose.
	\end{rmk}
	
	\begin{proof}
		Using the estimates of   \cref{prop.LB.estimates.Moi}, we see that for $(u,v) \in \LB$: $$ |\partial_v \nu|(u,v) \lesssim v^{-2s},$$
		which shows, by integrability, that for all $u \leq u_s$ there exists $\nu_{CH}(u)$ such that $\lim_{v \rightarrow+\infty} \nu(u,v)=\nu_{CH}(u)$. Now take again $u_{\infty} < u_s$ and two sequences $u_i \rightarrow u_{\infty}$, $V_i \rightarrow 1$, $V_i <1$ and write \begin{align*} &		\left|\int_{u_i}^{u_s} \frac{\nu}{r} |\phi|^2(u',V_i)\d u' - \int_{u_i}^{u_s} \frac{\nu_{CH}(u')}{r_{CH}(u')} |\phi|_{CH}^2(u')\d u'\right| \\ \leq\ & \left|\int_{u_i}^{u_{\infty}} \frac{\nu}{r} |\phi|^2(u',V_i)\d u'\right|+\left|\int_{u_i}^{u_{\infty}} \frac{\nu_{CH}(u')}{r_{CH}(u')} |\phi|_{CH}^2(u')\d u'\right|+ \left|\int_{u_{\infty}}^{u_s} \left( \frac{\nu}{r} |\phi|^2(u',V_i)- \frac{\nu_{CH}(u')}{r_{CH}(u')} |\phi|_{CH}^2(u') \right)\d u'\right|. 
		\end{align*}   
		
		Now both functions $\frac{\nu}{r} |\phi|^2(u,V)$ and $ \frac{\nu_{CH}(u')}{r_{CH}(u')} |\phi|_{CH}^2(u) $ are uniformly bounded in $u$ and $v$ on a set of form $(u,V) \in [u_{\infty}-\epsilon,u_s] \times [1-\epsilon,1]$ and $ \lim_{i \rightarrow +\infty} \frac{\nu}{r} |\phi|^2(u',V_i)=\frac{\nu_{CH}(u')}{r_{CH}(u')} |\phi|_{CH}^2(u')$ so by the dominated convergence theorem, the last term tends to $0$ as $i$ tends to $+\infty$.
		
		Moreover, the integrands of the first two terms are uniformly bounded, and thus these two terms tend to $0$ as $i$ tends to $+\infty$. This concludes the proof of the lemma.
	\end{proof}
	
	\begin{cor} \label{metric.ext.cor}
		Assume that the function $(u,v) \in \LB \rightarrow |\phi|(u,v)$ extends continuously to $\CH\cap \{u\leq u_s\}$ as a continuous function $|\phi|_{CH}(u)$. Then the metric $g$ admits a continuous extension $\tilde{g}$, which can be chosen to be $C^0$-admissible (Definition \ref{C0admissible}).
	\end{cor}
	\begin{proof}
		It follows  from   \cref{coordinatesphiboundedprop} and   \cref{lemma.ext.easy} that $\Omega^2_{CH}$ extends continuously to $\CH \cap \{u\leq u_s\}= \{u\leq u_s\} \times \{V=1\}$. We know already that $r$ extends continuously to $\CH \cap \{u\leq u_s\}= \{u\leq u_s\} \times \{V=1\}$, therefore, in view of the form of the metric \eqref{metric}, the corollary is proved.
	\end{proof}

	\subsection{Difference-type estimates on the scalar field and metric difference estimates} \label{difference.estimates.section}
	In this section, we carry out the nonlinear difference estimates. To do this, we have to introduce a new coordinate involving $h(v)$ defined in \eqref{h.def} (see already the difference estimate \eqref{HH.gauge}, to compare with \eqref{Omega.HH.1}): \begin{equation}\label{tv.def}
	\tilde{v}(v):= v+h(v).
	\end{equation}
	Recalling \eqref{h.bound}, it is clear that  $\tilde{v} = v \cdot (1+O(v^{1-2s}))$ and $\partial_{\tilde{v}} f = \partial_v f \cdot (1+O(v^{1-2s}))$ for all $f$. Note also \begin{align*}
	&\tilde{\Omega}^2(u,\tv(v))= \frac{\Omega^2(u,v)}{1+h'(v)}=(1+O(v^{1-2s}))\cdot  \Omega^2(u,v),\\ & \partial_{\tv} \log(\tilde{\Omega}^2)(u,\tv(v))= \frac{\partial_v \log(\Omega^2)(u,v)}{1+h'(v)}- \frac{h''(v)}{[1+h'(v)]^2}= (1+O(v^{1-2s}))\cdot  \partial_v \log(\Omega^2)(u,v)+ O(v^{-2s}),
	\end{align*}where  $\tilde{\Omega}^2:=-2g(\partial_u,\partial_{\tv})$. Estimates from \cref{recall.section} can be easily translated into $(u,\tv)$ coordinates: \begin{lem}\label{lem.trivial}
		Defining  $\tilde{\Omega}^2_H:=-2g(\partial_U,\partial_{\tv})$, the estimate \eqref{Omega.HH.1} on $\HH$ is replaced by:
		\begin{equation} \label{HH.gauge}
		\left|	\log( \frac{\tilde{\Omega}^2(0,\tv)}{\Omega^2_{RN}(0,\tv)})\right|=\left|	\log( \frac{\tilde{\Omega}^2_H(0,\tv)}{(\Omega^2_{RN})_H(0,\tv)})\right|\ls \tv^{1-2s},   \hskip 3 mm|\partial_{\tv} \log(\tilde{\Omega}_H^2)(0,\tv)- 2K_+| \ls \tv^{-2s}.
		\end{equation}
		Moreover, \eqref{OmegaPropRedshift}, \eqref{partialvRSOmegaprop} are replaced  by the following estimates valid in the spacetime region $\mathcal{R}$ \begin{equation}\label{Omega.new}
		e^{2K_+(u+\tv)}\ls   \tilde{\Omega}^2(U,\tv) \ls e^{2K_+(u+\tv)}, \hskip 5 mm |\partial_{\tv} \log(\tilde{\Omega}^2)(U,\tv) - 2K_+| \ls \tv^{-2s}.
		\end{equation}
		Finally, \eqref{partialvNSOmegaprop2} and  \eqref{partialvOmegatransition} are replaced by the following (weaker) estimates in the regions $\mathcal{N} \cup \mathcal{EB}$: \begin{equation} \label{Omega.ueless}
		|\partial_{\tv} \log(\tilde{\Omega}^2)(U,\tv) - 2K(U,\tv)| \ls \tv^{1-2s}.
		\end{equation}
		All the others estimates of \cref{recall.section} are still valid replacing $v$ by $\tv$, $\Omega^2$ by $\tilde{\Omega^2}$ and so on (adjusting the constants with no loss of generality, i.e.\ replacing $D_H$ by $2D_H$, $D_R$ by $2D_R$, $\frac{1}{3}$ by $\frac{2}{3}$ etc.).
	\end{lem}
	
	\begin{proof} This follows from  the equation $\tilde{\Omega}^2_H(0,\tv)= \frac{e^{2K_+\cdot( \tv-v_0)}}{1+h'(v)}$ (using the identity  \eqref{h.def}) and \eqref{Omega.HH.1.2}, \eqref{h.bound}.	\end{proof}
	\paragraph{Notation.}
	\emph{
		In view of \cref{lem.trivial}, from now on and until the end of the paper, we make a mild abuse of notation and redefine $v$ to be this new $\tilde{v}$ given by \eqref{tv.def} with the necessary adjustments, i.e.\ $\lambda$ becomes the notation for $\partial_{\tv}r$, $\Omega^2$ the notation for $-2g(\partial_u,\partial_{\tv})$, etc.  We will not use the old definition of $v$ any longer in what follows.}
	
	The goal of this section is to take the difference between  $\phi(u,v)$ and $\phiNl(u,v)$ and estimate the quantity: \begin{align}\dphi(u,v):=\phi(u,v)-\phiNl(u,v),\end{align}
	where $\phiNl$ solves the linear equation
	\begin{align}
	\label{eq:linearequationsection5}
	(\nabla_\mu +i q_0 (A^{RN})_\mu)(\nabla^\mu + iq_0 (A^{RN})^\mu) \phiNl  - m^2 \phiNl  =0
	\end{align}
	on the fixed Reissner--Nordström background \eqref{RN} in the gauge $A^{RN}$ as in \eqref{ARN}.  More precisely, we will define data for $\phiNl$ on $\HH$ (data on $\underline{C}_{in}$ is irrelevant) so as to match the data $\phiH \in \Sl$ for $\phi$  on $\HH$ (see already the paragraph immediately below): Our goal is then to prove that $\dphi$ is bounded and continuously extendible (for $q_0=0$), and similar estimates featuring nonlinear backreaction if $q_0 \neq 0$.

	We now define $\phiNl$ on $\mathcal Q^+$ as the unique solution of \eqref{eq:linearequationsection5} on the fixed Reissner--Nordström metric \eqref{RN} with parameters $(M,e)$ and with data \begin{align*}
	&  \phiNl(u,v_0) \equiv 0 \hskip 5 mm \forall u \in (-\infty,u_s], \\
	& (\phiNl)_{|\mathcal{H}^+} (v)\equiv \chi_{\geq v_0 + 3}(v) \phiH(v)  \hskip 5 mm \forall v \in [v_0,+\infty),
	\end{align*} 
	where $\chi_{v_0 + \geq 3}$ is the smooth cut-off supported on $v\geq v_0 +  2 $ and $\chi_{ \geq v_0 +  3} =1 $ for $v\geq v_0 + 3$ as defined in \cref{linearcor}.
	\begin{rmk} Note that the unique solution $\phil$ arising from the above data in the gauge \eqref{gauge_linear}, which is used in \cref{linearsection}, agrees with $\phiNl$ up to a gauge transformation as the gauges agree for the initial data, in particular, $A^{RN}_v = ( A_{RN}^\prime )_v =0$ on the event horizon by construction. 
	\end{rmk}
	Recall that $\phiNl$ is also a solution of \eqref{Field}, \eqref{Field2}, \eqref{Field4}, \eqref{Field5} where $(r,\Omega^2, A, D, \phi)$ are all replaced by their Reissner--Nordstr\"{o}m analogs $(r_{RN}, \Omega^2_{RN}, A^{RN}, D^{RN}, \phiNl)$. Similarly, $r_{RN}$, $\Omega^2_{RN}$, $A^{RN}$ also satisfy the equations of   \cref{eqcoord} with $\phi \equiv 0$ (i.e.\ \eqref{RN} satisfies the Einstein--Maxwell equations in spherical symmetry), a fact we will repetitively use.
	
	The estimates of \cite{Moi}, that are recalled in \cref{recall.section}  and stated in \cref{lem.trivial} in our new coordinate system, are key to our new difference estimates. We will use these estimates throughout the argument, without necessarily referring to them explicitly.
	\subsubsection{Difference estimates in the red-shift region}
	\begin{prop} \label{RS.diff.prop}
		There exists $D_H'(M,e,q_0,m^2,s,D_1,D_2)>0$ such that for all $(u,v) \in \Rs$: \begin{align} \nonumber
		|r(u,v)- & r_{RN}(u,v)| + |\lambda(u,v)-\lambda_{RN}(u,v)|\\ 
		\label{diff1.RS} & +  |Q(u,v)-e|+ |\log(\Omega^2)(u,v)-\log(\Omega^2_{RN})(u,v)| \leq D_H' \cdot v^{1-2s}, 
		\end{align}	\begin{align}\nonumber
		|\partial_u \log(\Omega^2)(u,v)- & \partial_u \log(\Omega^2_{RN})(u,v)|  +	|\nu(u,v)-\nu_{RN}(u,v)|	 \\&	+|A_u(u,v)-A_u^{RN}(u,v)| \leq D_H' \cdot e^{2K_+(u+v)} \cdot v^{1-2s}, \label{diff2.RS}
		\end{align}
		\begin{equation} \label{diff3.RS}
		|\partial_u \dphi| \leq D_H' \cdot e^{2K_+(u+v)} \cdot v^{1-3s},
		\end{equation}
		\begin{equation} \label{diff4.RS}
		|\dphi|	+	|\partial_v \dphi| \leq D_H' \cdot  v^{1-3s}.
		\end{equation}
	\end{prop}
	\begin{proof}
		
		First, recall that $r_{RN} \equiv r_+(M,e)$, $Q_{RN}\equiv e$, $\varpi_{RN}\equiv M$, and $\lambda_{RN} \equiv 0$ on the event horizon $\HH$, by definition.
		Lastly, recall that $A_u = A_u^{RN} $ on $\underline{C}_{in}$ by the gauge choice \eqref{gaugeAponctuelle}. Recalling that $D_H>0$ is defined in   \cref{prop.HH.estimates.Moi}, we bootstrap the following estimates: \begin{equation} \label{B.diff1.RS}
		|r(u,v)-r_{RN}(u,v)|  \leq 4 D_H \cdot v^{1-2s},
		\end{equation}
		\begin{equation} \label{B.diff2.RS}		|\log(\Omega^2)(u,v)-\log(\Omega^2_{RN})(u,v)| \leq 4 B_H \cdot v^{1-2s}	\end{equation}
		for $B_H(M,e,q_0,m^2,D_1,D_2)>0$ defined as the  constant in \eqref{HH.gauge} such that  $|\log(\frac{\Omega^2(0,v)}{\Omega^2_{RN}(0,v)})| \leq B_H \cdot v^{1-2s}$ in the new coordinate $v$. Plugging these bootstraps into \eqref{Radius2} and using \eqref{phi.HH}, \eqref{Omega.new}, we find that \begin{align*} |\partial_u (r\partial_{v}r -r_{RN} \partial_{v} r_{RN})| = |\partial_u  (r\lambda-r_{RN}\lambda_{RN})| &\ls |\Omega^2-\Omega^2_{RN}|+ \Omega^2 \cdot \left(|\phi|^2+ |r-r_{RN}|+ |Q-e|\right) \\ &  \ls  e^{2K_+(u+v)} \cdot v^{1-2s},\end{align*}
		where we used $|\Omega^2- \Omega^2_{RN}| \ls \Omega^2 \cdot |\log(\frac{\Omega^2}{\Omega^2_{RN}})| \ls   \Omega^2 \cdot v^{1-2s}$. This is also equivalent (recalling \eqref{eq:relationomegas}) to  $$ |\partial_U  (r\lambda-r_{RN}\lambda_{RN})| \ls  e^{2K_+ v} \cdot v^{1-2s}.$$
		Integrating the above using \eqref{lambda.HH} we get  \begin{equation} \label{rlambda.diff}
		|r\lambda-r\lambda_{RN}|\ls v^{-2s}+   \Omega^2 \cdot v^{1-2s}.
		\end{equation} Writing now the difference for \eqref{Radius}, taking advantage of \eqref{rlambda.diff} and the bootstraps gives 
		$$|\partial_v \partial_{U} (r- r_{RN})|  \ls |\lambda_{RN}| \cdot  | \partial_U r -\partial_U r_{RN}| + e^{2K_+v } v^{-2s} +  e^{2K_+v } v^{1-2s}.$$
		
		Integrating in $v$ using a Gronwall estimate and  the boundedness of $\partial_U r$ on $\underline{C}_{in}$ we get 
		$$ |\partial_U r-\partial_U r_{RN}| \ls 1+    e^{2K_+v} \cdot v^{1-2s} ,$$ which, upon integrating in $U$ this time and using \eqref{r.HH} gives $$ | r- r_{RN}| \leq D_H \cdot v^{1-2s}+ D \cdot e^{2K_+(u+v)} e^{-2K_+ v}+  D  \cdot e^{2K_+(u+v)} \cdot v^{1-2s} ,$$ where $D(M,e,q_0,m^2,D_1,D_2)>0$. Choosing $\Delta$ sufficiently large 
		such that \begin{align*}e^{2K_+(u+v)} \leq e^{-2K_+ \Delta} <  D^{-1} \cdot D_H \end{align*} allows us to retrieve bootstrap \eqref{B.diff1.RS}.

		Similarly plugging \eqref{phivRS}, \eqref{phiURS}, \eqref{B.diff2.RS} and the previously proven estimates into \eqref{Omega} we get  \begin{align*}
		& |\partial_u \partial_v (\log(\Omega^2)-\log(\Omega^2_{RN}))|   \ls |D_u \phi| \cdot |\partial_v \phi| + |\Omega^2-\Omega^2_{RN}| + \Omega^2 \cdot v^{1-2s} \ls \Omega^2  \cdot   v^{1-2s}, \end{align*} or equivalently using \eqref{Omega.new} \begin{align*}
		& |\partial_v \partial_U (\log(\Omega^2)-\log(\Omega^2_{RN}))|  \ls e^{2K_+v} \cdot v^{1-2s}. \end{align*}
		Integrating in $v$ using the boundedness of  $\partial_U \log(\Omega^2)$ and  $\partial_U \log(\Omega^2_{RN})$ on $\underline{C}_{in}$ we get $$|\partial_U (\log(\Omega^2)-\log(\Omega^2_{RN}))| \ls 1+  e^{2K_+v} \cdot v^{1-2s} \ls e^{2K_+v} \cdot v^{1-2s} ,$$  from which we retrieve bootstrap \eqref{B.diff2.RS}, using the  smallness of $e^{-2K_+ \Delta}$ as we did above. 
		
		Now all bootstraps are closed and we continue with the proof of the claimed difference estimates. Taking the difference between \eqref{eq:maxwellAu} and its Reissner--Nordstr\"{o}m version, and integrating in $v$ using $A_u(u,v_0)-A_u^{RN}(u,v_0)=0$,  we obtain $$|A_u(u,v)-A_u^{RN}(u,v)|  \ls e^{2K_+ \cdot (u+v)} \cdot v^{1-2s}.$$

		For $\dphi$, we introduce a new bootstrap assumption (completely independently from the other bootstraps assumptions that have already been retrieved), which is true on $\underline{C}_{in}$ by assumption: \begin{equation} \label{B.R.dphi}
		|\partial_u\dphi|(u,v) \leq B_1 \cdot e^{2K_+(u+v)} \cdot v^{1-3s},
		\end{equation} for some $B_1>0$ large enough to be chosen later. Integrating in $u$ and using $|\dphi|\ls v^{1-3s}$ on the event horizon $\HH$ (since  $\dphi_{|\HH} \equiv 0$ for $v\geq 3$) gives \begin{equation} \label{B.R.dphi2}
		|\dphi|(u,v) \lesssim (1+B_1) \cdot e^{2K_+(u+v)} \cdot v^{1-3s} \lesssim  (1+B_1) \cdot e^{-2K_+\Delta} \cdot v^{1-3s}.
		\end{equation}
		
		Now we take the difference of \eqref{Field} obeyed by $\phi$ and the corresponding equation obeyed by $\phiNl$, namely
		
			\begin{align*}
	&	\partial_{u}\partial_{v} (\dphi) =-\frac{\partial_{u}\dphi\ \partial_{v}r}{r}-\frac{ \partial_{v}\dphi\ \partial_{u}r}{r} +\frac{ q_{0}i \Omega^{2}}{4r^{2}}Q\ \dphi
		-\frac{ m^{2}\Omega^{2}}{4}\dphi- i q_{0} A_{u}\frac{\dphi\  \partial_{v}r}{r}-i q_0 A_{u}\ \partial_{v}\dphi\\&  - \partial_{u}\phiNl\ [\frac{ \partial_{v}r}{r}-\frac{ \partial_{v}r_{RN}}{r_{RN}}] - \partial_{v}\phiNl\ [\frac{ \partial_{u}r}{r}-\frac{ \partial_{u}r_{RN}}{r_{RN}}] +[\frac{ q_{0}i \Omega^{2}}{4r^{2}}Q-\frac{ q_{0}i \Omega^{2}_{RN}}{4r^{2}_{RN}}e]\ \phiNl \\ &
		-\frac{ m^{2}[\Omega^{2}-\Omega^2_{RN}]}{4}\phiNl- i q_{0} [A_{u}\frac{ \partial_{v}r}{r}-A_{u}^{RN}\frac{ \partial_{v}r_{RN}}{r_{RN}}]\ \phiNl\ -i q_0 [A_{u}-A_{u}^{RN}]\ \partial_{v}\phiNl. \end{align*}

		 We get, using also \eqref{B.R.dphi}, \eqref{B.R.dphi2} and \eqref{Omega.new}   (note that one can write $\phiNl = \phi-\dphi$ and use \eqref{phivRS}, \eqref{phiURS} to bound $\phi$ and  \eqref{B.R.dphi}, \eqref{B.R.dphi2} to bound $\phiNl$) 
		\begin{align}
		&|\partial_u \partial_v \dphi| \ls  e^{2K_+(u+v)} \cdot ( 1+B_1) \cdot v^{1-3s} + v^{-s} \cdot e^{2K_+(u+v)}  \cdot|\partial_v \dphi|. \label{partialuvphi}	\end{align}  Integrating in $u$ and using Gronwall's estimate we get (recalling that  $|\partial_v\dphi|\ls v^{1-3s}$ on $\HH$) we get: $$ |\partial_v \dphi| \ls (1+ B_1\cdot  e^{-2K_+ \Delta}) \cdot v^{1-3s},$$ and using this in \eqref{partialuvphi} we get  	\begin{align*} 	&|\partial_v \partial_u \dphi| \ls e^{2K_+(u+v)}  \cdot ( 1+B_1) \cdot v^{1-3s} +  B_1 \cdot e^{-2K_+ \Delta} \cdot v^{1-4s}\cdot e^{2K_+(u+v)} . 	\end{align*}
		Integrating in $v$ this time, choosing $B_1$ appropriately and using the smallness of  $e^{-2K_+ \Delta}$, retrieves, together with another integration in $u$, bootstrap \eqref{B.R.dphi}, gives the claimed estimates on $\dphi$ and concludes the proof.
	\end{proof}
	\subsubsection{Difference estimates in the no-shift region}
	
	\begin{prop} \label{N.diff.proposition}
		There exists $C_N=C_N(M,e,q_0,m^2,s,D_1,D_2)>0$ such that the following estimates are satisfied for all $(u,v) \in \NN$: \begin{align*} 
		&	|r(u,v)-r_{RN}(u,v)| +  |Q(u,v)-e|+ |\log(\Omega^2)(u,v)-\log(\Omega^2_{RN})(u,v)| \leq C_N \cdot v^{1-2s}, \\ & 
		|\partial_v\log(\Omega^2)(u,v)-\partial_v\log(\Omega^2_{RN})(u,v)|+ |\partial_u\log(\Omega^2)(u,v)-\partial_u\log(\Omega^2_{RN})(u,v)| \\ &+  |\lambda(u,v)-\lambda_{RN}(u,v)|+		|\nu(u,v)-\nu_{RN}(u,v)|	+	|A_u(u,v)-A_u^{RN}(u,v)| \leq C_N \cdot  v^{1-2s},\\ & 
		|\dphi|	+	|\partial_u \dphi|+ 	|\partial_v \dphi| \leq C_N \cdot  v^{1-3s}.
		\end{align*}
	\end{prop}
	
	\begin{proof}
		The proof  consists of combination of the proof of   \cref{RS.diff.prop} with that of in \cite[Proposition 4.7]{Moi}: we partition $\NN$ into smaller regions $\mathcal{N}_k:=\{ -\Delta+ (k-1) \epsilon \leq u+v \leq -\Delta+ k \epsilon\}$, for $k \in [[1,N]]$ and $N \cdot \epsilon  = \Delta'$. We will prove the result by finite induction on $k$: the induction hypothesis is that the following estimates hold in  $\mathcal{N}_k$: \begin{equation} \label{Ind.diff1.N}
		|r(u,v)-r_{RN}(u,v)| + |\lambda(u,v)-\lambda_{RN}(u,v)|+ |\log(\Omega^2)(u,v)-\log(\Omega^2_{RN})(u,v)| \leq C_k \cdot v^{1-2s},
		\end{equation}
		\begin{equation} \label{Ind.diff2.N}
		|\nu(u,v)-\nu_{RN}(u,v)|	+	|A_u(u,v)-A_u^{RN}(u,v)| \leq C_k \cdot  v^{1-2s},
		\end{equation}
		where $C_k = 2^k  \cdot B_N$ for a large enough constant $B_N>0$ to be determined later. The estimates of   \cref{RS.diff.prop} render the initialization of the induction true for $B_N$ large enough. So we assume that \eqref{Ind.diff1.N}, \eqref{Ind.diff2.N} hold for $\mathcal{N}_k$ and we prove them in $\mathcal{N}_{k+1}$. As before we bootstrap \begin{equation} \label{B.Ind.diff1.N}
		|r(u,v)-r_{RN}(u,v)| + |\lambda(u,v)-\lambda_{RN}(u,v)|+ |\log(\Omega^2)(u,v)-\log(\Omega^2_{RN})(u,v)| \leq 4 C_k \cdot v^{1-2s},
		\end{equation}
		\begin{equation} \label{B.Ind.diff2.N}
		|\nu(u,v)-\nu_{RN}(u,v)|	+	|A_u(u,v)-A_u^{RN}(u,v)| \leq4 C_k \cdot  v^{1-2s},
		\end{equation}
		We treat one typical term, to show the specificity of the no-shift region $\NN$ compared to $\Rs$: under the bootstraps and \eqref{Omega.r.boring}, \eqref{QNS2} we have $$ |\partial_u \partial_v r | \ls  (1+C_k) \cdot v^{1-2s} \sim C_k \cdot |u|^{1-2s} .$$ Upon integration in the $u$ direction,  it gives, using \eqref{Ind.diff1.N} in the past: for some $E(M,e,q_0,m^2,D_1,D_2)>0$ $$ |\lambda-\lambda_{RN}| \leq C_k \cdot v^{1-2s} +E \cdot \epsilon \cdot C_k \cdot v^{1-2s} ,$$ thus for $\epsilon>0$ sufficiently small, so that $E \cdot \epsilon<1$, we close the part of bootstrap \eqref{B.Ind.diff1.N} relative to $\lambda-\lambda_{RN}$. The other terms are addressed similarly, we omit the details. Such estimates allow us to retrieve bootstraps \eqref{B.Ind.diff1.N}, \eqref{B.Ind.diff2.N} and prove the induction hypothesis. Once this is done, we can prove difference estimates for $\dphi$ exactly as in    \cref{RS.diff.prop}.
	\end{proof}
	\subsubsection{Difference estimates in the early blue-shift region}
	
	\begin{prop}  \label{EB.diff.proposition}
		There exists a constant $C_E = C_E(M,e,q_0,m^2,s,D_1,D_2)>0$	 such that the following estimates are satisfied for all $(u,v) \in \EB$: \begin{align} \label{diff1.EB}
		&	|r(u,v)-r_{RN}(u,v)| \leq  C_E \cdot v^{1-2s},\\ & \label{diff2.EB}
		|\nu(u,v)-\nu_{RN}(u,v)|	+	|A_u(u,v)-A_u^{RN}(u,v)| + |\lambda(u,v)-\lambda_{RN}(u,v)| \leq  C_E \cdot  v^{1-2s}, \\ &
		\label{diff2.5.EB}
		|\partial_u\log(\Omega^2)(u,v)-\partial_u\log(\Omega^2_{RN})(u,v)|+ |\partial_v\log(\Omega^2)(u,v)-\partial_v\log(\Omega^2_{RN})(u,v)| \leq C_E \cdot  v^{1-2s}, \\ \label{diff3.EB}
		&	|\partial_u \dphi|+ 	|\partial_v \dphi| \leq  C_E \cdot  v^{1-3s},\\
		\label{diff4.EB}
		&	| \dphi| \leq C_E \cdot  v^{1-3s} \cdot \log(v),\\
		&	 \label{diff5.EB}
		|\log(\Omega^2)(u,v)-\log(\Omega^2_{RN})(u,v)| \leq C_E \cdot  v^{1-2s} \cdot \log(v).
		\end{align}
		
	\end{prop}
	
	\begin{proof} Note that in $\EB$, as in $\NN$, we have $v \sim |u|$ and that the size of the region is logarithmic i.e.\ $u-u_{\gamma_N}(v) \ls \log(v) \sim \log(|u|)$ and  $v-v_{\gamma_N}(u) \ls \log(|u|)\sim \log(v)$. As before, we start with bootstraps:
		
		\begin{equation} \label{B.diff2.EB}
		|\lambda(u,v)-\lambda_{RN}(u,v)|+ |\nu(u,v)-\nu_{RN}(u,v)|
		\leq 4 C_N \cdot v^{1-2s},
		\end{equation}
		\begin{equation} \label{B.diff3.EB}		\Omega^2(u,v) \cdot |\log(\Omega^2)(u,v)-\log(\Omega^2_{RN})(u,v)|\leq 4C_N \cdot  v^{1-3s},		\end{equation}
		\begin{equation} \label{B.diff4.EB} |r(u,v)-r_{RN}(u,v)|	 \leq B_N \cdot v^{1-2s} ,	\end{equation}  for some $B_N > C_N$ to be determined later. The set of $(u,v)$ for which these bootstraps are satisfied is non-empty by the estimates of   \cref{N.diff.proposition}.

		Retrieving the bootstrap on $r-r_{RN}$ is the most delicate.  We use \eqref{mu} and write the difference of the two identities below \begin{align*} 
		&	\lambda \cdot \kappa^{-1}= \nu \cdot \iota^{-1} =\frac{-2\lambda \nu}{\Omega^2}=\frac{1}{2}-\frac{\varpi}{r}+\frac{Q^2}{2r^2}, \\ &   \nu_{RN}=\lambda_{RN}=-\frac{\Omega^2_{RN}}{2}=\frac{1}{2}-\frac{M}{r_{RN}}+ \frac{e^2}{2r_{RN}^2}   	\end{align*} 	Thus, we have \begin{align*}	&(\lambda- \lambda_{RN}) \cdot \kappa^{-1}+ \lambda_{RN} \cdot (\kappa^{-1}-1)\\=\ & (\lambda- \lambda_{RN}) \cdot \kappa^{-1}+  (\nu_{RN}-\nu)+ \nu \cdot (1-e^{-\log(\Omega^2)+ \log(\Omega^2_{RN})}))\\=\ & -\frac{\varpi}{r} + \frac{Q^2}{2r^2}+  \frac{M}{r_{RN}}-\frac{e^2}{2r_{RN}^2}\frac{M-\varpi}{r}+ \frac{M}{r \cdot r_{RN}} \cdot (r-r_{RN})+\frac{Q^2-e^2}{2r^2}- \frac{e^2 \cdot (r+r_{RN})}{2r^2 \cdot r_{RN}^2} \cdot (r-r_{RN}),
		\end{align*} 
	 	hence, combined with the $(r-r_{RN})$ terms we have   \begin{align*} &
		\left(\frac{M}{r \cdot r_{RN}} - \frac{e^2 \cdot (r+r_{RN})}{2r^2 \cdot r_{RN}^2} \right)\cdot (r-r_{RN})= \frac{2M \cdot r \cdot r_{RN} - e^2 \cdot (r+r_{RN})}{2r^2 \cdot r_{RN}^2} \cdot (r-r_{RN})\\ &\;\;\;= (\lambda- \lambda_{RN}) \cdot \kappa^{-1}+  (\nu_{RN}-\nu)+ \nu \cdot (1-e^{-\log(\Omega^2)(u,v)+ \log(\Omega^2_{RN})}) +\frac{-M+\varpi}{r}+ \frac{-Q^2+e^2}{2r^2} .
		\end{align*}
		To conclude, we have to prove that the pre-factor of the left-hand-side $\frac{2M \cdot r \cdot r_{RN} - e^2 \cdot (r+r_{RN})}{r^2 \cdot r_{RN}^2}$ is bounded away from zero: for this, notice that, since $0<|e|< M$ we have $ r_-(M,e)= M-\sqrt{M^2-e^2} < \frac{e^2}{M}$,  which is equivalent  to  $$ 2M \cdot r_-^2 < 2e^2 \cdot r_-.$$ By \eqref{r.nodiff.EB} and choosing $\Delta_{N}$ sufficiently large,  there exists a small constant $\alpha(M,e)>0$ such that in $\EB$: $$ \left| \frac{2M \cdot r \cdot r_{RN} - e^2 \cdot (r+r_{RN})}{2r^2 \cdot r_{RN}^2} \right|> \alpha.$$
		Thus, as a consequence of bootstrap \eqref{B.diff2.EB} and \eqref{MTrans}, \eqref{QTrans}  and \eqref{kappatransition}, there exists $C_N'(M,e,q_0,m^2,s,D_1,D_2)>0$ such that
		$$ |r-r_{RN}| \leq C'_N \cdot  v^{1-2s} +   C'_N\cdot |\nu| \cdot  |\log(\Omega^2)- \log(\Omega^2_{RN})| \leq C'_N \cdot  v^{1-2s} +  2 C'_N \cdot 4 C_N \cdot v^{1-2s} < B_N \cdot v^{1-2s} ,$$ where we chose $B_N = 2 C'_N + 4 C'_N  \cdot 4 C_N$ for the last inequality to be true. Therefore, bootstrap \eqref{B.diff4.EB} is retrieved.
		
		Now we turn to bootstrap \eqref{B.diff3.EB}, which is equally delicate (because we want to avoid a logarithmic loss). As in   \cref{RS.diff.prop}, we write the difference between \eqref{Omega} satisfied by $\Omega^2$ and the analogous equation satisfied by $\Omega^2_{RN}$. Using also \eqref{QTrans} and bootstrap \eqref{B.diff2.EB}, \eqref{B.diff3.EB}, \eqref{B.diff4.EB} we obtain: \begin{align*}			    &|\partial_v\partial_u ( \log(\Omega^2)- \log(\Omega^2_{RN}))| \\ \ls\ & |D_u\phi|\cdot |\partial_v \phi| + \Omega^2 \cdot ( |Q-e|+ |r-r_{RN}|)+ |\Omega^2-\Omega^2_{RN}|+ |\lambda| \cdot |\nu - \nu_{RN}|   + |\lambda-\lambda_{RN}| \cdot |\nu| \\\ls\ & v^{-2s}+ \Omega^2 \cdot v^{1-2s} + \Omega^2 \cdot |\log(\frac{\Omega^2}{\Omega^2_{RN}})|+ \Omega^2 \cdot |\nu - \nu_{RN}|   +\Omega^2 \cdot |\lambda-\lambda_{RN}| \ls v^{-2s}+ \Omega^2 \cdot v^{1-2s},  			\end{align*}				where in the last line we have used \eqref{iotatransition}, \eqref{kappatransition} as $|\lambda|,\ |\nu| \ls \Omega^2$ and the usual inequality $|\Omega^2-\Omega^2_{RN}| \ls \Omega^2\cdot |\log(\frac{\Omega^2}{\Omega^2_{RN}})|$ (which is true because $\frac{\Omega^2_{RN}}{\Omega^2} \geq \frac{1}{10}$, an estimate which follows directly from \eqref{Omegatransition}). Integrating in $v$ (recall  the $v$-difference is of size $\log(v)$), we get, using   \cref{N.diff.proposition}: \begin{equation}	\label{partialu.omega.int}	    |\partial_u \log(\Omega^2)-\partial_u \log(\Omega^2_{RN})|(u,v) \ls v^{1-2s}.	\end{equation}
		
		Instead of integrating \eqref{partialu.omega.int} directly (and incur a logarithmic loss), we write an identity: for any $\eta>0$:	$$ \partial_u  \left[ \Omega^{\eta}\cdot (\log(\Omega^2)- \log(\Omega^2_{RN})) \right]= \Omega^{\eta}\cdot \frac{\eta}{2} \cdot \partial_u \log(\Omega^2) \cdot ( \log(\Omega^2)- \log(\Omega^2_{RN}))+  \Omega^\eta \cdot  \partial_u[ \log(\Omega^2)- \log(\Omega^2_{RN})], $$  from which we deduce, using also $\partial_u \log(\Omega^2)<0$ (see   \cref{prop.EB.estimates.Moi}): \begin{align*}			 &   \partial_u \left[ \Omega^{2\eta}  \cdot ( \log(\Omega^2)- \log(\Omega^2_{RN}))^2\right]   \\ =\ & 2 \eta \cdot \Omega^{2\eta} \cdot \partial_u \log(\Omega^2) \cdot ( \log(\Omega^2)- \log(\Omega^2_{RN}))^2  +  2\Omega^\eta \cdot \partial_u[ \log(\Omega^2)- \log(\Omega^2_{RN})] \cdot ( \log(\Omega^2)- \log(\Omega^2_{RN})) \\ \leq\  & 2\Omega^{2\eta} \cdot |\partial_u[ \log(\Omega^2)- \log(\Omega^2_{RN})]| \cdot | \log(\Omega^2)- \log(\Omega^2_{RN})| ,	\end{align*}  which in turn implies, using \eqref{partialu.omega.int}:  \begin{align*}			 &   \partial_u \left[ \Omega^{\eta}  \cdot | \log(\Omega^2)- \log(\Omega^2_{RN})|\right]    =  \Omega^{\eta} \cdot |\partial_u[ \log(\Omega^2)- \log(\Omega^2_{RN})]| \ls \Omega^{\eta} \cdot v^{1-2s}.	\end{align*} 
		
		Integrating  the above in $u$ using $\partial_u \log(\Omega^2) \in (3K_-,K_-)$ (see   \cref{prop.EB.estimates.Moi}) and the bounds from   \cref{N.diff.proposition}, we get for some $E'(M,e,q_0,m^2,s,D_1,D_2)>0$ \begin{equation}		\label{prelim.0.5.EB}
		\Omega^{\eta}(u,v) \cdot| \log(\Omega^2)(u,v)-\log(\Omega^2_{RN})(u,v)| \leq  C_N \cdot v^{1-2s}+  E' \cdot \eta^{-1} \cdot 	\Omega^{\eta}(u_{\gamma_{\mathcal{N}}}(v),v)\cdot  v^{1-2s}.
		\end{equation} 
		
		Applying \eqref{prelim.0.5.EB} for $\eta=2$, choosing $\Delta_N$ large enough so that
		$\Omega^{\eta}(u_{\gamma_{\mathcal{N}}}(v),v) \approx e^{2K_- \Delta_N}< \frac{C_N}{10 E'} $ retrieves bootstrap \eqref{B.diff3.EB}.
		
		Retrieving bootstrap \eqref{B.diff2.EB} is done similarly: we integrate the difference between \eqref{Radius} satisfied by $r$ and the analog satisfied by $r_{RN}$, using   \cref{N.diff.proposition}, and we prove:	\begin{equation*} 
		|\lambda(u,v)-\lambda_{RN}(u,v)|+ |\nu(u,v)-\nu_{RN}(u,v)|
		\leq 3 C_N \cdot v^{1-2s},
		\end{equation*} which closes all the  bootstrap assumptions.

		Now we turn to the rest of the differences estimates claimed in the statement of the proposition. Integrating the differences into \eqref{eq:maxwellAu}, \eqref{Omega} as we did in   \cref{RS.diff.prop} gives straightforwardly \begin{align} & |A_u(u,v)-A_u^{RN}(u,v)| \ls v^{1-2s}, \nonumber\\& |\partial_v \log(\Omega^2)(u,v)- \partial_v\log(\Omega^2_{RN})(u,v)| \ls v^{1-2s}, \nonumber
		\end{align} where we also used that the  size of the region of integration is logarithmic i.e. $\int^u_{u_{\gamma}(v)} v^{-2s} \d u \ls v^{-2s}\log(v)$. 
		
		For $\dphi$, we proceed as in   \cref{RS.diff.prop} and make the following bootstrap assumptions for some $B'>0$ \begin{equation}\label{phi.B.new}
		\Omega(u,v) \cdot |\dphi|(u,v)  \leq B' \cdot v^{1-3s},
		\end{equation}
		\begin{equation}\label{phi.B.new2}
		|\partial_v \dphi|(u,v)  \leq B' \cdot v^{1-3s}.
		\end{equation}
		Plugging differences into \eqref{Field} satisfied by $\phi$ and the analogous equation satisfied by $\dphi$, we get, using \eqref{phi.B.new}, \eqref{phi.B.new2} and the previously proven difference estimates \begin{equation} \label{partialuvphi.EB}
		|\partial_u \partial_v \dphi | \ls B' \cdot \Omega \cdot v^{1-3s} + \Omega^2 \cdot  |\partial_u \dphi|,
		\end{equation} from which we deduce, upon integrating in $v$ and using a Gronwall estimate:  \begin{equation} \label{partialuphi.EB}
		|\partial_u  \dphi |(u,v) \ls (1+B' \cdot \Omega(u,v_{\gamma_{\mathcal{N}}}(u)) )\cdot v^{1-3s},
		\end{equation} and plugging \eqref{partialuphi.EB} into \eqref{partialuvphi.EB} and integrating in $u$ this time we get   \begin{equation} \label{partialvphi.EB}
		|\partial_v  \dphi |(u,v) \ls (1+B' \cdot [\Omega(u,v_{\gamma_{\mathcal{N}}}(u))+ \Omega(v_{\gamma_N}(v),v)] )\cdot v^{1-3s},
		\end{equation} which is sufficient to retrieve bootstrap \eqref{phi.B.new2} after an appropriate choice of $B'$ and choosing also $\Delta_N$ large enough (to obtain a small constant from $\Omega(u_{\gamma_{\mathcal{N}}}(v),v)$ as we did above).
		
		To retrieve bootstrap \eqref{phi.B.new}, we proceed as with $\partial_u \log(\Omega^2)$ earlier, with the following identity $$ \partial_u ( \Omega^\eta \dphi)=  \frac{\eta}{2} \cdot\partial_u \log(\Omega^2) \cdot \Omega^\eta \cdot \dphi+ \Omega^\eta  \cdot \partial_u\dphi, $$  which also implies, using  \eqref{partialuphi.EB} and by the same reasoning as for $\partial_u \log(\Omega^2)$ above: 
		$$ \partial_u ( \Omega^{\eta} |\dphi|) \leq    \Omega^{\eta}  \cdot   |\partial_u(\dphi)| \ls  \Omega^{\eta}  \cdot  (1+B' \cdot \Omega(u,v_{\gamma_{\mathcal{N}}}(u)) )\cdot v^{1-3s}  .$$  Integrating this inequality in $u$ for $\eta=1$,  after an appropriate choice of $B'$ and choosing also $\Delta_N$ large enough as we did above allows to retrieve bootstrap \eqref{phi.B.new} and concludes the proof.
	\end{proof}
	\subsubsection{Difference estimates in the late blue-shift region} \label{diff.blue.section}
	In this section, we will not need to estimate metric differences anymore (although we will use the difference estimates from past sections): therefore, we do not require a bootstrap method and proceed directly.
	\begin{prop} \label{LB.prop}
		There exists a constant $C_L= C_L(M,e,q_0,m^2,s,D_1,D_2)>0$	 such that the following are satisfied for all $(u,v) \in \LB$: 
		\begin{align} \label{diff1.LB}	&|A_u(u,v)-A_u^{CH}(u) |+ |A_u^{RN}(u,v)-(A_u^{RN})^{CH}(u) | \leq C_L \cdot \Omega^2(u,v) \leq C_L^2 \cdot v^{-2s},\\ & \label{diff3.LB}	|A_u(u,v)-A_u^{RN}(u,v) |+ |A_u^{CH}(u)-(A_u^{RN})^{CH}(u) | \leq C_L \cdot  |u|^{1-2s},\\ &
		\label{diff8.LB}
		\bigl |\frac{\d}{\d u}( A_u^{CH}-(A_u^{RN})^{CH}) \bigr|(u) \leq C_L \cdot |u|^{1-2s},\\ \nonumber
		&\bigl|	|D_v \psi|(u,v)- 	| 	D_v \psil|(u,v) 	\bigr|\\	 	\label{diff4.LB}	 & \leq  \bigl| e^{iq_0 \int_{u_{\gamma}(v)}^u A_u^{CH}(u') \d u'} 	\partial_v \psi(u,v)- 		 e^{iq_0 \int_{u_{\gamma}(v)}^u (A_u^{RN})^{CH}(u') \d u'}  \partial_v \psil(u,v) 	\bigr|   \leq  C_L \cdot  v^{1-3s}, \\ 
		\label{diff6.LB}
		&	\bigl|  \psi(u,v)-  \int_{v_{\gamma}(u)}^{v} e^{iq_0 \int_{u_{\gamma}(v')}^u [(A_u^{RN})^{CH}-A_u^{CH}](u') \d u'}	\partial_v \psil(u,v') \d v'  	\bigr|  \leq C_L \cdot  |u|^{2-3s},\\ &		\label{diff7.LB}
		\bigl| |D_u \psi|(u,v)	- |D_u^{RN} \psi_{\mathcal{L}}|(u,v) \bigr|  \leq	\bigl| D_u \psi(u,v)	- D_u^{RN} \psi_{\mathcal{L}}(u,v) 	\bigr|   \leq  C_L \cdot  |u|^{1-3s}	\cdot \log|u|.	\end{align}	
		
		Moreover, for every fixed $u<u_s$, there exists  $f(u) \in \mathbb{C}$ such that \begin{equation} \label{ext.difference}
		\lim_{v \rightarrow +\infty} \psi(u,v)-  \int_{v_{\gamma}(u)}^{v} e^{iq_0 \int_{u_{\gamma}(v')}^u [(A_u^{RN})^{CH}-A_u^{CH}](u') \d u'}	\partial_v \psil(u,v') \d v'= f(u).
		\end{equation} 
	\end{prop}
	\begin{proof}
		
		We start with estimates on the potentials: by \eqref{eq:maxwellAu} and  \eqref{QLB} we have for $\eta=0.01$: $$ |\partial_v (A_u-A_u^{RN})| \leq  |\partial_v A_u|+  |\partial_v A_u^{RN}| \ls \Omega^{2-\eta}+ \Omega^{2-2\eta}_{RN} , $$ which we can integrate from the curve $\gamma$, using \eqref{Omega.gamma.EB} and \eqref{partialvOmegaLB} using   \cite[Lemma 4.1]{Moi} as before, we obtain, using also   \cref{EB.diff.proposition}, the following bound: \begin{equation} \label{Au.diff.LB}		|A_u-A_u^{RN}| \ls |u|^{1-2s}.		\end{equation}
		Moreover, recall that we proved in   \cref{Dupsi.prop} that $A_u(u,v)$ and $A_u^{RN}(u,v)$ extend to $\CH$ as bounded functions $(A_u)^{CH}(u)$ and $(A_u^{RN})^{CH}(u)$, respectively. Integrating \eqref{eq:maxwellAu} towards the past from the Cauchy horizon $\CH$ we also obtain the following estimates for all $(u,v) \in \LB$: \begin{equation} \label{Au.diff.2LB}		|A_u(u,v)-A_u^{CH}(u)|+ |A_u^{RN}(u,v)-(A_u^{RN})^{CH}(u)| \ls \Omega^2(u,v)\ls v^{-2s},	\end{equation} 
		\begin{equation} \label{Au.diff.3LB}		 \int_{u_{\gamma}(v)}^u |A_u(u',v)-A_u^{CH}(u')|\d u'+ \int_{u_{\gamma}(v)}^u |A_u^{RN}(u',v)-(A_u^{RN})^{CH}(u')|\d u' \ls v^{-2s}.	\end{equation}
		
		To obtain \eqref{diff8.LB}, note the following identity obtained using \eqref{eq:maxwellAu} with \eqref{gaugeAponctuelle} (note that $A_u(u,v_0)=A_u^{RN}(u,v_0)$): \begin{equation} \label{potential.ID}
		A_u^{CH}(u)-(A_u^{RN})^{CH}(u)= \int_{v_0}^{+\infty}- \frac{\Omega^2(u,v'
			) Q(u,v')}{r^2(u,v')}+ \frac{\Omega^2_{RN}(u,v'
			) e}{r^2_{RN}(u,v')} \d v'.
		\end{equation} We now commute \eqref{eq:maxwellAu} with $\partial_u$ to estimate $\frac{d}{du}(	A_u^{CH}(u)-(A_u^{RN})^{CH}(u))$ and we obtain an formula analogous to \eqref{potential.ID}. Using the fact that $\partial_u \log(\Omega^2) \Omega^{0.1}$ is bounded (by \cref{prop.LB.estimates.Moi}) to estimate the parts of the integral lying in $\LB$, and we obtain an estimate only involving the regions strictly to the past of $\LB$: \begin{equation} 
		|\frac{\d}{\d u}(	A_u^{CH}(u)-(A_u^{RN})^{CH}(u))| \leq \bigl| \int_{v_0}^{v_{\gamma}(u)} \partial_u\left(- \frac{\Omega^2(u,v'
			) Q(u,v')}{r^2(u,v')}+ \frac{\Omega^2_{RN}(u,v'
			) e}{r^2_{RN}(u,v')} \right)\d v' \bigr| + |u|^{-2s}.
		\end{equation} 
		Therefore, it is sufficient to control the above integral in $\Rs\cup \NN\cup \EB$. Note that the difference $\Omega^2-\Omega^2_{RN}$, $\partial_u \Omega^2-\partial_u \Omega^2_{RN}$, $Q-e$, $\nu - \nu_{RN}$ and $r-r_{RN}$ have been controlled with $|u|^{1-2s}$ weights in Propositions \ref{RS.diff.prop}, \ref{N.diff.proposition} and \ref{EB.diff.proposition}: this gives \eqref{diff8.LB}.

		Now we turn to the $\phi$ estimates. We write \eqref{Field5}  for $u_0=u_{\gamma}(v)$ and using the estimates from   \cref{prop.LB.estimates.Moi} (notably \eqref{QLB} and \eqref{phiLB} with $\eta=0.1$) we obtain:
		
		\begin{align*} 
		&	|\partial_u ( e^{iq_0 \int_{u_{\gamma}(v)}^u A_u(u',v) \d u' }\partial_v \psi- e^{iq_0 \int_{u_{\gamma}(v)}^u A_u^{RN}(u',v) \d u' } \partial_v \psil)|  \\  & \ls   |\partial_u ( e^{iq_0 \int_{u_{\gamma}(v)}^u A_u(u',v) \d u' }\partial_v \psi)| +|\partial_u ( e^{iq_0 \int_{u_{\gamma}(v)}^u A_u^{RN}(u,v) \d u' } \partial_v \psil)|  \\&  \ls\  |u|^{-2s} \cdot v^{1-3s} + (\Omega^{1.9}+ \Omega_{RN}^{1.9} )\cdot v^{-s}.
		\end{align*} Integrating in $u$ and using \eqref{Omega.gamma.EB} with the usual integration rules (i.e.\  \cite[Lemma 4.1]{Moi}) we obtain  	\begin{align} 	\label{prelim.1.LB} 
		& | e^{iq_0 \int_{u_{\gamma}(v)}^u A_u(u',v) \d u' } \partial_v \psi(u,v)- e^{iq_0 \int_{u_{\gamma}(v)}^u A_u^{RN}(u',v) \d u' } \partial_v \psil(u,v)|  \\ \ls\ & 	| \partial_v \psi(u_{\gamma}(v),v)- \partial_v \psil(u_{\gamma}(v),v)|+ |u|^{1-2s} \cdot v^{1-3s}+ v^{-2.8s} \ls v^{1-3s}, \nonumber
		\end{align}	where we also used \eqref{diff3.EB}.	Then by \eqref{prelim.1.LB}, \eqref{Au.diff.3LB}, we obtain:
		\begin{align} 
		\nonumber
		&	| \partial_v \psi(u,v)-  e^{iq_0 \int_{u_{\gamma}(v)}^u [(A_u^{RN})^{CH}(u')-A_u^{CH}(u')] \d u' } \partial_v \psil(u,v)|  \\ \leq\ &	| e^{iq_0 \int_{u_{\gamma}(v)}^u A_u(u',v) \d u' }\partial_v \psi(u,v)- e^{iq_0 \int_{u_{\gamma}(v)}^u A_u^{RN}(u',v) \d u' } \partial_v \psil(u,v)|\nonumber \\	+\ & |e^{iq_0 \int_{u_{\gamma}(v)}^u \left[A_u^{RN}(u',v)-(A_u^{RN})^{CH}(u')-A_u(u',v)+A_u^{CH}(u')\right] \d u' } -1| \cdot |\partial_v \psil|(u,v) \nonumber \\  \ls\ & v^{1-3s}+  |e^{iq_0 \int_{u_{\gamma}(v)}^u \left[A_u^{RN}(u',v)-(A_u^{RN})^{CH}(u')-A_u(u',v)+A_u^{CH}(u')\right] \d u' } -1| \cdot |\partial_v \psil|(u,v) \nonumber \\
		\ls\  & v^{1-3s}+ v^{-2s}\cdot v^{-s} \ls   v^{1-3s}, \label{prelim.12.LB}
		\end{align} 
		where in the first line we multiplied by the phase $e^{iq_0 \int_{u_{\gamma}(v)}^u A_u(u',v) \d u' }$ inside the absolute value and we used \eqref{phiVLB} (applied to $\phiNl$) in the last line. This implies \eqref{diff4.LB} (the first inequality is obtained by the reverse triangular inequality). Integrating in $v$ from $\gamma$ then gives \eqref{diff6.LB} and \eqref{ext.difference}, using also \eqref{phiVTransition} to control the boundary term $|\psi(u,v_{\gamma}(u))|\ls |u|^{-s} \ls |u|^{2-3s}$ (recall that $s\leq 1$).

		For \eqref{diff7.LB} we estimate \eqref{Field4} using the estimate of   \cref{prop.LB.estimates.Moi} (naively, without taking advantage of a difference structure) and $A_v=A_v^{RN}=0$ we get $$ | \partial_v ( D_u \psi-  D_u^{RN} \psi_{\mathcal{L}})| \ls |u|^{-2s} \cdot v^{1-3s}+ (\Omega^{1.9}+ \Omega^{1.9}_{RN} )\cdot v^{-s} .$$ Integrating in $v$, using the bounds of   \cref{EB.diff.proposition} and \eqref{Au.diff.LB} (to control the difference on $\gamma$, similarly to what was done earlier in the proof) allows us to prove \eqref{diff7.LB} thus concluding the proof.	\end{proof}
	
	\subsection{Combining the linear and the nonlinear estimates} \label{combining.section}
	
	In this section, we combine the nonlinear difference estimates of  \cref{difference.estimates.section} with the linear estimates on a fixed Reissner--Nordstr\"{o}m background obtained in  \cref{linearsection}. This allows us to conclude the proof of the boundedness of $\phi$ if $\phiH\in \OO$ and if $q_0=0$, blow up if  $\phiH\notin \OO$.

	\subsubsection{Boundedness and extendibility of the matter fields for oscillating data and proof of \cref{main.theorem}} \label{boundedness.combining.section}
	
	\begin{prop} \label{boundedness.prop}
		Assume the following gauge invariant condition: there exists $u_0 \leq u_s$ such that \begin{equation} \label{phil.estimate.bounded.end0}
		\lim_{v \rightarrow  +\infty} \int_{v_0}^{v} e^{iq_0\sbr(v')} e^{iq_0\int_{v_0}^{v'} A_v(u_0,v'') \d v''}	D_v \psil(u_0,v') \d v'
		\end{equation}exists and is finite for all $\sbr$ satisfying \eqref{sigma_err1}, \eqref{sigma_err2}, then $\phi$ in the gauge \eqref{GaugeAv}, \eqref{gaugeAponctuelle} admits a continuous extension to $\CH$. Moreover the gauge-independent quantities $|\phi|$ and the metric $g$ also admit a continuous extension to $\CH$ and the extension of $g$ can be chosen to be $C^0$-admissible as in \cref{C0admissible}.
		
		If we additionally assume the following gauge-invariant condition: for all $D_{br}>0$, there exists $\eta_0(D_{br})>0$ such that  for all $\sbr$ satisfying \eqref{sigma_err1}, \eqref{sigma_err2} and for all $(u,v)\in \LB$: \begin{equation} \label{phil.estimate.bounded.end}
		\left| \int_{v_{\gamma}(u)}^{v} e^{iq_0\sbr(v')} e^{iq_0\int_{v_0}^{v'} A_v(u_0,v'') \d v''}	D_v \psil(u_0,v') \d v'\right| \ls D' \cdot |u|^{s-1- \eta_0}.
		\end{equation}
		Then $Q$ and $\phi$ are bounded and the following estimates are true for all $(u,v) \in \LB$:  \begin{equation}\label{phi.boundedness.final}
		|\phi|(u,v) \ls |u|^{-1+s-\eta_0},
		\end{equation} 
		\begin{equation} \label{Q.boundedness.final}
		|Q-e|(u,v) \ls |u|^{-\eta_0},
		\end{equation} where the implicit constants are allowed to depend on $\eta_0>0$. Moreover, $Q$ extends to a continuous function $Q_{CH}(u)$ on $\CH$.

	\end{prop} \begin{proof}
		Applying the assumption to $\sbr(v)= \int_{u_{\gamma}(v)}^{u_0} [(A_u^{RN})^{CH}-A_u^{CH}](u') \d u'$ (which satisfies \eqref{sigma_err1} and \eqref{sigma_err2} by   \cref{LB.prop}) we get by   \cref{LB.prop} that for $\psi$ in the gauge \eqref{GaugeAv} (note that $A_v \equiv 0$): $$ \lim_{v \rightarrow +\infty} \psi(u_0,v):= \psi_{CH}(u_0)$$ exists and is finite.	 Recall also from   \cref{Dupsi.prop} that $D_u \psi$ and $A_u$ admit (in the gauge \eqref{GaugeAv}, \eqref{gaugeAponctuelle}) a bounded extension to $\CH$ which we denoted respectively $(D_u \psi)_{CH}$ and $(A_{u})^{CH}$.  
		Recall also that one can write for any $u_0 \in \RR$ the following identity $$  \partial_u( e^{iq_0 \int_{u_0}^u A_u(u',v) \d u'} \psi(u,v)) =  e^{iq_0 \int_{u_0}^u A_u(u',v) \d u'} D_u \psi(u,v),$$   which upon integration gives  $$   \psi(u,v) = e^{-iq_0 \int_{u_0}^u A_u(u',v) \d u'}  \psi(u_0,v)+  e^{-iq_0 \int_{u_0}^u A_u(u',v) \d u'} \int_{u_0}^ue^{iq_0 \int_{u_0}^{u'} A_u(u'',v) \d u''} D_u \psi(u',v) \d u'.$$
		
		Now note by   \cref{LB.prop}, $A_u \in L^{\infty}_{loc}$ therefore by dominated convergence, the function $(u,v) \mapsto \int_{u_0}^u A_u(u',v) \d u'$ extends continuously to  $\int_{u_0}^u (A_u)^{CH}(u') \d u'$ at $\CH$. Since $D_u \psi \in L^{\infty}_{loc}$ as well (by \eqref{Dupsiestimate}), an other use of dominated convergence, together with the existence of the limit  $\lim_{v \rightarrow +\infty}\psi(u_0,v)$ shows that $\psi(u,v)$ admits a continuous extension to $\CH$ denoted $\psi_{\CH}(u)$. By   \cref{CH.stab.thm}, $r$ admits a continuous extension $r_{\CH}(u)$ to $\CH$ which is bounded away from zero. Therefore, $\phi(u,v)$ also admits a continuous extension to $\CH$ denoted $\phi_{\CH}(u)$. The continuous extendibility of the metric $g$ (and the $C^0$-admissible character of the extension) follows immediately as a consequence of   \cref{metric.ext.cor}

		Now we make the additional assumption \eqref{phil.estimate.bounded.end}. We define  $(\sbr)_u(v):= \int_{u_{\gamma}(v)}^{u} [(A_u^{RN})^{CH}-A_u^{CH}](u') \d u'$ for each $u \leq u_s$. It follows from \eqref{diff3.LB} and \eqref{diff8.LB} that  $(\sbr)_u$ satisfies \eqref{sigma_err1}, \eqref{sigma_err2} with a constant $D_{br}(M,e,q_0,m^2,s,D_1,D_2)>0$ \emph{that is independent of $u$}. In view of this, \eqref{phi.boundedness.final} follows from \eqref{phil.estimate.bounded.end} combined with \eqref{diff6.LB} and the fact that $s>\frac{3}{4}$. Now we plug \eqref{phi.boundedness.final}, the boundedness of $r$, and \eqref{Dupsiestimate} into \eqref{chargeUEinstein2} to obtain the estimate in $\LB$: $$ |\partial_u Q| \ls |u|^{-1-
			\eta_0}.$$ Integrating this estimate from $\gamma$ we obtain \eqref{Q.boundedness.final}, in view of the estimate on $\gamma$ from   \cref{prop.EB.estimates.Moi}.

		For the continuous extendibility of $Q$, we start integrating \eqref{chargeUEinstein2} to get for all $(u,v) \in \LB$  $$ Q(u,v)= Q(u_{\gamma}(v),v)+ q_0\int_{u_{\gamma}(v)}^u \Im (\bar{\psi} D_u \psi)(u',v) \d u'.$$ Note that the function $u\rightarrow \Im (\bar{\psi} D_u \psi)(u,v)$ is  dominated by the integrable function $|u|^{-1-\eta_0}$ therefore by the dominated convergence theorem, $\int_{u_{\gamma}(v)}^u \Im (\bar{\psi} D_u \psi)(u',v) \d u'$ extends continuously to the function $\int_{-\infty}^u \Im (\overline{\psi_{CH}} (D_u \psi)_{CH})(u') \d u'.$ Therefore,  $Q$ admits a continuous extension to $\CH$, which concludes the proof.
		
	\end{proof}
	\begin{cor} \label{boundedness.cor}
		\begin{enumerate}
			\item 	Assume that $\phiH \in \OO$. Then $\phi$ is uniformly bounded on $\LB$ and thus, \eqref{phi.bounded.main.thm}	 holds true.
			
			\item Assume additionally that $\phiH \in \OOp$, then $|\phi|$ and $g$ are continuously extendible, and the extension of $g$ can be chosen to be $C^0$-admissible.
			\item  Assume additionally that $\phiH \in \OOpp$, then \eqref{phi.boundedness.final} and \eqref{Q.boundedness.final} are true for all $(u,v) \in \LB$ and moreover $Q$  admit a continuous extension to $\CH$.
		\end{enumerate}
	\end{cor}
	
	\begin{proof}The first statement follows from \eqref{diff6.LB} of \cref{LB.prop} and \cref{linearcor}.
		The others ar direct applications of   \cref{boundedness.prop} and   \cref{linearcor} (using that \eqref{phil.estimate.bounded.end0} and \eqref{phil.estimate.bounded.end} are gauge invariant conditions).
	\end{proof}
	
	In particular,  \cref{boundedness.cor} shows \cref{main.theorem}.

	\subsubsection{Blow-up of the scalar field for $\phiH \notin \OO$ (non-oscillating data) if $q_0=0$ and proof of \cref{main.theorem2}} \label{blow.up.section}
	
	\begin{lem} \label{stupid.lemma} Assume that there exists $ u_0\leq u_{_s}$ such that $$ \limsup_{v \rightarrow+\infty} |\phi|(u_0,v) = +\infty.$$
		
		Then for all $u  \leq u_s$ we have  $$ \limsup_{v \rightarrow+\infty} |\phi|(u,v) =  \limsup_{v \rightarrow+\infty} |\psi|(u,v)=+\infty.$$
		
		Moreover we have the following bounds: for all $u\leq u_s$, there exists $f(u)>0$ for all $v>v_{\gamma}(u)$:  \begin{equation} \label{asymptotics.phi.stupid}\begin{split}
		\bigl||\phi|(u,v)- \frac{r(u_0,v)}{r(u,v)}  |\phi|(u_0,v) \bigr| \leq f(u),\\ \liminf_{v \rightarrow +\infty} \frac{|\phi|(u,v)}{|\phi|(u_0,v)}= \frac{r_{CH}(u_0)}{r_{CH}(u)}>0.
		\end{split}
		\end{equation}

	\end{lem}
	\begin{proof}
		This is an immediate consequence of the integrating of \eqref{Dupsiestimate} and the continuous extendibility of $r$ to a function which is bounded away from zero (by definition of $\CH$). 
	\end{proof}
	We will not use \eqref{asymptotics.phi.stupid} in the present work, but it is an important estimate for our companion paper \cite{MoiChristoph2}.
	
	\begin{cor}\label{cor:proofofthmaii}
		Assume that $q_0=0$	and that $\phiH \in \Sl - \OO$. Then for all $u \leq u_{s}$ we have the following blow-up $$ \limsup_{v \rightarrow+\infty} |\phi|(u,v) =  \limsup_{v \rightarrow+\infty} |\psi|(u,v)=+\infty,$$ and moreover the asymptotics \eqref{asymptotics.phi.stupid} are satisfied. 
	\end{cor}
	
	\begin{proof}
		The result follows from a combined application of   \cref{linearcor} (using that $\phil $ and $ \phiNl$ relate by a gauge transformation, hence $|\phil| = |\phiNl| $), \eqref{ext.difference} in \cref{LB.prop} and   \cref{stupid.lemma}.
	\end{proof}	
	In particular, \cref{cor:proofofthmaii} shows \cref{main.theorem2}. 
	
	\subsubsection{Proof of   \cref{corollary.conj}} \label{computation.proof.section}
	
	Before turning to the proof of \cref{corollary.conj}, we prove  the following.
	\begin{lem} \label{lemma.example}
		Let  $s>\frac{3}{4}$ and $\omega_1 \in \RR-\{\omega_{res}\}$ and $\phiH$ be given by 
	 \begin{align} \label{eq:formforphih}
		\phiH (v) =    e^{-i\omega_1 v+ \omega_{err}(v)}   v^{-s }  + \phi_{err}
		\end{align} 
	for any $\phi_{err}\in C^1([v_0,+\infty),\RR)$ satisfying  \eqref{decay.s} with  $s>1$ and any $ \omega_{err} \in C^2([v_0,+\infty),\RR)$ such that $\omega_{err}'(v)\to 0 $ as $v \rightarrow +\infty$ and such that  $ |\omega_{err}''|(v) \leq D \cdot v^{-2+2s-\eta_0}$ for $v\geq v_0$ and some constants $D>0$ and $\eta_0>0$.  Then $\phiH \in \OOpp$, where we assume without loss of generality   that $\phiH \in \Sl$ (by choosing $D_1>0$ possibly larger). 
	\end{lem}
	\begin{proof}
 Since $\phiH \in \Sl$ (by possibly choosing $D_1>0$ larger) it suffices to check \eqref{eq:strong.oscillation.condition} independently for $ e^{-i\omega_1 v+ \omega_{err}(v)}   v^{-s } $ and $\phi_{err}$. 	  First note that $\phi_{err} $ satisfies  \eqref{eq:strong.oscillation.condition} since it satisfies \eqref{decay.s} with  $s>1$. 

	For $e^{-i\omega_1 v+ \omega_{err}(v)}   v^{-s }$ we can assume with no loss of generality that  $\frac{3}{4}<s\leq 1$ (since the case $s>1$ follows immediately from integrability).   It suffices to prove that there exists $\eta>0$, $E>0$ such that for all large enough $\tilde v,v$ with $\tilde v < v$
	 \begin{equation}\label{decay.extra}
		 \bigl| \int_{\tilde v}^{v} e^{ i\omega_{res}v'-i \omega_1 v'+ i\sbr(v') + i \omega_{err}(v')} (v')^{-s} \d v'\bigr| \leq E \cdot \tilde v^{-1+s-\eta}
		\end{equation} for all $\sbr$ satisfying \eqref{sigma_err1} and \eqref{sigma_err2}. For conciseness, we will introduce the notation $ \omega= \omega_{res}-\omega_1 \neq 0$.  	We make use of integration by parts: \begin{align*}  \int_{\tilde v}^v e^{ i  \omega v'+ i\sbr(v') + i \omega_{err}(v')} (v')^{-s} \d v'  = &   -i\int_{\tilde v}^{v} \frac{\d}{\d v'}(e^{ i  \omega v'+ i\sbr(v') + i \omega_{err}(v')}) \frac{(v')^{-s}}{ \omega  + \sbr'(v') + \omega_{err}'(v') } \d v'\\ = &  - i\frac{ v^{-s}  e^{ i  \omega v+ i\sbr(v) + i \omega_{err}(v)}}{ \omega + \sbr'(v) + \omega_{err}'(v) } + i \frac{  \tilde v^{-s}  e^{ i  \omega \tilde v + i\sbr(\tilde v) + i \omega_{err}(\tilde v)}}{ \omega + \sbr'(\tilde v) + \omega_{err}'(\tilde v) }\\
		& -i  s\int_{\tilde v}^v e^{ i  \omega v'+ i\sbr(v') + i \omega_{err}(v')} \frac{(v')^{-s-1}}{ \omega + \sbr'(v') + \omega_{err}'(v') } \d v' \\  & -i\int_{\tilde v}^v e^{ i  \omega v'+ i\sbr(v') + i \omega_{err}(v')} \frac{(v')^{-s} \cdot (\sbr''(v') + \omega_{err}''(v'))}{( \omega + \sbr'(v') + \omega_{err}'(v') )^2} \d v'.
		\end{align*}   Note that, using \eqref{sigma_err2} and decay the assumption on $\omega_{err}'$, that $ \omega + \sigma_{br}'(v')+ \omega_{err}'(v')$ is bounded away from zero for $\tilde v$ large enough (since $\omega \neq 0$).  The first two terms obviously obey \eqref{decay.extra} since  $s>\frac{1}{2}$. Similarly, the third term can be integrated to show $$ \bigl| \int_{\tilde v}^v e^{ i  \omega v'+ i\sbr(v') + i \omega_{err}(v')} \frac{(v')^{-s-1}}{ \omega + \sbr'(v') + \omega_{err}'(v') } \d v' \bigr| \lesssim \tilde v^{-s}.$$ For the last term,  we write using $|\omega_{err}''(v)| \ls v^{-2+2s-\eta_0}$ and \eqref{sigma_err2} 
		\begin{align*}
	&	\bigl| \int_{\tilde v}^v e^{ i  \omega v'+ i\sbr(v') + i \omega_{err}(v')} \frac{(v')^{-s} \cdot (\sbr''(v') + \omega_{err}''(v'))}{( \omega + \sbr'(v') + \omega_{err}'(v') )^2} \d v' \bigr| \ls  \int_{\tilde v}^v  (v')^{-s} \cdot (v^{-2+2s-\eta_0}+ v^{1-2s} )\d v' \\ \ls\ &   \tilde v^{-1+s-\eta_0}+\tilde v^{2-3s} \lesssim \tilde v^{-1+s-\eta_0},
		\end{align*} for some $\eta_0>0$, where to obtain this  estimate, we used the fact that $s<1+\eta_0$ for some $\eta_0>0$ and also $2-3s <  -1+s- \eta_0$ (since we assumed $s>\frac{3}{4}$).  
	\end{proof}
	\begin{prop} \label{prop.specific.profiles}
		Assume that the parameters $(M,e,q_0,m^2)$ are such that $$ |q_0 e| \neq r_-(M,e) |m|.$$ Let  $\phiH$ be given by either the profile of \eqref{conj.decay.1.eq} (if $m^2 > 0$, $q_0=0$) or \eqref{conj.decay.2.eq} (if $m^2 = 0$, $q_0\neq 0$) or \eqref{conj.decay.3.eq} (if $m^2 > 0$, $q_0\neq0$). Then  $\phiH \in \OO''$,  where we again assume without loss of generality   that $\phiH \in \Sl$ (by choosing $D_1>0$ possibly larger). 
	\end{prop}
	\begin{proof}
		If $m^2=0$, $|q_0e|<\frac{1}{2}$, then $\phiH$  satisfies \eqref{polynomialdecay} for $s>1$ and thus $\phiH \in  \OO''$. Otherwise, we have three different cases \begin{enumerate}
			\item $q_0=0$, $m^2\neq 0$: It suffices to prove that $e^{\pm i ( m v + \omega_{err}(v))} \cdot v^{-\frac{5}{6}} \in \OOpp$, where $\omega_{err}(v)= -\frac{3m}{2} (2\pi M)^{\frac{2}{3}} \cdot v^{\frac{1}{3}} + \omega(m \cdot M) $. Note that $\omega_{err}'(v)\to 0 $ as $v \rightarrow +\infty$ and such that  $ |\omega_{err}''|(v) \ls  v^{-\frac{5}{3}} \ls v^{-2+2 \cdot \frac{5}{6}-\eta_0}$ for any $0 < \eta_0 < \frac{4}{3}$, therefore by   \cref{lemma.example}, $e^{\pm i ( m v + \omega_{err}(v))} \cdot v^{-\frac{5}{6}} \in \OOpp$.
			
			\item  $|q_0e|\geq \frac{1}{2}$, $m^2= 0$. Then $\delta= \pm i \sqrt{ 4(q_0e)^2-1}$ and $\phiH$ is of the form \eqref{eq:formforphih} with $\omega_1 =-\frac{q_0 e}{r_+} \neq \omer$, $\omega_{err}= - (\sqrt{ 4(q_0e)^2-1}) \log(v)$ and $s=1$. Indeed we have $\omega_{err}'(v)=o(1)$ and $|\omega_{err}''|(v) \lesssim v^{-2} \lesssim v^{-2+2s -\eta_0}$ for  $\eta_0>0$ since $2s-2 = 0$. Therefore, $\phiH \in \OOpp$ by   \cref{lemma.example}.
			\item  $q_0\neq 0$, $m^2\neq 0$: as in the  case $q_0=0$, $m^2\neq 0$, $\phiH$ is a linear combination of  two profiles of the form \eqref{eq:formforphih}  with $\omega_1= \pm m -\frac{q_0e}{r_+}$. Since the parameters $(M,e,q_0,m^2)$ do not satisfy $ |q_0 e| \neq r_-(M,e) |m|$, we know that $\omega_1 \neq \omer$. The rest of the argument follows as above.
		\end{enumerate}
	\end{proof}
	
	\begin{cor}\label{cor:proofofthmb}
		Assume that the parameters $(M,e,q_0,m^2)$ are such that $$ |q_0 e| \neq r_-(M,e) |m|.$$ Let  $\phiH$  by either the profile of \eqref{conj.decay.1.eq} (if $m^2 > 0$, $q_0=0$) or \eqref{conj.decay.2.eq} (if $m^2 = 0$, $q_0\neq 0$) or \eqref{conj.decay.3.eq} (if $m^2 > 0$, $q_0\neq0$). Then,   \eqref{phi.boundedness.final} and \eqref{Q.boundedness.final} are true for all $(u,v) \in \LB$. Moreover, $|\phi|$, $Q$ and the metric $g$ admit a continuous extension to $\CH$ and the extension of $g$ can be chosen to be $C^0$-admissible.
	\end{cor}
	\begin{proof}
		This is an immediate application of   \cref{prop.specific.profiles} and   \cref{boundedness.cor} (using that $|\phil| = |\phiNl|$ since $\phil$ and $\phiNl$ only differ from  gauge transformation).
	\end{proof}
	In particular, 	\cref{cor:proofofthmb} shows \cref{corollary.conj}.
	\subsubsection{$\dot{W}^{1,1}_{loc}$ blow-up of the scalar field on outgoing cones: proof of \cref{W11.main.thm}} \label{W11.blowup.section}
	
	\begin{prop} \label{W11.OO.prop}
		
		Assume that for all $u\leq u_s$ we have the following  blow up: \begin{align} \label{Dv.lin.W11.blowup.est}
		\int_{v_0}^{+\infty} |D_v^{RN} \phiNl|(u,v') \d v'=+\infty.
		\end{align}
		
		Then  for all $u\leq u_s$: \begin{align} \label{DvW11.blowup.est}
		\int_{v_0}^{+\infty} |D_v \phi|(u,v') \d v'= 	\int_{v_0}^{+\infty} |D_v \psi|(u,v') \d v'=+\infty.
		\end{align}
		Conversely, \eqref{DvW11.blowup.est}   implies \eqref{Dv.lin.W11.blowup.est}.
	\end{prop} \begin{proof}
		
		Note that $D_v^{RN} \psil= r   D_v^{RN} \phiNl -\frac{\Omega^2_{RN}}{2} \phiNl$.  Since $r$ is lower bounded on $\LB$ and in view of \eqref{phiLB} (which also applies to $\phiNl$), for all $u \leq u_s$:  $$  \int_{v_0}^{+\infty}  |D_v^{RN} \psil|(u,v) \d v =+\infty.$$
		
		Therefore, integrating \eqref{diff4.LB} (since $s>\frac{3}{4} > \frac{2}{3}$) we also obtain for all $u \leq u_s$:  $$  \int_{v_0}^{+\infty}  |D_v\psi|(u,v) \d v =+\infty.$$
		
		Since  $D_v \psi= r   D_v \phi + \lambda \phi$ and by \eqref{phiblowupVLB}, \eqref{lambdaLB} $$ | \lambda \phi| \ls v^{1-3s}$$ is integrable, therefore for all $u \leq u_s$:  $$  \int_{v_0}^{+\infty}  |D_v\phi|(u,v)\d v =+\infty.$$
		The above also shows that \eqref{DvW11.blowup.est}   implies \eqref{Dv.lin.W11.blowup.est}.
	\end{proof}
	\begin{cor} \label{W11.OO.cor}
		Assume $\phiH \in \Sl-H$ (defined in the proof of Corollary \ref{linear.cor.W11}). Then \eqref{DvW11.blowup.est} holds true.

		In the particular case $|q_0 e| \leq \epsilon(M,e,m^2)$ (in particular if $q_0=0$),  where  $\epsilon>0$ is defined in the proof of Corollary \ref{linear.cor.W11}, then for all  $\phiH \in \Sl-L^1$, \eqref{DvW11.blowup.est} is satisfied.
	\end{cor}
	\begin{proof}
		This follows from \cref{linear.cor.W11} (using that $\phil $ are $ \phiNl$ relate by a gauge transformation, hence $|\phil| = |\phiNl| $ and $ |D_v \phil| = |D_v \phiNl| $) and  \cref{W11.OO.prop}. 
	\end{proof}
	
	\cref{W11.OO.cor} thus concludes the proof of {Part 1.} of \cref{W11.main.thm}. Now we turn to the proof of  {Part 2.} of \cref{W11.main.thm}.

	\begin{cor}\label{lem:w11blowupofconjrate}
		Let  $\phiH$ be given by either the profile of \eqref{conj.decay.1.eq} (if $m^2 > 0$, $q_0=0$) or \eqref{conj.decay.2.eq} (if $m^2 = 0$, $q_0\neq 0$) or \eqref{conj.decay.3.eq} (if $m^2 > 0$, $q_0\neq0$). Assume the   condition $\mathcal Z_{\mathfrak t} \cap \varTheta  = \emptyset$.

		Then, there exists a $\delta (M,e,q_0,m^2)>0$ sufficiently small such that $P_\delta \phiH \in L^1(\mathbb R)$.

		Moreover, the condition $\mathcal Z_{\mathfrak t}(M,e,q_0,m^2) \cap \varTheta (M,e,q_0,m^2)  = \emptyset$ is generic in the sense that for given $m^2\geq 0$, $q_0 \in \mathbb R$ with $m^2\neq q_0^2$, the set of parameters $(M,e) $ satisfying the conditions is 
 		the  zero set of a nontrivial real analytic function on $\{0<|e|<M\}$. 
		In particular, in view of Part 1.\ of \cref{W11.main.thm}, we obtain Part 2.\ of \cref{W11.main.thm}.
		\begin{proof}
			We start with the second claim. 
	 Fix $m^2\geq0$, $q_0\in \mathbb R$ with $q_0^2 \neq m^2$. We define $f_{\pm, m^2,q_0}(M,e):= \mathfrak t( \pm m- \frac{q_0 e}{r_+} , M,e,q_0,m^2)$. By analyticity of   $\mathfrak t$ (note that $\mathfrak t$ is the Wronskian of solutions to an o.d.e.\ with analytic coefficients depending analytically on $(\omega,M,e)$), we have that both $f_{\pm,m^2,q_0} \colon \{ (M,e) \in \mathbb R^2, 0<|e|<M\} \to \mathbb R$ are analytic. It suffices to show that both $f_\pm$ are non-trivial. From the  o.d.e.\ energy identity, $|\mathfrak t|^2 = |\mathfrak r|^2 + \omega (\omega - \omer) \geq \omega (\omega - \omer)$   we conclude $|f_\pm|^2 \geq (\pm m - \frac{q_0 e}{r_+}) ( \pm m - \frac{q_0 e}{r_-}) \to (\pm m - eq_0/|e|)^2 >0 $ as  $|e|\to M$. We used here that  $m^2\neq q_0^2$.

			Now, fix $0 < \delta < \operatorname{dist}(\mathcal Z_{\mathfrak t} , \varTheta ) $. 
			By Plancherel's theorem  and the Cauchy--Schwarz inequality, it suffices to show that $\chi_{\delta}(\omega) \mathcal{F}(\phiH)$ is in $H^{\frac 12 +\tau}$ for some $\tau>0$ (recalling the definition of $\chi_{\delta}(\omega)$ from \cref{sec:w11blowup}). Further, since $\chi_{\delta}$ is smooth and has compact support ($\mathcal Z_{\mathfrak t}^\delta \subset [-|\omer|-\delta, |\omer|+\delta]$), and $\mathcal{F}(\phiH) \in L^2$ , it suffices (e.g.\ by the Kato--Ponce inequality) to show that   $\chi_{\delta}(\omega) \langle\partial_\omega\rangle^{\frac 12 + \tau} \mathcal{F}(\phiH)$ is in $L^2$. Thus, we need to show that $\mathcal{F}(\langle v \rangle^{\frac 12 + \tau} \phiH) \in L^2(\mathcal Z_{\mathfrak t}^{\delta})$ for some $\tau >0$. We now fix  $0< \tau < s- \frac 12 $.   A direct adaption of the proofs of \cref{lemma.example} and \cref{prop.specific.profiles} then shows   $ \mathcal{F}(\langle v \rangle^{\frac 12 + \tau} \phiH) \in L^\infty(\mathcal Z_{\mathfrak t}^{\delta})$ from which the claim follows.
		\end{proof}
	\end{cor}
	\printbibliography[heading=bibintoc]
\end{document}

%% file: spacetime_constructed.pdf_tex
\begingroup%
  \makeatletter%
  \providecommand\color[2][]{%
    \errmessage{(Inkscape) Color is used for the text in Inkscape, but the package 'color.sty' is not loaded}%
    \renewcommand\color[2][]{}%
  }%
  \providecommand\transparent[1]{%
    \errmessage{(Inkscape) Transparency is used (non-zero) for the text in Inkscape, but the package 'transparent.sty' is not loaded}%
    \renewcommand\transparent[1]{}%
  }%
  \providecommand\rotatebox[2]{#2}%
  \newcommand*\fsize{\dimexpr\f@size pt\relax}%
  \newcommand*\lineheight[1]{\fontsize{\fsize}{#1\fsize}\selectfont}%
  \ifx\svgwidth\undefined%
    \setlength{\unitlength}{391.54989004bp}%
    \ifx\svgscale\undefined%
      \relax%
    \else%
      \setlength{\unitlength}{\unitlength * \real{\svgscale}}%
    \fi%
  \else%
    \setlength{\unitlength}{\svgwidth}%
  \fi%
  \global\let\svgwidth\undefined%
  \global\let\svgscale\undefined%
  \makeatother%
  \begin{picture}(1,0.39269956)%
    \lineheight{1}%
    \setlength\tabcolsep{0pt}%
    \put(0,0){\includegraphics[width=\unitlength,page=1]{spacetime_constructed.pdf}}%
    \put(0.24743363,0.37840655){\color[rgb]{0,0,0}\rotatebox{-45}{\makebox(0,0)[lt]{\lineheight{0}\smash{\begin{tabular}[t]{l}$\mathcal{CH}_{i^+}$\end{tabular}}}}}%
    \put(0,0){\includegraphics[width=\unitlength,page=2]{spacetime_constructed.pdf}}%
    \put(0.20155636,0.13787503){\color[rgb]{0,0,0}\rotatebox{45}{\makebox(0,0)[lt]{\lineheight{0}\smash{\begin{tabular}[t]{l}$\mathcal{H}^+$\end{tabular}}}}}%
    \put(0,0){\includegraphics[width=\unitlength,page=3]{spacetime_constructed.pdf}}%
    \put(0.00125395,0.16203068){\color[rgb]{0,0,0}\rotatebox{-45}{\makebox(0,0)[lt]{\lineheight{0}\smash{\begin{tabular}[t]{l}$\underline{C}_{in}$\end{tabular}}}}}%
    \put(0,0){\includegraphics[width=\unitlength,page=4]{spacetime_constructed.pdf}}%
    \put(0.32914913,0.28576275){\color[rgb]{0,0,0}\makebox(0,0)[lt]{\lineheight{1.25}\smash{\begin{tabular}[t]{l}$i^+$\end{tabular}}}}%
    \put(0.52565105,0.15907071){\color[rgb]{0,0,0}\makebox(0,0)[lt]{\lineheight{1.25}\smash{\begin{tabular}[t]{l}Slowly decaying data $\phiH$.\end{tabular}}}}%
    \put(0,0){\includegraphics[width=\unitlength,page=5]{spacetime_constructed.pdf}}%
    \put(0.4433996,0.29858715){\color[rgb]{0,0,0}\makebox(0,0)[lt]{\lineheight{1.25}\smash{\begin{tabular}[t]{l}Cauchy horizon $\CH$ exists with possibly $C^0$-singular metric.\end{tabular}}}}%
    \put(0,0){\includegraphics[width=\unitlength,page=6]{spacetime_constructed.pdf}}%
    \put(0.27884825,0.00595511){\color[rgb]{0,0,0}\transparent{0.50160801}\makebox(0,0)[lt]{\lineheight{1.25}\smash{\begin{tabular}[t]{l}$\Sigma$\end{tabular}}}}%
    \put(0.40622774,0.21225175){\color[rgb]{0,0,0}\makebox(0,0)[lt]{\lineheight{1.25}\smash{\begin{tabular}[t]{l}$\mathcal I^+$\end{tabular}}}}%
    \put(0,0){\includegraphics[width=\unitlength,page=7]{spacetime_constructed.pdf}}%
  \end{picture}%
\endgroup%

%% file: spacetime_constructed_oscillation_true.pdf_tex
\begingroup%
  \makeatletter%
  \providecommand\color[2][]{%
    \errmessage{(Inkscape) Color is used for the text in Inkscape, but the package 'color.sty' is not loaded}%
    \renewcommand\color[2][]{}%
  }%
  \providecommand\transparent[1]{%
    \errmessage{(Inkscape) Transparency is used (non-zero) for the text in Inkscape, but the package 'transparent.sty' is not loaded}%
    \renewcommand\transparent[1]{}%
  }%
  \providecommand\rotatebox[2]{#2}%
  \newcommand*\fsize{\dimexpr\f@size pt\relax}%
  \newcommand*\lineheight[1]{\fontsize{\fsize}{#1\fsize}\selectfont}%
  \ifx\svgwidth\undefined%
    \setlength{\unitlength}{532.95744535bp}%
    \ifx\svgscale\undefined%
      \relax%
    \else%
      \setlength{\unitlength}{\unitlength * \real{\svgscale}}%
    \fi%
  \else%
    \setlength{\unitlength}{\svgwidth}%
  \fi%
  \global\let\svgwidth\undefined%
  \global\let\svgscale\undefined%
  \makeatother%
  \begin{picture}(1,0.28976427)%
    \lineheight{1}%
    \setlength\tabcolsep{0pt}%
    \put(0,0){\includegraphics[width=\unitlength,page=1]{spacetime_constructed_oscillation_true.pdf}}%
    \put(0.15403969,0.26425653){\color[rgb]{0,0,0}\rotatebox{-45}{\makebox(0,0)[lt]{\lineheight{0}\smash{\begin{tabular}[t]{l}$\mathcal{CH}_{i^+}$\end{tabular}}}}}%
    \put(0.15765424,0.1151776){\color[rgb]{0,0,0}\rotatebox{45}{\makebox(0,0)[lt]{\lineheight{0}\smash{\begin{tabular}[t]{l}$\mathcal{H}^+$\end{tabular}}}}}%
    \put(0,0){\includegraphics[width=\unitlength,page=2]{spacetime_constructed_oscillation_true.pdf}}%
    \put(0.24181724,0.21120059){\color[rgb]{0,0,0}\makebox(0,0)[lt]{\lineheight{1.25}\smash{\begin{tabular}[t]{l}$i^+$\end{tabular}}}}%
    \put(0.36905631,0.12789578){\color[rgb]{0,0,0}\makebox(0,0)[lt]{\lineheight{1.25}\smash{\begin{tabular}[t]{l}Slowly decaying data $\phiH$ with \textbf{oscillation condition satisfied}.\end{tabular}}}}%
    \put(0,0){\includegraphics[width=\unitlength,page=3]{spacetime_constructed_oscillation_true.pdf}}%
    \put(0.20564102,0.00376812){\color[rgb]{0,0,0}\transparent{0.50160801}\makebox(0,0)[lt]{\lineheight{1.25}\smash{\begin{tabular}[t]{l}$\Sigma$\end{tabular}}}}%
    \put(0.51946762,0.18126646){\color[rgb]{0,0,0}\makebox(0,0)[lt]{\lineheight{1.25}\smash{\begin{tabular}[t]{l}$\big\Uparrow$	\\\end{tabular}}}}%
    \put(0,0){\includegraphics[width=\unitlength,page=4]{spacetime_constructed_oscillation_true.pdf}}%
    \put(0.00038538,0.12399425){\color[rgb]{0,0,0}\rotatebox{-45}{\makebox(0,0)[lt]{\lineheight{0}\smash{\begin{tabular}[t]{l}$\underline{C}_{in}$\end{tabular}}}}}%
    \put(0,0){\includegraphics[width=\unitlength,page=5]{spacetime_constructed_oscillation_true.pdf}}%
    \put(0.26377225,0.23659901){\color[rgb]{0,0,0}\makebox(0,0)[lt]{\lineheight{1.25}\smash{\begin{tabular}[t]{l}Spacetime is (highly non-uniquely) $C^0$-\textbf{extendible} across Cauchy horizon $\CH$. \end{tabular}}}}%
    \put(0,0){\includegraphics[width=\unitlength,page=6]{spacetime_constructed_oscillation_true.pdf}}%
  \end{picture}%
\endgroup%

%% file: spacetime_constructed_oscillation_false.pdf_tex
\begingroup%
  \makeatletter%
  \providecommand\color[2][]{%
    \errmessage{(Inkscape) Color is used for the text in Inkscape, but the package 'color.sty' is not loaded}%
    \renewcommand\color[2][]{}%
  }%
  \providecommand\transparent[1]{%
    \errmessage{(Inkscape) Transparency is used (non-zero) for the text in Inkscape, but the package 'transparent.sty' is not loaded}%
    \renewcommand\transparent[1]{}%
  }%
  \providecommand\rotatebox[2]{#2}%
  \newcommand*\fsize{\dimexpr\f@size pt\relax}%
  \newcommand*\lineheight[1]{\fontsize{\fsize}{#1\fsize}\selectfont}%
  \ifx\svgwidth\undefined%
    \setlength{\unitlength}{447.34151912bp}%
    \ifx\svgscale\undefined%
      \relax%
    \else%
      \setlength{\unitlength}{\unitlength * \real{\svgscale}}%
    \fi%
  \else%
    \setlength{\unitlength}{\svgwidth}%
  \fi%
  \global\let\svgwidth\undefined%
  \global\let\svgscale\undefined%
  \makeatother%
  \begin{picture}(1,0.34051421)%
    \lineheight{1}%
    \setlength\tabcolsep{0pt}%
    \put(0,0){\includegraphics[width=\unitlength,page=1]{spacetime_constructed_oscillation_false.pdf}}%
    \put(0.17512448,0.30319804){\color[rgb]{0,0,0}\rotatebox{-45}{\makebox(0,0)[lt]{\lineheight{0}\smash{\begin{tabular}[t]{l}$\mathcal{CH}_{i^+}$\end{tabular}}}}}%
    \put(0,0){\includegraphics[width=\unitlength,page=2]{spacetime_constructed_oscillation_false.pdf}}%
    \put(0.1917894,0.13900754){\color[rgb]{0,0,0}\rotatebox{45}{\makebox(0,0)[lt]{\lineheight{0}\smash{\begin{tabular}[t]{l}$ \mathcal{H}^+$\end{tabular}}}}}%
    \put(0,0){\includegraphics[width=\unitlength,page=3]{spacetime_constructed_oscillation_false.pdf}}%
    \put(0.00109756,0.13861395){\color[rgb]{0,0,0}\rotatebox{-45}{\makebox(0,0)[lt]{\lineheight{0}\smash{\begin{tabular}[t]{l}$\underline{C}_{in}$\end{tabular}}}}}%
    \put(0,0){\includegraphics[width=\unitlength,page=4]{spacetime_constructed_oscillation_false.pdf}}%
    \put(0.28809824,0.24691438){\color[rgb]{0,0,0}\makebox(0,0)[lt]{\lineheight{1.25}\smash{\begin{tabular}[t]{l}$i^+$\end{tabular}}}}%
    \put(0.41533207,0.14848925){\color[rgb]{0,0,0}\makebox(0,0)[lt]{\lineheight{1.25}\smash{\begin{tabular}[t]{l}Slowly decaying uncharged data $\phiH$ with \textbf{oscillation condition violated}.\end{tabular}}}}%
    \put(0,0){\includegraphics[width=\unitlength,page=5]{spacetime_constructed_oscillation_false.pdf}}%
    \put(0.39118824,0.25478277){\color[rgb]{0,0,0}\makebox(0,0)[lt]{\lineheight{1.25}\smash{\begin{tabular}[t]{l}Null contraction singularity at Cauchy horizon $\CH$: metric is $C^0$-\textbf{singular}.\end{tabular}}}}%
    \put(0,0){\includegraphics[width=\unitlength,page=6]{spacetime_constructed_oscillation_false.pdf}}%
    \put(0.24030345,0.00629533){\color[rgb]{0,0,0}\transparent{0.50160801}\makebox(0,0)[lt]{\lineheight{1.25}\smash{\begin{tabular}[t]{l}$\Sigma$\end{tabular}}}}%
    \put(0.62609886,0.20161446){\color[rgb]{0,0,0}\makebox(0,0)[lt]{\lineheight{1.25}\smash{\begin{tabular}[t]{l}$\big\Uparrow$	\\\end{tabular}}}}%
    \put(0,0){\includegraphics[width=\unitlength,page=7]{spacetime_constructed_oscillation_false.pdf}}%
  \end{picture}%
\endgroup%

%% file: rn_full.pdf_tex
\begingroup%
  \makeatletter%
  \providecommand\color[2][]{%
    \errmessage{(Inkscape) Color is used for the text in Inkscape, but the package 'color.sty' is not loaded}%
    \renewcommand\color[2][]{}%
  }%
  \providecommand\transparent[1]{%
    \errmessage{(Inkscape) Transparency is used (non-zero) for the text in Inkscape, but the package 'transparent.sty' is not loaded}%
    \renewcommand\transparent[1]{}%
  }%
  \providecommand\rotatebox[2]{#2}%
  \newcommand*\fsize{\dimexpr\f@size pt\relax}%
  \newcommand*\lineheight[1]{\fontsize{\fsize}{#1\fsize}\selectfont}%
  \ifx\svgwidth\undefined%
    \setlength{\unitlength}{228.07629182bp}%
    \ifx\svgscale\undefined%
      \relax%
    \else%
      \setlength{\unitlength}{\unitlength * \real{\svgscale}}%
    \fi%
  \else%
    \setlength{\unitlength}{\svgwidth}%
  \fi%
  \global\let\svgwidth\undefined%
  \global\let\svgscale\undefined%
  \makeatother%
  \begin{picture}(1,0.59269207)%
    \lineheight{1}%
    \setlength\tabcolsep{0pt}%
    \put(0.29096127,0.10103178){\color[rgb]{0,0,0}\makebox(0,0)[lt]{\lineheight{0}\smash{\begin{tabular}[t]{l}$\Sigma$\end{tabular}}}}%
    \put(0.77082787,0.25125806){\color[rgb]{0,0,0}\makebox(0,0)[lt]{\lineheight{0}\smash{\begin{tabular}[t]{l}$\mathcal{I}^+$\end{tabular}}}}%
    \put(0,0){\includegraphics[width=\unitlength,page=1]{rn_full.pdf}}%
    \put(0.09743412,0.25748379){\color[rgb]{0,0,0}\makebox(0,0)[lt]{\lineheight{0}\smash{\begin{tabular}[t]{l}$\mathcal{I}^+$\end{tabular}}}}%
    \put(0.55059231,0.49438177){\color[rgb]{0,0,0}\rotatebox{-45}{\makebox(0,0)[lt]{\lineheight{0}\smash{\begin{tabular}[t]{l}$\mathcal{CH}_{i^+}$\end{tabular}}}}}%
    \put(0.26816695,0.40848627){\color[rgb]{0,0,0}\rotatebox{45}{\makebox(0,0)[lt]{\lineheight{0}\smash{\begin{tabular}[t]{l}$\mathcal{CH}_{i^+}$\end{tabular}}}}}%
    \put(0.32489804,0.24612118){\color[rgb]{0,0,0}\rotatebox{-45}{\makebox(0,0)[lt]{\lineheight{0}\smash{\begin{tabular}[t]{l}$\mathcal{H}^+$\end{tabular}}}}}%
    \put(0.50536985,0.17474905){\color[rgb]{0,0,0}\rotatebox{45}{\makebox(0,0)[lt]{\lineheight{0}\smash{\begin{tabular}[t]{l}$\mathcal{H}^+$\end{tabular}}}}}%
    \put(0,0){\includegraphics[width=\unitlength,page=2]{rn_full.pdf}}%
    \put(0.70759959,0.32592237){\color[rgb]{0,0,0}\makebox(0,0)[lt]{\lineheight{1.25}\smash{\begin{tabular}[t]{l}$i^+$\end{tabular}}}}%
    \put(0.16533192,0.32115923){\color[rgb]{0,0,0}\makebox(0,0)[lt]{\lineheight{1.25}\smash{\begin{tabular}[t]{l}$i^+$\end{tabular}}}}%
    \put(0,0){\includegraphics[width=\unitlength,page=3]{rn_full.pdf}}%
  \end{picture}%
\endgroup%

%% file: general_spacetime.pdf_tex
\begingroup%
  \makeatletter%
  \providecommand\color[2][]{%
    \errmessage{(Inkscape) Color is used for the text in Inkscape, but the package 'color.sty' is not loaded}%
    \renewcommand\color[2][]{}%
  }%
  \providecommand\transparent[1]{%
    \errmessage{(Inkscape) Transparency is used (non-zero) for the text in Inkscape, but the package 'transparent.sty' is not loaded}%
    \renewcommand\transparent[1]{}%
  }%
  \providecommand\rotatebox[2]{#2}%
  \newcommand*\fsize{\dimexpr\f@size pt\relax}%
  \newcommand*\lineheight[1]{\fontsize{\fsize}{#1\fsize}\selectfont}%
  \ifx\svgwidth\undefined%
    \setlength{\unitlength}{178.70429772bp}%
    \ifx\svgscale\undefined%
      \relax%
    \else%
      \setlength{\unitlength}{\unitlength * \real{\svgscale}}%
    \fi%
  \else%
    \setlength{\unitlength}{\svgwidth}%
  \fi%
  \global\let\svgwidth\undefined%
  \global\let\svgscale\undefined%
  \makeatother%
  \begin{picture}(1,0.84927124)%
    \lineheight{1}%
    \setlength\tabcolsep{0pt}%
    \put(0,0){\includegraphics[width=\unitlength,page=1]{general_spacetime.pdf}}%
    \put(0.4211149,0.74142853){\color[rgb]{0,0,0}\rotatebox{-45}{\makebox(0,0)[lt]{\lineheight{0}\smash{\begin{tabular}[t]{l}$\mathcal{CH}_{i^+}$\end{tabular}}}}}%
    \put(0.32497084,0.27913127){\color[rgb]{0,0,0}\rotatebox{45}{\makebox(0,0)[lt]{\lineheight{0}\smash{\begin{tabular}[t]{l}$\mathcal{H}^+$\end{tabular}}}}}%
    \put(0.57597379,0.56550473){\color[rgb]{0,0,0}\makebox(0,0)[lt]{\lineheight{1.25}\smash{\begin{tabular}[t]{l}$i^+$\end{tabular}}}}%
    \put(0,0){\includegraphics[width=\unitlength,page=2]{general_spacetime.pdf}}%
    \put(0.42177278,0.05311322){\color[rgb]{0,0,0}\makebox(0,0)[lt]{\lineheight{1.25}\smash{\begin{tabular}[t]{l}$\Sigma$\end{tabular}}}}%
    \put(0,0){\includegraphics[width=\unitlength,page=3]{general_spacetime.pdf}}%
    \put(0.72958769,0.43013845){\color[rgb]{0,0,0}\rotatebox{-45}{\makebox(0,0)[lt]{\lineheight{1.25}\smash{\begin{tabular}[t]{l}$\mathcal I^+$\end{tabular}}}}}%
    \put(0.05159869,0.49687344){\color[rgb]{0,0,0}\makebox(0,0)[lt]{\lineheight{1.25}\smash{\begin{tabular}[t]{l}$\Gamma$\end{tabular}}}}%
    \put(0,0){\includegraphics[width=\unitlength,page=4]{general_spacetime.pdf}}%
    \put(0.18036665,0.80680847){\color[rgb]{0,0,0}\makebox(0,0)[lt]{\lineheight{1.25}\smash{\begin{tabular}[t]{l}$\mathcal S$\end{tabular}}}}%
    \put(-0.00054528,0.77927653){\color[rgb]{0,0,0}\makebox(0,0)[lt]{\lineheight{1.25}\smash{\begin{tabular}[t]{l}$b_\Gamma$\end{tabular}}}}%
  \end{picture}%
\endgroup%

%% file: regions.pdf_tex
\begingroup%
  \makeatletter%
  \providecommand\color[2][]{%
    \errmessage{(Inkscape) Color is used for the text in Inkscape, but the package 'color.sty' is not loaded}%
    \renewcommand\color[2][]{}%
  }%
  \providecommand\transparent[1]{%
    \errmessage{(Inkscape) Transparency is used (non-zero) for the text in Inkscape, but the package 'transparent.sty' is not loaded}%
    \renewcommand\transparent[1]{}%
  }%
  \providecommand\rotatebox[2]{#2}%
  \newcommand*\fsize{\dimexpr\f@size pt\relax}%
  \newcommand*\lineheight[1]{\fontsize{\fsize}{#1\fsize}\selectfont}%
  \ifx\svgwidth\undefined%
    \setlength{\unitlength}{156.43465043bp}%
    \ifx\svgscale\undefined%
      \relax%
    \else%
      \setlength{\unitlength}{\unitlength * \real{\svgscale}}%
    \fi%
  \else%
    \setlength{\unitlength}{\svgwidth}%
  \fi%
  \global\let\svgwidth\undefined%
  \global\let\svgscale\undefined%
  \makeatother%
  \begin{picture}(1,0.78705698)%
    \lineheight{1}%
    \setlength\tabcolsep{0pt}%
    \put(0,0){\includegraphics[width=\unitlength,page=1]{regions.pdf}}%
    \put(0.62355446,0.75128215){\color[rgb]{0,0,0}\rotatebox{-45}{\makebox(0,0)[lt]{\lineheight{0}\smash{\begin{tabular}[t]{l}$\mathcal{CH}_{i^+}$\end{tabular}}}}}%
    \put(0,0){\includegraphics[width=\unitlength,page=2]{regions.pdf}}%
    \put(0.47459388,0.13917407){\color[rgb]{0,0,0}\rotatebox{45}{\makebox(0,0)[lt]{\lineheight{0}\smash{\begin{tabular}[t]{l}$\mathcal{H}^+$\end{tabular}}}}}%
    \put(0,0){\includegraphics[width=\unitlength,page=3]{regions.pdf}}%
    \put(0.00102928,0.21916659){\color[rgb]{0,0,0}\rotatebox{-45}{\makebox(0,0)[lt]{\lineheight{0}\smash{\begin{tabular}[t]{l}$\underline{C}_{in}$\end{tabular}}}}}%
    \put(0,0){\includegraphics[width=\unitlength,page=4]{regions.pdf}}%
    \put(0.2619647,0.19171163){\color[rgb]{0,0,0}\makebox(0,0)[lt]{\lineheight{1.25}\smash{\begin{tabular}[t]{l}$\mathcal R$ \end{tabular}}}}%
    \put(0.50403633,0.59923008){\color[rgb]{0,0,0}\makebox(0,0)[lt]{\lineheight{1.25}\smash{\begin{tabular}[t]{l}$\mathcal{LB}$ \end{tabular}}}}%
    \put(0.31305441,0.46642571){\color[rgb]{0,0,0}\makebox(0,0)[lt]{\lineheight{1.25}\smash{\begin{tabular}[t]{l}$\mathcal{EB}$ \end{tabular}}}}%
    \put(0.24698235,0.32775089){\color[rgb]{0,0,0}\makebox(0,0)[lt]{\lineheight{1.25}\smash{\begin{tabular}[t]{l}$\mathcal N$ \end{tabular}}}}%
    \put(0.82516129,0.52409941){\color[rgb]{0,0,0}\makebox(0,0)[lt]{\lineheight{1.25}\smash{\begin{tabular}[t]{l}$i^+$\end{tabular}}}}%
  \end{picture}%
\endgroup%